\pdfoutput=1
\RequirePackage{latexml}
\iflatexml
	\documentclass{article}
\else
	\documentclass[a4paper,11pt]{scrartcl}
\fi

\usepackage{amsmath, amssymb}

\usepackage[T1]{fontenc}
\usepackage{lmodern}

\usepackage[a4paper, total={16cm, 23cm}]{geometry}
\usepackage{microtype}
\usepackage{booktabs}
\usepackage{graphicx, xcolor}
\usepackage{array, enumerate}

\usepackage{wrapfig}
\usepackage{etoolbox}
\iflatexml\else
\makeatletter
\patchcmd\WF@putfigmaybe{\lower\intextsep}{}{}{\fail}%
\AddToHook{env/wrapfigure/begin}{\setlength{\intextsep}{0pt}}
\makeatother
\fi
\iflatexml\else
\usepackage{titletoc}
\usepackage{shorttoc}
\fi
\usepackage{marvosym} %

\usepackage[round]{natbib}

\usepackage{tikz}
\usetikzlibrary{positioning, graphs, matrix, backgrounds, graphs.standard, calc, decorations.pathreplacing,decorations.markings, fit}
\tikzset{vertex/.style={draw=black, fill=black, circle, inner sep=0.4mm}}
\tikzset{->-/.style={decoration={
			markings,
			mark=at position .7 with {\arrow{>}}},postaction={decorate}}}

\usepackage[
	colorlinks,
	citecolor=green!45!black,
	linkcolor=red!60!black,
	urlcolor=blue!90!black,
	pagebackref]{hyperref}
\hypersetup{breaklinks=true}
\usepackage[vlined, ruled, noend]{algorithm2e}

\SetCommentSty{mycommfont}

\usepackage{amsthm}

\usepackage[nameinlink]{cleveref}
\iflatexml\else
\usepackage{crossreftools}
\fi
\usepackage{subcaption}
\Crefname{algocf}{Algorithm}{Algorithms}

\newtheoremstyle{smallcapstheorem} %
{\topsep}                    %
{\topsep}                    %
{\itshape}                   %
{}                           %
{\sffamily}                   %
{.}                          %
{.5em}                       %
{}  %
\theoremstyle{smallcapstheorem}
\newtheorem{theorem}{Theorem}[section]
\newtheorem{lemma}[theorem]{Lemma}
\newtheorem{proposition}[theorem]{Proposition}
\newtheorem{claim}{Claim}
\newtheorem{corollary}[theorem]{Corollary}
\newtheorem{open}{Open Problem}
\newtheoremstyle{smallcapsdefinition} %
{\topsep}                    %
{\topsep}                    %
{}                   %
{}                           %
{\sffamily}                   %
{.}                          %
{.5em}                       %
{}  %
\theoremstyle{smallcapsdefinition}
\newtheorem{definition}{Definition}

\iflatexml

	\newtheorem{example}{Example}

	\newenvironment{examplebox}[3][0]
	{%
	  \begin{example}[#2]%
	  \label{ex:#3}%
	}
	{%
	  \end{example}%
	}
	
	\newenvironment{tikzbackground}{}{}

\else

	\newenvironment{tikzbackground}{%
		\begin{scope}[on background layer]%
	}{%
		\end{scope}%
	}

	\usepackage{tcolorbox}
	\tcbuselibrary{breakable,theorems,skins}

	\tcolorboxenvironment{proof}{
		blanker,breakable,left=3.5mm,
		top=2pt, bottom=1pt,
		parbox=false,
		before skip=8pt,after skip=10pt,
		borderline west={1mm}{0.3mm}{black!13!white}}
	
	\tcolorboxenvironment{definition}{
		blanker,
		interior engine=standard,
		breakable,
		left=3.5mm, right=7pt,
		top=6pt, bottom=6pt,
		parbox=false,
		before skip=8pt,after skip=10pt,
		colback=blue!6!white,
		borderline west={1mm}{0mm}{blue!30!white}}
	
	\tcolorboxenvironment{theorem}{
		blanker,breakable,left=3.5mm,
		top=2pt, bottom=1pt,
		parbox=false,
		before skip=8pt,after skip=10pt,
		borderline west={1mm}{0.3mm}{red!25!white}}
	
	\tcolorboxenvironment{proposition}{
		blanker,breakable,left=3.5mm,
		top=2pt, bottom=1pt,
		parbox=false,
		before skip=8pt,after skip=10pt,
		borderline west={1mm}{0.3mm}{red!25!white}}
	
	\tcolorboxenvironment{lemma}{
		blanker,breakable,left=3.5mm,
		top=2pt, bottom=1pt,
		parbox=false,
		before skip=8pt,after skip=10pt,
		borderline west={1mm}{0.3mm}{red!25!white}}
	
	\tcolorboxenvironment{corollary}{
		blanker,breakable,left=3.5mm,
		top=2pt, bottom=1pt,
		parbox=false,
		before skip=8pt,after skip=10pt,
		borderline west={1mm}{0.3mm}{red!25!white}}
	
	\tcolorboxenvironment{open}{
		blanker,breakable,left=3.5mm,
		top=2pt, bottom=1pt,
		parbox=false,
		before skip=8pt,after skip=10pt,
		borderline west={1mm}{0.3mm}{green!30!black!30!white}}
	
	\newtcbtheorem[crefname={example}{examples}]%
	{examplebox}{Example}{
		new/auto counter,
		colback=red!2!white,
		colbacktitle=red!8!white,
		colframe=red!80!black,
		coltitle=black, fonttitle=\sffamily,
		arc=0mm,
		enlarge top by=0.3mm, %
		enlarge bottom by=0.3mm, %
		toptitle=0.8mm,bottomtitle=0.35mm,left=2mm,right=2mm,
		boxrule=\heavyrulewidth,titlerule=\lightrulewidth,leftrule=0mm,rightrule=0mm,
		description font=\rmfamily\itshape,
		before upper={\parindent15pt\noindent}, %
		title=#1}{ex}
\fi

\iflatexml
	\newcommand{\defemph}[1]{\emph{#1}}
\else
	\newcommand{\defemph}[1]{\textcolor{blue!75!black}{\emph{#1}}}
\fi

\newcommand{\pref}{\succcurlyeq}

\renewcommand{\top}{\operatorname{top}}
\newcommand{\rank}{\operatorname{rank}}

\newcommand{\calO}{{\mathcal{O}}}
\newcommand{\Gammasp}{{\Gamma_\textrm{SP}}}
\newcommand{\Gammasc}{{\Gamma_\textrm{SC}}}

\newcommand{\problemname}[1]{\textup{\textsc{#1}}}

\renewcommand*{\le}{\leqslant}
\renewcommand*{\leq}{\leqslant}
\renewcommand*{\ge}{\geqslant}
\renewcommand*{\geq}{\geqslant}

\newcommand{\problem}[3]{\bigskip\begin{tabular}{lp{0.7\textwidth}}
		\toprule
		\multicolumn{2}{l}{\textsc{#1}} \\
		\midrule
		\textbf{Instance:} & #2 \\
		\textbf{Question:} & #3 \\
		\bottomrule
	\end{tabular}\bigskip}

\usepackage{fancyhdr}
\title{Preference Restrictions \\ in Computational Social Choice: \\[3pt] A Survey}
\author{
	\textsf{\:\:\:\textbf{Edith Elkind}\:\:\:} \\
	\textsf{\clap{Northwestern University}}
	\and 
	\textsf{\textbf{Martin Lackner}} \\
	\textsf{WU Wien}
	\and 
	\textsf{\textbf{Dominik Peters}} \\
	\textsf{CNRS}
}
\date{\textsf{Version v2: March 2025}}

\begin{document}

\maketitle
\thispagestyle{empty}

\hrule

\begin{abstract}
	Social choice becomes easier on restricted preference domains such as single-peaked, single-crossing, and Euclidean preferences.
	Many impossibility theorems disappear, the structure makes it easier to reason about preferences, and computational problems can be solved more efficiently.
	In this survey, we give a thorough overview of many classic and modern restricted preference domains and explore their properties and applications.
	We do this from the viewpoint of computational social choice, letting computational problems drive our interest, but we include a comprehensive discussion of the economics and social choice literatures as well.
	Particular focus areas of our survey include algorithms for recognizing whether preferences belong to a particular preference domain, and algorithms for winner determination of voting rules that are hard to compute if preferences are unrestricted.
\end{abstract}

\iflatexml\else

\hrule

\shorttoc{}{1}
\clearpage
\setcounter{tocdepth}{3}
\linespread{1.05}
\tableofcontents
\linespread{1.0}
\clearpage
\newpage
\fi

\section{Introduction}
\label{sec:intro}

Social choice studies aggregation of preferences with the aim of making group decisions. An 
important part of social choice is voting theory, which designs voting rules that identify 
high-quality societal decisions, and formulates desirable properties 
of such rules. Computational social choice builds on this theory by studying computational 
tasks that arise in voting contexts. Among these tasks, the most important one 
is winner determination, i.e., computing the output of a voting rule given voters' preferences. 
While for many well-known voting rules
winner determination is a straightforward matter of counting ballots and adding numbers, for other rules identifying the winner(s) 
is computationally challenging,
especially when the decision space is combinatorial in nature. Other computational tasks 
concern elicitation, processing and analysis of preference data.

Famously, voting theory is riddled with impossibility theorems. 
\citeauthor{arrow1950difficulty}'s \citeyearpar{arrow1950difficulty} impossibility theorem 
concerning rank aggregation is the best-known, but for the purposes of voting, a theorem of 
\citet{gibbard1973manipulation} and \citet{satterthwaite1975strategy} is perhaps more problematic. 
This theorem shows that when there are three or more alternatives, every non-trivial voting rule can 
be manipulated by strategic voters who misrepresent their preferences. The underlying problem 
was identified much earlier, by \citet{Cond85a}. He observed that even when individual voters 
are fully rational (in the sense of having transitive preferences), their collective judgment 
may be irrational and contain cyclic (paradoxical) preferences.

\begin{examplebox}
	{Condorcet cycles: Pairwise majority judgments can be cyclic.}
	{condorcet}
	\begin{minipage}{0.15\linewidth}
		\begin{tabular}{ccc}
			\toprule
			$v_1$ & $v_2$ & $v_3$ \\
			\midrule
			$a$ & $b$ & $c$ \\
			$b$ & $c$ & $a$ \\
			$c$ & $a$ & $b$ \\
			\bottomrule
		\end{tabular}
	\end{minipage}
	\hfill
	\begin{minipage}{0.8\linewidth}
		Three voters rank the alternatives $a$, $b$, and $c$ from the most-preferred to the least-preferred. A majority of voters prefers $a$ over $b$, a majority prefers $b$ over $c$, and a majority prefers $c$ over $a$. Thus, the majority judgment contains a cycle.
	\end{minipage}
\end{examplebox}

Condorcet cycles are the primary source of paradoxes in voting: in their absence,
there is a clear winning 
alternative and the majority judgment is well-behaved. Some empirical results suggest that Condorcet
cycles are rare in real elections \citep{niemi1970occurrence,gehrlein1983condorcet,feld1992bigbadcycle,radcliff1994collective,
van1998empirical}. Thus, it is natural to look for contexts in which we can be assured that 
paradoxical majorities are avoided and in which the impossibility theorems do not apply.

The first, and most famous, such context was identified by \citet{black1948rationale} (and 
independently discovered by \citealp{arr:b:sc}). Black considers the case where the 
alternatives to be ranked lie on a one-dimensional axis; this occurs, for example, 
when we vote over the 
value of some numerical quantity. In those cases, it is natural to assume that votes will be 
\emph{single-peaked}, i.e., voters prefer values that are close to their favorite value. When 
preferences follow this pattern, it can be shown that Condorcet cycles are impossible. In 
addition, there are voting rules that cannot be manipulated.

Following \citet{black1948rationale}, social choice theorists have identified other domains of 
preferences that lead to similar positive results. Among others, such domains include 
generalizations of single-peakedness to allow for tree-shaped alternative spaces, the 
single-crossing property which assumes that voters (rather than alternatives) form a 
one-dimensional spectrum, and value restricted domains, obtained by simply forbidding the 
Condorcet cycle from \Cref{ex:condorcet} to occur as a subprofile.

Starting with the seminal work of \citet{bartholdi1989voting}, computer scientists began to 
realize that important voting-related computational problems may be algorithmically challenging, in 
the sense of being NP-complete or even harder. For example, it is difficult to determine the 
winner in elections that use the voting rules proposed by \citet{Dodg76a} or  
\citet{kemeny1959mathematics}. In computational complexity theory, when faced with a formally 
intractable problem, a common approach is to look for ``islands of tractability''. These are 
restricted classes of inputs for which polynomial-time algorithms exist. In the context of 
voting, a natural place to look for such islands is to start with classes that had already 
been suggested by social choice theorists in their quest to circumvent impossibility theorems. 
This research program was initiated by \citet{walsh2007uncertainty} and turned out to be fruitful: 
in many cases restrictions like single-peakedness do make computation easier.

Interestingly, while purely social choice-theoretic issues (such as manipulability or majority cycles)
vanish as soon as we assume that voters' preferences belong to a suitable restricted domain, many of the
algorithms for voting-related problems require the knowledge of the respective structural relationship among
voters and alternatives (such as the order of candidates witnessing that the profile is single-peaked).
This means that, to make use of these algorithms, one also needs an efficient procedure to discover
whether a given preference profile has the required structural property and to find a respective witness.
Consequently, the problem of designing such procedures has received a considerable amount of attention, too, 
resulting in polynomial-time algorithms
for recognizing preferences that belong to several prominent restricted domains.
These recognition algorithms are also useful for gaining a deeper understanding 
of a preference data set, by identifying unexpected structure underlying the preferences. 
This can make the data set more interpretable.

\subsection*{Contents of This Survey}

In this survey, we provide an overview of research on preference restrictions. 
While we focus on contributions of computational social choice, we briefly discuss
the most important results from `pure' social choice theory as well. 
For many results, we have included proofs or detailed proof sketches, to provide a convenient 
reference. Throughout, we state open questions that provide important avenues for future research.

\begin{itemize}
\item \Cref{sec:def} begins the survey by defining major preference domains that have been studied. We point out useful equivalent definitions and logical relationships. We also discuss implications on the structure of the majority relation, and the existence of strategyproof voting rules.

\item \Cref{sec:recog} studies recognition algorithms that decide, for a given preference 
profile, whether it is a member of one of the preference domains defined in \Cref{sec:def}. 
This topic is a major focus of this survey, and in many cases we give full details of the best-known algorithms for these tasks. In several cases, the analysis of these algorithms further illuminates properties of the different domains; in particular we can sometimes characterize all possible witnesses for membership of a given profile in a given domain.

\item \Cref{sec:subprofiles} gives an overview of several papers that have found characterizations of preference domains in terms of forbidden subprofiles.

\item \Cref{sec:problems} investigates the winner determination problem under the assumption 
that voters' preferences are drawn from a restricted domain. This section starts by 
considering several famous single-winner voting rules that are known to be NP-hard
to evaluate for unrestricted preferences, and then moves on to computationally 
challenging multi-winner rules. Multi-winner rules elect a committee of candidates, 
and this class of rules has been intensely studied in recent years.

\item 
\Cref{sec:manip} focuses on the computational complexity of various 
problems concerning manipulation and control. This literature has found that some types of manipulative 
attacks on elections are computationally difficult to pull off (in the worst case), 
which might provide some protection against these attacks. However, subsequent work showed that many 
of those hardness results do not hold when preferences are structured. 

\item
In \Cref{sec:further-topics} we briefly discuss other contexts
in which restricted preference domains have been considered, including 
weak orders (and specifically dichotomous preferences) and preference elicitation. 
We also mention results on counting structured profiles and 
computing the probability of random profiles being structured.

\item \Cref{sec:conclusion} concludes the survey by listing several directions 
for future research.
\end{itemize}

\subsection*{Technical Contributions}

In writing this survey, we have included some new results that are not published elsewhere. 
We feel that this survey is an appropriate place to describe these results.

\begin{itemize}
\item In \Cref{sec:recog:sp}, we present a linear-time algorithm for recognizing 
whether a preference profile is single-peaked. While this algorithm is directly based on 
the algorithms proposed by \citet{doignon1994polynomial} and \citet{escoffier2008single}, 
we believe that our version is the simplest to implement, and our analysis of this algorithm 
is the clearest one available.

\item In \Cref{sec:recog:sc}, we describe an algorithm that recognizes whether a 
preference profile is single-crossing. This algorithm runs in time $O(mn \log m)$, and is  
faster than all previously known algorithms.

\item In \Cref{sec:def:multidim-sp}, we discuss multidimensional extensions of the single-peaked domain. We identify a new variant of an existing definition, and prove \Cref{prop:multisp-triviality}, which shows that all profiles with at most $2^{2^{d-1}}$ alternatives are $d$-dimensional single-peaked.
\end{itemize}

\subsection*{Feedback Welcome!}

This is a draft version. In assembling this survey, we will doubtlessly have 
omitted relevant work, and introduced errors. We would 
appreciate any feedback and suggestions you might have for us. 
Please email to mail@dominik-peters.de.
We hope you enjoy reading the survey!
 
\newpage
\section{Preliminaries and Notation}
\label{sec:prel}

We write $[n]$ to denote the set $\{1,\dots,n\}$.

\paragraph{Relations and Profiles}
Let $A$ be a finite set of $m$ \emph{alternatives} (or \emph{candidates}).
Our survey will deal with ways to analyze \emph{preferences} over this set. Formally,  
we will consider binary relations over $A$, i.e., subsets of $A\times A$. 
Given a binary relation $R$ and alternatives $x \in A$ and 
$y \in A$, we interpret $(x,y)\in R$ to mean that $x$ is favored over $y$ by $R$. 
We will often use symbols such as $\pref$ and $\succ$ to denote binary relations, 
and write $a\pref b$ instead of $(a, b)\in\, \pref$ in this case.

A \emph{weak order} is a binary relation $\pref$ over $A$ that is reflexive, 
complete and transitive. A weak order can be used to describe 
a preference relation: we have $x \pref y$ if $x$ is weakly preferred to $y$. 
It could be that both $x\pref y$ and $y\pref x$, in which case
$\pref$ is indifferent between $x$ and $y$. 
A weak order with no indifferences is a \emph{linear order}. 
In other words, a linear order is a weak order that, in addition, is 
antisymmetric, so that $x\pref y$ and $y\pref x$ imply $x = y$.
A linear order is a \emph{ranking} of the 
alternatives, and thus there are $|A|!$ many different linear orders on $A$.
In what follows, we will often use the terms `preferences', 
`(preference) ordering', `(preference) order' and `(preference) ranking'
interchangeably.

If $\pref$ is a weak order, then we use 
$\succ$ to refer to its \emph{strict part}, that is, $x \succ y$ if and only if 
$x \pref y$ but $y \not\pref x$. When 
$\pref$ is a linear order, then $\succ$ is its irreflexive part. 
Whenever we use a strict relation symbol such as 
$\succ$, $>$, or $\lhd$, we are referring to the strict part of an order, 
so that these relations are always irreflexive.

Given a weak order $\pref$ over $A$ and two disjoint sets $A_1,A_2\subseteq A$, 
we write $A_1\pref A_2$ if for all alternatives $a\in A_1$ and $b\in A_2$ it holds that $a\pref b$, 
i.e., every element of $A_1$ is weakly preferred to every element of $A_2$.
If $\pref$ is a binary relation, then the relation $\pref'$ is the \emph{reverse} of $\pref$ if 
$a\pref b \Leftrightarrow b\pref' a$.

Typically, we will be interested in \emph{collections} of orders, 
and notions of what it means for such a collection to be structured. 
Such collections are referred to as \emph{preference profiles}, or simply \emph{profiles}.

\begin{definition}
 A \defemph{profile} $P = (v_1,\dots,v_n)$ over $A$ is a list of linear orders over $A$.
 The elements of $N = \{1,\dots,n\}$ are called \defemph{voters}, and we associate voter $i\in N$ 
 with the order $v_i$, which we call the \defemph{vote} of voter $i$. For convenience, 
 we write $a \succ_i b$ whenever $(a,b)\in v_i$, i.e., when voter $i$ 
 strictly prefers alternative $a$ to alternative $b$.
\end{definition}

We will always write $m$ for the number of alternatives and $n$ for the number of voters.
We denote the \emph{rank} of alternative $a$ in the vote of voter $i$ by $\rank_i(a)$;
we have $\rank_i(a) = |\{b: b\succ_i a\}|+1$.
Further, we write $\top(v_i)$ for the alternative that is ranked first in the linear order 
$v_i$, i.e., $\rank_i(\top(v_i)) = 1$.

\paragraph{Majority Relation and Condorcet Winners}
When there are only two alternatives $x$ and $y$, then following May's seminal analysis \citep{mayAxiomatic1952}, 
we should compare them by checking whether there is a majority of voters 
who prefer $x$ to $y$, in which case $x$ beats $y$ in a majority comparison. 
More generally, we can analyze a given profile in a pairwise fashion, considering pairs 
$\{x,y\}\subseteq A$ of alternatives separately. 
This approach gives rise to the idea of the majority relation.
Formally, given a profile $P$ over $A$, we define its \emph{majority relation} 
$\pref_{\text{maj}}$ to be a relation over $A$ such that
\[ 
a \pref_{\text{maj}} b \iff |i \in N : a \succ_i b| \ge |i \in N : b \succ_i a|. 
\]
We write $a \succ_{\text{maj}} b$ if $a \pref_{\text{maj}} b$ and not $b \pref_{\text{maj}} a$. 
Thus, $a \succ_{\text{maj}} b$ if and only if a strict majority of voters has $a \succ b$ 
in the profile $P$.
Alternative $a$ is a \emph{weak Condorcet winner} if $a \pref_{\text{maj}} b$ 
for all $b\in A\setminus\{a\}$ and a \emph{strong Condorcet winner} if $a \succ_{\text{maj}} b$ 
for all $b\in A\setminus\{a\}$.
As is well known, the relation $\succ_{\text{maj}}$ need not be transitive, a feature that leads to a lot of trouble. If the majority relation is not transitive, we say that $\pref_{\text{maj}}$ contains a \emph{Condorcet cycle}, as in the following example.

\begin{examplebox}
	{Condorcet winners and Condorcet cycles.}
	{cond-winners-and-cycles}
	\begin{minipage}{0.30\linewidth}
	\begin{tabular}{cccccc}
	\toprule
	$v_1$ & $v_2$ & $v_3$ & $v_4$ & $v_5$ & $v_6$ \\
	\midrule
	$a$ & $a$ & $e$ & $b$ & $d$ & $b$ \\
	$c$ & $b$ & $c$ & $a$ & $e$ & $e$ \\
	$d$ & $d$ & $d$ & $c$ & $c$ & $c$ \\
	$e$ & $e$ & $a$ & $d$ & $b$ & $d$ \\		
	$b$ & $c$ & $b$ & $e$ & $a$ & $a$ \\
	\bottomrule
	\end{tabular}
	\end{minipage}
	\hfill
	\begin{minipage}{0.67\linewidth}
	In profile $P=(v_1,\dots,v_6)$ there are two weak Condorcet winners, namely, 
	alternatives $a$ and $b$, as it holds that $\{a,b\} \succ_{\text{maj}} \{c,d,e\}$, 
	$a \pref_{\text{maj}} b$, and $b \pref_{\text{maj}} a$. The majority relation is, 
	however, not transitive: it holds that 
	$c \succ_{\text{maj}} d \succ_{\text{maj}} e  \succ_{\text{maj}} c$, a Condorcet cycle. 
	Also, as there are two weak Condorcet winners, it follows that there is no strong 
	Condorcet winner.
	\end{minipage}
\end{examplebox}

\paragraph{Alternative and Voter Deletion}
Suppose that $A' \subseteq A$ and let $v$ be a linear order. Then $v$ can be restricted to $A'$ 
in a natural way: define $v|_{A'} := v \cap (A' \times A')$, i.e., intuitively, we remove all 
candidates in $A\setminus A'$. Similarly, we can restrict a profile $P$ by restricting each 
vote in it: $P|_{A'} := (v_1|_{A'}, \dots, v_n|_{A'})$. We say that $P|_{A'}$ is obtained from 
$P$ by \emph{alternative deletion}. Similarly, if $N' \subseteq N$, the profile 
$P' = (v_i : i \in N')$ is said to be obtained from $P$ by \emph{voter deletion}.

\paragraph{Algorithms and Computational Complexity}

We assume that the reader is familiar with basic concepts of algorithm analysis 
(e.g., big-$\calO$ notation and runtime analysis) and computational complexity 
(e.g., NP-hardness and reductions).
 
\newpage
\section{Domain Restrictions}
\label{sec:def}

\begin{figure}[!ht]
	\begin{center}
		\begin{tikzpicture}[scale=.7,shift={(7,2)}]
			\node (cond) at (3,5) {weak Condorcet winner exists};
			\node (sct) at (-3,2) {single-crossing on trees};
			\node (vr) at (0,3.5) {value-restricted};
			\node (spt) at (6,3.5) {single-peaked on trees};
			\node (sc) at (-3,.5) {single-crossing};
			\node (sp) at (2,0.5) {single-peaked};
			\node (eucl) at (0,-1) {1-Euclidean};
			\draw (cond) -- (vr) -- (sp) %
			(sp) -- (eucl)
			(sc) -- (eucl)
			(spt) -- (sp)
			(sct) -- (sc)
			(sct) -- (vr)
			(spt) -- (cond);
		\end{tikzpicture}
	\end{center}
	\caption{Some domain restrictions that guarantee a weak Condorcet winner, and inclusion relationships.}\label{fig:domains}
\end{figure}

In this section, we will present the definitions of several prominent domain restrictions, and 
discuss some of their basic properties. Within the economics and social choice literature, 
domain restrictions were found to be interesting primarily because of their implications for 
the structure of the majority relation. For example, the majority relation of a single-peaked 
profile cannot have a Condorcet cycle. Also, some domain restrictions admit 
well-behaved voting rules: e.g., for the single-peaked domain we have the median
voter rule, which is strategyproof.

As an overview, in \Cref{fig:domains} we display the relationship among 
domain restrictions that guarantee a weak Condorcet winner.
The 1-Euclidean domain is the most restrictive domain that we consider here, 
whereas the value-restricted and single-peaked on trees domains are very general.
The figure does not mention the multidimensional single-peaked domain, 
the $d$-Euclidean domain for $d\ge 2$, 
and the domain of profiles that are single-peaked on a circle, 
because these domains do not rule out the presence of Condorcet cycles.

Before we start our discussion of specific domain restrictions, let us take a meta-level 
view, and define what we formally mean by a \emph{domain restriction} (which we also call a \emph{restricted domain} or a \emph{preference domain}). Most abstractly, a 
domain restriction is a property of a profile. Typically, this property imposes some 
structure on this profile. For most of this survey, we will consider notions of structure only 
in the context of profiles of \emph{linear} orders. Thus, except for some remarks in 
\Cref{sec:further-topics}, we will not consider preferences that have indifferences. 
Hence, our formal definition of a domain restriction is as follows.

\begin{definition}
	A \defemph{domain restriction} is a set of profiles of linear orders.
\end{definition}

For example, the domain of single-peaked profiles is the set of profiles $P$ for which 
there exists an axis $\lhd$ such that $P$ is single-peaked on $\lhd$ (we will discuss the meaning of single-peakedness just below in \Cref{sec:def:sp}).

Our definition of a domain restriction is a departure from much of the discussion of domain 
restrictions in the social choice literature. There, a domain typically refers to a set of allowed
\emph{votes}, rather than a set of profiles. Under this view, a profile is 
structured if each voter reports a linear order that belongs to the respective domain. 
Our notion of a domain is 
more general; in particular, we allow domain restrictions that are not Cartesian products of 
restrictions on individual voters' preferences. 
To appreciate this distinction, observe that
the domain of profiles that are single-peaked on some \emph{fixed} axis $\lhd$ 
is a Cartesian product, whereas the domain of all single-peaked profiles is not.

An important property that is satisfied by many---but not all---domain restrictions 
is closure under deleting voters and alternatives. A domain with this property 
is called \emph{hereditary}; for such domains, taking a subprofile of a structured profile 
yields another structured profile.

\begin{definition}\label{def:hereditary}
A domain restriction $X$ is \defemph{hereditary} if for every profile $P$ in $X$ 
and every profile $P'$ that can be obtained from $P$ by voter and alternative deletion 
it holds that $P'$ belongs to~$X$. 
\end{definition}

As we will see, the domain of single-peaked profiles is hereditary, and the same is true for 
Euclidean or single-crossing profiles. A counter-example is the domain of profiles that admit a 
weak Condorcet winner.

\begin{examplebox}
	{The property of having a weak Condorcet winner is not hereditary.}
	{condorcet-not-hereditary}
	\begin{minipage}{0.20\linewidth}
	\begin{tabular}{cccc}
	\toprule
	$v_1$ & $v_2$ & $v_3$ & $v_4$ \\
	\midrule
	$a$ & $a$ & $b$ & $c$ \\
	$b$ & $b$ & $c$ & $a$ \\
	$c$ & $c$ & $a$ & $b$ \\
	\bottomrule
	\end{tabular}
	\end{minipage}
	\hfill
	\begin{minipage}{0.77\linewidth}
	Alternative $a$ is a weak Condorcet winner. However, if we remove $v_1$, 
        we obtain the Condorcet cycle 
        $a \succ_{\text{maj}} b \succ_{\text{maj}} c \succ_{\text{maj}} a$, 
        and consequently this reduced profile does not have a weak Condorcet winner.
	\end{minipage}
\end{examplebox}

Hereditary domain restrictions are of particular importance for characterizations 
via forbidden subprofiles, as we will see in \Cref{sec:subprofiles}.

\subsection{Condorcet Winners}
\label{sec:def:condorcet-winners}
The concept of a Condorcet winner dates back to the 18th century;  
it was proposed by the Marquis de Condorcet, who argued that a voting rule should 
elect a Condorcet winner whenever it exists. However, it is well-known that Condorcet 
winners may fail to exist, since there are profiles whose majority 
relation contains cycles. Thus, this concept defines an interesting domain restriction: 
let $\mathcal D_{\text{Condorcet}}$ denote the set of profiles 
that have a (strong) Condorcet winner. For an odd number of voters, 
all the other domain restrictions 
that we consider in this survey (such as, e.g., single-peaked profiles) 
are \emph{subdomains} of $\mathcal D_{\text{Condorcet}}$, 
because each of these restrictions guarantees the existence of a strong Condorcet 
winner as long as the number of voters is odd. However, while 
$\mathcal D_{\text{Condorcet}}$ is a large domain compared to the other 
domains considered here, it is somewhat less interesting to us, because it does not imply much 
\emph{structure} in individuals' preferences, and because it does not imply tractability 
results for the winner determination problems of many important voting rules that we consider 
in \Cref{sec:problems}. Nevertheless, in this section, we will briefly discuss some 
properties of $\mathcal D_{\text{Condorcet}}$ 
since they have implications for the subdomains we will study in more detail later on.

There exists an obvious voting rule for the profiles in $\mathcal D_{\text{Condorcet}}$: 
the \emph{Condorcet rule}, which elects the Condorcet winner of the profile. Formally, a 
\emph{voting rule} is a function $f : \mathcal D \to A$ that maps each profile from some domain 
$\mathcal D$ to a unique winning alternative. We write 
$f_{\text{Condorcet}} : \mathcal D_{\text{Condorcet}} \to A$ for the Condorcet rule,
where $f_{\text{Condorcet}}(P)$ is the Condorcet winner of $P \in \mathcal D_{\text{Condorcet}}$. 
An important, but 
difficult to achieve, property of a voting rule is \emph{strategyproofness}, which requires 
that voters cannot manipulate the voting rule by misrepresenting their preferences. 
Formally, we say that a voting rule $f$ is \emph{manipulable} if there exist profiles 
$P = (\succ_1, \dots, \succ_i, \dots, \succ_n)$ and 
$P' = (\succ_1, \dots, \succ_i', \dots, \succ_n)$, both members of the domain, that only differ in the preferences of the 
$i$th voter, and such that $f(P') \succ_i f(P)$. 
Thus, if $\succ_i$ is the truthful preference order of voter $i$, this 
voter can obtain a strictly preferred outcome by submitting the non-truthful preference 
order $\succ_i'$ instead. If a voting rule is not manipulable, then it is 
\emph{strategyproof}.

While \citet{gibbard1973manipulation} and \citet{satterthwaite1975strategy} showed that any 
(surjective) voting rule defined on the domain of \emph{all} preference profiles must be 
dictatorial, there are restricted domains where strategyproofness can be achieved. The domain 
of profiles admitting a Condorcet winner is an example; in fact, one can check that the 
Condorcet rule is strategyproof.

\begin{proposition}
	\label{prop:condorcet-sp}
	The Condorcet rule $f_{\textup{Condorcet}}$ is strategyproof.
\end{proposition}
\begin{proof}
	For a contradiction, suppose there are profiles 
	$P = (\succ_1, \dots, \succ_i, \dots, \succ_n)$ and 
	$P' = (\succ_1, \dots$, $ \succ_i'$, $\dots$, $\succ_n)$ such that
	\[ f_{\text{Condorcet}}(P) = a, \quad f_{\text{Condorcet}}(P') = b,\quad \text{ and } b\succ_i a.  \]
	Since $a$ is the Condorcet winner at $P$, 
	there is a strict majority $N' \subseteq N$ of voters who prefer $a$ to $b$ in $P$. 
	As $b\succ_i a$, we have $i \not\in N'$, so all voters in $N'$ also prefer $a$ to $b$ in $P'$, 
	forming a strict majority. This is a contradiction with $b$ being the Condorcet winner at $P'$.
\end{proof}

A similar argument also establishes that $f_{\text{Condorcet}}$ is resistant to manipulation by  
\emph{groups of voters} \citep[Lemma~10.3]{mou:b:axioms}.

It turns out that $f_{\text{Condorcet}}$ is essentially the only voting rule defined on 
$\mathcal D_{\text{Condorcet}}$ that is strategyproof, at least if the number $n$ of voters is 
odd.

\begin{theorem}[\citealp{CaKe02,CaKe16}]
	Suppose that the number of voters is odd. Let $f$ be a non-dictatorial and surjective 
	voting rule defined on $\mathcal D_{\textup{Condorcet}}$. Then $f$ is strategyproof 
	if and only if $f = f_{\textup{Condorcet}}$.
\end{theorem}

\citet{CaKe15} show an analog of their theorem for an even number of voters if one strengthens 
non-dictatorship to anonymity, and surjectivity to neutrality. 
\citet[Theorem~3.6]{peters2019thesis} proves a different version for an even number of voters 
using anonymity and Pareto optimality.

Notice that the Campbell--Kelly theorem may not hold for domains that are smaller than 
$\mathcal D_{\textup{Condorcet}}$. For example, as we will discuss below, there are many more 
strategyproof voting rules defined for single-peaked profiles only.

\subsection{Single-Peaked Preferences} 
\label{sec:def:sp}

Let us now turn to what is arguably the most famous domain restriction: the single-peaked domain.
Consider a situation where voters need to decide among different possible quantities
of a numerical measure: it might be a parliament deciding on a tax rate, a firm's board deciding
on the price for a new product, or housemates deciding on the optimal setting of the thermostat. 
As \citet{black1948rationale} noted, in such situations it is reasonable to expect that each 
decision maker first identifies her optimum value of the measure under consideration, and that 
she is less and less happy the further the chosen quantity deviates from her optimum. For 
example, it would be surprising if a politician whose most preferred income tax rate is 20\% 
had 70\% as his second-favorite tax rate, and 40\% as his third-favorite tax rate. Rather, we 
expect his preference curve to have a single peak like the solid lines in \Cref{fig:sp-ex}. 
Another situation where preferences can be expected to have this shape is in a political 
election where the candidates can be placed on a left-to-right spectrum, 
with left-wing voters preferring left-wing candidates, 
and right-wing voters preferring right-wing candidates.

\begin{figure}
	\centering
	\begin{minipage}{0.2\textwidth}
		\centering
		\begin{tabular}{>{\color{green!50!black}}c>{\color{blue}}cc}
		\toprule
		$v_1$ & $v_2$ & $v_3$ \\
		\midrule
		$b$ & $d$ & $c$  \\
		$a$ & $e$ & $e$  \\
		$c$ & $f$ & $d$  \\
		$d$ & $c$ & $b$  \\
		$e$ & $b$ & $f$  \\
		$f$ & $g$ & $a$  \\
		$g$ & $a$ & $g$  \\						
		\bottomrule
		\end{tabular}
	\end{minipage}
	\quad
	\begin{minipage}{0.55\textwidth}
		\centering
		\begin{tikzpicture}[yscale=0.65,xscale=0.95]
		
		\def\xmin{1}
		\def\xmax{7}
		\def\ymin{0}
		\def\ymax{7}
		
		\draw[step=1cm,black!20,very thin] (\xmin,\ymin) grid (\xmax,\ymax);

		\draw[->] (\xmin -0.5,\ymin) -- (\xmax+0.5,\ymin) node[right] {};
		\foreach \x/\xtext in {1/a, 2/b, 3/c, 4/d, 5/e, 6/f, 7/g}
		\draw[shift={(\x,\ymin)}] (0pt,2pt) -- (0pt,-2pt) node[below] {$\strut\xtext$};
		\foreach \x/\xtext in {1, 2,3,4,5,6}
		\node[below] at (\x+0.5,\ymin) {$\strut\lhd$};  
		
		\foreach \x/\y in {4/7,5/6,3/4, 6/5, 2/3, 7/2, 1/1}
		\node[fill=blue, circle, inner sep=0.6mm] at (\x,\y) {};
		
		\draw[thick,blue] (1,1)--(2,3)--(3,4)--(4,7) -- (5,6)--(6,5)--(7,2);
		
		\foreach \x/\y in {4/4,5/3, 6/2, 3/5, 2/7, 7/1, 1/6}
		\node[fill=green!50!black, circle, inner sep=0.6mm] at (\x,\y) {};
		
		\draw[thick,green!50!black] (1,6)--(2,7)--(3,5)--(4,4) -- (5,3)--(6,2)--(7,1);  
		
		\foreach \x/\y in {1/2,2/4,3/7, 4/5, 5/6, 6/3, 7/1}
		\node at (\x,\y) {$\star$};
		
		\draw[dashed,black] (1,2)--(2,4)--(3,7)--(4,5) -- (5,6)--(6,3)--(7,1);
		
		\end{tikzpicture}
	\end{minipage}
	\caption{Votes $v_1$ and $v_2$, shown as solid lines, are single-peaked with respect to 
		$\lhd$. The vote $v_3$, depicted as a dashed line, is not single-peaked with respect 
		to $\lhd$. While a profile consisting of $v_2$ and $v_3$ is single-peaked 
		($d$ and $e$ have to be flipped on $\lhd$), there is no ordering of the candidates 
		for which the profile $(v_1,v_3)$ is single-peaked.
	}
	\label{fig:sp-ex}
\end{figure}

\paragraph{Definition}
As we will see, single-peaked profiles have many desirable properties. Let us start by giving a 
formal definition. See \Cref{fig:sp-ex} for an example.
\begin{definition}
Let $P$ be a profile over $A$ and let $\lhd$ be a linear order over $A$.
A linear order $v_i$ over $A$ is \defemph{single-peaked on $\lhd$} if 
for every pair of alternatives $a,b\in A$ with $\top(v_i) \lhd b\lhd a$ or 
$a\lhd b\lhd \top(v_i)$ we have $b\succ_i a$.
A profile $P$ over $A$ is \defemph{single-peaked with respect to $\lhd$} if every vote in $P$ 
is single-peaked on $\lhd$, and it is \defemph{single-peaked} if 
there exists a linear order $\lhd$ over $A$ such that $P$ is single-peaked with respect to $\lhd$.
\label{def:sp}
\end{definition}

In what follows we will often refer to a linear order $\lhd$ over $A$ that witnesses the 
single-peaked property as an \emph{axis}.

This definition requires that, as we move away from a voter's most-preferred alternative either 
to the left or to the right, the voter becomes less and less enthusiastic. Thus, plotting 
a ``preference curve'', which depicts the voter's preference intensity, as in \Cref{fig:sp-ex}, 
yields a single-peaked shape, giving rise to this domain restriction's name \citep{black1948rationale}.

\begin{wrapfigure}[6]{r}{0.22\linewidth}
	\centering
	\begin{tikzpicture}[yscale=0.65,xscale=0.95]
	\def\xmin{1}
	\def\xmax{3}
	\def\ymin{0}
	\def\ymax{2}
	
	\draw[step=1cm,black!20,very thin] (\xmin,\ymin) grid (\xmax,\ymax);

	\draw[->] (\xmin -0.5,\ymin) -- (\xmax+0.5,\ymin) node[right] {};
	\foreach \x/\xtext in {1/a, 2/b, 3/c}
	\draw[shift={(\x,\ymin)}] (0pt,2pt) -- (0pt,-2pt) node[below] {$\strut\xtext$};
	\foreach \x/\xtext in {1,2}
	\node[below] at (\x+0.5,\ymin) {$\strut\lhd$};  
	
	\foreach \x/\y in {3/2, 2/0, 1/1}
	\node[fill=blue, circle, inner sep=0.6mm] at (\x,\y) {};
	
	\draw[thick,blue] (1,1)--(2,0)--(3,2);
	\end{tikzpicture}
	\vspace{-24pt}
	\caption{A valley.}
	\label{fig:valley}
\end{wrapfigure}
To reason about this domain restriction, it is useful to have several alternative ways 
of thinking about it. For example, a feature of the preference curves shown in 
\Cref{fig:sp-ex} is that they do not contain any valleys, as shown 
in \Cref{fig:valley}
on the right.
Formally, given an axis $\lhd$, a vote $\succ_i$ has a \emph{valley} 
if there is a triple $a \lhd b \lhd c$ of alternatives such that $a \succ_i b$ and $c\succ_i b$. 
\Cref{prop:sp-equiv} below shows that a linear order has no valleys with respect to $\lhd$
if and only if it is single-peaked on $\lhd$. It also shows that it suffices to check that there are no `local' valleys where the alternatives $a$, $b$, $c$ are next to each other on the axis $\lhd$. This condition is
used by some of the dynamic programming algorithms 
presented later in this survey (see, e.g., \Cref{thm:alt-del-fixedorder}).

One can also view single-peakedness as a \emph{convexity} condition. 
Given an axis $\lhd$, a subset $A'\subseteq A$ of alternatives is \emph{convex} if and only if 
it is an interval of the axis $\lhd$, that is, 
for any triple of alternatives $a, b, c$ with $a\lhd b \lhd c$ and $a,c\in A'$ 
we have $b\in A'$. By looking at \Cref{fig:sp-ex}, we can see that if we take any 
\emph{prefix} of the votes $v_1$ and $v_2$, then this prefix is an interval of $\lhd$. For 
example, the four most-preferred alternatives in $v_2$ are $\{d,e,f,c\}$, and this set forms an 
interval of $\lhd$. Again, \Cref{prop:sp-equiv} establishes that
all prefixes of a vote are intervals of $\lhd$ 
if and only if this vote is single-peaked on $\lhd$.

\begin{proposition}\label{prop:sp-equiv}
	Let $\succ$ be a vote over $A$, and
        let $\lhd$ be a linear order over $A$. The following conditions are equivalent:
	\begin{enumerate}[(1)]
		\item The vote $\succ$ is single-peaked on $\lhd$.
		\item For all $a,b,c\in A$ such that $a \lhd b \lhd c$ and these alternatives form an interval of $\lhd$, 
        we do not have both $a \succ b$ and $c \succ b$, that is, there is no `local' valley
		in $\succ$. 
		\item For all $a,b,c\in A$ such that $a \lhd b \lhd c$, 
		we do not have both $a \succ b$ and $c \succ b$, that is, there is no valley
		in $\succ$. 
		\item For each $c \in A$, the set $\{a\in A : a \pref c\}$ is an interval on $\lhd$.
	\end{enumerate}
\end{proposition}

\begin{proof}
$(1)\rightarrow(2)$: Assume $\succ$ is single-peaked on $\lhd$. We will actually prove the stronger condition $(3)$. Assume towards a contradiction that there exists
a triple $a,b,c\in A$ with $a \lhd b \lhd c$, $a \succ b$ and $c \succ b$.
Let $x$ be the top-ranked candidate in $\succ$; note that $x\neq b$.
If $x \lhd b$ then $x \lhd b\lhd c$ and $c\succ b$; 
if $b \lhd x$ then $a \lhd b\lhd x$ and $a\succ b$.
In either case we obtain a contradiction with the assumption 
that $\succ$ is single-peaked on $\lhd$.

$(2)\rightarrow(3)$: We show that if $\succ$ has a valley, then it also has a local valley.
Define the \emph{width} of a valley $a\lhd b\lhd c$ as the number of alternatives
that appear between $a$ and $c$ on $\lhd$; a valley is local if and only if its width
is $1$. Now, suppose $a, b, c$ form a valley of width $w>1$ in $\succ$.
Assume without loss of generality that $a$ and $b$ are non-adjacent on $\lhd$, 
i.e., $a\lhd a' \lhd b\lhd c$ for some $b'\in A$. But then either $a'\succ b$, 
in which case $a', b, c$ form a valley of width at most $w-1$, 
or $b\succ a'$, in which case $a, a', b$ form a valley of width at most $w-1$; 
proceeding in this fashion, we arrive at a local valley. 

$(3)\rightarrow(4)$: Assume towards a contradiction that there exists 
an alternative $c$ such that $S=\{a\in A : a \pref c\}$ is not an interval of $\lhd$, 
i.e., there exist alternatives $d,e,f\in A$ such that $d\in S$, $e\notin S$, $f\in S$, 
and $d \lhd e\lhd f$.
Consequently, $d\pref c \succ e$ and $f\pref c \succ e$, which contradicts Condition (3).

$(4)\rightarrow(1)$: Let $a,b\in A$ with $x \lhd b\lhd a$ or $a\lhd b\lhd x$,
where $x$ is the top-ranked candidate in $\succ$.
By assumption, $S=\{c\in A : c \pref a\}$ is an interval of $\lhd$.
By construction of $S$ we have $a\in S$, $x\in S$ and hence $b\in S$.
Thus, $b\succ a$, and  therefore $\succ$ is single-peaked on $\lhd$.
\end{proof}

\Cref{def:sp} follows the spirit of the original definition due to \citet{black1948rationale}. 
Condition (3) is essentially how \citet{arr:b:sc} formulated single-peakedness. Condition (4), 
which is stated in terms of the connectedness of prefixes 
(also known as \emph{upper contour sets}) is often an 
elegant way of reasoning about single-peaked preferences, and it generalizes well to other 
settings (see \Cref{sec:def:spt} and \Cref{sec:further-topics}).

Consider a profile $P$ over $A$ that is single-peaked with respect to an axis $\lhd$,
and a triple of distinct alternatives $A'=\{a, b, c\}\subseteq A$.
The no-valley property implies that there
exist an alternative $x\in A'$ such that no voter ranks $x$ below all other
alternatives in $A'$: indeed, if $a\lhd b\lhd c$ then a vote $\succ$ 
with $a\succ b$, $c\succ b$ would have a valley at $b$. In particular, 
there can be at most two different alternatives in $A$ that are ranked last by some voter. This is also easy to conclude from the convexity 
property: only the endpoints of the axis can be ranked last. This observation
can often be used to show that a given profile is not single-peaked.

\paragraph{Majority Relation}
The most famous property of single-peaked profiles is that they admit a Condorcet winner, and 
that the Condorcet winner is the top alternative of the \emph{median voter}, 
i.e., the voter in the middle when one orders the voters according to the position 
of their top-ranked alternatives on $\lhd$. 
This property is often used by political scientists and in public choice to reason 
about decision over quantities such as tax rates.

\begin{proposition}[Median Voter Theorem]
	\label{prop:median-voter-theorem}
	Every single-peaked profile has a weak Condorcet winner.
\end{proposition}
\begin{proof}
	In what follows, given an axis $\lhd$ and two votes $v$, $v'$, we write 
        $\top(v)\unlhd \top(v')$ to denote that either $\top(v)\lhd \top(v')$ or $\top(v)=\top(v')$.
	First, we reorder the votes in $P = (v_1,\dots,v_n)$ so that 
	$\top(v_1) \unlhd \top(v_2) \unlhd \cdots \unlhd \top(v_n)$, 
	set $\ell := \lceil \frac{n+1}{2}\rceil$, and let 
	$c := \top(v_\ell)$ be the top-ranked alternative of (one of) 
	the median voter(s). We claim that $c$ is a weak Condorcet winner. 
	To see this, consider any other alternative $a\in A$.
	\begin{itemize}
		\item If $a \lhd c$, then the voters $\ell, \ell + 1, \dots, n$ all prefer $c$ to $a$, and these voters form a weak majority.
		\item If $c \lhd a$, then the voters $1,\dots,\ell - 1, \ell$ all prefer $c$ to $a$, and these voters form a weak majority.
	\end{itemize}
	Hence, $c$ is a weak Condorcet winner, as required.
\end{proof}
If the number of voters in $P$ is odd, then there is a unique Condorcet winner. 
If the number of voters in $P$ is even, then the set of weak Condorcet winners 
is the interval of $\lhd$ with endpoints given by the top choices of voters 
$\lfloor \frac{n+1}{2}\rfloor$ and $\lceil \frac{n+1}{2}\rceil$.

\begin{examplebox}
	{Median voters}
	{sp-median}
	Consider a single-peaked profile with the axis $a\lhd b\lhd c\lhd d\lhd e$.
	We have six voters with $\top(v_1)=\top(v_2)=a$, $\top(v_3)=b$, $\top(v_4)=\top(v_5)=d$, and $\top(v_6)=e$.
	Consequently, it holds that \[\top(v_1)\unlhd \top(v_2)\unlhd \top(v_3) \unlhd \top(v_4) \unlhd\top(v_5)\unlhd\top(v_6).\]
	This profile has three weak Condorcet winners: alternatives $b$, $c$, and $d$.
    Indeed, since $\lfloor \frac{n+1}{2}\rfloor=3$ and $\lceil \frac{n+1}{2}\rceil=4$, all alternatives in the interval $[\top(v_3), \top(v_4)]=\{b, c, d\}$ are weak Condorcet winners.
	
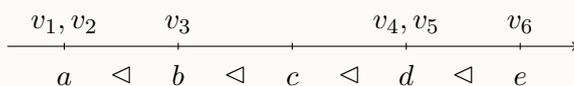
\begin{center}
 		\begin{tikzpicture}[yscale=0.65,xscale=1.5]
		
		\def\xmin{1}
		\def\xmax{5}
		\def\ymin{0}
		\def\ymax{5}

		\draw[->] (\xmin -0.5,\ymin) -- (\xmax+0.5,\ymin) node[right] {};
		\foreach \x/\xtext/\vtext in {1/a/{v_1,v_2}, 2/b/{v_3}, 3/c/{}, 4/d/{v_4,v_5}, 5/e/v_6}
		\draw[shift={(\x,\ymin)}] (0pt,2pt) -- (0pt,-2pt) node[below] {$\strut\xtext$} node[above] {$\strut\vtext$};
		\foreach \x/\xtext in {1, 2,3,4}
		\node[below] at (\x+0.5,\ymin) {$\strut\lhd$};  

		\end{tikzpicture}
	\end{center}
\end{examplebox}

As we will now see, the single-peaked property constrains the majority relation $\succ_\text{maj}$ 
even further: it implies that $\succ_\text{maj}$ is transitive as long as the number of
voters is odd. We give two different 
proofs, one of which works by repeated application of the Median Voter Theorem.

\begin{corollary}
	\label{cor:sp-transitive}
	If $P$ is a single-peaked profile with an odd number of voters then its 
	majority relation is transitive.
\end{corollary}
\begin{proof}[Proof by induction]
	By induction on $m$, the number of alternatives. The result is obvious for $m = 1$. 
	If $m > 1$, let $c$ be the (unique) Condorcet winner of $P$. Let $P' = P_{A\setminus \{c\}}$. 
	Since $P$ is single-peaked, $P'$ is single-peaked as well. By the inductive hypothesis, 
        the majority relation of $P'$ is transitive; denote it by $\succ'$. Define $\succ$ by 
	setting $a \succ b$ if and only if (i) $a \succ' b$ and $a,b\in A \setminus\{c\}$ or 
	(ii) $a=c$ and $b\in A\setminus\{c\}$. Then $\succ$ is the majority relation of 
	$P$, and it is transitive.
\end{proof}
\begin{proof}[Proof by contradiction]
	Suppose for a contradiction that the majority relation is not transitive, so that 
	$a \succ_{\text{maj}} b \succ_{\text{maj}} c \succ_{\text{maj}} a$ for some alternatives 
	$a,b,c$. Assume without loss of generality that $a\lhd b\lhd c$.
	Note that any two strict majorities intersect in at least one voter. 
	Thus, since $a \succ_{\text{maj}} b$ and $c \succ_{\text{maj}} a$, there exists a voter $i$ 
	with $c \succ_i a \succ_i b$, a contradiction with the no-valley property.
\end{proof}

\Cref{cor:sp-transitive} depends on the assumption that the number of voters is odd. 
In the following example with four voters, the relation $\pref_{\text{maj}}$ is not transitive. 
However, our second proof can be adjusted to show that its strict part $\succ_{\text{maj}}$ is always transitive. A relation whose strict part is transitive is called \emph{quasi-transitive}, and hence we can say that the majority relation of any single-peaked profile is quasi-transitive.

\begin{examplebox}
	{The weak majority relation of a single-peaked profile may not be transitive for even $n$}
	{sp-quasi-transitive}
	\begin{minipage}{0.20\linewidth}
		\begin{tabular}{cccc}
			\toprule
			$v_1$ & $v_2$ & $v_3$ & $v_4$ \\
			\midrule
			$b$ & $b$ & $c$ & $c$ \\
			$a$ & $a$ & $b$ & $b$ \\
			$c$ & $c$ & $a$ & $a$ \\
			\bottomrule
		\end{tabular}
	\end{minipage}
	\hfill
	\begin{minipage}{0.74\linewidth}
		This profile is single-peaked with respect to $a\lhd b \lhd c$.
		The majority relation~ $\pref_{\text{maj}}$ is not transitive as 
	        it satisfies $a \pref_{\text{maj}} c$ and $c \pref_{\text{maj}} b$, 
		yet all voters prefer $b$ to $a$. 
		Still, alternatives $b$ and $c$ are weak Condorcet winners, 
		and the strict majority relation $\succ_{\text{maj}}$ is transitive.
	\end{minipage}
\end{examplebox}

A domain restriction is said to have the \emph{representative voter property} 
\citep{rothstein1991representative} if for every structured profile $P$ it holds that 
its majority relation $\succ_{\text{maj}}$ is transitive, 
and there is some \emph{representative} voter $v_i$ in $P$ 
whose preferences coincide with $\succ_{\text{maj}}$. However, as the following example shows, 
the single-peaked domain does not have this property even if the number of voters is odd.

\begin{examplebox}
	{Single-peaked profiles may not have the representative voter property}
	{sp-not-repr-voter}
	\begin{minipage}{0.15\linewidth}
		\begin{tabular}{ccc}
			\toprule
			$v_1$ & $v_2$ & $v_3$ \\
			\midrule
			$a$ & $b$ & $c$ \\
			$b$ & $c$ & $b$ \\
			$c$ & $d$ & $a$ \\
			$d$ & $a$ & $d$ \\
			\bottomrule
		\end{tabular}
	\end{minipage}
	\hfill
	\begin{minipage}{0.8\linewidth}
		This profile is single-peaked with respect to $a\lhd b \lhd c \lhd d$.
		The majority relation of this profile satisfies $b \succ_{\text{maj}} c \succ_{\text{maj}} a \succ_{\text{maj}} d$, but no voter has this preference ranking.
	\end{minipage}
\end{examplebox}

The reason is that when we repeatedly applied the Median Voter Theorem in the first proof of 
\Cref{cor:sp-transitive}, different voters may have appeared in the 
median position at different stages of the process.
We will soon see that \emph{single-crossing} preferences do satisfy the representative 
voter property. However, for single-peaked preferences, a weaker property holds. 
Consider the majority relation in \Cref{ex:sp-not-repr-voter}: 
while no voter submitted this ordering as their vote, we can see that the majority 
ranking is single-peaked with respect to the axis $\lhd$. 
We will now argue that this is not a coincidence.

\begin{proposition}
	\label{prop:sp-maj-relation-is-sp}
	The single-peaked domain is closed. That is, if $P$ is a profile with an odd number of 
	voters that is single-peaked with respect to $\lhd$, then its majority relation $\succ_{\text{maj}}$ 
	is single-peaked on $\lhd$.
\end{proposition}
\begin{proof}
	One can prove this using the techniques of \citet{puppe2015condorcet}, i.e., by arguing that 
	the single-peaked domain is a \emph{maximal} Condorcet domain, and then appealing to their 
	Lemma~2.1. Here is a direct proof, which is due to \citet[Lemma 11.4]{mou:b:axioms}. 
	By \Cref{prop:median-voter-theorem} the majority relation has a unique peak 
	$\top(\succ_\text{maj})$, namely, the Condorcet winner. Now consider two alternatives 
	$a,b\in A$ with $\top(\succ_\text{maj}) \lhd b \lhd a$ 
	(the case $a \lhd b \lhd \top(\succ_\text{maj})$ is similar). 
	To establish that $\succ_\text{maj}$ is single-peaked on $\lhd$, 
	we need to show that $b \succ_{\text{maj}} a$. Now, since the 
	Condorcet winner is the peak of the median voter, a strict majority $N' \subseteq N$ of voters 
	have their peak located at or to the left of $\top(\succ_\text{maj})$. 
	Since all the voters in $N'$ have single-peaked preferences, 
	we then must have $b \succ_i a$ for all $i \in N'$. Since $N'$ forms 
	a strict majority, we have $b \succ_{\text{maj}} a$, as required.
\end{proof}
\Cref{prop:sp-maj-relation-is-sp} implies that some rank aggregation rules like Kemeny's rule always output a single-peaked ranking. \citet{bredereck2022opinionupdates} study the class of rank aggregation rules that preserve single-peakedness, and note their applications in an opinion diffusion model.

In addition to the majority relation, one can also consider a weighted version, namely the collection of \emph{majority margins}; the majority margin of $a$ over $b$ is defined to be $m_{ab} = |i \in N : a \succ_i b| - |i \in N : b \succ_i a|$. Thus, $m \succ_{\text{maj}} b$ if and only if $m_{ab} > 0$. \citet{smeulders2014non} and \citet{spanjaard2016netsinglepeaked} have characterized collections of majority margins that can be induced by a single-peaked preference profile and have given efficient algorithms for recognizing such collections.

\paragraph{Strategyproof Social Choice}
Consider the domain of single-peaked profiles with an odd number of voters. 
For odd $n$, every single-peaked profile admits a Condorcet winner.  
Hence, by \Cref{prop:condorcet-sp}, 
a voting rule that is defined on this domain and 
returns the Condorcet winner of the input profile is strategyproof.
Note that in this case the Condorcet winner is the top choice of the median voter.

\begin{examplebox}[breakable]
	{The median rule is strategyproof.}
	{sp-median-sp}
	It is instructive to see why the median voter rule is strategyproof.
	Consider a single-peaked profile with an odd number of voters. As an example, 
	consider the following distribution of voter peaks.
	
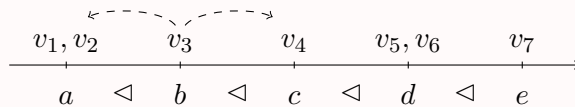
\begin{center}
	\begin{tikzpicture}[yscale=0.65,xscale=1.5]
	
	\def\xmin{1}
	\def\xmax{5}
	\def\ymin{0}
	\def\ymax{5}

	\draw[->] (\xmin -0.5,\ymin) -- (\xmax+0.5,\ymin) node[right] {};
	\foreach \x/\xtext/\vtext in {1/a/{v_1,v_2}, 2/b/{v_3}, 3/c/{v_4}, 4/d/{v_5,v_6}, 5/e/v_7}
	\draw[shift={(\x,\ymin)}] (0pt,2pt) -- (0pt,-2pt) node[below] {$\strut\xtext$} node[above] {$\strut\vtext$};
	\foreach \x/\xtext in {1, 2,3,4}
	\node[below] at (\x+0.5,\ymin) {$\strut\lhd$};  
	
	\node (v3) at (2,0.6) {};
	\node (v1) at (1.1,0.9) {};
	\node (v4) at (2.9,0.9) {};
	\draw[->, dashed] (v3) edge [bend left=60] (v4) edge [bend right=60] (v1);
	\end{tikzpicture}
	\end{center}
	The median voter is voter~4, and so the median voter rule returns $c$. Of course, voter~4 
        is not interested in manipulating, since her favorite alternative is elected. 
	Consider any other voter, say voter~3, whose truthful peak is $b$. 
	If voter~3 were to report a peak further to the left (such as $a$), 
	then this would not affect the position of the median voter, 
	so this is not a successful manipulation. If voter~3 were to report a peak further 
	to the right (such as $c$ or $d$), then either the median would not change, or it 
	would move further to the right. 
	Since the preferences of voter~3 are single-peaked, the latter change 
	would lead to a worse outcome for her. Hence voter~3 cannot manipulate, 
	and a similar argument shows that this is the case for all other voters as well. 	
\end{examplebox}

However, on this smaller domain, there may be additional strategyproof voting rules. In this 
context, instead of considering the domain of all profiles that are single-peaked, it is 
arguably more natural to fix an axis $\lhd$ over $A$, and take the domain of profiles 
single-peaked with respect to $\lhd$. We can interpret this setting as having a commonly-known 
structure of the alternative space; potential manipulators must submit preferences that conform 
with this structure. (If we considered voting rules defined for single-peaked profiles without 
a fixed axis, manipulations could change the axis.)
In this model, there are other examples of strategyproof rules: for instance, the rule 
returning the leftmost reported peak is strategyproof (the argument is the same
as for the median voter rule in the example above), as is any other order statistic.
\citet{moulin1980strategy} gave a characterization of strategyproof voting rules defined for 
profiles single-peaked with respect to a fixed axis $\lhd$. Note that his result does
not require the number of voters $n$ to be odd.

\begin{theorem}[\citealp{moulin1980strategy}]
	\label{thm:moulin-phantoms}
	Fix a set of voters $N$ and a set of alternatives $A$, 
	and let $\mathcal D$ be the set of profiles 
	single-peaked with respect to the axis $\lhd$. Then $f : \mathcal D \to A$ is anonymous, 
	Pareto-optimal, and strategyproof if and only if there exist alternatives 
	$\alpha_1, \dots, \alpha_{n-1} \in A$ such that for all profiles 
	$P \in \mathcal D$, we have
	\[ f(P) = \textup{median}_\lhd(\top(v_1), \dots, \top(v_n), \alpha_1, \dots, \alpha_{n-1}). \]
\end{theorem}
The social choice functions identified in this characterization are called \emph{generalized 
median rules}. These rules take a median of the voters' reported top choices together with 
$n-1$ fixed values. The fixed values $\alpha_1, \dots, \alpha_{n-1}$ are often interpreted as 
the reported top choices of $n-1$ ``phantom voters''. As 
Thomson~(\citeyear[p.~78]{thomson2018terminology}) suggests, one can also interpret this result as 
follows: for each ``extremist'' profile (in which $n-i$ voters report the leftmost and $i$ 
voters report the rightmost alternative of $\lhd$) we can choose an arbitrary output 
alternative $\alpha_i$ so that $\alpha_1 \lhd \cdots \lhd \alpha_{n-1}$; then there is a 
unique extension of these choices to the full domain $\mathcal D$ that is strategyproof.

\citet{moulin1980strategy} proves his result using the additional assumption of $f$ being 
``tops-only'' (so that $f$ only depends on voters' top-ranked alternatives). Later work has 
shown that any strategyproof and onto rule on $\mathcal D$ must be tops-only 
\citep{barbera1994characterization,weymark2011unified}, so this assumption can be dropped.

\Cref{thm:moulin-phantoms} continues to hold when the set of alternatives is 
$\mathbb R$, or any subset of $\mathbb R$ \citep{weymark2011unified}. It also holds when replacing 
strategyproofness by group strategyproofness (i.e., resistance to manipulation
by groups of voters), since these properties are equivalent on the 
single-peaked domain \citep{moulin1980strategy}. An analog of 
\Cref{thm:moulin-phantoms} holds when weakening Pareto optimality to the tops-only condition, in which 
case the class of anonymous and strategyproof rules consists of the generalized median rules 
with $n + 1$ rather than $n-1$ phantom voters \citep{moulin1980strategy}. The class of 
strategyproof rules without the anonymity requirement has also been characterized, though the 
description of this class involves $2^n$ parameters and is more complicated 
\citep{moulin1980strategy,weymark2011unified}. 

Within the class of generalized median rules, the most commonly considered ones are the order statistics (e.g., left-most peak, median peak) which can be obtained by having all the phantoms at extreme positions (only the left-most or right-most alternative). But other mechanisms in this class are also of interest. For example if the set of alternatives is the interval $[0,1]$, and we place phantom voters at $\alpha_1 = \frac1n, \alpha_2 = \frac2n, \dots, \alpha_{n-1} = \frac{n-1}{n}$, we obtain the \emph{uniform phantom mechanism} (also known as the \emph{linear median}), which approximates the rule that selects the average of the peaks \citep{CPS16b,freeman2021truthful,caragiannis2022truthful,jennings2023new}.

\paragraph{Counting}
It is easy to see that a single-peaked profile could be single-peaked with respect 
to several different axes. 
Indeed, \emph{every} single-peaked profile will be single-peaked with respect to 
at least two different axes, since reversing the axis preserves single-peakedness.
If we consider a profile consisting of just a single preference order $\succ$, we can see that 
this profile is single-peaked on $2^{m-1}$ different axes. Indeed, we can 
start building a partial axis for $\succ$
by placing the peak of $\succ$ on the line, and then process the rest of the vote 
from top to bottom. For each alternative, we can choose whether to put it 
to the left or to the right of the alternatives placed so far. 
For example, the profile $P = (a \succ b \succ c)$ is single-peaked with respect 
to the axes $a\lhd b\lhd c$, $b\lhd a\lhd c$, $c\lhd a\lhd b$, and $c\lhd b\lhd a$. 
On the other hand, a profile containing two reverse preference orders is only single-peaked 
on precisely two different axes (which coincide with these two orders).
We will study the collection of all axes that make a given profile single-peaked 
in more detail in \Cref{sec:recog:sp}, where we will see that the size of this collection 
is always a power of 2, and that any two axes can be obtained from each other by 
(repeatedly) reversing certain intervals of the axis. 

A converse counting problem fixes an axis $\lhd$ and asks how many preference orders 
are single-peaked on this axis. 
For this case, too, the answer is $2^{m-1}$: this is easy to see by induction on $m$, 
observing that for a preference order that is single-peaked 
on a fixed axis~$\lhd$, there are two choices available for the $m$-th position, 
namely, the two outermost alternatives of~$\lhd$.

\paragraph{Sampling}
Given a fixed axis $\lhd$, how can we randomly sample one of the $2^{m-1}$ single-peaked preferences on $\lhd$ uniformly at random? \citet{walsh2015generating} gives an algorithm that solves this problem. It works from the outside in: note that the last-ranked alternative in any single-peaked preference is either the left-most or the right-most alternative. Thus, the algorithm flips a coin to decide the last-ranked alternative. Say it is the left-most. Note that the second-to-last alternative in the preference ranking must now be either the second-left-most alternative or the right-most alternative. Again, the algorithm flips a coin to decide which of these two it is. The algorithm continues in the same fashion.

An earlier sampling algorithm of \citet{conitzer2009eliciting} can also be used to sample single-peaked preferences, but it doesn't sample them uniformly at random. Conitzer's algorithm works by first choosing one of the alternatives uniformly at random to be the most-preferred alternative (i.e., the peak). It then repeatedly decides by a coin flip whether the next-most-preferred alternative will be immediately to the left or immediately to the right of the set of alternatives the algorithm has already ranked. This algorithm oversamples rankings whose peaks are at the ends of the axis \citep{walsh2015generating}. For example, only 1 out of $2^{m-1}$ rankings has the left-most alternative as its peak, but Conitzer's algorithm gives this ranking a probability of $1/m$.

Both sampling algorithms are frequently used in numerical experiments within computational social choice \citep{boehmer2024guide}, since they tend to give qualitatively different results. Interestingly, Conitzer's algorithm produces elections that are similar to those obtained when sampling one-dimensional Euclidean preference profiles, with alternatives and voters placed on positions in the interval $[0,1]$ uniformly at random.

\paragraph{Further Properties}
\citet{puppe2016single} has characterized the domain of all preferences 
that are single-peaked on a fixed axis $\lhd$ among 
all (Cartesian) domains that guarantee a transitive majority relation: 
namely, this domain is the only one that is 
minimally rich (every alternative appears in a top position in some order), connected 
(with respect to the natural betweenness relation), 
and contains two completely reversed preference orders.

Finally, observe that the domain of single-peaked profiles is hereditary: closure under 
voter deletion is immediate from the definition, and closure under alternative deletion is 
clear from the no-valley condition.

\paragraph{Single-Caved Preferences}
\begin{wrapfigure}[7]{r}{0.32\linewidth}
	\vspace{1pt}
	\scalebox{0.8}{
		\begin{tikzpicture}[yscale=0.6,xscale=0.9]
		\def\xmin{1}
		\def\xmax{7}
		\def\ymin{0}
		\def\ymax{6}
		\draw[step=1cm,black!10,very thin] (\xmin,\ymin) grid (\xmax,\ymax);
		\draw[thick, ->] (\xmin -0.3,\ymin) -- (\xmax+0.3,\ymin) node[right] {};
		\foreach \x/\y in {4/1,5/2,3/4, 6/3, 2/5, 7/6, 1/7}
		\node[fill=blue, circle, inner sep=0.6mm] at (\x,\y-1) {};
		\draw[thick,blue] (1,6)--(2,4)--(3,3)--(4,0) -- (5,1)--(6,2)--(7,5);
		\foreach \x/\y in {1/6,2/4,3/1, 4/2, 5/3, 6/5, 7/7}
		\node[fill=red!50!black, circle, inner sep=0.6mm] at (\x,\y-1) {};
		\draw[thick,red!50!black] (1,5)--(2,3)--(3,0)--(4,1) -- (5,2)--(6,4)--(7,6);
		\end{tikzpicture}
	}
\end{wrapfigure}

A related domain restriction is the domain of \emph{single-caved} preferences, also known as 
\emph{single-dipped} preferences, which is obtained by reversing single-peakedness. 
Thus, a profile is single-caved if there exists an axis $\lhd$ such that every voter has a 
least-preferred alternative, and their preferences increase as we move away from this minimum. 
As an example, such preferences could arise in situations where we need to decide on the 
location of a public bad, such as a polluting factory, and people would like this facility to 
be placed as far away from their own location as possible.

Recall that if 
$R$ is a binary relation, then the relation $R'$ is the \emph{reverse} of $R$ if $(a,b) \in R 
\Leftrightarrow (b,a) \in R'$. The reverse of a profile $P$ is obtained by taking the reverse 
of every vote in $P$.

\begin{definition}\label{def:single-caved}
	A profile $P$ of linear orders is \defemph{single-caved} if the reverse of $P$ is single-peaked.
\end{definition}

The class of strategyproof voting rules on this domain is much smaller than in the case of 
single-peaked preferences: all such rules have a range of size at most two, and select only 
between the left-most and the right-most alternative 
\citep{manjunath2014efficient,barbera2012domains}.

\subsection{Single-Crossing Preferences}
\label{sec:def:sc}

In the previous section, we have considered a notion of structure that was imposed on the 
alternative set: we assumed that the alternative space is one-dimensional. Another natural idea 
is to require the set of \emph{voters} to be one-dimensional. This approach gives rise to the 
notion of \emph{single-crossing preferences}, which we will study in this section.

We defined a preference profile $P = (\succ_1, \dots, \succ_n)$ as an ordered list of 
votes. This already suggests a one-dimensional structure on the set $[n]$ of 
voters. What should it mean for the voters' preferences to respect this structure? Suppose that 
the leftmost voter and the rightmost voter disagree on the order of some alternatives, so that 
$a \succ_1 b$, but $b\succ_n a$ for some $a,b\in A$. Then we expect that voters who ``tend 
left'' agree with voter $1$, and that voters who ``tend right'' agree with voter $n$. Thus, we 
require that there is a voter $i\in[n]$ such that $a \succ_1 b, \dots, a\succ_i b$, and $b 
\succ_{i+1} a, \dots, b \succ_n a$. Thus, the preferences over the pair $\{a,b\}$ cross only 
once as we scan the profile from left to right. We will allow the cross-over voter $i$ to be 
different for different alternative pairs. Formally, a profile is single-crossing if it 
satisfies the following condition.

\begin{definition}
A profile $P = (v_1 ,\ldots, v_n)$ over $A$ is \defemph{single-crossing with respect to
the given ordering} if for every pair of alternatives
$a,b \in A$ both sets $\{i\in[n]:a\succ_i b\}$ and $\{i\in[n]:b\succ_i a\}$ are 
(possibly empty) intervals of $[n]$.
A profile $P = (v_1, \ldots, v_n)$ over $A$ is \defemph{single-crossing} 
if the votes in $P$ can be permuted so that the permuted profile is single-crossing 
with respect to the given ordering.
\end{definition}

\begin{examplebox}
	{A single-crossing profile}
	{sc-trajectories}
	\begin{minipage}{0.22\linewidth}
			\begin{tikzpicture}
		\matrix (m) [matrix of nodes] {
			\toprule
			$v_1$ & $v_2$ & $v_3$ & $v_4$ & $v_5$ \\
			\midrule
			$a$ & $b$ & $b$ & $d$ & $d$ \\
			$b$ & $a$ & $d$ & $b$ & $c$ \\
			$c$ & $d$ & $a$ & $c$ & $b$ \\
			$d$ & $c$ & $c$ & $a$ & $a$ \\				
			\bottomrule \\
		};
		
		\begin{tikzbackground}

		\draw[line width=7pt, red!50, draw opacity=0.5, transform canvas={yshift=1mm}] 
		(m-2-1.west) -- (m-2-1.center) -- (m-3-2.center) -- (m-4-3.center) -- (m-5-4.center) -- (m-5-5.east);

		\draw[line width=7pt, green!30, draw opacity=0.5, transform canvas={yshift=1mm}] 
		(m-3-1.west) -- (m-3-1.center) -- (m-2-2.center) -- (m-2-3.center) -- (m-3-4.center) -- (m-4-5.center) -- (m-4-5.east);

		\draw[line width=7pt, blue!40, draw opacity=0.5, transform canvas={yshift=1mm}] 
		(m-4-1.west) -- (m-4-1.center) -- (m-5-2.center) -- (m-5-3.center) -- (m-4-4.center) -- (m-3-5.center) -- (m-3-5.east);

		\draw[line width=7pt, black!20, draw opacity=0.5, transform canvas={yshift=1mm}] 
		(m-5-1.west) -- (m-5-1.center) -- (m-4-2.center) -- (m-3-3.center) -- (m-2-4.center) -- (m-2-5.center) -- (m-2-5.east);
		\end{tikzbackground}
		\end{tikzpicture}
	\end{minipage}
	\hfill
	\begin{minipage}{0.74\linewidth}
		When a profile is single-crossing with respect to the given ordering, 
		this admits an attractive visualization. For each alternative $a\in A$, 
		we draw a ``trajectory'' through the positions in the profile in which 
		$a$ appears. If the profile is single-crossing with respect to the given ordering, 
		any two trajectories will cross at most once.
	\end{minipage}
\end{examplebox}

In which situations can we expect to observe single-crossing preferences? Within economics, 
such preferences arise in models of income taxation under common assumptions
\citep{mirrlees1971exploration,roberts1977voting,rothstein1990order}. 
Specifically, if voters are ordered by 
increasing income, and some voter prefers a higher tax rate to a lower tax rate, then it stands 
to reason that all lower-income voters would agree that the higher rate is preferable to the 
lower rate (e.g., because those voters would obtain higher benefits under a redistributive 
regime). In many other contexts, there is a natural one-dimensional ordering of voters induced 
by a parameter $\theta$ (e.g., in terms of income, productivity, a discount factor, years of 
education, etc.); if the utility of the alternatives exhibits increasing differences in $\theta$, 
the resulting profile will be single-crossing. The foregoing discussion is based on the 
exposition of \citet{saporiti2009strategy}, who provides references to several models in which 
single-crossing preferences appear.

\begin{proposition}
	The domain of single-crossing preferences is hereditary, that is, closed under deleting voters and alternatives.
\end{proposition}

\paragraph{Majority Relation}
As we have seen, single-peaked profiles with an odd number of voters always have a transitive 
majority relation, and the Median Voter Theorem holds. This has been a major reason why social 
choice theorists have studied this domain restriction. Single-crossing profiles enjoy the same 
guarantee: for an odd number of voters, their majority relation is transitive. 
Moreover, in contrast to single-peaked profiles, 
single-crossing profiles with an odd number of voters 
always have a representative voter, i.e., a voter 
whose preference relation is identical to the majority relation. 

\begin{proposition}[Representative Voter Theorem, \citealp{rothstein1991representative}]
	\label{prop:sc-rvt}
	Suppose $P$ is single-crossing. If the number of voters is odd, $n = 2k-1$, 
	then the preference order of the median voter $k$ coincides with the majority relation. 
	Thus, the majority relation is transitive.
\end{proposition}
\begin{proof}
	We will argue that the majority relation agrees with the preferences of the $k$-th voter
	on every pair of alternatives.
	Let $a,b\in A$, and suppose $a \succ_k b$. By the single-crossing property, 
	it cannot be the case that both $b\succ_1 a$ and $b\succ_n a$; assume without loss
	of generality that $a\succ_1 b$. Then, by applying the single-crossing property
	again, we conclude that voters $2, \dots, k-1$ also prefer $a$
	to $b$, i.e., there are at least $k>n/2$ voters who rank $a$ above $b$.
\end{proof}
In particular, \Cref{prop:sc-rvt} implies that in a single-crossing profile 
with an odd number of voters the top alternative of the median voter is a strong
Condorcet winner.

When the number of voters is even, $n = 2k$, there are two median voters, i.e., $k$ and $k+1$, 
and the majority relation is the intersection of the relations $v_k$ and $v_{k+1}$. 
Thus, $a \succ_{\text{maj}} b$ if and only if both $a \succ_k b$ and $a \succ_{k+1} b$. 
This implies that the strict part of the majority relation is transitive, 
and hence the majority relation is quasi-transitive, though it may fail to be transitive 
(e.g., the profile in \Cref{ex:sp-quasi-transitive}, which was used to show 
that the weak majority relation of a single-peaked profile may fail to be transitive, 
is both single-peaked and single-crossing). 
Also, an argument similar to the proof of \Cref{prop:sc-rvt} shows that the top choices of voters 
$k$ and $k+1$ are weak Condorcet winners (and if these voters rank the same candidate first, 
then this candidate is a strong Condorcet winner).

\paragraph{Strategyproof Social Choice} 
The definition of a `strategyproof voting rule on the 
single-crossing domain' is somewhat subtle. Which rankings do we allow as admissible 
manipulative votes? Indeed, if we assume that the ordering of voters is given externally
(i.e., the voters are ordered by an observable parameter, such as their age, and cannot
change this parameter), then it is natural to require the manipulator's vote 
to be consistent with her position in the voter ordering. Under this assumption, if the number
of voters is odd, then the median voter rule is trivially strategyproof: the median
voter has no incentive to manipulate, and no other voter can change who the median voter is
(and hence voters cannot change the election outcome).

A more permissive approach is to allow the manipulator to submit a vote that is not necessarily
consistent with her original position; however, we still require 
the resulting profile to remain single-crossing (i.e., the profile should
become single-crossing with respect to the given ordering once the manipulator moves 
to an appropriate position in the voter ordering).
The median voter rule (for $n$ odd) remains strategyproof in this case: while a voter can now change
who the median voter is, no such change can be beneficial to her; the argument
is similar to the one for single-peaked preferences. To extend the analysis 
beyond the median voter rule, \citet{saporiti2009strategy} considers the following model.
Let $P$ be a maximal single-crossing 
profile, i.e., if we add to $P$ a linear order that is not present in $P$, then $P$ stops 
being single-crossing. Let $S$ be the set of linear orders present in $P$. Then we consider 
voting rules defined on the domain $\mathcal D = S^n$ of $n$-voter profiles where each voter's 
preference is taken from $S$. Further, let $T \subseteq A$ be the set of alternatives that are 
ranked in the top position by at least one order in $S$. 
\citet{saporiti2009strategy} proves that a voting rule 
$f$ is anonymous, unanimous, and strategyproof if and only if $f$ is a generalized median rule 
(as defined in the statement of \Cref{thm:moulin-phantoms}), where each phantom voter 
$\alpha_i$ votes for an alternative in $T$.

\paragraph{Counting}

In the case of single-peaked preferences, we have seen that a profile can be single-peaked with 
respect to exponentially many different axes. In contrast, a profile can only be 
single-crossing with respect to at most two essentially different orderings of the voters.

\begin{proposition}[see, e.g., \citealp{elkind2014recognizing}]
	\label{prop:sc-order-unique}
	Suppose $P = (v_1,\dots,v_n)$ is single-crossing with respect to the given ordering, 
	and all linear orders appearing in $P$ are pairwise distinct. Then the ordering of voters 
	making $P$ single-crossing is unique up to reversal. 
\end{proposition}
\begin{proof}
	Note first that if a profile $P=(v_1, \dots, v_n)$ is single-crossing with respect
	to the given ordering, then so is the reversed profile $(v_n, \dots, v_1)$. 

	We show that every other permutation of $P$ results in some pairs of alternatives
	crossing at least twice. We proceed by induction on $n$. 
	The case $n\le 2$ is trivial. Now, suppose $n\ge 3$.
	Since $v_1 \neq v_2$, there is a pair $a,b$ of alternatives such that $a \succ_1 b$, but $b \succ_2 a$.  
	As $P$ is single-crossing with respect to $(v_1,\dots,v_n)$, it follows that 
	$b\succ_j a$ for \emph{all} $j \neq 1$. Thus, voter~1 must be at one end of 
	any single-crossing ordering of $P$.
	
	Now, the profile $P' = (v_2,\dots,v_n)$ is also single-crossing with respect 
	to the given ordering. By the inductive hypothesis, the only other ordering making $P'$ 
	single-crossing is $(v_n,\dots,v_2)$. 
	Thus, we need to show that profiles $(v_2,\dots,v_n,v_1)$ and $(v_1,v_n,\dots,v_2)$
	are not single-crossing with respect to the given ordering. 
	Since $v_2 \neq v_n$, 
	there is a pair $c,d$ of alternatives such that $c \succ_2 d$, but $d \succ_n c$. 
	Since $P$ is single-crossing, we must have $c \succ_1 d$. But then both 
	in $(v_2, \dots, v_n, v_1)$ and in $(v_1, v_n, \dots, v_2)$ alternatives $c$ and $d$ 
	cross at least twice, as desired.
\end{proof}
If the preference orders in a profile are not pairwise distinct, then there are more than two
admissible orderings, because we can swap positions of identical votes. However, 
\Cref{prop:sc-order-unique} implies that this is the only freedom we have.

Another difference between single-peaked and single-crossing preferences is how `permissive' 
these preference restrictions are. Recall that a single-peaked profile can contain up to 
$2^{m-1}$ different preference orders. In contrast, a single-crossing profile 
contains at most quadratically many different votes.

\begin{proposition}\label{prop:sc-binom}
	A single-crossing profile contains at most $\binom{m}{2}+1$ distinct votes.
\end{proposition}
\begin{proof}
	Consider a single-crossing profile $P = (v_1,\dots,v_n)$ in which all votes are 
	pairwise distinct, 
	and a vote $v_i$ in $P$. As $v_i$ is different from its successor,
	there must be a pair of alternatives $a, b$  with $a\succ_i b$ but $b\succ_{i+1} a$.
	By the single-crossing property, $v_i$ is the \emph{last} vote in $P$ to rank $a$ above $b$. 
	Thus, we can label each of the first $n-1$ votes in $P$ with a distinct pair of alternatives, 
	so $n-1 \le \binom{m}{2}$.
\end{proof}

\begin{examplebox}
	{A single-crossing profile with $\binom{m}{2}+1$ distinct votes.}
	{max-sc}
	
	\begin{wrapfigure}[8]{l}{0.57\linewidth}
		\vspace{-1pt}
		\begin{tikzpicture}
			\matrix (m) [matrix of math nodes, column sep=1ex] {
				\toprule
				v_1 & v_2 & v_3 & v_4 & v_5 & v_6 & v_7 & v_8 & v_9 & v_{10} & v_{11} \\
				\midrule
				a & b & b & b & b & c & c & c & d & d & e \\
				b & a & c & c & c & b & d & d & c & e & d \\
				c & c & a & d & d & d & b & e & e & c & c \\
				d & d & d & a & e & e & e & b & b& b& b \\
				e & e & e & e & a & a & a & a& a& a& a \\
				\bottomrule \\
			};
			
			\begin{tikzbackground}
				\draw[line width=7pt, red!30, draw opacity=0.5, transform canvas={yshift=1mm}] 
				(m-2-1.west) -- (m-2-1.center) -- (m-3-2.center) -- (m-4-3.center) -- (m-5-4.center) -- (m-6-5.center) -- (m-6-11.east);

				\draw[line width=7pt, green!30, draw opacity=0.5, transform canvas={yshift=1mm}] 
				(m-3-1.west) -- (m-3-1.center) -- (m-2-2.center) -- (m-2-5.center) -- (m-3-6.center) -- (m-4-7.center) -- (m-5-8.center) -- (m-5-11.east);

				\draw[line width=7pt, blue!40, draw opacity=0.5, transform canvas={yshift=1mm}] 
				(m-4-1.west) -- (m-4-1.center) -- (m-4-2.center) -- (m-3-3.center) -- (m-3-5.center) -- (m-2-6.center)  -- (m-2-8.center)  -- (m-3-9.center) -- (m-4-10.center) -- (m-4-11.east);

				\draw[line width=7pt, purple!40, draw opacity=0.5, transform canvas={yshift=1mm}] 
				(m-5-1.west) -- (m-5-1.center) -- (m-5-3.center) -- (m-4-4.center) -- (m-4-6.center) -- (m-3-7.center) -- (m-3-8.center) -- (m-2-9.center) -- (m-2-10.center) -- (m-3-11.center) -- (m-3-11.east);

				\draw[line width=7pt, black!20, draw opacity=0.5, transform canvas={yshift=1mm}] 
				(m-6-1.west) -- (m-6-4.center) -- (m-5-5.center) -- (m-5-7.center) -- (m-4-8.center) -- (m-4-9.center) -- (m-3-10.center) -- (m-2-11.center) -- (m-2-11.east);
			\end{tikzbackground}
		\end{tikzpicture}
	\end{wrapfigure}
		For each $m\ge 1$, there is a single-crossing profile with $\binom{m}{2} + 1$ different votes. An example with $m = 5$ is shown. To construct such profiles, label $v_1$ so that $x_1 \succ_1 \cdots \succ_1 x_m$. In the following $m-1$ votes, alternative $x_1$ is moved down one rank at a time until it is at rank $m$. In the next $m - 2$ votes, alternative $x_2$ is moved down until it is at rank $m - 1$, and so on. See also \citet{bredereck2013characterization}.
\end{examplebox}

\paragraph{Comparison to Single-Peakedness}
Intuitively speaking, it seems that when a profile is single-crossing, 
this has implications on the structure of the alternative set as well: 
Alternatives that are highly-ranked by voters on the left seem like ``left-wing'' alternatives, 
and similarly for alternatives highly-ranked by voters towards the right. 
As the following example shows, there are single-crossing profiles that are not single-peaked. 
However, this implication \emph{almost} holds: if a single-crossing profile is also 
\emph{minimally rich} (meaning that every alternative is top-ranked by some voter; another
term used for this constraint is `narcissistic'), 
then it is also single-peaked, as we will see in \Cref{sec:def:spsc}.

\begin{examplebox}
	{A profile that is single-crossing, but not single-peaked.}
	{sc-but-not-sp}
	\begin{minipage}{0.18\linewidth}
		\begin{flushright}
			\begin{tikzpicture}
			\matrix (m) [matrix of math nodes, column sep=1ex] {
				\toprule
				v_1 & v_2 & v_3 \\
				\midrule
				a & c & c \\
				b & a & b \\
				c & b & a \\
				\bottomrule \\
			};
			
			\begin{tikzbackground}
			\draw[line width=7pt, red!20, draw opacity=0.5, transform canvas={yshift=1mm}] 
			(m-2-1.west) -- (m-2-1.center) -- (m-3-2.center) -- (m-4-3.center) -- (m-4-3.east);

			\draw[line width=7pt, green!30, draw opacity=0.5, transform canvas={yshift=1mm}] 
			(m-3-1.west) -- (m-3-1.center) -- (m-4-2.center) -- (m-3-3.center) -- (m-3-3.east);

			\draw[line width=7pt, blue!40, draw opacity=0.5, transform canvas={yshift=1mm}] 
			(m-4-1.west) -- (m-4-1.center) -- (m-2-2.center) -- (m-2-3.east);
			\end{tikzbackground}
			\end{tikzpicture}
		\end{flushright}
	\end{minipage}
	\hfill
	\begin{minipage}{0.77\linewidth}
		This profile is single-crossing with respect to the given ordering $(v_1,v_2,v_3)$. 
		It is not single-peaked, however, since three different alternatives occur 
		in bottom-most positions. (Further, as we will see in \Cref{sec:def:spsc}, 
	        this profile cannot be single-peaked because $b$'s trajectory has a valley.)
	\end{minipage}
\end{examplebox}

The reverse implication does not hold either: there are single-peaked profiles that are not 
single-crossing.

\begin{examplebox}
	{A profile that is single-peaked, but not single-crossing.}
	{sp-but-not-sc}
	\begin{minipage}{0.20\linewidth}
	\begin{tikzpicture}
	\matrix (m) [matrix of math nodes, column sep=1ex] {
		\toprule
		v_1 & v_2 & v_3 & v_4 \\
		\midrule
		a & a & b & b \\
		b & b & a & a \\
		c & d & c & d \\
		d & c & d & c \\			
		\bottomrule \\
	};
	
	\begin{tikzbackground}
	\draw[line width=7pt, red!20, draw opacity=0.5, transform canvas={yshift=1mm}] 
	(m-2-1.west) -- (m-2-2.center) -- (m-3-3.center) -- (m-3-4.east);

	\draw[line width=7pt, green!30, draw opacity=0.5, transform canvas={yshift=1mm}] 
	(m-3-1.west) -- (m-3-2.center) -- (m-2-3.center) -- (m-2-4.east);

	\draw[line width=7pt, blue!40, draw opacity=0.5, transform canvas={yshift=1mm}] 
	(m-4-1.west) -- (m-4-1.center) -- (m-5-2.center) -- (m-4-3.center) -- (m-5-4.center) -- (m-5-4.east);

	\draw[line width=7pt, black!20, draw opacity=0.5, transform canvas={yshift=1mm}] 
	(m-5-1.west) -- (m-5-1.center) -- (m-4-2.center) -- (m-5-3.center) -- (m-4-4.center) -- (m-4-4.east);
	\end{tikzbackground}
	\end{tikzpicture}
	\end{minipage}
	\hfill
	\begin{minipage}{0.77\linewidth}
	This profile is single-peaked with respect to $c \lhd b \lhd a \lhd d$. 
	It is not single-crossing, however. For an ordering of voters to be single-crossing, we need
	\begin{itemize}
		\itemsep0em
		\item $v_1$ and $v_2$ to be adjacent (because of $\{a,b\}$),
		\item $v_3$ and $v_4$ to be adjacent (because of $\{a,b\}$),
		\item $v_1$ and $v_3$ to be adjacent (because of $\{c,d\}$),
		\item $v_2$ and $v_4$ to be adjacent (because of $\{c,d\}$).
	\end{itemize}
	This is impossible.
	\end{minipage}
\end{examplebox}

The profile in \Cref{ex:sp-but-not-sc} is an important example of a profile that is not 
single-crossing. It will make another appearance in \Cref{sec:subprofiles} as a forbidden 
subprofile of the single-crossing domain.

Together, these examples show that the single-peaked and single-crossing conditions are 
logically independent. However, there are interesting domain restrictions that are 
strengthenings of both conditions, such as 1-Euclidean preferences 
(\Cref{sec:def:euclid}), the conjunction of the two conditions 
(\Cref{sec:def:spsc}), and top-monotonicity \citep{barbera2011top}.

\subsection{Euclidean Preferences}
\label{sec:def:euclid}

We now consider preference profiles that can be `embedded' into $d$-dimensional Euclidean 
space. Precisely, a preference profile is $d$-Euclidean if we can assign every voter and every 
alternative a point in $\mathbb R^d$ so that voters prefer those alternatives that are closer 
to them (according to the usual Euclidean metric $\rho_d$) to those that are further away. This 
characterization of preferences has intuitive appeal: considering $\mathbb R^d$ as a continuous 
`policy space', within which alternatives can vary along multiple dimensions, each voter is 
identified with an \emph{ideal point} \citep{bennett1960multidimensional}. The best alternative 
for the voter is the one that minimizes the distance to the ideal policy. We could also think 
of a facility location problem, where a single facility needs to be placed somewhere on a 
plane, with each decision maker preferring the facility to be placed as close to them as 
possible \citep{hotelling1929stability}.

\begin{examplebox}{2-Euclidean preferences}{2d-euclid}
	\begin{minipage}{0.77\linewidth}
Consider the 2-Euclidean profile illustrated in \Cref{fig:2deuclid}.
Both voters and alternatives are located in a two-dimensional plane. 
The (ordinal) preferences of voters are determined by their closeness to alternatives, 
based on the Euclidean distance $\rho_2$.
For example, for voter $1$ we have $\rho_2(1,a)<\rho_2(1,b)<\rho_2(1,c)$.
Hence, $a\succ_1 b\succ_1 c$.
	\end{minipage}
	\hfill
	\begin{minipage}{0.18\linewidth}
	\begin{tikzpicture}
	\matrix (m) [matrix of math nodes, column sep=1ex] {
		\toprule
		v_1 & v_2 & v_3  \\
		\midrule
		a & b & c \\
		b & c & a \\
		c & a & b \\
		\bottomrule \\
	};
	\end{tikzpicture}
	\end{minipage}
\end{examplebox}

Let $\|\cdot\|$ refer to the usual Euclidean $\ell_2$-norm on $\mathbb R^d$, that is \[\| (x_1,\dots,x_d) \| = \| (x_1,\dots,x_d) \|_2 = \sqrt{ x_1^2 + \cdots + x_d^2 }\text{.}\]  

\begin{figure}
	\centering
	\begin{tikzpicture}
		[alt/.style={circle, fill=black, inner sep=1.3pt}]
		\draw[thick, ->] (0,-2.2) -- (0,3); 
		\draw[thick, ->] (-3,0) -- (3,0);
		\node [alt, label=-90:{$a$}] (a) at (-1.7,0.7) {};
		\node [alt, label=0:{$b$}] (b) at (1.7,1.7) {};
		\node [alt, label=0:{$c$}] (c) at (0.8,-1.5) {};
		\node [label={[label distance=-3mm]+120:voter $1$}] (v1) at (-0.85,2) 
			{\color{red!60!white}\Large\Gentsroom};
		\node [label={[label distance=-3mm]-120:voter $2$}] (v2) at (2.5,0.7) {\color{blue!60!white}\Large\Gentsroom};
		\node [label={[label distance=-3mm]-120:voter $3$}] (v3) at (-0.4,-1) {\color{green!60!black!50!white}\Large\Gentsroom};
		\draw[dashed,->] (v1) -- (a);
		\draw[dashed,->] (v1) -- (b);
 		\draw[dashed,->] (v1) -- (c);		
	\end{tikzpicture}
	\caption{A 2-Euclidean embedding of the Condorcet profile.}\label{fig:2deuclid}
\end{figure}

\begin{definition}
\label{def:euclid}
A profile $P$ is \defemph{$d$-Euclidean} (where $d\ge 1$) if there is a map $x : N \cup A \to \mathbb R^d$ such that
\[ a \succ_i b \iff \|x(i) - x(a)\| < \| x(i) - x(b)\| \qquad \text{for all $i\in N$ and $a,b\in A$.}\]
Thus, voter $i$ prefers those alternatives which are closer to $i$ according to the embedding $x$. 
\end{definition}

This preference representation has been introduced by 
\citet{coombs1950psychological,coombs1964theory} as \emph{unidimensional unfolding} in the 
psychometrics literature (for the 1-dimensional case), and was later also studied in the 
multidimensional case \citep{bennett1960multidimensional,hays1961multidimensional}. The 
extensive multidimensional unfolding literature tries to find an approximate $d$-Euclidean 
representation of a given preference profile, usually by minimizing a loss function using local 
search.

\paragraph{Majority Relation}
The Condorcet paradox profile is 2-Euclidean, as seen in \Cref{ex:2d-euclid}. Therefore, being $d$-Euclidean does not guarantee a transitive majority relation or the existence of a Condorcet winner for every $d \ge 2$. (It does guarantee these properties for $d = 1$, as we discuss below.) 
In fact, \citet{EST21} prove that McGarvey's theorem \citep{mcgarvey1953} 
holds for 2-Euclidean preferences, so that every possible majority tournament 
can be induced as the majority relation of a profile that is 2-Euclidean. 
Thus, this preference restriction does not impose any structure on the majority relation. 

\paragraph{Choice of Metric}
Outside of facility location type problems, in many cases 
the space $\mathbb R^d$ does not have a natural geometric interpretation, 
and hence the $\ell_2$ distance is not necessarily a 
good modeling assumption: using  other metrics on $\mathbb R^d$, 
such as $\ell_1$ and $\ell_\infty$, 
may also be sensible. 
For example, in politics, the $d$ dimensions may encode the positions of a 
candidate on a variety of independent issues, and the best measure of a voter's distance 
to a candidate may be the sum of distances along each individual dimension, captured by the 
$\ell_1$ distance. For further discussion of the merits of using the $\ell_1$ norm, see 
\citet{eguia2011foundations} and the references therein. If voters are pessimistic, and focus 
on the issue where they have the most disagreement with a candidate, then the 
$\ell_\infty$ distance would be appropriate. However, preferences that are driven by these 
alternative metrics are not well-studied, and so for the rest of this survey we will focus on 
the usual $\ell_2$ distance. Note that for the case of $d = 1$, all three of these metrics 
give the same definition, because $\ell_1$, $\ell_2$, and $\ell_\infty$ coincide on the line.

\paragraph{Number of Distinct Preferences}
Given an embedding of the set of alternatives in $d$-dimensional space, how many different preference rankings are Euclidean? The answer is $O(m^{2d})$ \citep{jamieson2011active}, which one can determine by counting the number of cells in a hyperplane arrangement \citep[Section 28.1.1]{halperin2017arrangements}. For two dimensions, this means that there are up to $\Theta(m^4)$ different preference orders compatible with a given embedding of the alternatives \citep{bennett1960multidimensional}. For the two-dimensional case, the same bound also applies under the $\ell_1$ and $\ell_{\infty}$ metrics \citep[Theorem 6.1]{escoffier2022euclidean}. 

\paragraph{Sufficient Dimension}
Every preference profile is $d$-Euclidean for a sufficiently large dimension $d$.
In particular, it is easy to see that every profile over 
$m$ alternatives is $(m-1)$-Euclidean. For example, we can place the 
alternatives in the vertices of the standard $(m-1)$-dimensional simplex; then,  
by placing the voters at appropriate points within the simplex, 
we can induce all orders over the $m$ alternatives.

For given values of $n$ and $m$, \citet{bogomolnaia2007euclidean} ask which 
dimension $d = d(n,m)$ is sufficient to guarantee that all profiles with $m$ alternatives and 
$n$ voters are $d$-Euclidean. By the argument above, we know that $d(n,m) \le m-1$. Another 
simple argument establishes that $d(n,m) \le n$.

\begin{proposition}[\citealp{bogomolnaia2007euclidean}]
	Every profile with $n$ voters is $n$-Euclidean.
\end{proposition}
\begin{proof}
	We give a Euclidean embedding $x : N \cup A \to \mathbb R^n$. Each voter's preferences 
	are encoded on a separate axis. For $i\in N$ and $a\in A$, we set 
	$x(a)_i = -|\{ b \in A : b \pref_i a \}$. Thus, on axis $i$, voter $i$'s top-ranked 
	alternative has coordinate $-1$, the second-ranked alternative has coordinate $-2$, 
	and so on.
	
	Next, for some fixed $M>0$, and for all $i\in N$, we define $x(i)_i = M$ and $x(i)_j = 0$ 
	for $j\neq i$. Then it is easy to check that, for sufficiently large $M$, $x$ gives 
	an $n$-Euclidean embedding of the profile.
\end{proof}

To obtain a lower bound, we need examples of profiles that are not embeddable for a given 
dimension~$d$. \citet{bogomolnaia2007euclidean} show that a generalization of the Condorcet 
profile gives such an example.

\begin{examplebox}
	{A profile that is not $(d-2)$-Euclidean}
	{cyclic-non-euclidean}
	\begin{minipage}{0.27\linewidth}
		\begin{tabular}{cccc}
			\toprule
			$v_1$ & $v_2$ & $\cdots$ & $v_d$ \\
			\midrule
			$a_1$     & $a_2$    & $\cdots$ & $a_d$ \\
			$a_2$     & $a_3$    & $\cdots$ & $a_1$ \\
			$\vdots$  & $\vdots$ & $\ddots$ & $\vdots$  \\
			$a_{d-1}$ & $a_d$    & $\cdots$ & $a_{d-2}$ \\
			$a_d$     & $a_1$    & $\cdots$ & $a_{d-1}$ \\
			\bottomrule
		\end{tabular}
	\end{minipage}
	\hfill
	\begin{minipage}{0.7\linewidth}
		This profile with $d$ alternatives and $d$ voters is a generalization of 
		the Condorcet profile to $d$ alternatives.
		\citet{bogomolnaia2007euclidean} show that this profile is not $(d-2)$-Euclidean.
		Their argument rests on the observation that the points corresponding to 
		the $d$ alternatives must be affinely dependent in ${\mathbb R}^{d-2}$.
	\end{minipage}
\end{examplebox}

\Cref{ex:cyclic-non-euclidean} shows that $d(n,m) \ge \min\{n-1, m-1\}$. The 
previous arguments imply that $d(n,m) \le \min\{n, m-1\}$. This leaves a slight gap; closing it 
is an open problem, though \citet{bogomolnaia2007euclidean} give some further examples of 
profiles at either end of the gap. They also study the sufficient dimension for 
profiles of weak orders, and precisely characterize it to be equal to $\min\{n, m-1\}$.

\begin{figure}[t]
	\centering \def\sep{0.4}
	\begin{tikzpicture}[draw=black!70, xscale=0.25, yscale=-0.25]
		\node[] at (-4,-2){Worst-case}; 
		\node[] at (-4,0){dimension}; 
		
		\node[] at (18,-2){Number of alternatives $(m)$};
		\node[] at (-4,6){Number};  \node[] at (-4,8){of voters};    \node[] at (-4,10){$(n)$};
		
		\foreach \nc in {1,...,8} \node[] at (\nc*2,0) {$\nc$}; 
		\foreach \x/\nc in {9.5/\cdots,11/15,12/16, 13.5/\cdots} \node[] at (\x*2,0) {$\nc$};       
		\foreach \nv in {1,...,6} \node[] at (0,\nv*2) {$\nv$}; 
		\node[] at (0,14) {$\vdots$}; 
		
		\draw (1,-3)--(1,15);
		\draw (-8,1)--(30,1);

		\tikzset{ every path/.style={draw=black, color=black,line width=1.5pt}}
		
		\tikzset{ }
		\node at (4,4) {1D};
		\draw (5,15) -- (5,5) -- (7,5) -- (7,3) -- (30,3);      
		\node[] at (6,6) {$\bullet$};
		\node[] at (8,4) {$\bullet$};
		
		\tikzset{ every path/.style={draw=blue, color=blue,line width=1.5pt}}
		\node at (8,6) {2D};       
		\draw (7,15) -- (7,7) -- (15,7) -- (15,5) -- (30,5); %
		\node[] at (8,8) {$\bullet$};
		\node[] at (16,6) {$\bullet$};    
		
		\tikzset{ every path/.style={draw=red!80!black, color=red!80!black,line width=1.5pt}}
		\node at (12,8) {3D};     
		\node at (19,8) {? (3 or 4)};     
		\draw (9,15) -- (9,9) -- (15,9) -- (23,9) -- (23,7) -- (30,7);  \draw (15,9) -- (15,7) -- (23,7) -- (30,7); 
		\node[] at (10,10) {$\bullet$};
		\node[] at (24,8) {$\bullet$};

		\tikzset{ every path/.style={draw=black, color=black,line width=1.5pt}}
		
		\node at (19,12) {$\geq$ 4D};     
	\end{tikzpicture}
	\caption{Boundaries of non-$d$-Euclidean profiles with a given number of voters and alternatives. Each colored bullet point denotes the existence of such a non-Euclidean profile for the corresponding dimension. Figure reproduced from \citet{bulteau20222dimensional}.} %
	\label{fig:euclid-sufficient-dimension}
\end{figure}

\citet{bulteau20222dimensional} focus on the case of $d = 2$, and show that every profile with at most 3 voters and at most 7 alternatives is 2-Euclidean. They summarize the known results about the sufficient dimension for small $n$ and $m$ in \Cref{fig:euclid-sufficient-dimension} which we reproduce here. 

\citet{chen2022manhattan} study similar questions for Euclidean preferences defined with
respect to $\ell_1$ distances, obtaining similar bounds, and giving a precise analysis of the
$d = 2$ case.

\paragraph{1-Euclidean Preferences}
Since every profile is $d$-Euclidean for some $d$, the Euclidean domain restriction is only 
interesting when the dimension is small. We will be particularly interested in the 
one-dimensional case, and so we will focus on $1$-Euclidean profiles for the rest of this 
section. 

With single-peaked and single-crossing preferences, we have already seen two proposals 
for what it means for preferences to be one-dimensional. The $1$-Euclidean domain is a 
refinement of these notions: every $1$-Euclidean profile is both single-peaked and 
single-crossing.

\begin{proposition}
	\label{prop:euclid-implies-spsc}
	Let $P$ be a $1$-Euclidean profile, and let $x : N \cup A \to \mathbb R$ 
	be an embedding witnessing this. 
	Then $P$ is both single-peaked and single-crossing (with respect to, respectively,
	the alternative and voter orderings induced by $x$).
\end{proposition}
\begin{proof}
	$P$ is single-crossing since all voters $i$ with $x(i) < (x(a) + x(b))/2$ have $a \succ_i b$, 
	and all voters $i$ with $x(i) > (x(a) + x(b))/2$ have $b \succ_i a$. 
	
	Define an axis $\lhd$ by setting $a \lhd b$ if and only if $x(a) < x(b)$. 
	We claim that $P$ is single-peaked with respect to $\lhd$. 
	To see this, note that 
	for each $i\in N$ and each $c\in A$, the set $\{a \in C : a \pref_i c\}$ equals 
	$\{a \in C : x(a) \in I \}$ for some interval $I \subseteq \mathbb R$ of real numbers. 
	Thus $\{a \in C : a \pref_i c\}$ is an interval of $\lhd$ as well, 
	and hence $P$ is single-peaked with respect to $\lhd$ by \Cref{prop:sp-equiv}~(4).
\end{proof}

One might think that the conjunction of the single-peaked condition and the single-crossing condition 
is sufficient to guarantee the existence of a $1$-Euclidean embedding, i.e., that
the converse of \Cref{prop:euclid-implies-spsc} is true. However, this is not the case:
there are profiles that are single-peaked and single-crossing but not $1$-Euclidean, 
as the following example shows.
Thus, the $1$-Euclidean domain imposes additional geometric constraints beyond 
the combinatorial information in the single-peaked axis 
and the single-crossing voter ordering.

\begin{examplebox}
	{A profile that is single-peaked and single-crossing, but not $1$-Euclidean}
	{spsc-not-euclid}
\begin{wrapfigure}[10]{l}{0.16\linewidth}
	\begin{tikzpicture}
		\matrix (m) [matrix of math nodes, column sep=1ex] {
			\toprule
			v_1 & v_2 & v_3 \\
			\midrule
			b & d & d \\
			c & e & e \\
			d & c & f \\
			e & b & c \\
			a & a & b \\
			f & f & a \\
			\bottomrule \\
		};
		
		\begin{tikzbackground}
			\draw[line width=7pt, red!20, draw opacity=0.5, transform canvas={yshift=1mm}] 
			(m-6-1.west) -- (m-6-1.center) -- (m-6-2.center) -- (m-7-3.center) -- (m-7-3.east);

			\draw[line width=7pt, green!30, draw opacity=0.5, transform canvas={yshift=1mm}] 
			(m-2-1.west) -- (m-2-1.center) -- (m-5-2.center) -- (m-6-3.center) -- (m-6-3.east);

			\draw[line width=7pt, blue!40, draw opacity=0.5, transform canvas={yshift=1mm}] 
			(m-3-1.west) -- (m-3-1.center) -- (m-4-2.center) -- (m-5-3.center) -- (m-5-3.east);

			\draw[line width=7pt, red!20, draw opacity=0.5, transform canvas={yshift=1mm}] 
			(m-4-1.west) -- (m-4-1.center) -- (m-2-2.center) -- (m-2-3.center) -- (m-2-3.east);

			\draw[line width=7pt, green!30, draw opacity=0.5, transform canvas={yshift=1mm}] 
			(m-5-1.west) -- (m-5-1.center) -- (m-3-2.center) -- (m-3-3.center) -- (m-3-3.east);

			\draw[line width=7pt, blue!40, draw opacity=0.5, transform canvas={yshift=1mm}] 
			(m-7-1.west) -- (m-7-1.center) -- (m-7-2.center) -- (m-4-3.center) -- (m-4-3.east);
		\end{tikzbackground}
	\end{tikzpicture}
\end{wrapfigure}

	This profile is single-peaked with respect to the axis $a\lhd b\lhd c\lhd d\lhd e\lhd f$ 
	(and its reverse),
	and single-crossing with respect to the given ordering of the voters.
	Moreover, it can be verified that it is not single-peaked with respect to any other 
	axes (see, e.g., \Cref{sec:recog:sp}), so if there exists an embedding $x$ witnessing 
	that this profile is $1$-Euclidean, it has to be the case that 
	$x(a)<x(b)<\dots <x(f)$ or $x(f)<x(e)<\dots<x(a)$; without loss of 
	generality, we assume the former. 
	
	Since voter 1 prefers $e$ to $a$ and $b$ to $c$, we have
	\begin{align*}
	  (x(a) + x(e))/2 &<  x(1) < (x(b) + x(c))/2. 
	\end{align*}
	Considering the preferences of voters 2 and 3, we obtain
	\begin{align*}
	  (x(c) + x(d))/2 &< x(2) < (x(a) + x(f))/2 \\
	  (x(b) + x(f))/2 &< x(3) < (x(d) + x(e))/2.
 	\end{align*}
 	Adding these inequalities yields 
	$x(a) + x(b) + x(c) + x(d) + x(e) + x(f) < 
	 x(a) + x(b) + x(c) + x(d) + x(e) + x(f)$, a contradiction.
\end{examplebox}

The profile shown in \Cref{ex:spsc-not-euclid} is minimal in the sense that there are no 
examples with two voters and six alternatives, and no examples with three voters and five 
alternatives. Indeed, for five or fewer alternatives, \citet{chen2021small} prove that a profile is 1-Euclidean if and only if it is single-peaked and single-crossing. In addition, they show that every single-peaked profile with two voters is 1-Euclidean.

\subsection{Preferences Single-Peaked on a Tree}
\label{sec:def:spt}

\citet{demange1982single} observed that some of the good properties of the single-peaked domain can be preserved even when we allow the underlying axis to assume a more complicated shape. Specifically, Demange defined a notion of being single-peaked on a \emph{tree}.%
\footnote{One can also define single-peakedness on other graphs, such as on cycles (see \Cref{sec:def:spoc}) and on general graphs \citep[see, e.g.,][]{nehring2007structure}.}

\begin{definition}[\citealp{demange1982single}]
	Let $T = (A,E)$ be a tree with vertex set $A$. 
	A linear order $v_i$ over $A$ is \defemph{single-peaked on $T$} if for every pair of 
	distinct alternatives $a,b\in A$ such that $a$ lies on the unique $\top(v_i)$--$b$ path in $T$
	we have $a \succ_i b$.
	A profile $P$ over $A$ is \defemph{single-peaked on $T$} if every vote in $P$ 
	is single-peaked on $T$. 
	A profile $P$ is \defemph{single-peaked on a tree} if there exists a tree $T$ 
	such that $P$ is single-peaked on $T$.
\end{definition}

Thus, a profile is single-peaked (in the sense of \Cref{sec:def:sp}) if and only if it 
is single-peaked on a tree that is a path. In particular, the property of being single-peaked 
on a tree is less demanding than the property of being single-peaked.

One of the motivating examples for the single-peaked domain is the problem of locating a 
facility along a line (such as a road) when voters prefer the facility to be close to them. 
When the road network has a tree structure, we can expect the voters' preferences 
to be single-peaked on a tree.

Just like we did for single-peaked preferences, we can phrase the definition of being 
single-peaked on a tree in several different ways. In the following proposition, definition~(2) 
is the one used by \citet{demange1982single}, and definition~(4) is the one used by 
\citet{trick1989recognizing}.

\begin{proposition}
	\label{prop:spt-equiv}
	Let $T = (A,E)$ be a tree. The following are equivalent:
	\begin{enumerate}[(1)]
		\item $P$ is single-peaked on $T$.
		\item For every subtree $T' \subseteq T$ which is a path, $P|_{T'}$ is single-peaked.
		\item For each $i\in N$, and all $a,b,c\in A$ such that $a,b,c$ lie on a common path in $T$, we do not have both $a \succ_i b$ and $c \succ_i b$.
		\item For each $i\in N$ and each $c \in A$, the set $\{a\in C : a \pref_i c\}$ is connected in $T$.
	\end{enumerate}
\end{proposition}

\begin{examplebox}[breakable]
	{A vote single-peaked on a tree}
	{spt-one-vote}
	The vote $debca\mathit{f}$ is single-peaked on the below tree, as we can verify by using condition (4) of  \Cref{prop:spt-equiv}, i.e., by checking that the following sets are all connected in the tree: $\{d\}$, $\{d,e\}$, $\{d,e,b\}$, $\{d,e,b,c\}$, $\{d,e,b,c,a\}$, and $\{d,e,b,c,a,f\}$.
	\begin{center}
		\begin{tikzpicture}[visited/.style={fill=red!60, draw=red!60, thick}]
			\node [vertex, label={[label distance=2.5pt]below:$a$}] (a) {};
			\node [vertex, label={below:$b$}, right=of a] (b) {};
			\node [vertex, label={above:$c$}, above right=0.5cm and 1.1cm of b] (c) {};
			\node [vertex, visited, label={below:$d$}, below right=0.5cm and 1.1cm of b] (d) {};
			\node [vertex, label={above:$e$}, above right=0.5cm and 1.1cm of d] (e) {};
			\node [vertex, label={below:$f$}, below right=0.5cm and 1.1cm of d] (f) {};
			\draw (a) edge (b) (b) edge (c) edge (d) (d) edge (e) edge (f);
		\end{tikzpicture}
		\quad
		\begin{tikzpicture}[visited/.style={fill=red!60, draw=red!60, thick}]
			\node [vertex, label={[label distance=2.5pt]below:$a$}] (a) {};
			\node [vertex, label={below:$b$}, right=of a] (b) {};
			\node [vertex, label={above:$c$}, above right=0.5cm and 1.1cm of b] (c) {};
			\node [vertex, visited, label={below:$d$}, below right=0.5cm and 1.1cm of b] (d) {};
			\node [vertex, visited, label={above:$e$}, above right=0.5cm and 1.1cm of d] (e) {};
			\node [vertex, label={below:$f$}, below right=0.5cm and 1.1cm of d] (f) {};
			\draw (a) edge (b) (b) edge (c) edge (d) (d) edge [visited] (e) edge (f);
		\end{tikzpicture}
		\quad
		\begin{tikzpicture}[visited/.style={fill=red!60, draw=red!60, thick}]
			\node [vertex, label={[label distance=2.5pt]below:$a$}] (a) {};
			\node [vertex, visited, label={below:$b$}, right=of a] (b) {};
			\node [vertex, label={above:$c$}, above right=0.5cm and 1.1cm of b] (c) {};
			\node [vertex, visited, label={below:$d$}, below right=0.5cm and 1.1cm of b] (d) {};
			\node [vertex, visited, label={above:$e$}, above right=0.5cm and 1.1cm of d] (e) {};
			\node [vertex, label={below:$f$}, below right=0.5cm and 1.1cm of d] (f) {};
			\draw (a) edge (b) (b) edge (c) edge [visited] (d) (d) edge [visited] (e) edge (f);
		\end{tikzpicture}
		\quad
		\begin{tikzpicture}[ visited/.style={fill=red!60, draw=red!60, thick}]
			\node [vertex, label={[label distance=2.5pt]below:$a$}] (a) {};
			\node [vertex, visited, label={below:$b$}, right=of a] (b) {};
			\node [vertex, visited, label={above:$c$}, above right=0.5cm and 1.1cm of b] (c) {};
			\node [vertex, visited, label={below:$d$}, below right=0.5cm and 1.1cm of b] (d) {};
			\node [vertex, visited, label={above:$e$}, above right=0.5cm and 1.1cm of d] (e) {};
			\node [vertex, label={below:$f$}, below right=0.5cm and 1.1cm of d] (f) {};
			\draw (a) edge (b) (b) edge [visited] (c) edge [visited] (d) (d) edge [visited] (e) edge (f);
		\end{tikzpicture}
		\quad
		\begin{tikzpicture}[ visited/.style={fill=red!60, draw=red!60, thick}]
			\node [vertex, visited, label={[label distance=2.5pt]below:$a$}] (a) {};
			\node [vertex, visited, label={below:$b$}, right=of a] (b) {};
			\node [vertex, visited, label={above:$c$}, above right=0.5cm and 1.1cm of b] (c) {};
			\node [vertex, visited, label={below:$d$}, below right=0.5cm and 1.1cm of b] (d) {};
			\node [vertex, visited, label={above:$e$}, above right=0.5cm and 1.1cm of d] (e) {};
			\node [vertex, label={below:$f$}, below right=0.5cm and 1.1cm of d] (f) {};
			\draw (a) edge [visited] (b) (b) edge [visited] (c) edge [visited] (d) (d) edge [visited] (e) edge (f);
		\end{tikzpicture}
		\quad
		\begin{tikzpicture}[ visited/.style={fill=red!60, draw=red!60, thick}]
			\node [vertex, visited, label={[label distance=2.5pt]below:$a$}] (a) {};
			\node [vertex, visited, label={below:$b$}, right=of a] (b) {};
			\node [vertex, visited, label={above:$c$}, above right=0.5cm and 1.1cm of b] (c) {};
			\node [vertex, visited, label={below:$d$}, below right=0.5cm and 1.1cm of b] (d) {};
			\node [vertex, visited, label={above:$e$}, above right=0.5cm and 1.1cm of d] (e) {};
			\node [vertex, visited, label={below:$f$}, below right=0.5cm and 1.1cm of d] (f) {};
			\draw (a) edge [visited] (b) (b) edge [visited] (c) edge [visited] (d) (d) edge [visited] (e) edge [visited] (f);
		\end{tikzpicture}
	\end{center}
	On the other hand, the vote $de\mathit{f}cba$ is not single-peaked on this tree, because the set $\{d, e, f, c\}$ of the top 4 alternatives is not connected in the tree.
	\begin{center}
		\begin{tikzpicture}[visited/.style={fill=red!60, draw=red!60, thick}]
			\node [vertex, label={[label distance=2.5pt]below:$a$}] (a) {};
			\node [vertex, label={below:$b$}, right=of a] (b) {};
			\node [vertex, visited, label={above:$c$}, above right=0.5cm and 1.1cm of b] (c) {};
			\node [vertex, visited, label={below:$d$}, below right=0.5cm and 1.1cm of b] (d) {};
			\node [vertex, visited, label={above:$e$}, above right=0.5cm and 1.1cm of d] (e) {};
			\node [vertex, visited, label={below:$f$}, below right=0.5cm and 1.1cm of d] (f) {};
			\draw (a) edge (b) (b) edge (c) edge (d) (d) edge [visited] (e) edge [visited] (f);
		\end{tikzpicture}
	\end{center}
\end{examplebox}

\begin{examplebox}[breakable]
	{Profiles single-peaked on a star}
	{spt-star}
	Consider the star $T$ with center alternative $z$ and leaf alternatives $a,\dots,g$. 
	Which orders are single-peaked on $T$?
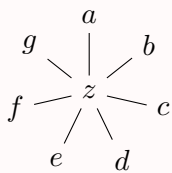
	\begin{center}
	\begin{tikzpicture}
	\graph {
		subgraph I_n[n=7, V={$a$,$b$,$c$,$d$,$e$,$f$,$g$}, clockwise] -- {z/$z$}
	};
	\end{tikzpicture}
	\end{center}
	Consider a voter whose preferences are single-peaked on this tree. If her ranking 
	begins with $z$, she can rank the other alternatives in an arbitrary order:
	any such ranking is single-peaked on $T$. 
	But suppose her top choice is a leaf alternative such as $a$. Then $z$ must be the 
	second alternative in her ranking, because 
	the set consisting of her top two choices must be connected in $T$
	(\Cref{prop:spt-equiv} (4)). 
        After ranking $a$ and $z$, the voter can order the remaining alternatives arbitrarily.
	
	Thus, precisely the orders in which the center vertex is ranked first or second 
	are single-peaked on a star. Hence, there are $2(m-1)! = \Theta(m!)$ orders 
	single-peaked on a star---many more than the $\Theta(2^m)$ orders 
	that are single-peaked on a line.
\end{examplebox}

Allowing trees as the underlying structure of the alternatives also enables us to handle certain 
configurations that are `almost' single-peaked on a line. For example, in politics, while the 
traditional left-to-right spectrum of political parties has substantial explanatory power, this 
classification can break down, especially at the extremes. For instance, it can be difficult to 
say whether a party focused on the environment is more left-wing 
than a party focused on women's rights, or vice versa.
The following kind of tree might capture the situation better:
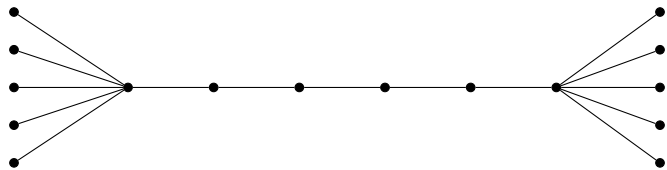
\begin{center}
\begin{tikzpicture}
[scale=1]
\node[vertex] (l0) {};
\node[vertex, right=1 of l0] (l1) {};
\node[vertex, right=1 of l1] (l2) {};
\node[vertex, right=1 of l2] (l3) {};
\node[vertex, right=1 of l3] (l4) {};
\node[vertex, right=1 of l4] (l5) {};
\node[vertex] at (-1.5,1) (lex1) {};
\node[vertex] at (-1.5,0.5) (lex2) {};
\node[vertex] at (-1.5,0) (lex3) {};
\node[vertex] at (-1.5,-0.5) (lex4) {};
\node[vertex] at (-1.5,-1) (lex5) {};
\node[vertex] at (7,1) (rex1) {};
\node[vertex] at (7,0.5) (rex2) {};
\node[vertex] at (7,0) (rex3) {};
\node[vertex] at (7,-0.5) (rex4) {};
\node[vertex] at (7,-1) (rex5) {};
\draw (l0) edge (lex1) edge (lex2) edge (lex3) edge (lex4) edge (lex5)
(l5) edge (rex1) edge (rex2) edge (rex3) edge (rex4) edge (rex5)
(l0) edge (l1)
(l1) edge (l2)
(l2) edge (l3)
(l3) edge (l4)
(l4) edge (l5);
\end{tikzpicture}
\end{center}
With this tree, voters are free to order the extreme alternatives as they like, but there is 
still a noticeable left-to-right ordering over the more moderate choices.

\paragraph{Majority Relation}
We have seen that the single-peaked domain is attractive in that, for an odd number of voters, 
it guarantees a Condorcet winner, and, moreover, a transitive majority relation.
The generalization to trees preserves the first of these guarantees.

\begin{figure}[t]
	\centering
	\begin{tikzpicture}[vertex/.style={circle, draw=black, fill=black, inner sep=0pt, minimum size=5pt},
		on path/.style={red!60!black, thick}]
		\node [style=vertex] (0) at (-4.75, 3.75) {}; %
		\node [style=vertex] (1) at (-4.75, 4.9) {};
		\node [style=vertex] (2) at (-3.75, 2.75) {};
		\node [style=vertex] (3) at (-2.5, 2.25) {};
		\node [style=vertex] (4) at (-3.5, 1.75) {};
		\node [style=vertex] (5) at (-3.25, 0.75) {};
		\node [style=vertex] (7) at (-5.75, 3) {}; %
		\node [style=vertex] (8) at (-6.75, 2.25) {};
		\node [style=vertex] (9) at (-6.25, 1.25) {};
		\node [style=vertex] (10) at (-7.75, 1.5) {}; %
		\node [style=vertex] (11) at (-8.5, 0.5) {};
		\node [style=vertex] (12) at (-7.25, 0.5) {};

		\node at (-4.4, 3.75) {$c$};
		\node at (-5.75, 3.3) {$x$};
		\node at (-7.75, 1.85) {$d$};

		\draw[-latex] (0) to (1);
		\draw[-latex] (0) to (2);
		\draw[-latex] (2) to (3);
		\draw[-latex] (2) to (4);
		\draw[-latex] (4) to (5);
		\draw[-latex,on path] (0) to (7);
		\draw[-latex,on path] (7) to (8);
		\draw[-latex] (8) to (9);
		\draw[-latex,on path] (8) to (10);
		\draw[-latex] (10) to (11);
		\draw[-latex] (10) to (12);
	\end{tikzpicture}
	\caption{Illustration of the proof of \Cref{prop:spt-condorcet}. 
	The edges of the tree $T$ are oriented according to the majority relation. 
	The source vertex of the resulting directed acyclic graph can be shown to be a Condorcet winner.}
	\label{fig:spt-condorcet-proof}
\end{figure}

\begin{proposition}[\citealp{demange1982single}]
	\label{prop:spt-condorcet}
	Every profile $P$ with an odd number of voters that is single-peaked on a tree 
	has a  Condorcet winner.
\end{proposition}
\begin{proof}
	Fix a profile $P$ that is single-peaked on a tree $T$,
	and turn $T$ into a directed graph by orienting its edges 
	according to the majority relation, so that there is an arc $a \to b$
	if and only if $\{a,b\}$ is an edge of $T$ and $a \succ_{\text{maj}} b$ in $P$. 
	See \Cref{fig:spt-condorcet-proof} for an illustration.
	Since $T$ is a tree, the resulting directed graph $D$ is acyclic. 
	Hence it must have a source vertex $c \in A$, i.e.,  
	a vertex that dominates all of its neighbors. We claim that $c$ is a Condorcet winner. 
	Indeed, consider an alternative $d\in A\setminus\{c\}$. If $d$ is a neighbor of $c$, 
	then $c \succ_{\text{maj}} d$ since $c$ is a source. Otherwise, consider the (unique) 
	path $c-x-\cdots-d$ between $c$ and $d$ in $T$. Since $c$ is a source, 
	we have $c \succ_{\text{maj}} x$, and so there is a strict majority $N' \subseteq N$ 
	of voters who all prefer $c$ to $x$. Consider a voter $i\in N'$.
	Since her preferences remain single-peaked when restricted to the path $c-x-\cdots-d$, 
	by the no-valley property she prefers $x$ to $d$, and therefore, 
	by transitivity, $c\succ_i d$. Hence $c\succ_{\text{maj}} d$. 
\end{proof}
Interestingly, to determine which alternative is the Condorcet winner, it suffices to know the plurality scores of all the alternatives (i.e., the number of voter ranking it top), since this information is enough to decide how to orient the edges in the tree according to the majority relation \citep{xu2024localstability}. Indeed, for an edge $\{a,b\}$ of the tree $T$, if we delete the edge, we obtain a graph with two connected components. We can compute the total plurality score of the alternatives in each component. Then $a \succ_{\text{maj}} b$ if and only if the total plurality score in the component including $a$ is strictly higher than the total plurality score of the other component.

A slight modification of this proof shows that profiles that are single-peaked 
on a tree and have an even number of voters 
still have a weak Condorcet winner \citep{demange1982single}. Further,  
this preference domain is maximal in guaranteeing the existence of a Condorcet winner, 
in the following sense. 
Take any tree $T = (A,E)$ and let $L$ be the set of all linear orders that 
are single-peaked on $T$. \citet{demange1982single} shows that if we add any linear order to $L$, 
then there exists a profile that consists of preferences drawn from this enlarged set 
and does not have a Condorcet winner.

However, if a profile is single-peaked on a tree, its 
majority relation may fail to be transitive (unlike for profiles that are single-peaked on a path).
Examples are easy to construct, for instance using a 
star: recall that single-peakedness on a star is not very restrictive 
(\Cref{ex:spt-star}). In fact, one 
can show that paths are the only trees that guarantee a transitive majority relation.

\begin{proposition}[\citealp{demange1982single}]
	No tree other than a path guarantees a transitive majority relation.
\end{proposition}
\begin{proof}
	We have seen in \Cref{cor:sp-transitive} that paths do guarantee 
	a transitive majority relation. Suppose $T = (A,E)$ is a tree that is not a path, 
	so that there is some vertex $x\in A$ with at least three neighbors. 
	Suppose $a,b,c\in A$ are three distinct neighbors of $x$ in $T$. 
	Then consider the following profile:
	\begin{center} 
	\begin{tabular}{ccc}
	\toprule
	$v_1$ & $v_2$ & $v_3$ \\
	\midrule
	$x$ & $x$ & $x$ \\
	$a$ & $b$ & $c$ \\
	$b$ & $c$ & $a$ \\
	$c$ & $a$ & $b$ \\
	$\vdots$ & $\vdots$ & $\vdots$ \\
	\bottomrule
	\end{tabular}
	\end{center}
	Assume that the votes have been completed so as to be single-peaked on $T$. 
	The resulting profile has $x$ as the Condorcet winner, 
	but contains a Condorcet cycle 
	$a \succ_{\text{maj}} b \succ_{\text{maj}} c \succ_{\text{maj}} a$ below $x$.
\end{proof}

Similarly to \Cref{prop:sp-maj-relation-is-sp}, one can show that if $P$ is 
single-peaked on $T$ and its majority relation $\succ_\text{maj}$ forms a linear order,
then $\succ_\text{maj}$ is also single-peaked on $T$.  
Indeed, a vote is single-peaked on $T$ if and only if it is single-peaked on every path in $T$, 
and we can apply \Cref{prop:sp-maj-relation-is-sp} to each such path.

\paragraph{Strategyproof Voting Rules}
\citet{danilov1994structure} characterizes the class of strategy-proof voting rules 
for preferences single-peaked on a specific tree $T$. 
Like in \citeauthor{moulin1980strategy}'s \citeyearpar{moulin1980strategy} characterization 
for the single-peaked case, any such rule is a median of voters' peaks and `phantom' peaks, 
for an appropriate definition of `median'. \citet{peters2020spt} generalize this result to 
randomized rules and to more general graphs.

\paragraph{Counting}
As we have seen, there are $\Theta(m!)$ different rankings that are single-peaked 
on a star, but only $\Theta(2^m)$ rankings that are single-peaked on a path. 
Conversely, since a profile can be single-peaked on exponentially many different axes, 
it can also be single-peaked on exponentially many different trees. 
The collection of all such trees admits a concise representation, 
as we will discuss in \Cref{sec:recog:spt}.

In our discussion of single-peaked preferences, we saw that every single-peaked profile is 
single-peaked with respect to at least two axes, since reversing an axis preserves 
single-peakedness. However, from a graph-based perspective, `reversing' the vertices of a 
path does not change the graph $T = (A,E)$. Thus, it is possible for a profile to be 
single-peaked on a unique tree. In particular, \citet{trick1989recognizing} has shown that 
every narcissistic profile that is single-peaked on a tree is single-peaked on a unique tree.

\paragraph{Sampling}
\citet{sliwinski2019spt} give an efficient algorithm that given a fixed tree $T$, samples a preference single-peaked on $T$ uniformly at random. The algorithm can be seen as a generalization of the idea behind the algorithm of \citet{conitzer2009eliciting} for single-peaked preferences, except that it samples the vote's peak with probabilities carefully chosen to ensure that preferences are sampled uniformly at random.

\paragraph{Hereditariness}
It is clear from the definition that the domain of profiles single-peaked on a tree 
is closed under voter deletion. However, as the following example shows, it is not closed 
under deleting alternatives. Thus, this domain is not hereditary.

\begin{examplebox}[height=4.5cm]
	{The domain of profiles single-peaked on trees is not closed under alternative deletion}
	{sp-trees-not-closed-alt-del}
	\begin{wrapfigure}{l}{0.15\linewidth}
		\begin{tabular}{ccc}
			\toprule
			$v_1$ & $v_2$ & $v_3$ \\
			\midrule
			$x$ & $x$ & $x$ \\
			$a$ & $b$ & $c$ \\
			$b$ & $c$ & $a$ \\
			$c$ & $a$ & $b$ \\
			\bottomrule
		\end{tabular}
	\end{wrapfigure}
	This profile is single-peaked on a star with center $x$. However, deleting 
	alternative $x$ results in a profile that does not have a Condorcet winner, 
	so it cannot be single-peaked on a tree.
\end{examplebox}

While the domain of preferences single-peaked on trees 
is not closed under alternative deletion, it has a weaker property: 
if a profile is single-peaked on some tree $T = (A,E)$, and some alternative $a\in A$ has degree $2$ 
in $T$, then deleting $a$ yields a profile that is again single-peaked on a tree 
(namely, we delete $a$ from $T$ and join the two neighbors of $a$ by an edge). Similarly, 
we can delete leaves (vertices of degree $1$) from $T$.

An interesting direction for future work is to consider the domain of profiles that are 
single-peaked on a tree, and, moreover, remain single-peaked on some tree, no matter which 
(and how many) alternatives we delete.

\subsection{Preferences Single-Peaked on a Circle}
\label{sec:def:spoc}
In the previous section, we explored the consequences of 
extending the notion of single-peaked preferences
from paths to trees. Another class of graphs that are, in some sense, 
similar to paths are cycles. Motivated by this intuition, 
\citet{peters2017spoc} introduced the domain of 
preferences single-peaked on a circle. In such profiles, the alternatives can be arranged in a 
circle, and each voter ``cuts'' the circle at some point, obtaining a line, so that the voter's 
preferences are then single-peaked on this line. Note that different voters may cut the circle 
at different points. In our discussion, we follow the exposition of \citet{peters2017spoc}.

\paragraph{Motivating Examples}
In some cases, the alternative set comes with an obvious cyclic structure, for example, 
when deciding on a recurring meeting time (with weekdays forming a cycle), 
or a time for a daily event (on a 24-hour cycle). In facility location, 
when locating a facility along the boundary of an area (such as a city), 
one can often view this boundary as a circle.

\begin{wrapfigure}[7]{r}{0.25\linewidth}
	\scalebox{0.6}{
		\begin{tikzpicture}[yscale=0.6,xscale=0.9]
		\def\xmin{1}
		\def\xmax{7}
		\def\ymin{0}
		\def\ymax{6}
		\draw[step=1cm,black!10,very thin] (\xmin,\ymin) grid (\xmax,\ymax);
		\draw[thick, ->] (\xmin -0.3,\ymin) -- (\xmax+0.3,\ymin) node[right] {};
		\foreach \x/\y in {4/7,5/6,3/4, 6/5, 2/3, 7/2, 1/1}
		\node[fill=blue, circle, inner sep=0.6mm] at (\x,\y-1) {};
		\draw[thick,blue] (1,0)--(2,2)--(3,3)--(4,6) -- (5,5)--(6,4)--(7,1);
		\foreach \x/\y in {1/6,2/4,3/1, 4/2, 5/3, 6/5, 7/7}
		\node[fill=red!50!black, circle, inner sep=0.6mm] at (\x,\y-1) {};
		\draw[thick,red!50!black] (1,5)--(2,3)--(3,0)--(4,1) -- (5,2)--(6,4)--(7,6);
		\end{tikzpicture}
	}
\end{wrapfigure}

Another important subclass of preferences single-peaked on a circle is formed 
by mixtures of single-peaked and 
single-dipped preferences on a common axis. Such profiles are single-peaked on the circle 
obtained by gluing together the ends of the axis. Mixtures of this form occur 
naturally in facility location on a line, where voters may disagree whether the facility 
(e.g., a school or a hospital) is a good (and want it to be as close as possible 
to their own position) or a bad (and want it to be far away).

\paragraph{Definition}
We say that two axes $\lhd$ and $\lhd'$ are \emph{cyclically equivalent} if there is 
a value $\ell\in [m]$ such that we can write $a_1 \lhd a_2 \lhd a_3 \lhd \dots \lhd a_m$ 
and $a_\ell \lhd' a_{\ell+1} \lhd' \dots \lhd' a_m \lhd' a_1 \lhd' \dots \lhd' a_{\ell-1}$. 
For an axis $\lhd$, we then define the \emph{circle $C(\lhd)$ of $\lhd$} to be the set of axes 
cyclically equivalent to $\lhd$. Any set $C$ of axes that can be written as $C=C(\lhd)$ 
for some $\lhd$ is called a \emph{circle}. 
For example, $C = \{ a \lhd b \lhd c, b \lhd' c \lhd' a, c \lhd'' a \lhd'' b  \}$ is a circle.
Note that ``cutting'' a circle $C$ at a point yields an axis ${\lhd}\in C$. We say that 
$\lhd$ \emph{starts at} $a\in A$ if $a \lhd b$ for all $b\in A \setminus \{a\}$.

\begin{definition}
	Let $C$ be a circle. A vote $v_i$ is \defemph{single-peaked on $C$} if there is an axis 
	${\lhd} \in C$ such that $v_i$ is single-peaked with respect to $\lhd$. A preference profile 
	$P$ is \defemph{single-peaked on a circle (SPOC)} if there exists a circle $C$ such that 
	every vote $v_i \in P$ is single-peaked on $C$.
\end{definition}

\begin{wrapfigure}[4]{r}{0.1\linewidth}
	\scalebox{0.75}{
		\begin{tikzpicture}
		\graph[nodes={draw=black, fill=black, circle, inner sep=1pt}, empty nodes, clockwise, n=13] {
			{[nodes={red!80!black, very thick}]a; b; c; d; e; f; g; h; i;};
			{[nodes={blue!80!black, very thick}] j; k; l; m;};
			{[edges={red!80!black, ultra thick}] a -- b -- c -- d -- e -- f -- g -- h -- i};
			i -- j;
			m -- a;
			{[edges={blue!80!black, ultra thick}] j -- k -- l -- m}
		};
		\end{tikzpicture}}
\end{wrapfigure}
Intuitively, a vote $v_i$ is single-peaked on $C$ if $C$ can be cut so that $v_i$ is single-peaked 
on the resulting line.
Again, we can give an equivalent definition in terms of connected sets. 
An \emph{interval} $I \subseteq A$ of a circle $C$ is a set that is an interval 
of one of the axes ${\lhd} \in C$ of the circle. Then a vote $v_i$ is single-peaked on a circle $C$ 
if and only if each prefix $\{ a \in A : a \succ_i c \}$ of $v_i$ is an interval of $C$. 
Note that the complement $A\setminus I$ of an interval $I$ of $C$ is again an interval. 
Thus, a weak order $\pref$ is single-peaked on $C$ if and only if its reverse 
is also single-peaked on $C$.

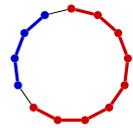
\begin{examplebox}[height=3.6cm]
	{Single-peaked preference orders on a circle with 5 alternatives}
	{spoc-orders}
	\begin{wrapfigure}[6]{l}{0.16\linewidth}
		\begin{tikzpicture}
			\graph[nodes={baseline, circle, inner sep=1pt}, edges={draw=black}, clockwise, n=5] {
				subgraph C_n [V={a,b,c,d,e}];
			};
		\end{tikzpicture}
	\end{wrapfigure}
	Consider the circle with candidates $a$, $b$, $c$, $d$, $e$ in that order. The eight preference orders that are single-peaked on that circle and that rank $a$ on top are: $abcde$, $abced$, $abecd$, $abedc$, $aebcd$, $aebdc$, $aedbc$, $aedcb$. Hence there are $5 \times 8 = 40$ preference orderings single-peaked on this circle.
\end{examplebox}

\paragraph{Majority Relation}
One can easily check that
the Condorcet paradox profile $(x\succ_1y\succ_1z, y\succ_2 z \succ_2 x, z \succ_3 x \succ_3 y)$ on 3 alternatives is single-peaked on the circle $x\lhd y \lhd z$. Therefore, being single-peaked on a circle does not guarantee a transitive majority relation or the existence of a Condorcet winner. 
In fact, \citet{peters2017spoc} prove that McGarvey's theorem \citep{mcgarvey1953} 
holds for preferences single-peaked on a circle: every possible majority tournament 
can be induced as the majority relation of a profile that is single-peaked on a circle. 
Thus, this preference restriction does not impose any structure on the majority relation. 
Notably, even if the majority relation is transitive, it need not be single-peaked 
on a circle itself, in contrast 
to the case of preferences single-peaked on lines and trees.

\paragraph{Strategyproof Social Choice}
We have seen that, for any given tree (e.g., a path), the domain of preferences 
single-peaked on that tree
admits a non-trivial strategyproof voting rule. In contrast,  
the results of \citet{kim1980special} and \citet{sato2010circular} show that there 
is no (resolute) voting rule defined on profiles single-peaked on a circle 
that is non-dictatorial, onto, and satisfies strategyproofness.
In fact, they prove that this result holds for an even more restricted domain 
consisting only of the $2m$ orders which traverse the circle clockwise and 
counterclockwise starting from each possible alternative.
Note that the orders used in this proof are heavily `directional':
if we associate each agent with a position on the circle, then for a `clockwise'
agent her top choice is the first alternative located clockwise from her and her bottom
choice in the first alternative located counterclockwise from her 
(and similarly for `counterclockwise' agents).
Still, a similar dictatorship result can be proved for orders that 
are `Euclidean' on a circle, where preferences decrease uniformly 
in both directions from the peak \citep{schummer2002strategy}. 
It can also be shown that, with these Euclidean orders, the \emph{random dictatorship} 
rule is \emph{group}-strategyproof \citep{alon2010walking}, and there is an intriguing 
randomized mechanism that is strategyproof and provides a 3/2-approximation to 
the egalitarian social welfare \citep{alon2010strategyproof}.

\paragraph{Counting and Sampling}
Given a fixed circle $C$, the number of preference orders single-peaked on $C$ is $m\cdot 2^{m-2}$ \citep{oeisA057711}.
One can sample among these orders uniformly at random by using the analog of the \citet{conitzer2009eliciting} algorithm for single-peaked preferences: sample the peak uniformly at random, then repeatedly flip a coin to decide whether the next-most-preferred alternative should be adjacent to the set of already-ranked alternatives in clockwise or counterclockwise direction. This leads to a selection that is uniformly at random for the circular context \citep{boehmer2024guide}. 

\paragraph{Single-Crossing on a Circle}
\citet{constantinescu2022voting} give a definition of preferences that are single-crossing on a circle, which is equivalent to preferences that cross at most twice. Like for single-peaked on a circle, McGarvey's theorem also holds for single-crossing on a circle. \citet{constantinescu2022voting} give a polynomial-time recognition algorithm for this domain, and propose polynomial-time algorithms for evaluating Young's rule and the Chamberlin--Courant rule on this domain.

\subsection{Preferences Single-Crossing on a Tree}
\label{sec:def:sct}

The single-crossing domain also allows a generalization from the line to trees, as proposed by 
\citet{clearwater2015generalizing} and \citet{kung2015sorting}. (In fact, 
\citet{clearwater2015generalizing}, and, subsequently, \citet{puppe2015condorcet}, 
also consider preferences that are single-crossing on 
\emph{median graphs}, but we will not discuss this larger domain.) 
This restricted domain consists of profiles where 
voters correspond to the vertices of a tree, and for each pair of alternatives, 
the voters' preferences cross along a single edge.

\begin{definition}
	\label{def:sct}
	A profile $P$ of linear orders is \defemph{single-crossing on the tree $T = (N,E)$} 
	if for each pair $a,b\in A$ of alternatives, the sets $\{ i \in N : a \succ_i b \}$ 
	and $\{ i \in N : b \succ_i a \}$ are connected in $T$. A profile $P$ is 
	\defemph{single crossing on a tree} if there exists a tree $T$ such that $P$ 
	is single-crossing on $T$.
\end{definition}

Clearly, a profile is single-crossing in the sense of \Cref{sec:def:sc} if and only if it 
is single-crossing on a tree that is a path.

\begin{examplebox}[breakable]
	{A profile that is single-crossing on a tree, but not on a path}
	{sct-but-not-sc}
	\begin{wrapfigure}[2]{l}{0.22\linewidth}
		\begin{tikzpicture}
		\matrix (m) [matrix of math nodes, column sep=1ex] {
			\toprule
			v_1 & v_2 & v_3 & v_4 \\
			\midrule
			a & a & a & b \\
			c & b & b & a \\
			b & c & d & c \\
			d & d & c & d \\			
			\bottomrule \\
		};
		
		\begin{tikzbackground}
		\draw[line width=7pt, red!20, draw opacity=0.5, transform canvas={yshift=1mm}] 
		(m-2-1.west) -- (m-2-2.center) -- (m-2-3.center) -- (m-3-4.center) -- (m-3-4.east);

		\draw[line width=7pt, red!20, draw opacity=0.5, transform canvas={yshift=1mm}] 
		(m-4-1.west) -- (m-4-1.center) -- (m-3-2.center) -- (m-3-3.center) --(m-2-4.center) -- (m-2-4.east);

		\draw[line width=7pt, red!20, draw opacity=0.5, transform canvas={yshift=1mm}] 
		(m-3-1.west) -- (m-3-1.center) -- (m-4-2.center) -- (m-5-3.center) -- (m-4-4.center) -- (m-4-4.east);

		\draw[line width=7pt, red!20, draw opacity=0.5, transform canvas={yshift=1mm}] 
		(m-5-1.west) -- (m-5-1.center) -- (m-5-2.center) -- (m-4-3.center) -- (m-5-4.center) -- (m-5-4.east);
		\end{tikzbackground}
		\end{tikzpicture}
	\end{wrapfigure}
	It can be checked that this profile is not single-crossing. However, it is 
	single-crossing on the following tree:
	
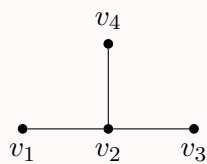
\begin{center}
	\begin{tikzpicture}
		\node [vertex, label={below:$v_1$}] (1) {};
		\node [vertex, label={below:$v_2$}, right=of 1] (2) {};
		\node [vertex, label={below:$v_3$}, right=of 2] (3) {};
		\node [vertex, label={above:$v_4$}, above=of 2] (4) {};
		\draw (1) edge (2) (2) edge (3) edge (4);
	\end{tikzpicture}
	\end{center}
	Indeed, for each pair of alternatives, the relevant sets of voters are all connected:
	\begin{center}
		\begin{tikzpicture}
			\node [vertex, label={below:$v_1$}] (1) {};
			\node [vertex, label={below:$v_2$}, right=of 1] (2) {};
			\node [vertex, label={below:$v_3$}, right=of 2] (3) {};
			\node [vertex, label={above:$v_4$}, above=of 2] (4) {};
			\draw (1) edge (2) (2) edge (3) edge (4);
			
			\begin{tikzbackground}
			\begin{scope}[every pin/.style={transform shape,scale=0.6,inner sep=2pt},every pin edge/.style={thin}]
				\node [fit={(1) (2)}, inner sep=2pt, draw=none, fill=red!20, rounded corners]  {};
				
				\node [fit={(2) (3)}, pin={above right:$a \succ b$}, inner sep=2pt, draw=none, fill=red!20, rounded corners]  {};
				\node [fit={(4)}, pin={220:$b \succ a$}, inner sep=2pt, draw=none, fill=red!20, rounded corners]  {};
			\end{scope}
			\end{tikzbackground}
		\end{tikzpicture}
		\quad
		\begin{tikzpicture}
			\node [vertex, label={below:$v_1$}] (1) {};
			\node [vertex, label={below:$v_2$}, right=of 1] (2) {};
			\node [vertex, label={below:$v_3$}, right=of 2] (3) {};
			\node [vertex, label={above:$v_4$}, above=of 2] (4) {};
			\draw (1) edge (2) (2) edge (3) edge (4);
			
			\begin{tikzbackground}
			\begin{scope}[every pin/.style={transform shape,scale=0.6,inner sep=2pt},every pin edge/.style={thin}]
				\node [fit={(2) (4)}, inner sep=2pt, draw=none, fill=red!20, rounded corners]  {};
				
				\node [fit={(2) (3)}, pin={above right:$b \succ c$}, inner sep=2pt, draw=none, fill=red!20, rounded corners]  {};
				\node [fit={(1)}, pin={50:$c \succ b$}, inner sep=2pt, draw=none, fill=red!20, rounded corners]  {};
			\end{scope}
			\end{tikzbackground}
		\end{tikzpicture}
		\quad
		\begin{tikzpicture}
			\node [vertex, label={below:$v_1$}] (1) {};
			\node [vertex, label={below:$v_2$}, right=of 1] (2) {};
			\node [vertex, label={below:$v_3$}, right=of 2] (3) {};
			\node [vertex, label={above:$v_4$}, above=of 2] (4) {};
			\draw (1) edge (2) (2) edge (3) edge (4);
			
			\begin{tikzbackground}
			\begin{scope}[every pin/.style={transform shape,scale=0.6,inner sep=2pt},every pin edge/.style={thin}]
				\node [fit={(1) (2)}, inner sep=2pt, draw=none, fill=red!20, rounded corners]  {};
				
				\node [fit={(3)}, pin={100:$d \succ c$}, inner sep=2pt, draw=none, fill=red!20, rounded corners]  {};
				\node [fit={(2) (4)}, pin={150:$c \succ d$}, inner sep=2pt, draw=none, fill=red!20, rounded corners]  {};
			\end{scope}
			\end{tikzbackground}
		\end{tikzpicture}
	\end{center}
\end{examplebox}

The following observation shows that
a profile is single-crossing on a tree if and only if it is single-crossing
on every path in that tree; it can be seen as a partial analog
of \Cref{prop:spt-equiv}.

\begin{proposition}
	\label{prop:sct-equiv}
A profile $P$ is single-crossing on a tree $T$ if and only if for every
path $Q$ in $T$ the restriction of $P$ to $Q$ is single-crossing.
\end{proposition}
\begin{proof}
Suppose that for some pair of alternatives $a, b\in A$ the set of voters
who prefer $a$ to $b$ in $T$ is not connected. Let $i$ and $j$ be two such voters, 
and let $Q$ be the (unique) path in $T$ that connects them. Then some voter on $Q$
prefers $j$ to $i$ and hence the restriction of $P$ to $Q$ is not single-crossing.

Conversely, consider a path $Q$ in $T$ such that the restriction of $P$ to that
path is not single-crossing. Then there are two voters $i$ and $j$ on $Q$
and a pair of alternatives $a, b\in A$ such that $a\succ_i b$, $a\succ_j b$,
but there is a voter $k$ that lies between $i$ and $j$ on $Q$ such that
$b\succ_k a$. Since $T$ is a tree, one can only travel from $i$ to $j$
in $T$ along $Q$, so the set of voters who prefer $a$ to $b$ is not connected in $T$.
\end{proof}

\paragraph{Majority Relation}
For single-peakedness, we saw that, in generalizing from paths to trees, we lost transitivity of 
the majority relation, retaining only the existence of Condorcet winners. 
For single-crossingness, on the other hand, the move to trees preserves the transitivity guarantee. 
In fact, the Representative Voter Theorem still holds; in \Cref{ex:sct-but-not-sc}, 
voter $v_2$ is the representative voter.

\begin{theorem}
	[\citealp{clearwater2014single}]
	Let $P$ be a profile with an odd number of voters that is single-crossing on a tree 
	$T = (N,E)$. Then the majority relation of $P$ is transitive, and $P$ contains a vote
	that coincides with the majority relation.
\end{theorem}
\begin{proof}
	For a pair $a,b\in A$ of alternatives, let $N_{ab} = \{ i \in N : a \succ_i b \}$. 
	Since $P$ is single-crossing on $T$, each set $N_{ab}$ is connected in $T$. 
	
	Write $\mathcal M = \{ N_{ab} : a,b\in A \text{ such that } a \succ_{\text{maj}} b \}$. 
	Note that any two sets $N_{ab}, N_{cd} \in \mathcal M$ must have non-empty intersection, 
	since any two strict majorities intersect. Hence, by the Helly property of connected 
	subsets of a tree \citep[see, e.g.,][]{golumbic2004algorithmic}, 
	we must have $\bigcap_{a \succ_{\text{maj}} b} N_{ab} \neq \varnothing$. 
	The voters in this intersection are representative voters.
\end{proof}

\paragraph{Hereditariness}
It is easy to see that a profile single-crossing on a tree remains single-crossing
on the same tree if we delete alternatives. 
However, the domain of preferences single-crossing on trees 
is not closed under deleting voters. In fact, \citet{clearwater2014single} 
show that if $P$ is a profile single-crossing on a tree, then all subprofiles of $P$ are 
single-crossing on a tree if and only if $P$ is in fact single-crossing on a path.

\begin{examplebox}
	{The domain of preferences single-crossing on trees is not closed under voter deletion}
	{sct-voter-del}
	\begin{wrapfigure}[8]{l}{0.2\linewidth}
		\begin{tikzpicture}
		\matrix (m) [matrix of math nodes, column sep=1ex] {
			\toprule
			v_1 & v_2 & v_3 & v_4 \\
			\midrule
			a & a & a & b \\
			c & b & b & a \\
			b & c & d & c \\
			d & d & c & d \\			
			\bottomrule \\
		};
		
		\begin{tikzbackground}
		\draw[line width=7pt, red!20, draw opacity=0.5, transform canvas={yshift=1mm}] 
		(m-2-1.west) -- (m-2-2.center) -- (m-2-3.center) -- (m-3-4.center) -- (m-3-4.east);

		\draw[line width=7pt, red!20, draw opacity=0.5, transform canvas={yshift=1mm}] 
		(m-4-1.west) -- (m-4-1.center) -- (m-3-2.center) -- (m-3-3.center) --(m-2-4.center) -- (m-2-4.east);

		\draw[line width=7pt, red!20, draw opacity=0.5, transform canvas={yshift=1mm}] 
		(m-3-1.west) -- (m-3-1.center) -- (m-4-2.center) -- (m-5-3.center) -- (m-4-4.center) -- (m-4-4.east);

		\draw[line width=7pt, red!20, draw opacity=0.5, transform canvas={yshift=1mm}] 
		(m-5-1.west) -- (m-5-1.center) -- (m-5-2.center) -- (m-4-3.center) -- (m-5-4.center) -- (m-5-4.east);
		\end{tikzbackground}
		\end{tikzpicture}
	\end{wrapfigure}
	Revisiting the profile in \Cref{ex:sct-but-not-sc}, we see that this profile 
	is no longer single-crossing on a tree if we delete the central voter $v_2$.
	To see this, notice that a three-voter profile is single-crossing on a tree 
	if and only if it is single-crossing (because a tree with three vertices must be a path). Can we order $v_1$, $v_3$ and $v_4$
	so that the resulting profile is single-crossing in the given order?
	In particular, which voter can be placed second? 
	Not $v_1$ (because of $b, c$), 
	not $v_3$ (because of $c, d$), and not $v_4$ (because of $a, b$), 
	so the answer is `no'.
	We conclude that $(v_1,v_3,v_4)$ is not single-crossing (on any tree) 
	and hence the domain of preferences single-crossing on trees 
	is not closed under voter deletion.
\end{examplebox}

\subsection{Single-Peaked and Single-Crossing (SPSC) Preferences}
\label{sec:def:spsc}

Every 1-Euclidean profile is single-peaked and single-crossing, but the converse is not true  when there are 6 or more alternatives (\Cref{ex:spsc-not-euclid}, \citealp{chen2021small}).
Even though the 1-Euclidean domain does not coincide with the domain of preference profiles 
that are both single-peaked and single-crossing (SPSC for short), we can still view SPSC as a 
`combinatorial'  version of 1-Euclidean preferences. \citet{elkind2020spsc} undertook 
the task of understanding more closely how the two concepts interact. A key role 
in their approach is played by so-called narcissistic profiles: A preference profile $P$ 
over $A$ is \emph{narcissistic} if for every $c\in A$, there is a voter $i \in N$ 
with $\top(i) = c$. The term `narcissistic' is chosen because such profiles arise in situations 
where voters and alternatives coincide, and each voter ranks herself first; some authors
use the term `minimally rich' instead. 
We start with an interesting observation on the impact of narcissism on single-crossing profiles, 
which is discussed by \citet{puppe2016single} and implicit in 
\citet[Theorem~1]{barbera2011top}.

\begin{proposition}
	\label{prop:nsc-to-sp}
	A narcissistic single-crossing profile $P = (v_1,\dots,v_n)$ is single-peaked 
	with respect to the axis ${\lhd} = {\succ_1}$ given by the preference order of the first voter.
\end{proposition}
\begin{proof}
	Suppose for the sake of contradiction that 
	there are alternatives $a \succ_1 b \succ_1 c$ with respect to which 
	voter $j>1$ has a valley, i.e., $a \succ_j b$ and $c\succ_j b$. 
	Because $P$ is narcissistic, there is some voter $i$ with $\top(i) = b$, 
	and hence $b \succ_i c$. Since $P$ is single-crossing we must have $i < j$. 
	But also $b \succ_i a$, and so, since $P$ is single-crossing, we must have $i > j$. 
	This is a contradiction.
\end{proof}

However, the converse is not true: clearly, there are SPSC profiles that are not narcissistic.
Indeed, we can start with a narcissistic SPSC profile and delete all voters whose
top choice belongs to some non-empty subset of $A$. 
The resulting profile is not narcissistic, but it remains 
SPSC, because both the single-peaked domain and the single-crossing domain 
are closed under voter deletion. It turns out that all SPSC profiles
can be generated in this way.

\begin{theorem}[\citealp{elkind2020spsc}, Theorem 9]
	A preference profile is both single-peaked and single-crossing if and only if 
	it can be obtained from a narcissistic single-crossing profile by deleting voters.
\end{theorem}
\begin{proof}[Proof sketch]
	We already know one direction from \Cref{prop:nsc-to-sp}: every profile obtained 
	by deleting voters from a narcissistic single-crossing profile is SPSC. 
	The other direction can be shown by a constructive argument: given an SPSC profile, 
	we can extend it to a narcissistic one, while preserving the single-crossing property.
\end{proof}

\citet[Theorem~10]{elkind2020spsc} also give an $O(m^2n)$ time algorithm that checks whether 
a given profile $P$ is single-peaked and single-crossing, and, if so, returns a narcissistic 
single-crossing profile $P' \supseteq P$.

As an alternative characterization, \citet[Corollary~12]{elkind2014characterization} 
show that a profile $P$ is SPSC if and only if for every alternative $c\in A$, 
there is a vote $v_i$ with $\top(v_i) = c$ such that $P + v_i$ is single-crossing 
(for an appropriate reordering of the voters).
Finally, \citet[Proposition~11]{elkind2020spsc} show that in an SPSC profile written 
in a single-crossing order, the `trajectories' of all alternatives are single-peaked, 
in the sense that if $a=\text{top}(v_k)$ for some $k\in [n]$ then 
for all $i, j\in [n]$  such that $1\le i< j\le k$ or
$k\le j< i\le n$ we have $\rank_i(a)\ge \rank_j(a)$.

\subsection{Multidimensional Single-Peaked Preferences}
\label{sec:def:multidim-sp}

\newcommand{\multisp}{\mathbin{\vec{\lhd}}}

Single-peaked preferences are described in terms of ordering the alternatives 
along a single dimension---the axis. 
It is natural to ask whether this concept can be extended to higher dimensions
in a way that preserves some of its desirable features.
\citet{barbera1993generalized} were the first to introduce a multidimensional analog 
of the single-peaked domain.
Their definition of multidimensional single-peaked preferences is based on the assumption 
that alternatives are points on a grid and every grid point is an alternative; 
thus, the alternative set has a very specific structure.
\citet{sui2013multi} define $d$-dimensional single-peaked preferences 
for a general model with any alternative set.
Here, we consider their definition and a slight variant.
The domain defined by \citet{sui2013multi}, which corresponds to condition (MD1) 
in the following definition, turns out not to be closed under alternative deletion. 
Our variant (MD2) is hereditary.

\begin{definition}
	Let $P$ be a profile over $A$ and let $\multisp = (\lhd_1,\ldots,\lhd_d)$ 
	be a $d$-tuple of linear orders over~$A$.
	For three alternatives $a,b,c$, we write $a \multisp b \multisp c$ if $b$ 
	is between $a$ and $c$ on every axis, i.e., for every $k\in[d]$ 
	it holds that either $a\lhd_k b\lhd_k c$ or $c\lhd_k b\lhd_k a$.
	\begin{enumerate}[(MD1)]
		\item Vote $v_i$ over $A$ is \defemph{$d$-dimensional single-peaked with respect to $\multisp$} if 
		for every pair of alternatives $a,b\in A$ with $a\multisp b\multisp \top(v_i)$ 
		it holds that $b\succ_i a$.
		\item Vote $v_i$ over $A$ is \defemph{$d$-dimensional hereditary single-peaked 
		with respect to $\multisp$} if 
		for every triple of alternatives $a,b,c\in A$ with $a\multisp b\multisp c$ 
		we do not have both $a\succ_i b$ and $c\succ_i b$.
	\end{enumerate}
	A profile $P$ over $A$ is \emph{$d$-dimensional (hereditary) single-peaked with respect to $\multisp$} 
	if every vote in $P$ is $d$-dimensional (hereditary) single-peaked with respect to $\multisp$, and 
	it is \emph{$d$-dimensional (hereditary) single-peaked} if there exists a $d$-tuple $\multisp$ 
	of linear orders over $A$ such that $P$ is $d$-dimensional (hereditary) single-peaked 
	with respect to $\multisp$.
	\label{def:multidim-sp}
\end{definition}

It is easy to see that if $P$ satisfies (MD2) then it also satisfies (MD1): 
the constraint in (MD1) is obtained by reformulating the constraint in (MD2)
for triples of the form $(a, b, \top(v_i))$ only.

By \Cref{prop:sp-equiv}, in the one-dimensional case the conditions (MD1) 
and (MD2) are equivalent, and indeed single-peakedness (according to \Cref{def:sp}) is 
the same as one-dimensional (hereditary) single-peakedness.

\begin{examplebox}
	{Preferences that are two-dimensional single-peaked}
	{md-sp}
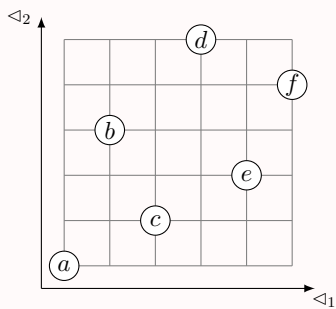
	\begin{wrapfigure}[9]{l}{0.29\linewidth}
		\begin{tikzpicture}
		[scale=0.6,
		alternative/.style={circle,draw,fill=white,inner sep=0pt,minimum size=13pt,scale=0.85}]
		\foreach \y in {0,1,2,3,4,5}
		\draw[gray] (0,\y) -- (5,\y);
		\foreach \x in {0,1,2,3,4,5}
		\draw[gray] (\x,0) -- (\x,5);
		
		\draw[-latex] (-0.5,-0.5) -- (5.5,-0.5);
		\draw[-latex] (-0.5,-0.5) -- (-0.5,5.5);
		\node[scale=0.7] at (5.7,-0.75) (lhd1) {$\lhd_1$};
		\node[scale=0.7] at (-0.99,5.5) (lhd2) {$\lhd_2$};
		
		\node[alternative] at (0,0) (a) {$a$};
		\node[alternative] at (1,3) (b) {$b$};
		\node[alternative] at (2,1) (c) {$c$};
		\node[alternative] at (3,5) (d) {$d$};
		\node[alternative] at (4,2) (e) {$e$};
		\node[alternative] at (5,4) (f) {$f$};	
		\end{tikzpicture}
	\end{wrapfigure}
	Take axes $a\lhd_1 b\lhd_1 c\lhd_1 d\lhd_1 e\lhd_1 f$ and $a\lhd_2 c\lhd_2 e\lhd_2 b\lhd_2 f\lhd_2 d$.
	Then the alternatives can be displayed on a grid as shown. 
	The order $a \succ c \succ e \succ b \succ d \succ f$ satisfies both (MD1) and (MD2). 
	The order $d \succ f \succ c \succ e \succ b \succ a$ satisfies (MD1), but it fails (MD2) 
	because $c \multisp e \multisp f$ yet $e$ is ranked below both $c$ and $f$. 
	(To convince yourself that $e$ is between $c$ and $f$ in the sense that 
	$c \multisp e \multisp f$, note that $e$ is in the box spanned 
	by corner points $c$ and $f$). 
	Finally, the order $a \succ b \succ e \succ d \succ c \succ f$ fails (MD1) 
	(and hence also (MD2)), because $a \multisp c \multisp e$ yet $c$ is ranked below 
	both $a = \top(\succ)$ and $e$.
\end{examplebox}

It is immediate that (MD2) is hereditary: removing an alternative $x$ does not change the set of 
triples of alternatives $a, b, c\in A\setminus\{x\}$ such that $a\multisp b \multisp c$,
and for each such triple the voters' preferences do not change either. However, for (MD1) this argument
does not go through: if we remove an alternative $x$ such that $x=\top(v_i)$ for some voter $i$,
in the new profile the constraint in (MD1) should now apply to triples of alternatives
involving $i$'s top choice in $A\setminus\{x\}$. Indeed,   
the following example shows that (MD1) is not hereditary for $d \ge 2$, 
and thus (MD1) and (MD2) are not equivalent for $d \ge 2$.

\begin{examplebox}
	{The two-dimensional single-peaked domain defined by (MD1) is not hereditary}
	{md1-not-hereditary}
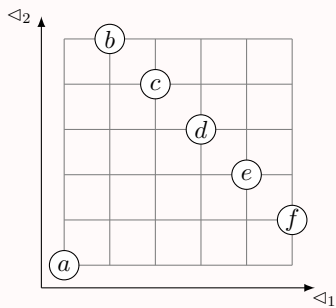
	\begin{wrapfigure}[9]{l}{0.29\linewidth}
		\begin{tikzpicture}
		[scale=0.6,
		alternative/.style={circle,draw,fill=white,inner sep=0pt,minimum size=13pt,scale=0.85}]
		\foreach \y in {0,1,2,3,4,5}
		\draw[gray] (0,\y) -- (5,\y);
		\foreach \x in {0,1,2,3,4,5}
		\draw[gray] (\x,0) -- (\x,5);
		
		\draw[-latex] (-0.5,-0.5) -- (5.5,-0.5);
		\draw[-latex] (-0.5,-0.5) -- (-0.5,5.5);
		\node[scale=0.7] at (5.7,-0.75) (lhd1) {$\lhd_1$};
		\node[scale=0.7] at (-0.99,5.5) (lhd2) {$\lhd_2$};
		
		\node[alternative] at (0,0) (a) {$a$};
		\node[alternative] at (1,5) (b) {$b$};
		\node[alternative] at (2,4) (c) {$c$};
		\node[alternative] at (3,3) (d) {$d$};
		\node[alternative] at (4,2) (e) {$e$};
		\node[alternative] at (5,1) (f) {$f$};	
		\end{tikzpicture}
	\end{wrapfigure}
	Consider the profile $P$ over $A = \{a,b,c,d,e,f\}$ that contains all $5!$ votes 
	that rank $a$ first.
	Consider axes $a\lhd_1 b\lhd_1 c\lhd_1 d\lhd_1 e\lhd_1 f$ and $a\lhd_2 f\lhd_2 e\lhd_2 d\lhd_2 c\lhd_2 b$.
	Then $P$ satisfies (MD1), since for all $x, y\in A\setminus\{a\}$ 
	it is not the case that $x\multisp y\multisp a$.
	Hence, this profile is two-dimensional single-peaked.
	However, if we remove candidate $a$ from the profile, the resulting profile 
	$P|_{A\setminus\{a\}}$ is no longer two-dimensional single-peaked.
	Indeed, this profile contains all possible orderings of alternatives 
	in $A\setminus\{a\}$. Hence, for it to be two-dimensional single-peaked, 
	there has to exist a 2-tuple of linear orders $\multisp=(\lhd_1, \lhd_2)$ 
	such that there is no triple of alternatives $x, y, z\in A\setminus\{a\}$ 
	with $x\multisp y\multisp z$.
	However, it can be shown that for five alternatives no such 2-tuple exists.
	(That the restricted profile violates (MD1) also follows from the tightness claim 
	of \Cref{prop:multisp-triviality}.)
\end{examplebox}

We observe that if a profile is $d$-dimensional (hereditary) single-peaked, 
then it is also $d'$-dimensional (hereditary) single-peaked for $d'>d$.
Indeed, if a $d'$-tuple $\multisp'=(\lhd_1, \dots, \lhd_d, \allowbreak \lhd_{d+1}, \dots, \lhd_{d'})$ is obtained
from a $d$-tuple $\multisp = (\lhd_1, \dots, \lhd_d)$ by adding additional axes, 
then $a\multisp' b\multisp' c$ implies $a\multisp b\multisp c$, i.e., 
adding axes can only reduce the number of triples of alternatives 
to which the constraints in (MD1) and (MD2) apply, making these constraints
easier to satisfy. This suggests that, as $d$ grows, the domain of $d$-dimensional 
single-peaked profiles becomes quite large. The following result confirms this intuition: it shows that 
every profile is $d$-dimensional (hereditary) single-peaked for sufficiently large $d$.%
\footnote{We thank Jan Kyncl for pointing us to the results used in the proof.}

\begin{theorem}
Let $d\ge 1$.
All profiles with at most $2^{2^{d-1}}$ alternatives are 
$d$-dimensional (hereditary) single-peaked. 
This bound is tight:
The profile containing all linear orders over $2^{2^{d-1}}+1$ alternatives 
is not $d$-dimensional (hereditary) single-peaked.%
\label{prop:multisp-triviality}
\end{theorem}
\begin{proof}
Let $d\ge 1$ be fixed.
We rely on a result by de Bruijn (published by \citet{kruskal1953monotonic}, 
with further proof details provided by \citet{alon1985separating}), stated here in a restricted form.
We say that a sequence $(x_1,\dots,x_m)$ with $x_j\in \mathbb{R}^d$, $j\in[m]$, 
is \emph{monotonic} if it is component-wise monotonic, i.e., 
for all $k\in[d]$ it holds that $(x_{1,k},\dots,x_{m,k})$ is 
either weakly increasing or weakly decreasing.
The result states that every sequence $(x_1,\dots,x_m)$ with $x_j\in \mathbb{R}^d$ and 
$m\ge 2^{2^{d}}+1$ contains a monotonic subsequence of length 3. 
On the other hand, there exists a sequence of length $2^{2^{d}}$ that 
does not contain monotonic subsequences of length 3; \citet{alon1985separating} 
explicitly construct a sequence with this property that additionally satisfies
$x_{i, k}\neq x_{j, k}$ for all $k\in [d]$ and all $i, j\in [m]$ with $i\neq j$.

Let $m\le 2^{2^{d-1}}$ and $A=[m]$.
To prove that all profiles with $m$ alternatives are $d$-dimensional (hereditary) single-peaked, 
we will construct a $d$-tuple of linear orders $\multisp=(\lhd_1,\dots,\lhd_d)$ such that
there is no triple of alternatives $a, b, c\in A$ with $a\multisp b\multisp c$; then the constraints
in (MD1) and (MD2) are trivially satisfied.

To this end, let $(x_1,\dots,x_m)$ be a sequence with $x_j\in \mathbb{R}^{d-1}$ 
that does not contain monotonic subsequences of length 3 
and such that for each $k\in [d-1]$ all elements of $X_k=\{x_{1, k}, \dots, x_{m, k}\}$
are pairwise distinct.
We transform this sequence into a $d$-tuple $\multisp$ as follows. 
First, we take $\lhd_1$ to be $1\lhd 2 \lhd\dots\lhd m$.
Then, for each $k\in[d-1]$, we set $i\lhd_{k+1} j$ if $x_{i, k} < x_{j, k}$.
Note that $i\lhd_{k+1} j\lhd_{k+1} \ell$ if the sequence $(x_{i, k}, x_{j, k}, x_{\ell, k})$
is increasing, and
	  $\ell\lhd_{k+1} j\lhd_{k+1} i$ if the sequence $(x_{i, k}, x_{j, k}, x_{\ell, k})$
is decreasing.
Let $\multisp=(\lhd_1,\dots,\lhd_d)$. 
Then a triple of alternatives $i,j,\ell\in [m]$ satisfies $i\multisp j\multisp \ell$
if and only if $(x_i, x_j, x_\ell)$ is a monotone subsequence of length three 
in the input sequence, and such a subsequence does not exist.

Conversely, suppose that $m>2^{2^{d-1}}$\!, and let $P$ be the profile containing all possible votes 
over $[m]$. Furthermore, let $\multisp=(\lhd_1,\dots,\lhd_d)$ be an arbitrary $d$-tuple of linear orders 
over $[m]$; by renaming the alternatives, we can assume that $1\lhd_1\dots\lhd_1 m$.
We are going to identify a triple $i, j, \ell\in [m]$ with $i\multisp j\multisp \ell$; 
this shows that $P$ is not $d$-dimensional (hereditary) single-peaked with respect to $\multisp$, 
because $P$ contains a vote of the form $i\succ\dots\succ j$.  
For this purpose we construct a sequence $(x_1,\dots,x_m)$ of vectors in $[m]^{d-1}$ as follows.
For each $k\in [d]$, we set $x_{i, k}$ to be the rank of $i$ in $\lhd_{k+1}$.
By de Bruijn's theorem, the sequence $(x_1, \dots, x_m)$
contains a monotone subsequence $(x_i, x_j, x_\ell)$.
Further, we have $i\lhd_1 j\lhd_1 \ell$, and for each $k\in [d-1]$
we have  $x_{i, k}<x_{j, k}<x_{\ell, k}$ if and only if $i\lhd_{k+1} j\lhd_{k+1} \ell$ and
                  $x_{\ell, k}<x_{j, k}<x_{i, k}$ if and only if $\ell\lhd_{k+1} j\lhd_{k+1} i$. 
Thus, we have $i\multisp j \multisp \ell$.
This concludes the proof.
\end{proof}

Consequently, all profiles with $4$ alternatives are 2-dimensional single-peaked, 
all profiles with up to $16$ alternatives are 3-dimensional single-peaked, etc.
This immediately implies that many nice properties of the 1-dimensional 
single-peaked domain do not extend to two or more dimensions.
For example, since all profiles on 3 alternatives are 2-dimensional single-peaked, 
so is a profile that induces a cyclic majority relation, and hence transitivity 
of the majority relation cannot be guaranteed.
The result might also partially explain the empirical findings of \citet{sui2013multi}, 
who show that in real-world election datasets, one can find 2-dimensional axes on which 
47--65\% voters are single-peaked (in the (MD1) sense); 
2-dimensional single-peakedness is perhaps a less demanding condition than one might expect.

\paragraph{$d$-Euclidean Preferences and $d$-Dimensional Single-Peaked Preferences}
We know that every $1$-Euclidean profile is single-peaked. It is thus natural to ask
if this relationship holds in higher dimensions. 

Suppose first that we have a $d$-Euclidean profile $P$ where the set of voters
coincides with the set of candidates. Then the answer to our question is `yes'.
Indeed, we can define the $d$ axes $\lhd_1, \dots, \lhd_d$
in a natural way: given two candidates $a$ and $b$, located at $(x_1, \dots, x_d)$
and $(y_1, \dots, y_d)$, respectively, 
for each $k\in [d]$ we set $a\lhd_k b$ if and only if $x_k \le y_k$.
Let $\multisp=(\lhd_1, \dots, \lhd_d)$.
Denote the vote of candidate $c$ by $\succ_c$. We have $c\succ_c a$
for all $a\in A\setminus\{c\}$. Further, $c\multisp a\multisp b$
implies that for each $k\in [d]$ the $k$-th coordinate of $a$ lies between
the $k$-th coordinate of $c$ and the $k$-th coordinate of $b$, 
which means that $c$ prefers $a$ to $b$. Thus, $P$ is $d$-dimensional
single-peaked. 

\begin{figure}
	\centering
	\begin{tikzpicture}
		[scale=0.3, alt/.style={circle, fill=black, inner sep=1.3pt}]
		\draw[thick, ->] (0,-0.3) -- (0,12); 
		\draw[thick, ->] (-0.5,0) -- (12,0);
		\node [alt, label=0:{$a$}] (a) at (1,1) {};
		\node [alt, label=0:{$b$}] (b) at (10,2) {};
		\node [alt, label=0:{$c$}] (c) at (11,11) {};
		\node [label={[label distance=0mm]0:voter $i$}] (v1) at (1,10) 
		{\color{red!80!black}\Large\Gentsroom};
	\end{tikzpicture}
	\caption{A difficulty in translating $2$-Euclidean profiles to 2-dimensional single-peaked profiles.}\label{fig:2d-sp-to-euclid}
\end{figure}

However, if the set of voters may differ from the set of candidates, 
this construction no longer works. Indeed, suppose that $d=2$, 
candidate $a$ is located at $(0, 0)$, candidate $b$ is located at $(9, 1)$,
and candidate $c$ is located at $(10, 10)$ (compare \Cref{fig:2d-sp-to-euclid}). If we define axes
$\lhd_1$ and $\lhd_2$ as we did in the previous paragraph, we obtain
$a\multisp b\multisp c$. However, a voter $i$ with $2$-Euclidean
preferences who is located at $(0, 9)$ has $a$
as her top choice and prefers $c$ to $b$, so her preferences
are not two-dimensional single-peaked with respect to $\multisp$.
Moreover, if this voter were to join the election as a candidate, 
her preferences would not be two-dimensional hereditary single-peaked with respect
to $\multisp$.

Of course, the example in the previous paragraph only shows that
voter $i$'s preferences fail to be two-dimensional single-peaked
for a specific choice of axes; clearly,  a single vote over three
alternatives is always single-peaked with respect to a suitable axis.
However, it is not clear whether every $d$-dimensional Euclidean
profile admits a suitable collection of $d$ axes. 

\subsection{Value Restriction}

Since \citet{black1948rationale} and \citet{arr:b:sc} introduced the single-peaked condition, 
social choice theorists looked for less restrictive conditions that still entailed  
transitivity of the majority relation. \citet{sen1966possibility} proposed value 
restriction, and \citet{sen1969necessary} showed that value restriction is, in a sense we will 
discuss below, the largest preference domain that guarantees transitivity.

\begin{definition}
	A profile $P$ over $A$ is \defemph{value-restricted} if for every triple $a,b,c\in A$ 
	of alternatives, there is an element in $\{a,b,c\}$ that among the alternatives 
	$\{a,b,c\}$ is either never ranked first, or never ranked second, or never ranked third 
	in any vote $v_i \in P$.
\end{definition}

This definition can be equivalently phrased in terms of a \emph{forbidden subprofile}: a profile 
is value-restricted if and only if it does not contain the Condorcet profile as a subprofile. In 
\Cref{sec:subprofiles} we will see similar characterizations of other domain 
restrictions.

\begin{proposition}\label{prop:vr-cycle}
	A profile $P$ over $A$ is value-restricted if and only if there do not exist alternatives 
	$a,b,c\in A$ and voters $i,j,k\in N$ such that the restriction of $P$ to these voters
	and alternatives forms a Condorcet profile, i.e., 
	
	\begin{minipage}{0.7\textwidth}
		\begin{itemize}
			\item $a \succ_i b \succ_i c$,
			\item $b \succ_j c \succ_j a$,
			\item $c \succ_k a \succ_k b$.
		\end{itemize}
	\end{minipage}
	\hfill
	\begin{minipage}{0.2\textwidth}
	\begin{flushright}                                      
	\begin{tabular}{ccc}
	\toprule
	$v_i$ & $v_j$ & $v_k$ \\
	\midrule
	$a$ & $b$ & $c$ \\
	$b$ & $c$ & $a$ \\
	$c$ & $a$ & $b$ \\
	\bottomrule
	\end{tabular}
	\end{flushright}
	\end{minipage}
\end{proposition}
\begin{proof}
If there is a triple of alternatives $a, b, c$ and a triple of voters $i, j, k$
such that the restriction of $P$ to these voters and alternatives forms a Condorcet
profile, then clearly $P$ is not value-restricted. 

Conversely, suppose $P$ is not value-restricted, as witnessed by some triple
$a, b, c\in A$, and let $P'$ be the restriction of $P$ to $\{a, b, c\}$.
If $P'$ contains all six different linear orders over $\{a, b, c\}$ then
in particular it contains a Condorcet profile. Thus, suppose some order, say, 
$a\succ b\succ c$, does not appear in $P'$. As $P'$ is not value-restricted, 
it has to contain a linear order where $a$ is ranked first, and it can only be 
$a\succ c\succ b$. Also, it has to contain a linear order where $b$ is ranked
second, and it can only be $c\succ b\succ a$. Finally, it has to contain
a linear order where $c$ is ranked last, and it can only be $b\succ a\succ c$.
But then $a\succ c\succ b$, $c\succ b\succ a$ and $b\succ a\succ c$ form a Condorcet profile.
\end{proof}

\paragraph{Majority Relation}
The defining feature of value-restricted profiles is that, because they exclude Condorcet 
subprofiles, they avoid cycles in the majority relation, as we will now formally show.

\begin{proposition}
	\label{prop:vr-transitive}
	The majority relation of a value-restricted profile with an odd number of voters is transitive.
\end{proposition}
\begin{proof}
	Consider a profile $P$ with an odd number of voters 
	whose majority relation is not transitive, so that there are $a,b,c\in A$ with 
	$a \succ_{\text{maj}} b \succ_{\text{maj}} c \succ_{\text{maj}} a$.
	We will show that $P$ is not value-restricted.
	Indeed, for a majority of voters we have $a \succ b$, 
	and for a majority of voters we have $b \succ c$; 
	as any two strict majorities intersect, 
	there must be a voter $i$ with $a \succ_i b \succ_i c$. 
	Similarly, there are majorities with $b\succ c$ and with 
	$c \succ a$, so there is a voter $j$ with $b \succ_j c \succ_j a$. 
	And again, similarly, there is a voter $k$ with $c \succ_k a \succ_k b$. 
	Together, voters $i$, $j$, and $k$ witness that $P$ is not value-restricted.
\end{proof}

Fix a set $A$ of alternatives and a number $n$ of voters. A domain $\mathcal D$ of $n$-voter 
preference profiles over $A$ is said to have \emph{product structure} if there exists a subset 
$L$ of the set of linear orders over $A$ such that $\mathcal D = L^n$. Thus, a profile is part 
of $\mathcal D$ if and only if each voter chooses their preference order from the set $L$. 
For example, one such domain is the domain of all $n$-voter profiles where each vote
is single-peaked on a specific axis $\lhd$. Within the universe of 
domains with product structure, the value-restricted domain turns out to be 
the unique maximal domain that guarantees a transitive majority relation (for $n$ odd).

\begin{proposition}
	\label{prop:vr-maximal}
	Let $\mathcal D$ be a domain of $n$-voter profiles that has product structure,
	where $n\ge 3$ is odd. If every profile in $\mathcal D$ has 
	a transitive majority relation, then every profile in $\mathcal D$ is value-restricted.
\end{proposition}
\begin{proof}
	Since $\mathcal D$ has product structure, we can write $\mathcal D = L^n$ for some set 
	$L$ of linear orders. Suppose for the sake of contradiction that there is some profile 
	in $\mathcal D$ that is not value-restricted. Then, by \Cref{prop:vr-cycle}, 
	$L$ must contain three orders 
	${\succ_i},{\succ_j},{\succ_k} \in L$ such that $a \succ_i b \succ_i c$, 
	$b \succ_j c \succ_j a$, and $c \succ_k a \succ_k b$. 
	Since $n$ is odd and $n\ge 3$, we can write $n = 2s+3$ for some $s\ge 0$. 
	Now, consider the profile $P$ containing $s+1$ copies of $\succ_i$, 
	$s+1$ copies of $\succ_j$, and one copy of $\succ_k$. Then $P \in \mathcal D$, 
	but the majority relation of $P$ satisfies 
	$a \succ_{\text{maj}} b \succ_{\text{maj}} c \succ_{\text{maj}} a$, a contradiction.
\end{proof}

The result of \Cref{prop:vr-maximal} has led some authors to discount conditions such 
as single-peakedness as less interesting or useful, because value restriction allows strictly 
more profiles while retaining the guarantee of a transitive majority relation. This is a fair 
point if one is mainly interested in avoiding the impossibilities of Arrow and 
Gibbard--Satterthwaite, but value-restriction does not imply enough structure on the profile to 
be useful for many other applications, such as efficiently evaluating multi-winner rules.

It is easy to misinterpret \Cref{prop:vr-transitive,prop:vr-maximal} 
as saying that a profile $P$ is value-restricted if and only if it has a transitive majority 
relation. The following example shows that this is not true, and in particular the converse of 
\Cref{prop:vr-transitive} does not hold. However, if a profile $P$ has the property 
that its majority relation is transitive, and every way of changing the multiplicities (weights) 
of the preference orders occurring in $P$ leads to a profile whose majority relation is still 
transitive, then $P$ must be value-restricted by \Cref{prop:vr-maximal}.

\begin{examplebox}
	{A profile whose majority relation is transitive, but which is not value-restricted}
	{val-restr-transitive}
	\begin{minipage}{0.3\linewidth}
	\begin{tabular}{cccccc}
	\toprule
	$v_1$ & $v_2$ & $v_3$ & $v_4$ & $v_5$ & $v_6$ \\
	\midrule
	$a$ & $b$ & $c$ & $c$ & $c$ & $c$ \\
	$b$ & $c$ & $a$ & $a$ & $a$ & $a$ \\
	$c$ & $a$ & $b$ & $b$ & $b$ & $b$ \\
	\bottomrule
	\end{tabular}
	\end{minipage}
	\hfill
	\begin{minipage}{0.65\linewidth}
		There are profiles that are not value-restricted, but still have a transitive 
		majority relation: Here, the voters with $c \succ a \succ b$ `overwhelm' 
		the Condorcet cycle.
	\end{minipage}
\end{examplebox}

Three notions that are closely related to value restriction are sometimes discussed in the 
literature: best-restricted, medium-restricted, and worst-restricted profiles. Each of these 
conditions is stronger than being value-restricted, so trivially each 
best/medium/worst-restricted profile again admits a transitive majority relation.

\begin{definition}
A profile $P$ over $A$ is \defemph{best-/medium-/worst-restricted} if for every triple $a,b,c\in A$ 
of alternatives, there is an element in $\{a,b,c\}$ that is never ranked first/second/last, respectively,
in the restriction of $P$ to $\{a, b, c\}$. 
\end{definition}

Observe that, by the no-valley property, every single-peaked profile is worst-restricted. However, 
the converse is not true: the two-voter profile with votes $a\succ b\succ c\succ d$ and 
$d\succ b\succ c \succ a$ is worst-restricted (indeed, any two-voter profile is worst-restricted), 
but it is easy to see that it is not single-peaked. We will encounter this profile again
in \Cref{sec:subprofiles}, where we discuss forbidden subprofiles. Similarly, every
single-caved profile is best-restricted, but the converse is not true.

\begin{wrapfigure}[9]{r}{0.345\linewidth}
	\vspace{1pt}
	\begin{tabular}{crrr}
		\toprule
		$m$ & $f(m)$ & $f_{\text{SP}}(m)$ & $f_{\text{SC}}(m)$ \\
		\midrule
		$3$ & 4 & 4 & 4 \\
		$4$ & 9 & 8 & 7 \\
		$5$ & 20 & 16 & 11 \\
		$6$ & 45 & 32 & 16 \\
		$7$ & 100 & 64 & 22 \\
		$8$ & 224 & 128 & 29 \\
		\bottomrule
	\end{tabular}
\end{wrapfigure}
\paragraph{Counting} 
Value-restriction makes most sense in the context of domains $\mathcal D = L^n$ with 
product structure. In this case, sets $L$ that induce value-restricted profiles are known as \emph{Condorcet domains}, and they have been intensely studied; see \citet{puppe2024maximal} for a recent survey.
It is particularly interesting to ask how big the set $L$ can be: 
for a given $m$, what is the maximum number $f(m)$ of different orders 
that $L$ can contain without including the forbidden subprofile? 
This question is combinatorially challenging and has inspired much work 
\citep[see][for a survey]{monjardet2009acyclic}. It is easy to see that $f(3) = 4$: the 6 
possible linear orders over three profiles split into two Condorcet cycles, and we need to 
eliminate one order from each. It is also not hard to check that $f(4) = 9$. 
\citet{fishburn1996acyclic} showed that $f(5) = 20$ and later that $f(6) = 45$ 
\citep{fishburn2002acyclic}. Computing the values of $f(m)$ for $m\ge 7$ remains open. 
To prove the lower bounds, \citet{fishburn1996acyclic,fishburn2002acyclic} 
introduced two schemes that produce large $L$-sets: 
the \emph{alternating scheme} and the \emph{replacement scheme}. The alternating 
scheme produces the unique maximum value-restricted sets for $m = 5$ and $m = 6$, 
but for $m \ge 16$, the replacement scheme is better. 
Fishburn conjectured that the alternating scheme is also 
optimal for $m = 7$, which would imply that $f(7) = 100$ which was confirmed by \citet{galambos2008acyclic}. However, for $m = 8$, the alternating scheme has size $222$, while the largest domain contains $224$ orders \citep{leedham2024largest}. In the table, we show the known values 
of $f(m)$ and contrast with the maximum number $f_{\text{SP}}(m)=2^{m-1}$ of different orders in 
a single-peaked profile and with the maximum number $f_{\text{SC}}(m)={m\choose 2} + 1$ of 
different orders in a single-crossing profile.

\subsection{Group-Separable Preferences}

Group-separable preferences were introduced by \citet{inada1964note,inada1969simple}.
Intuitively, the voters' preferences are group-separable if every subset of at least
two alternatives can be split into two groups, so that each voter ranks one group
above the other (but voters may differ in which of the two groups they prefer).
This idea is formalized as follows.

\begin{definition}\label{def:gs}
A profile $P$ is \defemph{group separable} if for every set of alternatives $A'\subseteq A$ 
there is a proper subset $B\subset A'$ such that for every voter $i\in N$ we have either 
$B\succ_i (A'\setminus B)$ or $(A'\setminus B)\succ_i B$.
\end{definition}

\begin{examplebox}
	{Group-separable preferences.}
	{group-separable}
	\begin{wrapfigure}[9]{l}{0.46\linewidth}
		\begin{tabular}{cccc}
			\toprule
			$v_1$ 	& $v_2$ & $v_3$ & $v_4$ \\
			\midrule
			Austin  & Nice    & Calgary & Bergen \\
			Calgary & Madrid  & Austin  & Lisbon \\
			Bergen  & Lisbon  & Nice    & Madrid \\
			Nice	& Bergen  & Madrid  & Nice \\
			Madrid  & Austin  & Lisbon  & Calgary \\
			Lisbon  & Calgary & Bergen  & Austin \\
			\bottomrule
		\end{tabular}
	\end{wrapfigure}
	In this instance, when comparing cities, voters first consider the continent 
	they are located on, 
	with some voters preferring Europe to North America, and some voters preferring
	North America to Europe. Within Europe, the voters distinguish between Southern 
	Europe	(Nice, Lisbon, and Madrid) and Northern Europe (Bergen). Within 
	Northern America the voters distinguish between the USA and Canada, and within 
	Southern Europe they distinguish between France and the Iberic Peninsula;
	within the latter, they distinguish between Spain and Portugal. 
	
	Note that 
	Austin, Bergen, Calgary and Nice are weak Condorcet winners, 
	but Madrid and Lisbon lose to Nice in a pairwise comparison. 
\end{examplebox}

\paragraph{Hereditariness and Relationship to Other Domains}
It is immediate from the definition that the domain of group-separable
preferences is closed under voter and alternative deletion. Moreover,
group separability is related to another domain restriction we have discussed, 
namely the domain of medium-restricted preferences.
  
\begin{proposition}\label{prop:gs-medium}
Every group-separable profile is medium-restricted.
\end{proposition}
\begin{proof}
Suppose a profile $P$ fails to be medium-restricted, 
i.e., there is a triple of alternatives $a, b, c$ such that
in the restriction of $P$ to $A'=\{a, b, c\}$ each alternative
is ranked second in at least one vote. Now suppose 
for the sake of contradiction that $P$ is group-separable, 
and let $(B, A'\setminus B)$ be the associated partition of $A'$;
we can assume without loss of generality that $B=\{a, b\}$.
But then no voter can rank $c$ second among $a$, $b$, and $c$, 
a contradiction. 
\end{proof}
We note that the converse of \Cref{prop:gs-medium}
is not true: a two-voter profile $a\succ_1 b \succ_1 c\succ_1 d$, 
$c\succ_2 a\succ_2 d\succ_2 b$ is medium-restricted, 
but not group-separable. This profile will appear again
in \Cref{sec:subprofiles}, where we characterize
the group-separable domain in terms of forbidden minors.

\paragraph{Majority Relation}
Let us apply \Cref{def:gs} to the entire set of alternatives $A$: 
let $B$ be a proper subset of $A$ such that some voters rank $B$ above $A\setminus B$,
while other voters rank $A\setminus B$ over $B$. We can assume that a weak majority of voters
prefer $B$ to $A\setminus B$. Thus, if $B$ is a singleton, then the unique candidate in 
$B$ is a weak Condorcet winner, and otherwise we can recursively partition $B$
so as to eventually arrive at a weak Condorcet winner. As the domain of group-separable
preferences is hereditary, we can remove the weak Condorcet winner and apply the same
argument to the remaining alternatives; this establishes that the weak majority relation
of a group-separable profile is transitive.

\paragraph{Characterization via Clone Tree Decomposition}
Group-separable profiles are conveniently described in terms of \emph{clone sets}.
We say that a set of alternatives $C\subseteq A$ is a \emph{clone set with respect
to a profile $P$} if $1<|C|<|A|$ and for every voter $i$, every alternative
$b\in A\setminus C$ and every pair of alternatives $c, c'\in C$ it is not the case
that $c\succ_i b\succ_i c'$; that is, alternatives in $C$
form a contiguous block in each voter's ranking. Observe that 
a clone set may be a subset 
of another clone set: e.g., in \Cref{ex:group-separable} both 
\{Lisbon, Madrid\} and \{Lisbon, Madrid, Nice\} are clone sets.

The clone sets of a given profile can be succinctly described by 
its \emph{clone decomposition tree} \citep{elkind2012clone}.
To define this concept formally, we need the machinery
of PQ-trees \citep{booth1976testing}. A \emph{PQ-tree} over a set of alternatives $A$
is an ordered rooted tree whose leaves are labeled with elements of $A$
so that each label is used exactly once (i.e., there is a bijection between
$A$ and the leaves of $T$) 
and whose internal nodes belong to one of the two types
(P-nodes and Q-nodes). There are two types of permissible
operations on a PQ-tree: we can either reorder the children of 
a P-node in an arbitrary way, or flip the order of the children of a Q-node.
A PQ-tree $T$ 
\emph{generates} a ranking $r$ of $A$ if $r$ can be obtained by performing
zero or more permissible operations and then listing the leaves of the 
resulting ordered tree from left to right; let $R(T)$ denote the set
of all rankings generated by $T$. Given a node $x$ of a tree $T$ over $A$, 
we let $L(x)$ denote the set of labels of all leaves of $T$ that are descendants
of $x$. Note that for each node $x$ the alternatives in $L(x)$ 
form a contiguous block in every ranking in $R(T)$;
in fact, this is also the case for any set $L(x_i)\cup\dots\cup L(x_j)$,
where $x_i, \dots, x_j$ are consecutive children of some Q-node of $T$.

\begin{figure}[t]
	\centering
	\scalebox{0.9}{
		\begin{tikzpicture}
			[qnode/.style={rectangle,fill=red!50,draw},
			leaf/.style={anchor=base},
			node distance=0.8cm and 0.5cm]
			\node(A)[qnode,text width=10.5cm]{};
			\node(B)[qnode,below right=0.5cm and -6.85cm of A, text width=6.6cm]{};
			\node(NA)[qnode,below left=0.5cm and -3.05cm of A, text width=2.8cm]{};
			\node(C)[qnode,below right=0.5cm and -4.8cm of B, text width=4.5cm]{};
			\node(D)[qnode,below right=0.5cm and -3cm of C, text width=2.75cm]{};
			\node(7)[left=of C]{Bergen};
			\node(2)[left=of D]{Nice};
			\node(3)[leaf,below left=0.32cm and -1.3cm of D]{Madrid};
			\node(4)[leaf,below right=0.32cm and -1.2cm of D]{Lisbon};
			\node(Austin)[leaf,below left=0.32cm and -1.15cm of NA]{Austin};
			\node(Calgary)[leaf,below right=0.32cm and -1.35cm of NA]{Calgary};
			\draw[->](A.south)-|(NA.north);
			\draw[->](A.south)-|(B.north);
			\draw[->](NA.south)-|(Austin.north);
			\draw[->](NA.south)-|(Calgary.north);
			\draw[->](B.south)-|(7.north);
			\draw[->](B.south)-|(C.north);
			\draw[->](C.south)-|(2.north);
			\draw[->](C.south)-|(D.north);
			\draw[->](D.south)-|(3.north);
			\draw[->](D.south)-|(4.north);
		\end{tikzpicture}
	}
	\caption{The clone tree decomposition of \Cref{ex:group-separable}.}
	\label{fig:clone-tree}
\end{figure}

We say that a PQ-tree $T$ is a \emph{clone tree decomposition} 
for a profile $P$ over $A$ if for each $B\subseteq A$ it holds that $B$ is a clone set
with respect to $P$ if and only if $B=L(x)$ for some node $x$ of $T$
or $B=L(x_i)\cup\dots\cup L(x_j)$, where $x_i, \dots, x_j$
are consecutive children of some Q-node.
Each profile admits a unique clone decomposition tree, 
which can be computed in polynomial time \citep{elkind2012clone}.
The clone tree decomposition of the profile in \Cref{ex:group-separable} is shown in \Cref{fig:clone-tree}.

The definition of group-separable preferences can be rephrased 
in terms of clone sets: it says that for each $A'\subseteq A$
there is a proper subset $B\subset A'$ such that both $B$
and $A'\setminus B$ are clone sets with respect to $P|_{A'}$.
Building on this idea,
\citet{karpov:j:gs} provides a characterization of group-separable
profiles in terms of the properties of their clone decompositions trees:
specifically, he establishes that a profile is group-separable
if and only if all internal nodes of its clone decomposition tree 
are Q-nodes. 

\paragraph{Counting}
Recall that a single-peaked profile can contain at most $2^{m-1}$ distinct
rankings, whereas a single-crossing profile may contain at most ${m\choose 2}+1$
distinct rankings. Interestingly, the maximum number of distinct rankings in a group-separable 
profile over a set $A$ is the same as in a single-peaked profile over $A$.

\begin{proposition}\label{prop:gs-count}
A group-separable profile over $m$ alternatives may contain at most $2^{m-1}$
distinct rankings, and this bound is tight.
\end{proposition}
\begin{proof}
For the upper bound, we proceed by induction on $m$. For $m=2$, the profile that
consists of two distinct rankings is group-separable. Now, suppose our claim 
has been established for $m'<m$. Consider a group-separable 
profile $P$ over a set of alternatives $A$, $|A|=m$, such that all rankings in $P$
are pairwise distinct. Since $P$
is group-separable, there exists a set $B$, $1\le |B|\le m-1$
such that in each ranking in $P$ either $B$ appears above
$A\setminus B$ or $A\setminus B$ appears above $B$.
By construction, the profiles $P|_B$ and $P_{A\setminus B}$
are group-separable. Thus, by the inductive hypothesis, 
$P|_B$ contains at most $2^{|B|-1}$ distinct rankings and 
$P|_{A\setminus B}$ contains at most $2^{m-|B|-1}$ distinct rankings. 
Now, note that each ranking in $P$ is obtained
by `stacking' a ranking from $P|_B$ and a ranking from $P|_{A\setminus B}$
in one of the two possible ways. Hence, $P$ contains
at most $2\cdot 2^{|B|-1}\cdot 2^{m-|B|-1}=2^{m-1}$ distinct rankings.

To show that this bound is tight, we explain how to build a group-separable
profile over $m$ alternatives that contains $2^{m-1}$ distinct rankings.
Again, we proceed by induction. The case $m=2$ is immediate.
Now, suppose we have built a group-separable profile $P_\textrm{gs-max}^{m-1}$ over $\{a_1, \dots, a_{m-1}\}$
that contains $2^{m-2}$ distinct rankings. For each vote $v$ in $P_\textrm{gs-max}^{m-1}$, 
we construct two rankings over $\{a_1, \dots, a_m\}$: both rank the alternatives
 $a_1, \dots, a_{m-1}$ in the same order as $v$ does, but one ranks $a_{m}$
first, while the other ranks $a_m$ last. By construction, 
the resulting profile $P_\textrm{gs-max}^m$ contains $2^{m-1}$ distinct rankings.

To see that $P_\textrm{gs-max}^m$ is group-separable, consider a subset of alternatives
$A'\subseteq \{a_1, \dots, a_m\}$ with $|A'|\ge 2$.
We need to argue that there is a subset $B$, $1\le |B|\le |A'|-1$,
such that each voter's preferences are either of the form $B\succ A'\setminus B$
or of the form $A'\setminus B\succ B$. Now, if $a_m\in A'$, we can simply take 
$B=\{a_m\}$. On the other hand, if $a_m\not\in A'$ then 
$A'\subseteq \{a_1, \dots, a_{m-1}\}$ and the existence of $B$ follows
from the assumption that $P_\textrm{gs-max}^{m-1}$ is group-separable.
\end{proof} 

The clone decomposition tree of the profile $P_\textrm{gs-max}^m$
is a caterpillar, i.e., a binary tree such that each internal node
has one child that is a leaf. Each vote in this profile can be encoded
by a binary string of length $m-1$. Specifically, to generate a vote from
a binary string $(b_1, \dots, b_{m-1})$, we think of this vote
as an $m$-by-$1$ array, with $i$-th entry containing the alternative ranked 
in the $i$-th position, and then place $a_m, \dots, a_1$
in that array, so that $a_i$ is placed in the top-most available
position if $b_i=1$ and in the bottom-most available position if $b_i=0$
(with $a_1$ being placed in the last available position).
Yet another way to think about $P_\textrm{gs-max}^m$ is that 
there is a one-to-one correspondence between votes in $P_\textrm{gs-max}^m$
and subsets of $A\setminus \{a_1\}$: a subset $B\subseteq A\setminus\{a_1\}$
corresponds to a vote in which all elements of $B$ are ranked 
in top $B$ positions in decreasing order of indices, all elements
of $(A\setminus\{a_1\})\setminus B$ are ranked in the bottom $m-1-|B|$
positions in increasing order of indices, and $a_1$ is ranked 
in between in position $|B|+1$. These observations show
that $P_\textrm{gs-max}^m$ is neither single-peaked nor 
single-crossing for $m=3$; also, they will be useful 
in \Cref{sec:problems}, where we use them to show that
many computationally difficult problems remain hard
if voters' preferences are group separable.

We note that \citet{karpov:j:gs} provides an exact formula 
for the number of group-separable profiles with 
$n$ voters and $m$ alternatives, as well as an expression
for the number of such profiles that are narcissistic.

\newpage
\section{Recognition}
\label{sec:recog}

\begin{table}[!ht]
\centering
	\begin{tabular}{rll}
	\toprule
	\textbf{domain} 			& \textbf{recognition}		&  \\
	\midrule
	single-peaked 				& $O(mn)$				& \Cref{thm:recognizing-sp}\\
	single-crossing 			& $O(nm\log m)$ 				& \Cref{thm:recognizing-sc}\\
	1-Euclidean 				& in P  				& via solving an LP, see \Cref{sec:recog:euclid} \\
	$d$-Euclidean, $d\ge 2$ 		& $\exists\mathbb R$-complete 	&  \Cref{sec:recog:euclid}\\
	single-peaked on a tree 		& $O(mn)$				&  \Cref{thm:recognizing-spt}\\
	single-peaked on a circle   & $O(mn)$ &  Thm.\ 2 in \citet{peters2017spoc}\\	
	single-crossing on a tree 		&  $O(m^2n^3)$		&  Thm.\ 6 in \citet{clearwater2014single} \\
	single-peaked and single-crossing 	& $O(nm\log m)$ 			& \Cref{thm:recognizing-sp,thm:recognizing-sc} \\	
	single-caved 				& $O(mn)$				& \Cref{thm:recognizing-sp}\\
	value/best/medium/worst-restricted 			& $O(m^3n)$ 				& Prop.\ 6 in \citet{elkind2014detecting}\\
	group-separable 				& $O(mn)$				& \Cref{sec:recog:group-sep}\\
	$d$-dim.\ (hereditary) single-peaked	 	& open 				& \Cref{open:multi-dimsp}\\
	\bottomrule
	\end{tabular}
\caption{The fastest known algorithms for recognizing domain restrictions for profiles with $n$ voters and $m$ alternatives.}
\label{tab:recognition}
\end{table}

Each of the many domain restrictions introduced in \Cref{sec:def} is associated with a natural
algorithmic question: given a profile $P$, does $P$ belong to this restricted domain?
Such \emph{recognition} problems have been studied extensively and, 
in many cases, complexity classifications and algorithms exist.
Nonetheless, important questions remain open.
In this section, we consider restricted domains listed in \Cref{sec:def}, 
survey the known results for their corresponding recognition problems, 
and point out open problems along the way.
For an overview, we refer the reader to \Cref{tab:recognition}. 

For computational purposes, the appeal of the restricted domains is that many popular 
voting rules that are computationally hard in general 
admit efficient winner determination algorithms when their inputs are drawn
from such domains (see \Cref{sec:problems}). Many of these algorithms 
require a certificate that the input profile is structured; such a certificate 
can be provided, e.g., by an axis $\lhd$ that makes the input single-peaked, 
or an ordering of voters that makes it single-crossing. 
This is another reason why we are interested in recognition 
algorithms: without them, it would be impossible to make use of specialized winner 
determination procedures.

From this perspective, it is important to have recognition 
algorithms that are as fast as possible. Indeed, the cost of a slow recognition algorithm 
may outweigh the benefits of a dedicated fast winner determination procedure, 
so users might opt instead to pass their problem directly to a super-polynomial algorithm, such as an 
integer linear programming solver. Thus, in this section we will pay particular
attention to the running times of the algorithms we consider.
Nevertheless, we will also mention several slower algorithms, since they offer different 
perspectives, and provide a deeper understanding of the respective domains.
In particular, some (potentially slower) algorithms are \emph{certifying}, i.e., they 
return a short certificate in case the input profile \emph{does not} belong 
to the target domain.

\subsection{Algorithms for Single-Peaked Preferences}
\label{sec:recog:sp}

We will present two approaches to recognizing single-peaked profiles. 
Historically, the first algorithm to solve the recognition problem uses a simple reduction to the \emph{consecutive ones problem} which was put forward by \citet{bar-tri:j:sp}.
This approach leads to an $O(m^2n)$ runtime.
Later, faster runtimes were achieved by direct algorithms, which as a by-product also offer insight 
into the nature of single-peaked preferences.
In particular, an algorithm by \citet{doignon1994polynomial} runs in $O(mn+m^2)$ time, which was further improved to $O(mn)$ time by \citet{escoffier2008single}. 
Another linear-time algorithm, which is also applicable to certain profiles of weak orders,
was subsequently described by \citet{lackner2014incomplete}.

\subsubsection*{The Consecutive Ones Approach}

Let us consider a binary matrix, i.e., a matrix consisting of zeros and ones.
We say that this matrix possesses the \emph{consecutive ones property} if its columns can be permuted 
so that in each row all ones appear consecutively; see \Cref{fig:consones1} for an example.
This property gives rise to the following algorithmic problem.

\begin{figure}
	\centering
	\begin{tikzpicture}
	[highlight/.style={line width=11pt, black!10}]
	\matrix (m) [matrix of math nodes,left delimiter={[},right delimiter={]}] {
		0 & 0 & 0 & 1 & 1 & 1 & 1 & 1 & 0 & 0 \\
		0 & 1 & 1 & 1 & 1 & 1 & 0 & 0 & 0 & 0 \\
		0 & 0 & 0 & 0 & 0 & 0 & 0 & 1 & 1 & 1 \\
		0 & 0 & 1 & 1 & 1 & 1 & 1 & 1 & 1 & 0 \\
		0 & 0 & 0 & 0 & 0 & 0 & 1 & 1 & 0 & 0 \\
	};
	\begin{tikzbackground}
	\draw[highlight] (m-1-4.west) -- (m-1-8.east);
	\draw[highlight] (m-2-2.west) -- (m-2-6.east);
	\draw[highlight] (m-3-8.west) -- (m-3-10.east);
	\draw[highlight] (m-4-3.west) -- (m-4-9.east);
	\draw[highlight] (m-5-7.west) -- (m-5-8.east);
	\end{tikzbackground}		
	\end{tikzpicture}
	\caption{A matrix with the consecutive ones property.}
	\label{fig:consones1}
\end{figure}

\problem
{Consecutive Ones}
{An $n\times m$ matrix $M$ with $M(i,j)\in\{0,1\}$}
{Does $M$ have the consecutive ones property, i.e., 
is there a permutation of the columns of $M$ such that in every row, 
all $1$-entries appear consecutively?}

The consecutive ones property was first defined by \citet{fulkerson1965incidence}, who also described 
an $O(m^2n)$ time algorithm for solving \problemname{Consecutive Ones}.
Later, \citet{booth1976testing} introduced the PQ-tree data structure, 
which enabled them to prove the following result.

\begin{theorem}[\citealp{booth1976testing}]
	\problemname{Consecutive Ones} can be solved in time $O(m+n+f)$, where $f = \sum_{i,j} M(i,j)$ 
	is the number of $1$s occurring in $M$.\label{thm:consones}
\end{theorem}

In the applications that we consider, the matrices are dense, so the time bound boils down to 
$O(mn)$. In addition to allowing for a linear-time algorithm, PQ-trees also have the property that they 
generate all permutations that witness the consecutive ones property.
Subsequent work has improved this result in various ways \citep{meidanis1998consecutive,
habib2000lex,
mcconnell2004certifying}, and a survey of the literature on the consecutive ones problem is provided by \citet{dom2009consecutive}.

We will now explain how to reduce the problem of recognizing single-peaked preference
profiles to the consecutive ones problem; our exposition follows \citet{bar-tri:j:sp}.
As shown in \Cref{prop:sp-equiv}, 
a vote $v_i$ is single-peaked with respect to an order $\lhd$ if and only if 
for each $c \in A$, the set $\{a\in A : a \pref_i c\}$ is an interval of $\lhd$, or,
equivalently, for all $\ell\in[m]$ the top $\ell$ alternatives in $v_i$ form an interval of $\lhd$.
This observation enables us map a profile $P$ to a binary matrix $M$ so that
$P$ is single-peaked if and only if $M$ has the consecutive ones property.
The reduction proceeds as follows.

Assume that $A=[m]$; we create a matrix $M$ with $m$ columns and $nm$ rows, 
so that each column corresponds to an alternative and each group of $m$ rows corresponds
to a voter. Specifically, for each $i\in [n]$, $\ell\in [m]$ 
the row $(i-1)m+\ell$ encodes the top $\ell$ alternatives
in the preferences of voter $i$, i.e., for each $j\in [m]$ 
we set 
\begin{align*}
M((i-1)m+\ell, j) = 
	\begin{cases} 
	1 &\mbox{if $j$ is among the top $\ell$ alternatives of voter $i$}\\
	0 & \mbox{otherwise.} 
	\end{cases} 
\end{align*}
Now, observe that each 
permutation of columns that witnesses the consecutive ones property of $M$ 
corresponds to an ordering of alternatives such that every prefix of every vote forms
an interval of that ordering, and vice versa. 
By \Cref{prop:sp-equiv}, this establishes that our reduction is correct.

\begin{examplebox}
{Reduction of the single-peakedness recognition problem to the consecutive ones problem}
{sp-to-c1p}
\centering
\begin{minipage}{0.1\linewidth}
	\begin{tabular}{cc}
		\toprule
		$v_1$ & $v_2$ \\
		\midrule
		$b$ & $c$  \\
		$c$ & $d$  \\
		$a$ & $a$  \\
		$d$ & $b$  \\
		\bottomrule
	\end{tabular}
\end{minipage}
\qquad
\raisebox{-10pt}{\scalebox{2.5}{$\mapsto$}}
\qquad
\begin{minipage}{0.4\linewidth}
	\begin{tikzpicture}
	[decoration=brace,
	highlight/.style={line width=11pt, red!15},
	rowlabel/.style={anchor=base, text width=2cm}]
	\matrix (m) [matrix of math nodes,left delimiter={[},right delimiter={]}] {
		0 & 1 & 0 & 0 \\
		0 & 1 & 1 & 0 \\
		1 & 1 & 1 & 0 \\
		1 & 1 & 1 & 1 \\
		0 & 0 & 1 & 0 \\
		0 & 0 & 1 & 1 \\
		1 & 0 & 1 & 1 \\
		1 & 1 & 1 & 1 \\
	};
	\node [anchor=base] (a) at ($(m-1-1.north)+(0,5pt)$) {$a$};
	\node [anchor=base] (b) at ($(m-1-2.north)+(0,5pt)$) {$b$};
	\node [anchor=base] (c) at ($(m-1-3.north)+(0,5pt)$) {$c$};
	\node [anchor=base] (d) at ($(m-1-4.north)+(0,5pt)$) {$d$};
	\node [rowlabel] (row1) at ($(m-1-4.east)+(1.5cm,-3pt)$)  {$\{b\}$};
	\node [rowlabel] (row2) at ($(m-2-4.east)+(1.5cm,-3pt)$)  {$\{b,c\}$};
	\node [rowlabel] (row3) at ($(m-3-4.east)+(1.5cm,-3pt)$)  {$\{a,b,c\}$};
	\node [rowlabel] (row4) at ($(m-4-4.east)+(1.5cm,-3pt)$)  {$\{a,b,c,d\}$};
	\node [rowlabel] (row5) at ($(m-5-4.east)+(1.5cm,-3pt)$)  {$\{c\}$};
	\node [rowlabel] (row6) at ($(m-6-4.east)+(1.5cm,-3pt)$)  {$\{c,d\}$};
	\node [rowlabel] (row7) at ($(m-7-4.east)+(1.5cm,-3pt)$)  {$\{a,c,d\}$};
	\node [rowlabel] (row8) at ($(m-8-4.east)+(1.5cm,-3pt)$)  {$\{a,b,c,d\}$};
	\draw[decorate,transform canvas={xshift=-1.3em},thick] ($(m-4-1.south west)+(0,2pt)$) -- node[left=2pt] {$v_1$} (m-1-1.north west);
	\draw[decorate,transform canvas={xshift=-1.3em},thick] (m-8-1.south west) -- node[left=2pt] {$v_2$} ($(m-5-1.north west)+(0,-2pt)$);
	\begin{tikzbackground}
	\draw[highlight] (m-1-2.west) -- (m-1-2.east);
	\draw[highlight] (m-2-2.west) -- (m-2-3.east);
	\draw[highlight] (m-3-1.west) -- (m-3-3.east);
	\draw[highlight] (m-4-1.west) -- (m-4-4.east);
	\draw[highlight] (m-5-3.west) -- (m-5-3.east);
	\draw[highlight] (m-6-3.west) -- (m-6-4.east);
	\draw[highlight] (m-7-1.west) -- (m-7-1.east);
	\draw[highlight] (m-7-3.west) -- (m-7-4.east);
	\draw[highlight] (m-8-1.west) -- (m-8-4.east);
	\end{tikzbackground}		
	\end{tikzpicture}
\end{minipage}
\end{examplebox}
Thus, the problem of recognizing single-peaked profiles is polynomial-time solvable.
More precisely, 
since we can translate a profile into a binary matrix in  $O(m^2n)$ time 
and solve the resulting consecutive ones problem in $O(m^2n)$ time (\Cref{thm:consones}), 
we obtain a total runtime of $O(m^2n)$.
The PQ-tree data structure behind this approach yields a compact representation 
of all orderings that witness the single-peaked property.

\subsubsection*{A Direct Approach}

We now turn to a direct, linear-time algorithm for recognizing single-peaked profiles.
Our presentation builds on the algorithm by \citet{doignon1994polynomial}, 
infused with ideas from the linear-time algorithm by \citet{escoffier2008single}.
We follow the exposition of \citet{brandt2011lectures}.
The main idea, which is shared by both algorithms, is to build an axis $\lhd$  
by starting from the outside and iteratively placing alternatives proceeding inwards.
At each stage, we consider the set of alternatives that are ranked last in some vote. If there are one or two such alternatives, those are the alternatives that we place.
If there are three or more, 
the profile is not single-peaked: one of these three has to be placed 
in between the others on $\lhd$, creating a valley 
in the vote where it is ranked below the other two (cf.\ \Cref{prop:sp-equiv}~(3)).

\begin{examplebox}
	{Placing outermost alternatives}
	{sp-alg-1}
	\begin{minipage}{0.15\linewidth}
	\begin{tabular}{ccc}
		\toprule
		$v_1$ & $v_2$ & $v_3$ \\
		\midrule
		$e$ & $e$ & $d$ \\
		$d$ & $d$ & $e$ \\
		$b$ & $c$ & $c$ \\
		$c$ & $b$ & $a$ \\
		$a$ & $a$ & $b$ \\
		\bottomrule
	\end{tabular}
	\end{minipage}
	\hfill
	\begin{minipage}{0.8\linewidth}
	Consider the profile $P$ on the left.
	There are two alternatives that appear in the bottom-most position: 
	alternative $a$ in $v_1$ and $v_2$ and alternative $b$ in $v_3$.
	These two alternatives have to be at outermost positions 
	in any order of alternatives witnessing the single-peaked property.
	As axes can be reversed without affecting single-peakedness, 
	we can restrict our attention to orders $\lhd$ satisfying $a\lhd\{c,d,e\}\lhd b$.
	\end{minipage}
\end{examplebox}

Let $B$ denote the set of alternatives that are ranked last by at least one voter
(among the alternatives that have not yet been placed).
If at any point in the execution of the algorithm we have $|B|\ge 3$, 
the algorithm returns `no' since these alternatives would form a valley in at least one vote, no matter
how we place them on $\lhd$; thus, from now on,
when describing the algorithm, we assume that $|B|\le 2$.

The algorithm proceeds in two stages. The goal of the first stage is to ensure that the leftmost
position and the rightmost position on $\lhd$ are filled.
Thus, if initially $|B|=1$, we place the unique alternative in $B$ in the leftmost unfilled 
position on $\lhd$ and remove this alternative from the profile.
We repeat this step until $|B|=2$.
If $|B|=2$, so that $B=\{a, b\}$, we place $a$ into the leftmost unfilled position on $\lhd$
and $b$ into the rightmost unfilled position on $\lhd$. 
We are now guaranteed that there is at least one alternative on the left-hand side of the partial axis 
and at least one alternative on its right-hand side.

At the second stage, we proceed as follows.
At each iteration, we first check that
there are at least two alternatives that have not been placed yet;
otherwise, we can terminate.
Let $\ell$ denote the alternative last placed on the left-hand side 
and let $r$ denote the alternative last placed on the right-hand side,
so that the partial axis is of the form $\ell'\lhd\dots\lhd\ell\lhd\dots\lhd r\lhd\dots\lhd r'$,
where the positions between $\ell$ and $r$ are unfilled.
For each $x\in B$
let $N_x$ be the set of voters that do not rank $x$ first among the remaining alternatives.
Suppose that for some $a\in B$ there is a voter $i\in N_a$ with $\ell\succ_i a\succ_i r$. 
Then $a$ has to be placed into the rightmost available spot, i.e., just to the left of $r$.
Indeed, if, instead, we place $a$ just to the right of $\ell$, 
the top remaining alternative of voter $i$ (which is distinct from $a$ since $i\in N_a$) 
would create a valley together with $\ell$ and $a$.
Similarly, if there is a voter $j\in N_a$ with $r\succ_j a\succ_j \ell$, 
then $a$ has to be placed just to the right of $\ell$.

\definecolor{lgray}{gray}{0.7}
\iflatexml\addtocounter{example}{-1}\else\addtocounter{tcb@cnt@examplebox}{-1}\fi
\begin{examplebox}{continued --- placing inner alternatives}{sp-alg-2} 
	\begin{minipage}{0.14\linewidth}
		\begin{tabular}{ccc}
			\toprule
			$v_1$ & $v_2$ & $v_3$ \\
			\midrule
			$e$ 				   & $e$ 					& $d$ \\
			$d$ 				   & $d$ 					& $e$ \\
			\textcolor{lgray}{$b$} & $c$ 					& $c$ \\
			$c$ 			       & \textcolor{lgray}{$b$} & \textcolor{lgray}{$a$} \\
			\textcolor{lgray}{$a$} & \textcolor{lgray}{$a$} & \textcolor{lgray}{$b$} \\
			\bottomrule
		\end{tabular}
	\end{minipage}
	\hfill
	\begin{minipage}{0.8\linewidth}
	We consider the restriction of our profile $P$ to alternatives not yet placed, 
	i.e., the set $\{c,d,e\}$. The set $B$ of bottom-ranked alternatives is $\{c\}$.
	Our current partial axis is $a\lhd\dots\lhd b$, so $\ell=a$ and $r=b$.
	Since $b\succ_1 c \succ_1 a$, and voter $1$ does not rank $c$ first among $\{c, d, e\}$, 
	alternative $c$ has to be placed next to $a$.
	Otherwise, if we chose the partial axis $a\lhd\dots\lhd c \lhd b$, 
	the alternatives $e,c,b$ would form a valley for vote $v_1$.
	Votes $v_2$ and $v_3$ rank $c$ above $a$ and $b$, so they impose
	no constraints on the placement of $c$.
	We thus continue with the partial axis $a\lhd c\lhd\dots\lhd b$.
\end{minipage}
\end{examplebox}

If at some point $|B|=2$ and both alternatives in $B$ have to be placed next to $\ell$ 
or both have to be placed next to $r$, the profile is not single-peaked.
It may also be the case that all alternatives in $B$ can be placed in either position,
in which case we make the choice arbitrarily.

\iflatexml\addtocounter{example}{-1}\else\addtocounter{tcb@cnt@examplebox}{-1}\fi
\begin{examplebox}{continued --- placing inner alternatives}{sp-alg-3} 
	It remains to place the alternatives $\{d,e\}$. Since our partial axis is 
	$a\lhd c\lhd\dots\lhd b$, we have $\ell=c$ and $r=b$.
	Observe that there does not exist a voter $i$ with 
	$c\succ_i d \succ_i b$; nor is there a voter $j$ with $b\succ_j d \succ_j c$.
	The same holds for alternative $e$.
	Thus, both alternatives can be placed arbitrarily.
	The algorithm concludes that both $a\lhd c\lhd e \lhd d \lhd b$ and $a\lhd c\lhd d \lhd e \lhd b$ 
	are single-peaked axes for $P$, and returns one of those; the reader can verify
	that this is correct.
\end{examplebox}

A more precise description of the algorithm is given by
the pseudocode in \Cref{alg:sp}.
In the pseudocode, $N_x(A')$ denotes the set of voters that do not rank $x$ 
first among the alternatives in $A'\subseteq A$.
Furthermore, we use the $\oplus$ operator to concatenate two lists.
We also concatenate lists and sets with $\oplus$ 
if the sets are guaranteed to contain at most one element.

\newcommand{\longleft}{\makebox[0pt][l]{\texttt{left}}\phantom{\texttt{right}}}
\LinesNumbered
\SetNlSty{}{\color{green!30!black!70!white}}{}
\begin{algorithm}
	\DontPrintSemicolon
	\KwIn{A profile $P$ over $A$}
	\KwOut{An axis $\lhd$ on $A$ such that $P$ is single-peaked with respect to $\lhd$, if one exists}
	\SetKwRepeat{Do}{do}{while}
	$\texttt{left},\texttt{right}\gets$ empty lists\;
	$A' \gets A$\hfill /\!/ $A'$ is the set of alternatives that still need to be placed\\
	\smallskip
	\While(\hfill // first stage)
	{\textup{$A'\neq\varnothing$ and \texttt{right} is empty}}
	{
		$B \gets $ the set of bottom-ranked alternatives in $P|_{A'}$\;
		\lIf{$|B|>2$}{\Return{``$P$ is not single-peaked''}\nllabel{sp-alg:phase1:not-worst}}
		\lIf{$|B|=\{x\}$}{$\texttt{left}\gets\texttt{left}\oplus \{x\}$}
		\If{$|B|=\{x,y\}$\nllabel{sp-alg:phase1:two-bottom}}
		{
			$\longleft\gets\texttt{left}\oplus \{x\}$\;
			$\texttt{right}\gets \{y\}\oplus \texttt{right}$
		}
		$A'\gets A'\setminus B$
	}
	\smallskip
	\While(\hfill // second stage)
	{$|A'| \ge 2$}
	{%
		$\ell\gets$ rightmost element of $\texttt{left}$\;
		$r   \gets$ leftmost element of $\texttt{right}$\;
		$B   \gets $ the set of bottom-ranked alternatives in $P|_{A'}$\;
		\lIf{$|B|>2$}{\Return{``$P$ is not single-peaked''}\nllabel{sp-alg:phase2:not-worst}}
		$L \gets \{x\in B : \text{there is $i\in N_x(A')$ with $r    \succ_i x \succ_i \ell$}\}$\;
		$R \gets \{x\in B : \text{there is $i\in N_x(A')$ with $\ell \succ_i x \succ_i r   $}\}$\;
		\lIf{$|L|>1$ or $|R|>1$}{\Return{``$P$ is not single-peaked''}\nllabel{sp-alg:phase2:LRtoolarge}}
		\lIf{$L\cap R\neq\varnothing$}{\Return{``$P$ is not single-peaked''}\nllabel{sp-alg:phase2:contrad}}
		\If{$L\cup R\neq B$}{
			\nllabel{sp-alg:arbitrary}
			Arbitrarily assign elements of $B\setminus(L\cup R)$ to $L$ and $R$
					 so that $|L|\le 1$ and $|R|\le 1$.}
			$\longleft\gets\texttt{left}\oplus L$\;
			$\texttt{right}\gets R\oplus \texttt{right}$
			\hfill
			\iflatexml\else\begin{tikzpicture}[remember picture,overlay,x=.4cm,xshift=-8.7cm]
			\node[anchor=base] (axis-label) at (8.1cm,-0.35) {$\lhd$};
			
			\node[rectangle, fill=green!20!white, anchor=south west, minimum width=3.2cm, minimum height=0.6cm] (left) at (2.4cm,0) {$A'$};
			
			\node[rectangle, fill=red!40!white, anchor=south west, minimum width=2.4cm, 
				minimum height=0.6cm] (left) at (0,0) {\emph{left}};
			\node[rectangle, anchor=south west, minimum width=.4cm, minimum height=0.6cm] 
				(left) at (2cm,0) {$\ell$};
			\node[rectangle, anchor=south west, minimum width=.4cm, minimum height=0.6cm] 
				(left) at (2.4cm,0) {$L$};
		
			\node[rectangle, fill=red!40!white, anchor=south west, minimum width=2.4cm, 
				minimum height=0.6cm] (left) at (5.6cm,0) {\vphantom{$L$}\emph{right}};
			\node[rectangle, anchor=south west, minimum width=.4cm, minimum height=0.6cm] 
				(left) at (5.6cm,0) {$\vphantom{L} r$};
			\node[rectangle, anchor=south west, minimum width=.4cm, minimum height=0.6cm] 
				(left) at (5.1cm,0) {$R$};

			\draw[-latex] (0,0) -- (21,0);
			\foreach \x in {0,...,6}
				\draw (\x,1pt) -- (\x,-3pt);
			\foreach \x in {14,...,20}
				\draw (\x,1pt) -- (\x,-3pt);
			\end{tikzpicture}\fi
			\\
		$A'\gets A'\setminus B$
	}
	\Return $\texttt{left}\oplus A' \oplus\texttt{right}$ \hfill /\!/ {$A'$ contains at most one element}
	\caption{Recognizing single-peaked profiles in $O(mn)$ time}\label{alg:sp}
\end{algorithm}
\LinesNotNumbered

\begin{theorem}\label{thm:recognizing-sp}
\Cref{alg:sp} recognizes single-peaked profiles in time $O(mn)$.
\end{theorem}

\noindent
\emph{Proof.}\:
We proceed in three steps. First, we argue that if the algorithm outputs an axis $\lhd$
then the input profile $P$ is single-peaked with respect to $\lhd$. Second, we show that
if the algorithm fails, then the profile is not single-peaked. Finally, we analyze the running time.

\begin{claim}
If \Cref{alg:sp} outputs an axis $\lhd$ then $P$ is single-peaked
with respect to $\lhd$.
\end{claim}
\begin{proof}
Assume towards a contradiction that there is a vote $v_i$ in $P$ that is not single-peaked on $\lhd$.
Then there exists a valley, 
i.e., three alternatives $a,b,c$ such that $a\lhd b\lhd c$, $a\succ_i b$, and $c\succ_i b$.
We can assume that $a$, $b$, and $c$ appear consecutively on $\lhd$, i.e., 
the axis is of the form $\dots\lhd a\lhd b \lhd c \lhd \cdots$. (If there exists a valley, a valley also exists on consecutive candidates.) Let us consider the order
in which $a$, $b$, and $c$ were placed on the axis. 
We distinguish with four cases which candidates were placed first on the axis: $\{a, b\}$ or $\{b, c\}$ were placed simultaneously first, $\{a, c\}$  were placed simultaneously first, only $\{b\}$ was placed first, and either $\{a\}$ or  $\{c\}$ was placed first.

It cannot be the case that $a$ and $b$ were placed 
simultaneously, and before $c$ (or that $b$ and $c$ were placed simultaneously, and before $a$), 
since whenever $|B|=2$, we place the alternatives in $B$ at the opposite
ends of the unfilled part of the axis. Also, 
it cannot be the case that $a$ and $c$ were placed simultaneously, and before $b$: if $b$
has not been placed yet, we cannot have $B=\{a, c\}$, since voter $i$ ranks both $a$ and $c$
above $b$. Thus, there was an iteration where all three of $a$, $b$ and $c$ 
were available, and we placed exactly one of them. It could not have been $b$, 
as the unfilled positions should always form an interval of the axis. 

The only remaining possibility is that
one of $a$ and $c$ was placed strictly before the other two alternatives in $\{a, b, c\}$;
assume without loss of generality that it was $a$. Further, we note that $b$ and $c$ were placed during the second stage. This can be seen as follows: since $a\succ_i b$, $a$ was not placed as the only candidate in this round, i.e., in this round $|B|=2$. Thus, $a$ was placed either in the final iteration of the first stage (line~\ref{sp-alg:phase1:two-bottom}) or in the second stage, and consequently $b$ and $c$ were placed during the second stage.
Now, suppose that $c$ was placed strictly before $b$; we will show that this leads to a contradiction.
Indeed, given that $a$, $b$ and $c$ form a contiguous segment of $\lhd$, it follows that 
in the iteration in which $c$ was placed we had to have $A'=\{b, c\}$.
As $c\succ_i b$, it follows that we had $B=\{b, c\}$ at that point, so $b$ would have to be placed
in the same iteration.

Thus, it has to be the case that $b$ is placed before $c$ or simultaneously with $c$.
Consider the iteration in which $b$ is placed. 
At the start of this iteration, the partial axis is of the form 
$\cdots\lhd a \lhd\cdots\lhd r\lhd\cdots$, with positions between $a$ and $r$ unfilled.
Note that the alternatives placed in the previous iteration form a subset of $\{a, r\}$;
as $a\succ_i b$, it follows that the bottom-most alternative in $i$'s vote 
in the previous iteration was $r$, so we have $a\succ_i b\succ_i r$.
As $c$ has not been placed yet and voter $i$ prefers $c$ to $b$, 
we have $i\in N_b(A')$. Together with $a\succ_i b\succ_i r$ this implies (line 17)
$b\in R$. As $|A'|\ge 2$, this means that $b$ cannot be placed next to $a$, a contradiction.
We conclude that every vote in $P$ is single-peaked on $\lhd$.
\end{proof}

To prove the converse claim, i.e., that \Cref{alg:sp} only returns
`no' on profiles that are not single-peaked, we need a technical lemma.

\begin{lemma}
	\label{lem:sp-minors-are-bad}
	Let $P$ be a profile. 
	If there exist alternatives $a,b,c\in A$ and voters $i,j,k\in N$ such that
	\begin{align}
		\label{eq:sp-minor-1}
		\{b,c\} \succ_i a, \qquad \{a,c\}\succ_j b, \qquad\text{ and }\qquad \{a,b\}\succ_k c.
	\end{align}
	then $P$ is not single-peaked.
	If there exist alternatives $a,b,c,d\in A$ and voters $i,j\in N$ such that
	\begin{align}
		\label{eq:sp-minor-2}
		\{a,d\} \succ_i b  \succ_i c, \qquad\text{ and }\qquad
		\{c,d\} \succ_j b  \succ_j a,
	\end{align}
	then $P$ is not single-peaked.
\end{lemma}
\begin{proof}
	Assume for a contradiction that $P$ is single-peaked with respect to $\lhd$,  
	but condition \eqref{eq:sp-minor-1} holds.
	One of the alternatives $a, b, c$ appears between the other two on $\lhd$. This creates a valley
	in the preferences of the voter who ranks this alternative last, a contradiction.
	
	Suppose next that condition \eqref{eq:sp-minor-2} holds.
	Since $P$ is single-peaked, so is its restriction to $A'=\{a, b, c, d\}$.
	Thus we can assume without loss of generality that $A=\{a, b, c, d\}$.
	From \Cref{prop:sp-equiv}~(4), get that $\{a,d\}$ and $\{c,d\}$ must be intervals of $\lhd$.
	Since the alternatives $a$ and $c$ occur bottom-ranked, they must appear at the ends of $\lhd$, so without loss of generality we have $a \lhd \{b, d\} \lhd c$.
	But if $a \lhd b \lhd d \lhd c$, then $\{a,d\}$ is not an interval, and if $a \lhd d \lhd b \lhd c$, then $\{c,d\}$ is not an interval.
	This is a contradiction.
\end{proof}
\begin{claim}
	If \Cref{alg:sp} returns `no', then $P$ is not single-peaked.
\end{claim}
\begin{proof}
	We show that whenever \Cref{alg:sp} returns `no', then either situation \eqref{eq:sp-minor-1} or \eqref{eq:sp-minor-2} from \Cref{lem:sp-minors-are-bad} occurs and hence $P$ is not single-peaked.
	
	If \Cref{alg:sp} fails in Line~\ref{sp-alg:phase1:not-worst} 
	or in Line~\ref{sp-alg:phase2:not-worst}, so that $|B| \ge 3$, then we can take three alternatives $a,b,c \in B$ and deduce that $P$ is not single-peaked by \eqref{eq:sp-minor-1}.
	
	Next, assume that \Cref{alg:sp} fails in Line~\ref{sp-alg:phase2:LRtoolarge}, 
	i.e., either $|L|=2$ or $|R|=2$.
	Without loss of generality, suppose $|L|=2$ and write $L=\{a,b\}$.
	This means that there are voters $i\in N_a(A')$ with $r\succ_i a\succ_i \ell$
	and $j \in N_b(A')$ with $r\succ_j b\succ_j \ell$. 
	
	Suppose first that we can choose such voters to be the same ($i = j$), so there is a single voter $i \in N_a(A') \cap N_b(A')$ with $r \succ_i \{a,b\} \succ_i \ell$. Let $x$ be $i$'s most-preferred alternative among the set $A'$ of unplaced alternatives, which is different from $a$ and $b$. Assume without loss of generality that $a \succ_i b$, and thus $\{x,r\} \succ_i a \succ_i b \succ_i \ell$. Since $a \in B$ and $x, b \in A'$, there is a voter $k$ with $\{x,b\} \succ_k a$. If $a \succ_k r$, then we have
	\[
		\{x,r\} \succ_i a \succ_i b \text{ and } \{x,b\} \succ_k a \succ_k r,
	\]
	which is an instance of \eqref{eq:sp-minor-2}. Otherwise, $r \succ_k a$. In that case, take a voter $t$ with $A' \succ_t r$, where such a voter exists because alternative $r$ was placed in the previous iteration. Then we have
	\[
	r \succ_i a \succ_i b \text{ and } b \succ_k r \succ_k a \text{ and } \{a,b\} \succ_t r,
	\]
	which is an instance of \eqref{eq:sp-minor-1}.
	
	Otherwise, the voters $i$ and $j$ cannot be chosen to be the same. Thus, for neither voter do we have $r \succ \{a,b\} \succ \ell$. On the other hand, we do have $\{a,b\} \succ \ell$ for both $i$ and $j$. This can be seen as follows: Consider voter~$i$, the argument for $j$ is the same. We have $r \succ_i a \succ_i \ell$ but not $r \succ_i \{a,b\} \succ_i \ell$. So the remaining possibilities are $b \succ_i r \succ_i a \succ_i \ell$ and $r \succ_i a \succ_i \ell \succ_i b$.
	The latter is not possible because then $b$ would have been placed before (or together with) either $r$ or $\ell$.
	Consequently,
	$b \succ_i r \succ_i a \succ_i \ell$ and $a \succ_j r \succ_j b \succ_j \ell$.
	Now, take again a voter $t$ with $A' \succ_t r$, who exists because $r$ was placed in a previous iteration. Then we have
	\[
	b \succ_i r \succ_i a \text{ and } a \succ_j r \succ_j b \text{ and } \{a,b\} \succ_t r,
	\]
	which is an instance of \eqref{eq:sp-minor-1}.
	
	Finally, suppose the algorithm fails in Line~\ref{sp-alg:phase2:contrad} (and thus  did not fail in Line~\ref{sp-alg:phase2:LRtoolarge}), and let $a\in L\cap R$. Thus there are voters $i,j \in N_a(A')$ with $r \succ_i a \succ_i \ell$ and $\ell \succ_j a \succ_j r$.
	If $|B| = 1$, let $b \neq a$ be an arbitrary alternative in $A'$ (noting that $|A'| \ge 2$ by the condition of the while-loop). Then by definition of $B$, we have $b \succ_i a$ and $b \succ_j a$.
	If $|B| = 2$, let us write $B = \{a,b\}$. In that case we also have $b \succ_i a$ and $b \succ_j a$, because otherwise we would have, say, $a \succ_i b$ and thus $r \succ_i \{a,b\} \succ_i \ell$ and then we would have already failed in Line~\ref{sp-alg:phase2:LRtoolarge}, contradiction. Hence we have
	\[
	\{b, r\} \succ_i a \succ_i \ell \text{ and } \{b, \ell\} \succ_j a \succ_j r,
	\]
	which is an instance of \eqref{eq:sp-minor-2}.
\end{proof}

\begin{claim}
\Cref{alg:sp} runs in time $O(nm)$.
\end{claim}
\begin{proof}
Observe that both while loops induce at most $m$ iterations.
Thus, to prove the runtime bound, we have to show that all lines within these loops 
require at most $O(n)$ time.
This is straightforward for most lines, but the computation of the sets $B$, $L$, and $R$ 
requires some attention. The set $B$ contains all bottom-ranked alternatives in $P|_{A'}$.
Iterating through each vote until we found an alternative that is contained in $A'$
would require $O(m)$ time per voter.
This issue can be circumvented by storing---for each voter---a pointer to the bottom-ranked 
alternative in $A'$.
Similarly, we can store a pointer to the top-ranked alternative in $A'$ for each voter;
this way, for each $x\in A'$ and $i\in [n]$ we can easily check whether $i\in N_x(A')$.
We update these pointers at the end of each while loop, by iterating through all votes;
importantly, each pointer needs to be moved by at most two positions.

Recall that $L=\{x\in B : \text{there is $i\in N_x(A')$ with $r \succ_i x \succ_i \ell$}\}$.
At the time of computing $L$ and $R$, we know that $|B|\leq 2$.
Thus, under the assumption that $r \succ_i x \succ_i \ell$ can be checked in constant time, 
computing $L$ takes $O(n)$ time.
Satisfying this assumption is not entirely trivial: e.g.,
if we stored votes as lists, comparing two alternatives would require $O(m)$ time, 
as we would have to locate these alternatives in the list.
Thus, we precompute (outside of the while loops) for each vote 
a function from alternatives to positions that maps each alternative $c$ to 
its position in the vote; this can be done by scanning each vote once, so in $O(m)$ time per vote.
Then we can decide if, e.g., $r\succ_i x$, by comparing the positions of $r$ and $x$ in vote $v_i$.
Hence the set $L$---and similarly the set $R$---can be computed in $O(n)$ time.
\end{proof}
Together, the three claims establish \Cref{thm:recognizing-sp}.
\qed

\subsubsection*{Finding All Axes}

When \Cref{alg:sp} succeeds, it returns a single axis 
on which the input profile is single-peaked. 
We will now discuss how to modify the algorithm so that it can return all axes in a compact 
representation. It turns out that this is possible within the $O(mn)$ runtime bound.

A set $A' \subseteq A$ is a \emph{common prefix} of a profile $P$ 
if $A' \succ_i A\setminus A'$ for all $i \in N$, i.e., every voter in $P$
ranks the alternatives in $A'$ above all other alternatives. Suppose that 
$P$ is single-peaked on $\lhd$, and $A'$ is a common prefix of $P$. 
Since $P$ is single-peaked on $\lhd$, every prefix of every vote 
corresponds to an interval of $\lhd$. Hence $A'$ is an interval of $\lhd$. 
Now consider the axis $\lhd'$ obtained from $\lhd$ by reversing the 
order of alternatives in $A'$, and observe
that every prefix of every vote is still an interval of $\lhd'$. Thus, the collection of 
axes is closed under reversing common prefixes. In fact, 
the set of axes obtained from $\lhd$ in this way is complete, in the sense 
that it contains all valid axes. This can be seen by analyzing \Cref{alg:sp}.
During the first stage, each iteration except for the last one corresponds
to the case where the unranked alternatives form a common prefix. While 
during these iterations \Cref{alg:sp} places the alternative
it considers on the left-hand side of the axis, this choice is only made
for convenience, as it allows us to quickly decide when to move to the second stage;
placing any of these alternatives on the right (which is equivalent to reversing
the associated common prefix) would result in a valid axis as well.
During the second stage, 
the algorithm can only make choices in case $L = R = \varnothing$. 
In this case, it then follows from the definition of 
$L$ and $R$ that the remaining set of alternatives $A'$ is a common prefix of $P$.

\begin{proposition}[\citealp{doignon1994polynomial}, Prop.~7]
	\label{prop:sp-all-axes}
	Suppose a profile $P$ is single-peaked on $\lhd$. 
	Let $\varnothing \neq A_1 \subset A_2 \subset \dots \subset A_t = A$ 
	be the collection of common prefixes of $P$. 
	Then $P$ is single-peaked on an axis $\lhd'$ if and only if $\lhd'$ 
	can be obtained from $\lhd$ by the following process: for $j = 1, \dots, t$ sequentially, 
	reverse or do not reverse the interval $A_j$ in $\lhd$.
\end{proposition}

\begin{wrapfigure}[7]{r}{0.49\linewidth}
	\scalebox{0.85}{
		\begin{tikzpicture}
		[qnode/.style={rectangle,fill=red!50,draw},
		node distance=0.8cm and 0.5cm]
		\node(A)[qnode,text width=8.5cm]{};
		\node(B)[qnode,below=0.5cm of A, text width=6.6cm]{};
		\node(C)[qnode,below right=0.5cm and -4.8cm of B, text width=3.5cm]{};
		\node(D)[qnode,below=0.5cm of C, text width=1.5cm]{};
		\node(8)[left=of B]{$h$};
		\node(9)[right= of B]{$i$};
		\node(6)[left=of C]{$f$};
		\node(7)[left=of 6]{$g$};
		\node(1)[right=of C]{$a$};
		\node(2)[left=of D]{$b$};
		\node(5)[right=of D]{$e$};
		\node(3)[below left=0.4cm and -0.6cm of D]{$c$};
		\node(4)[below right=0.32cm and -0.6cm of D]{$d$};
		\draw[->](A.south)-|(8.north);
		\draw[->](A.south)-|(B.north);
		\draw[->](A.south)-|(9.north);
		\draw[->](B.south)-|(7.north);
		\draw[->](B.south)-|(6.north);
		\draw[->](B.south)-|(C.north);
		\draw[->](B.south)-|(1.north);
		\draw[->](C.south)-|(2.north);
		\draw[->](C.south)--(D.north);
		\draw[->](C.south)-|(5.north);
		\draw[->](D.south)-|(3.north);
		\draw[->](D.south)-|(4.north);
		\end{tikzpicture}
	}
\end{wrapfigure}
Another way to obtain a concise representation
of the set of all axes a given profile is single-peaked on is by leveraging
the connection between single-peakedness and the consecutive ones property of the associated 0-1 matrix 
(see \Cref{ex:sp-to-c1p}). \citet{booth1976testing} show that the collection of column 
reorderings making a 0-1 matrix exhibit the consecutive ones property can be represented 
by a data structure known as a \emph{PQ-tree}. 
In this context, \Cref{prop:sp-all-axes} can be viewed as 
a statement about the structure of this PQ-tree: all its non-leaf nodes are Q-nodes, 
and each Q-node has at most one non-leaf child. The pictured PQ-tree captures the structure of 
\Cref{ex:sp-all-axes}.

\begin{examplebox}
	{All single-peaked axes of a profile \citep[Example~9]{doignon1994polynomial}}
	{sp-all-axes}
	\centering
	\begin{minipage}{0.15\linewidth}
		\begin{tabular}{ccc}
			\toprule
			$v_1$ & $v_2$ & $v_3$ \\
			\midrule
			$c$ & $d$ & $c$ \\
			$d$ & $c$ & $d$ \\
			\midrule
			$b$ & $b$ & $e$ \\
			$e$ & $e$ & $b$ \\
			\midrule
			$f$ & $a$ & $a$ \\
			$g$ & $f$ & $f$ \\
			$a$ & $g$ & $g$ \\
			\midrule
			$h$ & $i$ & $h$ \\
			$i$ & $h$ & $i$ \\
			\bottomrule
		\end{tabular}
	\end{minipage}
	\quad
	\begin{minipage}{0.8\linewidth}
		The profile on the left is single-peaked on 
		$h \lhd g \lhd f \lhd b \lhd c \lhd d \lhd e \lhd a \lhd i$. 
		This profile has four non-empty common prefixes, indicated using horizontal lines: 
		$\{c,d\}, \{b,c,d,e\}, \{a,b,c,d,e,f,g\}, \{a, b, c, d, e, f, g, h, i\}$. 
		One can arrange these in nested `boxes', 
		where the order in each box follows $\lhd$, as below:
		\begin{center}
		\fbox{$h$ \fbox{$g$\:$f$ \fbox{$b$ \fbox{$c$\:$d$} $e$}  $a$} $i$}
		\end{center}
		Then, each axis on which the profile is single-peaked can be obtained by deciding, 
		for each box, whether to reverse it. Thus, there are $2^4 = 16$ different axes.
	\end{minipage}
\end{examplebox}

As we have seen, the problem of recognizing single-peaked preferences is very well understood.
However, this does not extend to multidimensional single-peaked preferences (\Cref{def:multidim-sp}).
Indeed, no algorithms or complexity classifications for two or more dimensions are known.

\begin{open}\label{open:multi-dimsp}
What is the complexity of recognizing $d$-dimensional (hereditary) single-peaked preferences 
for $d\ge 2$?
\end{open}

\subsection{Algorithms for Single-Crossing Preferences}
\label{sec:recog:sc}

For single-crossing preferences, we can use a reduction to the consecutive ones problem
to obtain an $O(nm^2)$ algorithm; a more direct 
combinatorial approach leads to an $O(nm\log m)$ algorithm. 
However, in contrast to single-peaked preferences,
a linear-time recognition algorithm for this domain is not known.
 
\subsubsection*{Reduction to Consecutive Ones}

We have already seen that for single-peaked preferences the recognition problem 
can be reduced to the consecutive ones problem. This is also the case for single-crossing  
preferences \citep{bredereck2013characterization}.

Given a profile $P$ with $m$ alternatives and $n$ voters, 
we construct an $m^2 \times n$ matrix $A$ as follows: 
we introduce a column for each voter, and one row for each ordered pair $(a,b)\in A\times A$. 
We set
\[ A( (a,b), i) = 
\begin{cases}
1 & \text{if } a \succ_i b, \\
0 & \text{if } b \succ_i a.
\end{cases} \]
Then any ordering of the columns of $A$ witnessing the consecutive ones property corresponds to a 
permutation of the profile $P$ making it single-crossing in the given order 
(see \Cref{ex:sc-c1p}). 
Since we can check the consecutive ones property in linear time (\Cref{thm:consones}), 
this gives an $O(nm^2)$ 
time algorithm for recognizing single-crossing preferences. 

\begin{examplebox}
	{Reducing the recognition of single-crossing profiles to the consecutive ones problem.}
	{sc-c1p}
	\centering
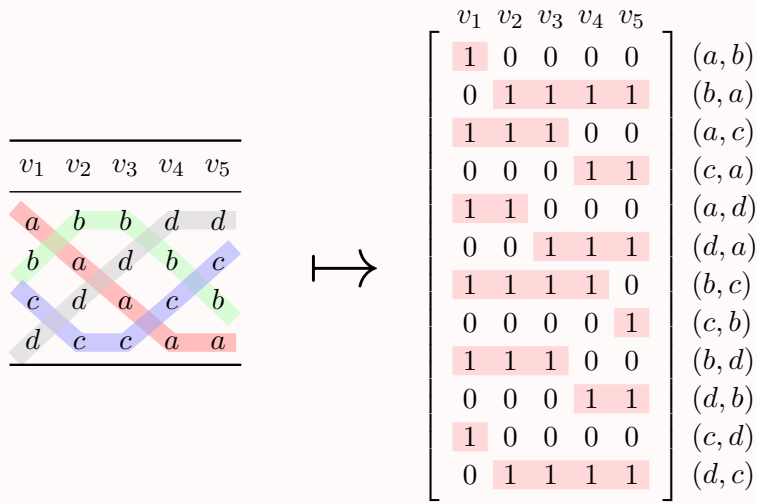
	\begin{minipage}{0.23\textwidth}
		\begin{tikzpicture}
			\matrix (m) [matrix of nodes] {
				\toprule
				$v_1$ & $v_2$ & $v_3$ & $v_4$ & $v_5$ \\
				\midrule
				$a$ & $b$ & $b$ & $d$ & $d$ \\
				$b$ & $a$ & $d$ & $b$ & $c$ \\
				$c$ & $d$ & $a$ & $c$ & $b$ \\
				$d$ & $c$ & $c$ & $a$ & $a$ \\				
				\bottomrule \\
			};
			
			\begin{tikzbackground}
				\draw[line width=7pt, red!50, draw opacity=0.5, transform canvas={yshift=1mm}] 
				(m-2-1.north west) -- (m-3-2.center) -- (m-4-3.center) -- (m-5-4.center) -- (m-5-5.east);

				\draw[line width=7pt, green!30, draw opacity=0.5, transform canvas={yshift=1mm}] 
				(m-3-1.south west) -- (m-2-2.center) -- (m-2-3.center) -- (m-3-4.center) -- (m-4-5.south east);

				\draw[line width=7pt, blue!40, draw opacity=0.5, transform canvas={yshift=1mm}] 
				(m-4-1.north west) -- (m-5-2.center) -- (m-5-3.center) -- (m-4-4.center) -- (m-3-5.north east);

				\draw[line width=7pt, black!20, draw opacity=0.5, transform canvas={yshift=1mm}] 
				(m-5-1.south west) -- (m-4-2.center) -- (m-3-3.center) -- (m-2-4.center) -- (m-2-5.east);
			\end{tikzbackground}
		\end{tikzpicture}
	\end{minipage}
	\quad
	\raisebox{-10pt}{\scalebox{2.5}{$\mapsto$}}
	\quad
	\begin{minipage}{0.4\textwidth}
		\begin{tikzpicture}
			[decoration=brace,
			highlight/.style={line width=11pt, red!15},
			rowlabel/.style={anchor=base, text width=1.8cm}]
			\matrix (m) [matrix of math nodes,left delimiter={[},right delimiter={]}, column sep=0.5ex] {
				1 & 0 & 0 & 0 & 0 \\ %
				0 & 1 & 1 & 1 & 1 \\ %
				1 & 1 & 1 & 0 & 0 \\ %
				0 & 0 & 0 & 1 & 1 \\ %
				1 & 1 & 0 & 0 & 0 \\ %
				0 & 0 & 1 & 1 & 1 \\ %
				1 & 1 & 1 & 1 & 0 \\ %
				0 & 0 & 0 & 0 & 1 \\ %
				1 & 1 & 1 & 0 & 0 \\ %
				0 & 0 & 0 & 1 & 1 \\ %
				1 & 0 & 0 & 0 & 0 \\ %
				0 & 1 & 1 & 1 & 1 \\ %
			};
			\node [anchor=base] (a) at ($(m-1-1.north)+(0,5pt)$) {$v_1$};
			\node [anchor=base] (b) at ($(m-1-2.north)+(0,5pt)$) {$v_2$};
			\node [anchor=base] (c) at ($(m-1-3.north)+(0,5pt)$) {$v_3$};
			\node [anchor=base] (d) at ($(m-1-4.north)+(0,5pt)$) {$v_4$};
			\node [anchor=base] (e) at ($(m-1-5.north)+(0,5pt)$) {$v_5$};
			\node [rowlabel] (row1) at ($(m-1-5.east)+(1.5cm,-3pt)$)  {$(a,b)$};
			\node [rowlabel] (row2) at ($(m-2-5.east)+(1.5cm,-3pt)$)  {$(b,a)$};
			\node [rowlabel] (row3) at ($(m-3-5.east)+(1.5cm,-3pt)$)  {$(a,c)$};
			\node [rowlabel] (row4) at ($(m-4-5.east)+(1.5cm,-3pt)$)  {$(c,a)$};
			\node [rowlabel] (row5) at ($(m-5-5.east)+(1.5cm,-3pt)$)  {$(a,d)$};
			\node [rowlabel] (row6) at ($(m-6-5.east)+(1.5cm,-3pt)$)  {$(d,a)$};
			\node [rowlabel] (row7) at ($(m-7-5.east)+(1.5cm,-3pt)$)  {$(b,c)$};
			\node [rowlabel] (row8) at ($(m-8-5.east)+(1.5cm,-3pt)$)  {$(c,b)$};
			\node [rowlabel] (row9) at ($(m-9-5.east)+(1.5cm,-3pt)$)  {$(b,d)$};
			\node [rowlabel] (row10) at ($(m-10-5.east)+(1.5cm,-3pt)$)  {$(d,b)$};
			\node [rowlabel] (row11) at ($(m-11-5.east)+(1.5cm,-3pt)$)  {$(c,d)$};
			\node [rowlabel] (row12) at ($(m-12-5.east)+(1.5cm,-3pt)$)  {$(d,c)$};
			\begin{tikzbackground}
				\draw[highlight] (m-1-1.west) -- (m-1-1.east);
				\draw[highlight] (m-2-2.west) -- (m-2-5.east);
				\draw[highlight] (m-3-1.west) -- (m-3-3.east);
				\draw[highlight] (m-4-4.west) -- (m-4-5.east);
				\draw[highlight] (m-5-1.west) -- (m-5-2.east);
				\draw[highlight] (m-6-3.west) -- (m-6-5.east);
				\draw[highlight] (m-7-1.west) -- (m-7-4.east);
				\draw[highlight] (m-8-5.west) -- (m-8-5.east);
				\draw[highlight] (m-9-1.west) -- (m-9-3.east);
				\draw[highlight] (m-10-4.west) -- (m-10-5.east);
				\draw[highlight] (m-11-1.west) -- (m-11-1.east);
				\draw[highlight] (m-12-2.west) -- (m-12-5.east);
			\end{tikzbackground}		
		\end{tikzpicture}
	\end{minipage}
\end{examplebox}

\subsubsection*{Refining Partitions}
\citet{bredereck2013characterization} also give a direct recognition algorithm, which stores the 
set of voters in an ordered partition. Initially, all the voters are in a single partition cell. 
Then, the algorithm iterates through all pairs $(a,b)$ of alternatives, at each step separating the 
voters with $a \succ b$ from the voters with $b \succ a$. If the algorithm does not discover a 
contradiction, it then returns a linear ordering of the voters that is compatible with the ordered 
partition. This algorithm takes $O(m^2n)$ time and can be implemented in a certifying way: 
if the given profile is not single-crossing, the algorithm returns a small \emph{forbidden subprofile} 
to witness this (see \Cref{sec:subprofiles}). Another advantage of this approach
is that it extends naturally to preferences that are single-crossing on trees 
(see \citet{kung2015sorting} and \Cref{sec:recog:sct}).

\subsubsection*{Recognizing Profiles That Are Single-Crossing in the Given Order}

Before we go into describing faster and more involved recognition algorithms, it will be useful to 
carefully consider the problem of determining whether a profile $P = (v_1,\dots,v_n)$ is 
single-crossing \emph{in the given order}, that is, whether we are already done. This task 
can be accomplished by a straightforward $O(m^2n)$ time algorithm: 
we transform each vote into
an $m\times m$ pairwise comparison matrix in time $O(m^2)$ per vote (so that we can decide whether $a\succ_i b$ in constant time), 
and then, for each of the $m^2$ pairs $(a,b)\in A 
\times A$, go through the profile in the given order and check that there is at most one crossing. 
We will now show that, with a bit more sophistication, 
this check can actually be performed in 
$O(nm\log m)$ time.

\begin{definition}\label{def:kendall}
The \defemph{Kendall-tau distance} between two rankings $\succ_1$ and $\succ_2$ over $A$
is the number of pairwise comparisons that they disagree on: 
$K({\succ_1},{\succ_2}) = |\{ (a,b)\in A \times A : a\succ_1 b \text{ and } b \succ_2 a \}|$. 
Let us write $\Delta({\succ_1},{\succ_2}) = \{ (a,b)\in A \times A : a\succ_1 b \text{ and } 
                                                                     b \succ_2 a \}$, 
so that $K({\succ_1},{\succ_2}) = |\Delta({\succ_1},{\succ_2})|$.
\end{definition}

In what follows, when considering the Kendall-tau distance between two votes $v_i$ and $v_j$ 
in a profile~$P$, for brevity, we will write $K[i,j]$ instead of $K(v_i, v_j)$.

Our algorithms make use of the observation that
if $P$ is single-crossing in the given order, then, as we go from left to right, the 
Kendall-tau distance $K[1,i]$ to the leftmost voter increases. 
To see this, observe that if voter $i$ disagrees
with voter $1$ on some pair of alternatives $(a, b)$ then so does each voter $j$ with $j>i$, i.e., 
$\Delta[v_1, v_i]\subseteq \Delta[v_1, v_j]$ and hence $K[1, i]\le K[1, j]$;
further, $K[1, i]=K[1, j]$ is only possible if $v_i=v_j$.
We will also need the following two results:

\begin{proposition}
	The Kendall-tau distance between two linear orders $u$ and $v$ over $m$ alternatives 
	can be computed in $O(m\log m)$ time.
\end{proposition}
\begin{proof}
	Treat $u$ as the `correct' ordering of the alternatives and use merge sort to sort 
	the list $v$ into the ordering $u$. While doing this, we can keep track of the number 
	of swaps required to do so.
\end{proof}

The fastest known algorithm for computing the Kendall-tau distance takes 
	$O(m\sqrt{\log m})$ time \citep{chanP10}.

\begin{proposition}
	\label{prop:kendall-tau-triangle-ineq}
	The Kendall-tau distance satisfies the triangle inequality
	\[  K(u, w) \le K(u, v) + K(v, w), \]
	with equality if and only if $\Delta(u,w) \supseteq \Delta(u, v)$.
\end{proposition}
\begin{proof}
	The Kendall-tau distance can be rewritten as 
	\[K(u,w) = \sum_{(a,b)\in A \times A} 
	{\mathbf{1}}_{(a,b) \in \Delta(u,w)},\] 
	where ${\mathbf{1}}_\phi$ is the indicator of whether $\phi$ is true.

	To see the triangle inequality, note that whenever $u$ and $w$ disagree on a pair $(a,b)$, 
	this disagreement must be present either between $u$ and $v$, or between $v$ and $w$. 
	Thus, ${\mathbf{1}}_{(a,b) \in \Delta(u,w)} \le 
	       {\mathbf{1}}_{(a,b) \in \Delta(u,v)} + 
	       {\mathbf{1}}_{(a,b) \in \Delta(v,w)}$ for all $(a,b)\in A\times A$. 
	Summing over all pairs, we obtain the triangle inequality.

	Given the inequality shown above, equality occurs if and only if 
	${\mathbf 1}_{(a,b) \in \Delta(u,w)} = {\mathbf 1}_{(a,b) \in \Delta(u,v)} + 
	 {\mathbf 1}_{(a,b) \in \Delta(v,w)}$ for all $(a,b)\in A\times A$. 
	
	($\Rightarrow$): Assume these equalities hold, and let $(a,b) \in \Delta(u,v)$. 
	Then we must have ${\mathbf 1}_{(a,b) \in \Delta(u,w)} = 1$, 
	so that $(a,b) \in \Delta(u,w)$. Hence $\Delta(u,w) \supseteq \Delta(u, v)$. 
	
	($\Leftarrow$): Suppose $\Delta(u,w) \supseteq \Delta(u, v)$. Let 
	$(a,b)\in A\times A$ be a pair of alternatives. 
	We show that ${\mathbf 1}_{(a,b) \in \Delta(u,w)} = 
	{\mathbf 1}_{(a,b) \in \Delta(u,v)} + {\mathbf 1}_{(a,b) \in \Delta(v,w)}$. 
	Suppose first that $(a,b) \in \Delta(u, v)$. By assumption, $(a,b) \in \Delta(u, w)$. 
	Since both $v$ and $w$ disagree with $u$ on $(a, b)$, they must agree
	with each other, so $(a,b) \not\in \Delta(v, w)$. 
	Hence the equality holds. 
	Alternatively, suppose that $(a,b) \not\in \Delta(u, v)$, so that $u$ and $v$ 
	agree on $(a,b)$. Then $u$ and $w$ disagree on $(a,b)$ if and only if 
	$v$ and $w$ disagree, again confirming the equality.
\end{proof}

The statement about equality in \Cref{prop:kendall-tau-triangle-ineq} suggests a way of reasoning 
about the set of up to $O(m^2)$ disagreements without having to store a set of this size 
explicitly. This gives us \Cref{alg:sc-given-order}, which achieves the promised time bound: for each 
voter, it calculates two Kendall-tau distances. The correctness of this algorithm 
follows immediately from the following proposition.

\begin{proposition}
	Let $P = (v_1, \dots, v_n)$ be a profile. The following statements are equivalent:
	\begin{enumerate}
		\item $P$ is single-crossing in the given order.
		\item $\Delta[1, i] \subseteq \Delta[1, i+1]$ for all $1 \le i \le n-1$.
		\item $K[1, i+1] = K[1, i] + K[i, i+1]$ for all $1 \le i \le n-1$.
	\end{enumerate}
\end{proposition}

\begin{algorithm}[H]
	\DontPrintSemicolon
	\KwIn{A profile $P=(v_1,\dots,v_n)$ over $A$}
	\KwOut{Is $P$ single-crossing in the given order?}
	\For{each vote $v_i$ in $P$ with $i\ge 2$}
	{
		calculate the distances $K[1, i]$, $K[i, i+1]$, and $K[1, i+1]$\;
		\If{$K[1, i] + K[i, i+1] \neq K[1, i+1]$}
		{ 
			\Return{``$P$ is not single-crossing in the given order''}\;
		}
	}
	\Return{``$P$ is single-crossing in the given order''}\;
	\caption{Recognizing profiles single-crossing in the given order in $O(nm\log m)$ time}
	\label{alg:sc-given-order}
\end{algorithm}

\subsubsection*{Guessing the Leftmost Voter}

Based on the insights of the above procedure, we can give a very simple (but slow) 
recognition algorithm, similar to an algorithm proposed by \citet{elkind2012clone}. It proceeds by 
guessing which voter will appear in the leftmost position in the single-crossing order. 
For each of the $n$ possible guesses, we then sort the 
remaining voters in increasing order of their Kendall-tau distance to the leftmost voter, and 
check whether the resulting profile is single-crossing in the given order. 

In each iteration, we calculate Kendall-tau distances in 
$O(nm\log m)$ time, sort the voters in $O(n \log n)$ time, and spend $O(nm 
\log m)$ time invoking \Cref{alg:sc-given-order}. Thus, overall this process can be 
implemented in $O(n(n\log n +nm\log m))$ time. If 
we assume that all the votes in $P$ are distinct (and hence $n = O(m^2)$ by 
\Cref{prop:sc-binom}), then this time bound can be simplified
to $O(n^2m\log m)$; the same time bound can be achieved without 
assuming distinctness of the votes using the approach we describe 
in the context of the next algorithm.

\begin{algorithm}[H]
	\DontPrintSemicolon
	\KwIn{A profile $P$ over $A$, in which all votes are pairwise distinct}
	\KwOut{A single-crossing ordering of $P$, if one exists}
	\For{each voter $v_i$ in $P$}
	{
		sort the voters in $P$ according to increasing Kendall-tau distance from $v_i$\;
		\If{this ordering is single-crossing}
		{ 
			\Return{this ordering of $P$}\;
		}
	}
	\Return{``$P$ is not single-crossing''}\;
\caption{Recognizing single-crossing profiles in $O(n^2m\log m)$ time}
\label{alg:sc-naive-sorting}
\end{algorithm}

\subsubsection*{Fast Recognition by Sorting}

The above na\"ive algorithm can be modified in a way that does not require us to guess the 
leftmost voter. The resulting algorithm runs in time $O(nm\log m)$, and is essentially a faster 
implementation of an algorithm due to \citet{doignon1994polynomial}
That algorithm
repositions voters at each iteration, leading to a 
worse time bound of $O(n^2 + nm\log m)$. 
Our algorithm calculates a value $\textit{score}[i]$ for each voter $i$ that is based on Kendall-tau distances, and then reorders the input profile $P$ in increasing order of the score. If $P$ is single-crossing then the resulting ordering will be single-crossing.

\begin{algorithm}[H]
	\DontPrintSemicolon
	\KwIn{A profile $P$ over $A$} %
	\KwOut{A single-crossing ordering of $P$, if one exists}
	\lnl{alg:sc-fastest:1-2-distinct}
	Ensure that voters 1 and 2 have different preference orders (otherwise relabel)\;
	Calculate $K[1, 2]$\;
	\textit{score} $\gets$ empty array indexed by the voters 
	\tcp*{will have $\textit{score}[i] = \pm K[1,i]$}
	\textit{score}$[1] \gets 0$; \textit{score}$[2] \gets +K[1, 2]$\;
	\For{each voter $i\in N \setminus \{1,2\}$}
	{
		Calculate the distances $K[1,i]$ and $K[2,i]$\;
		\lnl{alg:sc-fastest:if-between} \If{$K[1,2] = K[1,i] + K[i,2]$}
		{ 
			\textit{score}$[i] \gets +K[1,i]$ \tcp*{$i$ goes between $1$ and $2$}
		} 
		\lnl{alg:sc-fastest:if-right} \ElseIf{$K[1,i] = K[1,2] + K[2,i]$}
		{ 
			\textit{score}$[i] \gets +K[1,i]$ \tcp*{$i$ goes to the right of $2$} 
		}
		\lnl{alg:sc-fastest:if-left} \ElseIf{$K[2,i] = K[1,i] + K[1,2]$}
		{ 
			\textit{score}$[i] \gets -K[1,i]$ \tcp*{$i$ goes to the left of $1$} 
		} 
		\Else
		{
			\Return{``$P$ is not single-crossing''}\;}
	}
	// Order the voters by their score. Fastest way depends on whether $n < m$ or not. \;
	\If{$n < m$} {
		\lnl{alg:sc-fastest:reorder-1}
		$P' \gets$ the list $P$ sorted in order of $\textit{score}[i]$ using $O(n \log n)$ time\;}
	\ElseIf{$n\ge m$} {
	$B\gets$ an array indexed by $-m^2, \dots, 0, \dots, m^2$, each entry containing an empty list\;
	\For{each voter $i \in N$} {
		add $v_i$ to the end of the list $B[\textit{score}[i]]$\;}
	\lnl{alg:sc-fastest:reorder-2}
	$P'\gets B[-m^2] \oplus \cdots \oplus B[m^2]$\;
	}
	\lnl{alg:sc-fastest:check-scgo} \If{$P'$ is single-crossing in the given order (use \Cref{alg:sc-given-order})}
	{
		\lnl{alg:sc-fastest:yes}\Return{$P'$}\;}
	\Else{
		\lnl{alg:sc-fastest:check-sc}\Return{``$P$ is not single-crossing''}\;}
	\caption{Recognizing single-crossing profiles in $O(nm\log m)$ time}
	\label{alg:sc-fastest}
\end{algorithm}

\begin{theorem}\label{thm:recognizing-sc}
	\Cref{alg:sc-fastest} recognizes single-crossing profiles 
	in $O(nm\log m)$ time.
\end{theorem}
\begin{proof}
	First we establish the time bound.
	A voter whose vote is distinct from $v_1$
	(line~\ref{alg:sc-fastest:1-2-distinct}) can be found in $O(nm)$ time 
	by starting at voter~1 and then scanning the profile until we find a $v_i$
	such that $v_i\neq v_1$ (if no such voter can be found, the profile is trivially
	single-crossing).
	Now, for each voter in $P$ we calculate the Kendall-tau distances to $v_1$ and $v_2$, 
	taking $O(nm\log m)$ time.
	If $n < m$, creating the profile $P'$ takes $O(n \log n)$ time, which is in $O(m \log m)$ by the assumption that $n < m$.
	If $n \ge m$, creating the profile $P'$ takes $O(m^2)$ time which is in $O(nm)$ by the assumption that $n \ge m$. In either case this step was performed in time $O(nm\log m)$.
	Finally, we check if the profile is single-crossing in the given order
	in $O(nm\log m)$ time using \Cref{alg:sc-given-order}.
	
	For correctness, we show that the algorithm reports that $P$ is single-crossing if and only if 
	$P$ is indeed single-crossing. One direction is easy: if the algorithm returns `yes' 
	in line~\ref{alg:sc-fastest:yes}, then $P'$ is single-crossing in the returned order 
	by the check in line~\ref{alg:sc-fastest:check-sc} and the correctness of 
	\Cref{alg:sc-given-order}, and $P'$ is a single-crossing reordering of the votes in $P$.
	
	Conversely, assume that the input profile $P$ is single-crossing. Fix some reordering 
	$\widehat{P}$ of $P$ that is single-crossing in the given order and such that voter $2$ 
	is placed to the right of voter $1$. As in the algorithm, write $K[i,j]$ for the Kendall-tau distance between $v_i$ and $v_j$.
	We claim that, after running the for-loop, it holds that 
	for all $i\neq 1$ that are placed to the right of $1$ in $\widehat{P}$, we have $\textit{score}[i] = K[1,i]$, and for all $i \neq 1$ placed to the left of $1$ in $\widehat{P}$, we have $\textit{score}[i] = -K[1,i]$.
	
	Clearly this holds for voter 2.
	For a voter $i$ with the same vote as voter 1, the conditions of all three if-conditions are satisfied, but in each case the algorithm sets $\textit{score}[i] = 0$, which is compatible with our claim no matter if $i$ is placed to the left or to the right of voter 1. Next, for a voter $i$ with the same vote as voter $2$, the first two if-conditions are satisfied (but not the third since $K[1,2] > 0$) and in either case the algorithm sets $\textit{score}[i] = K[1,2]$, which agrees with our claim since $i$ must be placed to the right of $1$ in $\widehat{P}$.
	For all other voters $i$, at most one of the conditions in 
	lines~\ref{alg:sc-fastest:if-between}, \ref{alg:sc-fastest:if-right}, 
	and~\ref{alg:sc-fastest:if-left} can be true, because in each case the value 
	on the left-hand side needs to be the uniquely largest of the three.
	We will now do a case analysis on the position of voter $i$ in $\widehat{P}$.
	
	\begin{itemize}
		\item If $i$ is between $1$ and $2$ in $\widehat{P}$, then we have 
		      $\Delta(v_1, v_i) \subseteq \Delta(v_1,v_2)$ by the single-crossing property, 
		      and hence $K[1,2] = K[1,i] + K[i,2]$, which is the condition of 
		      line~\ref{alg:sc-fastest:if-between}. Thus the algorithm sets $\textit{score}[i] = K[1,2]$ as required by the claim.
		\item If $i$ is to the right of $2$ in $\widehat{P}$, then we have 
		      $\Delta(v_1,v_2) \subseteq \Delta(v_1,v_i)$ by the single-crossing property, 
		      and hence $K[1,i] = K[1,2] + K[2,i]$, which is the condition of 
		      line~\ref{alg:sc-fastest:if-right}. Thus the algorithm sets $\textit{score}[i] = K[1,2]$ as required by the claim.
		\item If $i$ is to the left of $1$ in $\widehat{P}$, then we have 
		      $\Delta(v_i, v_1) \subseteq \Delta(v_i,v_2)$ by the single-crossing property, 
		      and hence $K[i,2] = K[i,1] + K[1,2]$, which is the condition of 
		      line~\ref{alg:sc-fastest:if-left}. Thus the algorithm sets $\textit{score}[i] = -K[1,2]$ as required by the claim.
	\end{itemize}

	In each case, by choosing the sign of $\textit{score}[i]$, the algorithm correctly decides 
	whether to place $i$ to the right of voter $1$ (setting $\textit{score}[i] > 0$) or to the left  
	(setting $\textit{score}[i] < 0$). Now, because $\widehat{P}$ is in single-crossing order, we have that the Kendall-tau distance to voter $1$ increases as we scan from $1$ to the right or as we scan from $1$ to the left. It follows that we have $\textit{score}[i] = \textit{score}[j]$ if and only if $v_i = v_j$ (because by single-crossingness two voters on the same side of voter 1 with the same Kendall-tau distance to voter 1 must be identical) and that we have 
	$\textit{score}[i] < \textit{score}[j]$ if and only if $v_i \neq v_j$ and $i$ is positioned to the left of $j$	in $\widehat{P}$. Thus, any list $P'$ obtained by ordering votes in order of increasing $\textit{score}[i]$ coincides with $\widehat{P}$ (up to reordering voters with the same vote). Hence, the profile $P'$ constructed by the algorithm in line \ref{alg:sc-fastest:reorder-1} or line \ref{alg:sc-fastest:reorder-2} is single-crossing in the given order, and hence $P'$ passes the check in line~\ref{alg:sc-fastest:check-scgo} and the algorithm correctly returns $P'$ in line \ref{alg:sc-fastest:yes}.
 \end{proof}

\begin{open}
	Does there exist an $O(mn)$ time algorithm for recognizing single-crossing profiles?
\end{open}

\subsection{Algorithms and Complexity Results for Euclidean Preferences}
\label{sec:recog:euclid}
We start by sketching an algorithm for recognizing 1-Euclidean profiles; 
subsequently, we will discuss complexity results for $d$-Euclidean preferences, $d>1$.

\subsubsection*{1-Euclidean}

For $d=1$, we can capture the recognition problem by $n\cdot {\binom{m}{2}}$ 
constraints over $n+m$ variables: the embedding $x:N\cup A\to{\mathbb R}$ witnesses that the input profile
is $1$-Euclidean if for each voter $i$ and every pair of alternatives $(a, b)$ with $a\succ_i b$
we have 
$|x(i)-x(a)| < |x(i)-x(b)|$. 
While these constraints are not linear, they 
can be transformed into linear constraints if the order of the points $\{x(a): a\in A\}$
on the real line is known: the constraint $|x(i)-x(a)| < |x(i)-x(b)|$
can be replaced with $x(i)<\frac12(x(a)+x(b))$ if $x(a)<x(b)$
and $x(i)>\frac12(x(a)+x(b))$ if $x(a)>x(b)$. The problem of deciding if the input
profile is $1$-Euclidean then boils down to finding a feasible solution of the resulting
linear program (the reader may worry that our constraints involve strict inequalities, 
but this can be handled by replacing each constraint of the form $z>z'$ with $z\ge z'+1$,
as discussed, e.g., by \citet{elkind2014recognizing}). 

It remains to figure out how to order the alternatives on the line. 
One may be tempted to use an algorithm 
from \Cref{sec:recog:sp} for this purpose: indeed, a profile is $1$-Euclidean
only if it is single-peaked, so we can reject the input if it is not single-peaked, and otherwise
use the ordering provided by a single-peaked axis. The problem with this approach is 
that a $1$-Euclidean profile may be single-peaked with respect to several different
axes (indeed, potentially exponentially many of them), and some of these axes may be
unsuitable for our purposes. This issue is illustrated by the following example. 

\begin{examplebox}
	{Recognizing $1$-Euclidean preferences: choosing an axis}
	{1D-badaxis}
	\begin{wrapfigure}[7]{l}{0.1\linewidth}
	\begin{tabular}{cc}
		\toprule
		$v_1$ & $v_2$ \\
		\midrule
		$b$   & $c$   \\
		$c$   & $b$   \\
		$a$   & $d$   \\
		$d$   & $a$   \\
		\bottomrule
	\end{tabular}
	\end{wrapfigure}
	The profile $P$ on the left
	is single-peaked with respect to the axis $\lhd_1$
	given by $a\lhd_1 b\lhd_1 c\lhd_1 d$ as well as the axis $\lhd_2$
	given by $d\lhd_2 b\lhd_2 c\lhd_2 a$.
	There is a mapping $x: N\cup A\to{\mathbb R}$ that is consistent with $\lhd_1$
	and witnesses that $P$ is $1$-Euclidean: e.g., we can take 
	$x(a)=-4$, $x(b)=-1$, $x(c)=1$, $x(d)=4$ and co-locate each voter
	with its top alternative (i.e., $x(v_1)=-1$, $x(v_2)=1$).
	However, there is no mapping $x': N\cup A\to{\mathbb R}$ that is consistent
	with $\lhd_2$ and witnesses that $P$ is $1$-Euclidean. 

	To see this, consider a mapping $x': N\cup A\to {\mathbb R}$ 
	that is consistent with $\lhd_2$, i.e., satisfies $x'(d)<x'(b)<x'(c)<x'(a)$.
	Let $x_{bc}$ be the midpoint of the segment $[x'(b), x'(c)]$, and observe that 
	$b\succ_1 c$, $c\succ_2 b$  
	implies that $v_1$ needs to be placed to the left of $x_{bc}$, while
	             $v_2$ needs to be placed to the right of $x_{bc}$, so
			$x'(v_1) < x'(v_2)$.
	On the other hand, $a\succ_1 d$, $d\succ_2 a$,
	so, by considering the positions of $v_1$ and $v_2$ 
	with respect to the midpoint of the segment $[x'(d), x'(a)]$,
	we obtain $x'(v_1) > x'(v_2)$,
	a contradiction.
\end{examplebox}

The history of recognition algorithms for 1-Euclidean preferences is colorful. 
The first polynomial-time 
algorithm appears in the paper by 
\citet{doignon1994polynomial}, who start by developing an 
elegant theory of single-peaked and single-crossing preferences, and then use it to reason about 
1-Euclidean preferences. However, their work used the terminology of the field of psychometrics, 
such as `unidimensional unfolding' \citep{coombs1950psychological}. Perhaps as a consequence 
of that, until recently, economists and computational social choice theorists were not aware of this work. 
\citet{knoblauch2010recognizing} proposed a 
recognition algorithm that reasoned about single-peaked axes before turning to a linear program; 
in \citeyear{elkind2014recognizing}, \citeauthor{elkind2014recognizing} 
proposed a recognition algorithm that used a single-crossing order of the voters
before, again, turning to a linear program. Neither of these two papers cites
\citet{doignon1994polynomial}.

Out of all these algorithms, \citeauthor{doignon1994polynomial}'s way of \emph{combining} 
information derived from a single-peaked order of alternatives and 
a single-crossing order of voters is particularly elegant; below,  
we provide an outline of their proof.

\begin{proposition}[\citealp{doignon1994polynomial}]
	\label{prop:comp-axis-exists}
	Suppose $P = (v_1,\dots,v_n)$ is a profile that is single-crossing in the given order, 
	as well as single-peaked. Then there exists an axis $\lhd$ on which $P$ is single-peaked 
	such that the profile $P' = (\lhd, v_1,\dots, v_n, \textit{reverse}(\lhd))$ is still single-crossing. 
	Moreover, such an axis can be found in polynomial time.
\end{proposition}

We will say that an axis whose existence is established in \Cref{prop:comp-axis-exists}
is \emph{compatible} with $P$.
To prove existence, \citeauthor{doignon1994polynomial} use an inductive argument. Given the knowledge that a compatible axis exists, finding one is straightforward. First ensure that the input profile $P$ is ordered to be single-crossing, using an algorithm from \Cref{sec:recog:sc}. Then use \Cref{alg:sp} to obtain a single-peaked axis for $P$, but whenever we can make an arbitrary decision (in line \ref{sp-alg:arbitrary} of the algorithm), we make it in a way that allows $(\lhd, v_1, v_n, \textit{reverse}(\lhd))$ to be single-crossing.

To see the relevance of \Cref{prop:comp-axis-exists}, suppose we are considering a profile $P = (v_1,\dots,v_n)$ that is 1-Euclidean with embedding $x : N\cup A \to \mathbb R$. As we saw in the proof of \Cref{prop:euclid-implies-spsc}, 
the left-to-right ordering of the alternatives induced by $x$ is an axis $\lhd$ that $P$ 
is single-peaked on. Moreover, $\lhd$ is in fact compatible with $P$: just place a new voter $\ell$ 
to the left of the leftmost alternative in $x$, 
and place another voter $r$ to the right of the rightmost alternative. 
The profile described by this new embedding is exactly $(\lhd, v_1,\dots, v_n, \textit{reverse}(\lhd))$.

Thus, any ordering of the alternatives that can be induced by a 1-Euclidean embedding 
is a compatible axis. Very conveniently, the converse is also true.
\begin{theorem}[\citealp{doignon1994polynomial}]\label{thm:compatible}
	Suppose $P$ is a 1-Euclidean profile (and hence it is single-peaked and single-crossing). 
	Then for \emph{any} axis $\lhd$ that is compatible with $P$, there exists a 1-Euclidean embedding 
	$x: N\cup A \to \mathbb R$ of $P$  that satisfies
	\[ x(a) < x(b) \iff a \lhd b \qquad \forall a,b\in A. \]
\end{theorem}

\Cref{thm:compatible} tells us that, when ordering the alternatives
on a line for the purpose of constructing a linear program, 
we can take \emph{any} compatible axis $\lhd$: If our 
profile does not admit a Euclidean embedding agreeing with $\lhd$, 
then no other axis $\lhd'$ will work either. We can now see how to proceed:
check that the input profile is single-peaked and single-crossing, 
identify a compatible axis, and then use linear programming
to find the positions of voters and alternatives.

\begin{theorem}
	There is a polynomial-time algorithm that recognizes whether a given profile is 1-Euclidean.
\end{theorem}
\begin{proof}
	First, using the algorithms from \Cref{sec:recog:sp,sec:recog:sc}, check whether the given profile $P$
	is single-peaked and single-crossing. If not, $P$ is not 1-Euclidean. If yes, use the algorithm 
	from \Cref{prop:comp-axis-exists} to produce an axis $\lhd$ compatible with $P$. Then construct the 
	following linear feasibility program with variables $\{x(i): i\in N\}\cup\{x(a): a\in A\}$ 
	over $\mathbb R$ to check for the existence of an embedding 
	$x : N \cup A \to \mathbb R$ respecting $\lhd$.
	\begin{alignat*}{3}
	  x(i) + 1 &\le  \frac{x(a) + x(b)}{2} && \qquad\text{for all $a\lhd b$ and $i\in N$ with $a\succ_i b$,} \\
	  \frac{x(a) + x(b)}{2} +1 &\le x(i) && \qquad\text{for all $a\lhd b$ and $i\in N$ with $b\succ_i a$,} \\
	  x(a) + 1 &\le x(b) \vphantom{\dfrac12} && \qquad\text{for all $a\lhd b$.} %
	\end{alignat*}
\end{proof}

\begin{open}
	Can 1-Euclidean profiles be recognized by a purely `combinatorial' algorithm, 
	that is, an algorithm that does not depend on solving a linear program?
\end{open}
It should be noted that \Cref{thm:euclid-forbidden} from \Cref{sec:subprofiles} is, perhaps, bad news for 
this endeavor.

\subsubsection*{Higher Dimensions}

\paragraph{Complexity}
While the one-dimensional case admits a polynomial-time recognition algorithm, 
in higher dimensions this is unlikely to be the case. Indeed, by 
\Cref{thm:d-euclid-precision}, we might need exponentially many bits to write down a Euclidean embedding 
when one exists, so it is not even clear if the decision version of our problem is in NP. The right complexity 
class for this recognition problem turns out to be the little-known class $\exists\mathbb R$ (see 
\citealp{schaefer2013realizability} and \citealp{schaefer2015nash}).

\begin{theorem}[\citealp{peters2016recognising}]
	\label{thm:recognizing-d-Euclidean}
	For fixed $d \ge 2$, it is NP-hard to recognize whether a given profile is $d$-Euclidean. In fact, 
	this problem is equivalent to the existential theory of the reals (ETR), and thus 
	$\exists\mathbb R$-complete. It is contained in PSPACE.
\end{theorem}

Suppose $P$ is a $d$-Euclidean profile. Then, because \Cref{def:euclid} only uses strict inequalities, 
there must be a Euclidean embedding that only uses \emph{rational} coordinates; and by multiplying by the 
denominators we can assume that only integral coordinates are used. Thus, there is a Euclidean embedding 
$x : N\cup A \to \mathbb Z^d$. Now we can ask how big the integers in this representation need to be: 
do numbers that are polynomial in $n$, $m$ and $d$ suffice? 
Do singly exponential numbers, that is, numbers with polynomially many bits, suffice? 
The answer turns out to be `no':

\begin{theorem}[\citealp{peters2016recognising}, based on \citealp{mcdiarmid2013integer}]
	\label{thm:d-euclid-precision}
	For each fixed $d\ge 2$, there are $d$-Euclidean profiles with $n$ voters and $m$ alternatives such that every integral Euclidean embedding uses at least one coordinate that is $2^{2^{\Omega(n+m)}}$\!. On the other hand, every $d$-Euclidean profile can be realized by an integral Euclidean embedding whose coordinates are at most $2^{2^{O(n+m)}}$\!.
\end{theorem}

\paragraph{Heuristic Algorithms}
While the recognition problem is hard even for $d = 2$, heuristic algorithms have been proposed for this problem that categorize profiles into three groups: \emph{yes}-instances (together with an embedding), \emph{no}-instances (together with a forbidden subprofile), and \emph{unknown}.
The fist such algorithm was proposed by \citet{escoffier2023euclideanalgorithm}. It starts by searching for an occurrence of some known subprofiles with 4 candidates that are not 2-Euclidean. If it does not find such a \emph{no}-witness, it then attempts to find an embedding by repeatedly sampling potential locations for the alternatives using different probability distributions, until it reaches a time limit. \citet{dvořák2025practicalapproach} developed an improved heuristic algorithm using additional forbidden subprofiles, applying some reduction rules, and using ILP and QCP solvers. The resulting heuristic was able to classify all but 60 instances from PrefLib, with 98.7\% of PrefLib instances resolved in under 1 second.

\paragraph{Other Metrics}
When we defined $d$-Euclidean preferences in \Cref{sec:def:euclid}, we used the Euclidean $\ell_2$-metric. When using the arguably very natural $\ell_1$- or $\ell_\infty$-metrics for the definition, \citet{peters2016recognising} noted that the recognition problem is contained in NP.

\begin{open}
	Is the recognition problem for $d$-Euclidean preferences with the $\ell_1$- or $\ell_\infty$-metric NP-hard?
\end{open}

\subsection{Algorithms for Preferences Single-Peaked on a Tree}
\label{sec:recog:spt}

The linear-time algorithm for recognizing single-peaked profiles (\Cref{sec:recog:sp}) 
proceeds by building the underlying axis from the outside in, that is, starting at the 
extremes. We will now see that a similar approach is helpful for recognizing preferences that are 
single-peaked on a tree. The algorithm we present is due to 
\citet{trick1989recognizing}; it can decide
whether the input profile is single-peaked on a tree in time $O(m^2n)$.  
While it is not as fast as the linear-time algorithm for 
preferences single-peaked on a line, it is very intuitive and easy to describe.
The exposition in this section follows \citet{peters2020spt-journal}.

We have observed that if a profile is single-peaked on some axis $\lhd$, 
then every alternative that is ranked last by some voter must be placed 
at one of the endpoints of $\lhd$. The analog of this claim 
for a profile that is single-peaked on a tree can be formulated as follows.
\begin{proposition}
	Suppose $P$ is single-peaked on $T$, and suppose $a$ occurs as a bottom-most alternative
	of some voter $i$. Then $a$ is a leaf of $T$.
\end{proposition}
\begin{proof}
	The set $A\setminus\{a\}$ is a prefix of $v_i$ and hence must be connected in $T$. 
	This can only be the case if $a$ is a leaf of $T$.
\end{proof}
Note that this observation potentially allows for many different alternatives to appear 
in the bottom-most position. This is in contrast to preferences single-peaked on a line, 
where at most two alternatives can be bottom-ranked. Indeed, a path only has two leaves,
whereas an arbitrary tree can have up to $m-1$ many leaves.

Consider a bottom-ranked alternative $a$; we know that if our profile 
is single-peaked on some tree $T$, then $a$ is a leaf of $T$. Now, being a leaf, 
$a$ must have exactly one adjacent vertex. Which vertex could this be? 
The following simple observation provides an answer.

\begin{proposition}
	\label{prop:spt-leaf-attachment}
	Suppose a vote $v_i$ is single-peaked on the tree $T$, and $a\in A$ is a leaf of $T$, 
	adjacent to $b\in A$. Then either 
	\begin{enumerate}[(i)]
		\item $b \succ_i a$, or
		\item $a$ is $i$'s top-ranked alternative and $b$ is $i$'s second-ranked alternative.
	\end{enumerate}
\end{proposition}
\begin{proof}
	Suppose first that $a$ is not $i$'s top-ranked alternative, 
	and let $c= \top(v_i)$. Let $\pi$ be the (unique) path from $c$ to $a$; 
	note that $\pi$ passes through $b$. 
	Since $v_i$ is single-peaked on $T$, it is single-peaked on $\pi$, 
	and hence $i$'s preferences decrease along $\pi$. Since $\pi$ visits $b$ before $a$, 
	it follows that $b\succ_i a$.
	
	On the other hand, suppose that $a$ is $i$'s top-ranked alternative, 
	and $c$ is $i$'s second-ranked alternative. Then $\{a,c\}$ is a prefix segment of $v_i$. 
	Since $v_i$ is single-peaked on $T$, $\{a,c\}$ is connected in $T$, and hence forms an edge. 
	Thus, $a$ is adjacent $c$, so $c = b$, as required. 
\end{proof}

Thus, in our search for a neighbor of the leaf $a$, we can restrict our attention to alternatives 
$b$ such that for every voter $i$ one of the conditions (i) or (ii) in the proposition above
is satisfied. Let us write this down 
more formally. For each $i\in N$, we write $\operatorname{second}(v_i)$ for 
the alternative that is ranked second in $v_i$. For each $c\in A$ and $i\in N$, define
\[
B(v_i, c) = \begin{cases}
\{ c' : c' \succ_i c \} & \text{if } \top(v_i) \neq c, \\
\{ \operatorname{second}(v_i) \} & \text{if }  \top(v_i) = c.
\end{cases}
\]
Applying \Cref{prop:spt-leaf-attachment} to all voters gives us the following constraint 
for our choice of $b$.
\begin{corollary}
	Suppose a profile is single-peaked on $T$, and $a\in A$ is a leaf of $T$. 
	Then $a$ must be adjacent to an element of $B(a) := \bigcap_{i\in N} B(v_i,a)$.
\end{corollary}
We have established that it is necessary for a leaf $a$ to be adjacent to some alternative in $B(a)$. 
It turns out that if the profile is single-peaked on a tree, we can actually attach $a$ to \emph{any} 
of the alternatives in $B(a)$.

\begin{proposition}
	If a profile $P$ is single-peaked on some tree, and $a$ occurs bottom-ranked, 
	then for each $b\in B(a)$ there exists a tree $T$ in which $a$ is a leaf adjacent to $b$, 
	and $P$ is single-peaked on $T$.
\end{proposition}

These observations suggest that a recognition algorithm can proceed recursively: we identify a leaf $a$, 
compute $B(a)$, then obtain a tree $T$ for a restriction of $P$ to $A\setminus\{a\}$, pick
an element $b\in B(a)$ in this tree, and connect $a$ to $b$.
This is in fact precisely what \citeauthor{trick1989recognizing}'s algorithm (\Cref{alg:trick-spt}) 
does (we present it in non-recursive form to simplify the analysis of the running time).
However, to argue that \Cref{alg:trick-spt} is correct, we need
to be a bit more careful. Namely, in \Cref{sec:def:spt} we have seen that the class of profiles
single-peaked on trees is not closed under alternative deletion, so deleting $a$ may 
potentially result in a profile that is not single-peaked on a tree. 
However, since $a$ is chosen to be a \emph{leaf}, this is not an issue:

\begin{proposition}
	A profile $P$ over $A$ is single-peaked on a tree $T$ with leaf $a\in A$ if and only 
	if the profile $P|_{A\setminus \{a\}}$ is single-peaked on $T - \{a\}$.
\end{proposition}

Putting all these ideas together, we obtain \citeauthor{trick1989recognizing}'s algorithm.

\begin{algorithm}[H]
	\DontPrintSemicolon
	\KwIn{A profile $P$ over $A$}
	\KwOut{A tree on $A$ such that $P$ is single-peaked on $T$, if one exists}
	$T \gets $ the empty tree on vertex set $A$ \;
	$A' \gets A$\;
	\While{$|A'| \ge 3$}{
		$L \gets $ the set of last-ranked alternatives in $P|_{A'}$\;
		\For{each alternative $a\in L$}{
			$B_a \gets \bigcap_{i\in N} B(v_i|_{A'}, a)$\;
			\eIf{$B_a \neq \varnothing$}{
				$b \gets$ an arbitrary alternative from $B_a$\;
				add edge $\{b,a\}$ to $T$\;
				$A' \gets A' \setminus \{a\}$\;
			}{
				\Return{$P$ is not single-peaked on any tree}\;
}
		}
	}
	\If{$|A'| = 2$ with $A' = \{c,d\}$}{add edge $\{c,d\}$ to $T$\;}
	\Return{$T$}
\caption{Recognizing profiles single-peaked on trees in $O(m^2n)$ time \citep{trick1989recognizing}}
\label{alg:trick-spt}
\end{algorithm}

\medskip
Recall that a profile is \emph{narcissistic} if every alternative is ranked first by some voter, 
that is, if for each $a\in A$ there is an $i\in N$ with $\top(v_i) = a$. Interestingly, 
for narcissistic profiles the running time of \citeauthor{trick1989recognizing}'s recognition algorithm 
can be improved to $O(nm)$. Intuitively, this is because for narcissistic profiles, 
we can make heavy use of condition~(ii) in \Cref{prop:spt-leaf-attachment}.

\begin{theorem}[\citealp{trick1989recognizing}]\label{thm:recognizing-spt}
	For a narcissistic profile $P$, we can decide in time $O(mn)$ whether 
	$P$ is single-peaked on a tree $T$. Moreover, if $P$ is single-peaked
	on a tree, this tree is unique.
\end{theorem}

In \Cref{sec:problems}, we will see that many hard voting problems remain hard even for profiles 
single-peaked on a tree. However, one can find algorithms that perform well when given a profile that is 
single-peaked on a tree $T$ that satisfies some additional properties, such as having few leaves. In 
order to use these algorithms, we need to be able to find such a `nice' tree if it exists. It turns out 
that, by studying \Cref{alg:trick-spt} more closely, we can identify `nice' trees 
for various notions of niceness.

\begin{theorem}[\citealp{peters2020spt-journal}]
	\label{thm:nice-trees}
	Suppose a profile $P$ is single-peaked on some tree, and let ${\cal T}(P)$ be the set of all such trees 
	Then we can find in polynomial time 
	an element of ${\cal T}(P)$ with 
	(i)  a minimum number of leaves, 
	(ii) a minimum number of internal vertices, 
	(iii) a minimum diameter, 
	(iv) a minimum pathwidth, 
	(v) a minimum max-degree. 
We can also decide in polynomial time whether a profile is single-peaked on a star, 
caterpillar, lobster, or subdivision of a star.
\end{theorem}
Interestingly, \Cref{thm:nice-trees} does not extend to \emph{all} questions of this type.

\begin{theorem}[\citealp{peters2020spt-journal}]
	Given a profile $P$ over $A$ and a tree $T$ with $|V(T)| = |A|$, it is NP-complete to decide 
	whether the vertices of $T$ can be labeled with alternatives so that $P$ becomes single-peaked 
	on that (labeled) tree. 
\end{theorem}

\citet{peters2020spt-journal} 
also show that it is NP-complete to decide whether a profile is single-peaked on a tree whose 
non-leaf vertices all have degree~4.

\citet{escoffier2020recognizing} consider the recognition problem in the framework of single-peaked 
preferences on general graphs. One of their contributions is a linear program that can be used to 
recognize preferences single-peaked on a tree. More generally, they describe an integer linear program 
for finding a graph $G$ with a minimum number of edges or a minimum degree
such that the input profile is single-peaked on $G$, but prove that these
problems are NP-hard for general graphs.

\subsection{Algorithms for Preferences Single-Peaked on a Circle}
\label{sec:recog:spc}

In the beginning of \Cref{sec:recog:sp}, we saw that single-peaked preferences can be recognized 
using a straightforward reduction to the consecutive ones problem, which is known to be polynomial-time 
solvable. For the task of recognizing preference profiles that are single-peaked on a circle, 
the same reduction yields an instance of the \emph{circular} ones problem, 
which is also polynomial-time solvable by reduction to the consecutive ones problem 
\citep{booth1976testing,dom2009consecutive}. The resulting time complexity is $O(m^2n)$. This approach 
also works for weak orders. For the case of linear orders, \citet{peters2017spoc} develop a linear-time 
algorithm that runs in $O(mn)$ time, using a reduction to the problem of deciding
whether certain profiles of weak orders are single-peaked.

\subsection{Algorithms for Preferences Single-Crossing on a Tree}
\label{sec:recog:sct}

The observations that helped us develop efficient algorithms for recognizing single-crossing preference
profiles are also useful for recognizing preference profiles that are single-crossing on a tree.
The approach described below can be viewed as a specialization of the algorithm
by \citet{clearwater2015generalizing} (which works for all median graphs)
to the special case of trees; for trees, this approach results in a particularly
simple algorithm.

It will be convenient to assume that in the input profile every vote occurs at most once;
the following observation shows that this assumption is without loss of generality.

\begin{proposition}
Consider a profile $P$, and let $\overline{P}$ be a profile that contains exactly 
one copy of each linear order that appears in $P$. 
Then $P$ is single-crossing on a tree if and only if $\overline{P}$ is 
single-crossing on a tree.
\end{proposition}
\begin{proof}
If $\overline{P}$ is single-crossing on a tree $\overline{T}$,
we can transform $\overline{T}$ into a tree $T$ such that $P$
is single-crossing on $T$, as follows. Consider a vote $v$ that occurs
$k$ times in $P$ (and once in $\overline{P}$); assume for convenience that $v_1=\dots=v_k=v$.
We create a path of length $k-1$, label its vertices
as $1, \dots, k-1$, connect $k-1$ to the node of $\overline{T}$ that is associated
with the unique voter in $\overline{P}$ whose vote is equal to $v$, 
and label that node as $k$. By repeating this step for each vote
that occurs in $\overline{P}$, we obtain a tree for $P$.

Conversely, suppose that $P$ is single-crossing on a tree $T$. 
Then each set of nodes of $T$ that corresponds to a group
of voters with identical preferences must form a subtree of $T$.
By contracting this set of nodes into a single node, we obtain a tree
$\overline{T}$ such that $\overline{P}$ is single-crossing on $\overline{T}$.
\end{proof}

Now, given a profile $P$ in which each linear order occurs at most once, 
we construct a graph $G_P$ with vertex set $N$ so that
there is an edge between $i$ and $j$ if and only if $P$ does not
contain a vote $v_k$ such that $K[i, k]+K[k, j]=K[i, j]$.
We will now argue that if this graph is a tree, it is exactly 
the `right' tree for our profile.

First, we observe that $G_P$ is connected. 
\begin{proposition}\label{prop:GP-connected}
The graph $G_P$ is connected.
\end{proposition} 
\begin{proof}
Suppose for the sake of contradiction that $G_P$ has $s>1$ connected
components, namely, $G_1, \dots, G_s$. Let 
$q$ be the minimum Kendall-tau distance between two votes that lie
in different connected components.
Fix a pair of voters $i$ and $j$ that lie in different connected
components (say, $G_\ell$ and $G_r$) and satisfy $K[i, j]=q$. Since there is no edge
between $i$ and $j$, the profile $P$ contains a vote $v_k$
with $K[i, k]+K[k, j]=K[i, j]$. Since all votes in $P$ are distinct, 
all distances are positive, so $K[i, k]<q$, $K[k, j]<q$. By our choice of $q$, 
the first of these inequalities means that $k$ must be in $G_\ell$, while 
the second of these inequalities means that $k$ must be in $G_r$, a contradiction.
\end{proof}

Given this observation, we know that $G_P$ is a tree if and only if it is acyclic.
The following observation provides additional insights into the structure of $G_P$.

\begin{proposition}\label{prop:GP-convex}
Let $i$ and $j$ be two non-adjacent vertices of $G_P$.
Then there is a path in $G_P$ between $i$ and $j$ such 
that for each vertex $k$ on this path we have $K[i, k]+K[k, j]=K[i, j]$.
\end{proposition}
\begin{proof}
Suppose for the sake of contradiction that this is not the case; among all
pairs $(i, j)$ that violate this property, pick one with the minimum Kendall-tau distance
$K[i, j]$ and let $q=K[i, j]$. 
Since the edge $\{i, j\}$ is not in $G_P$, there exists a voter $v_k$ in $P$
such that $K[i, k]+K[k, j]=K[i, j]$. Since all votes in $P$ are distinct, 
we have $K[i, k]<q$. Hence there exists an $i$--$k$
path in $G_P$ 
such that for each vertex $\ell$ on this path 
we have $K[i, \ell]+K[\ell, k]=K[i, k]$. 
By the triangle inequality, we have
$$
K[i, \ell]+K[\ell, j]\le K[i, \ell]+K[\ell, k]+K[k, j]=K[i, k]+K[k, j]=K[i, j].
$$
On the other hand, by the triangle inequality we obtain 
$K[i, j]\le K[i, \ell]+K[\ell, j]$ and hence 
$K[i, \ell]+K[\ell, j]=K[i, j]$.
Similarly, since $K[k, j]<q$, there exists a $k$--$j$ path in $G_P$
such that for each vertex $\ell'$ on this path 
we have $K[k, \ell']+K[\ell', j]=K[k, j]$; arguing as above,
we conclude that this implies $K[i, \ell']+K[\ell', j]=K[i, j]$. 
By combining these two paths, 
we obtain an $i$--$j$ path with the desired property,
a contradiction.
\end{proof}

We will now show that cycles in $G_P$ correspond to violations
on the single-crossing property.

\begin{theorem}\label{thm:sct-criterion}
If $P$ is single-crossing on a tree then 
$G_P$ is a tree, and $P$ is single-crossing on $G_P$. 
Conversely, if $G_P$ is a tree, then $P$ is single-crossing on $G_P$.
\end{theorem}
\begin{proof}
Suppose that $P$ is single-crossing on some tree $T$. Let $i$ and $j$
be two adjacent vertices of $T$. We claim that $i$ and $j$ are adjacent in $G_P$.
Indeed, suppose for the sake of contradiction that $K[i, j] = K[i, k] + K[k, j]$
for some $k\in N$. Consider the path connecting $i$, $j$ and $k$ in $T$.
Since $i$ and $j$ are adjacent, this path is of the form $i$--$j$--$\dots$--$k$
or $j$--$i$--$\dots$--$k$; without loss of generality, assume the former is true.
But then the profile $(v_i, v_j, v_k)$ must be single-crossing
in the given order, and hence $K[i, j]+K[j, k] = K[i, k]$.
Substituting $K[i, j] = K[i, k]+K[k, j]$, we obtain
$K[j, k] = 0$, a contradiction with the assumption that all votes in $P$ are distinct. 
Thus, $T$ is a subgraph of $G_P$.

On the other hand, suppose that $i$ and $j$ are adjacent in $G_P$. We claim that
if $P$ is single-crossing on some tree $T$ then $i$ and $j$
are adjacent in $T$. Indeed, suppose some node $k$ appears on the 
(unique) path from $i$ to $j$ in $T$. Then the profile 
$(v_i, v_k, v_j)$ is single-crossing in the given order and hence
$K[i, k]+K[k, j]=K[i, j]$, a contradiction with $i$ and $j$
being adjacent in $G_P$. Thus, $G_P$ must be a subgraph of every tree $T$
such that $P$ is single-crossing on $T$. 

We conclude that if $P$ is single-crossing on a tree $T$, then $T=G_P$. 
In particular, this means that if $G_P$ contains a cycle 
then $P$ is not single-crossing on a tree.

Conversely, suppose that $G_P$ is a tree. We will argue that in this case
$P$ is single-crossing on $G_P$. Suppose for the sake of contradiction
that this is not the case, and let $i$--$\dots$--$j$ be a minimum-length
path in $G_P$ such that the profile $(v_i, \dots, v_j)$ obtained
by restricting $P$ to this path is not single-crossing in the given order.
Let $I$ denote the set of internal vertices of this path.
By minimality of this path, there exists a pair of alternatives
$(a, b)$ such that each vote $v_\ell$ with $\ell\in I$ disagrees with both $v_i$ and $v_j$
on this pair. Hence, $K[i, \ell]+K[\ell, j]>K[i, j]$ for each $\ell\in I$.
However, by \Cref{prop:GP-convex}, there is an $i$--$j$
path in $G_P$ such that $K[i, \ell]+K[\ell, j]=K[i, j]$ for every vertex
$\ell$ of this path. Since $G_P$ is a tree, and hence there is a unique $i$--$j$
path in $G_P$, we obtain a contradiction.
\end{proof}

The proof of \Cref{thm:sct-criterion} establishes
that if $P$ is single-crossing on a tree and 
contains at most one occurrence of each linear order
then it is single-crossing on a unique tree, namely, $G_P$.
Of course, this is no longer the case if $P$ may contain
multiple copies of the same linear order, as the respective
voters can be arranged in many different ways (with the constraint 
that they form a subtree of the resulting tree).

Further, \Cref{thm:sct-criterion} provides an efficient way
to decide if the input profile is single-crossing on a tree
and, if yes, to construct a suitable tree: all we need to do
is to build the graph $G_P$ and check whether it is acyclic.
A straightforward implementation of this approach results 
in an algorithm that runs in time $O(n^3+n^2m\log m)$.

We note that there are other approaches to detecting profiles that are single-crossing
on trees. For instance, \citet{kung2015sorting} describes a recursive algorithm, 
which works by splitting voters into cells based on their preferences.
That is, at the first step it picks a pair of alternatives $\{a, b\}$ and 
creates two cells---one containing voters who prefer $a$ to $b$
and one containing voters who prefer $b$ to $a$. It proceeds recursively in this manner
until each cell contains exactly one voter; the sequence of `cuts'
determines the edges of the tree. Yet another approach 
is proposed by \citet{clearwater2014single}, who start by identifying a leaf
of the tree, and then proceed recursively. Both procedures can be
implemented in polynomial time.

\subsection{Algorithms for Group-Separable Preferences}
\label{sec:recog:group-sep}
Recall the characterization of group-separable profiles in terms of their canonical
clone tree decompositions: a profile is group-separable if and only if its 
canonical clone tree decomposition does not contain P-nodes \citep{karpov:j:gs}.
Since a canonical clone tree decomposition of a given profile can be 
computed in linear time \citep{elkind2012clone}, we obtain a linear-time
algorithm for recognizing group-separable preferences. We will
now describe a simpler (but slower) algorithm, which does not require
the full power of PQ-trees.

\begin{theorem}\label{thm:gs-alg}
Given a profile $P$ over a set of alternatives $A$, 
we can decide in polynomial time whether $P$
is group-separable.
\end{theorem}
\begin{proof}
We can assume that the first voter in $P$ ranks the alternatives as
$a_1\succ\dots\succ a_m$. For each $i\in [m]$, let $A_i=\{a_1, \dots, a_i\}$.
We first check if there is an alternative 
$a_i$, $1\le i\le m-1$, such that each voter in $P$ either ranks
$A_i$ above $A\setminus A_i$ or ranks $A\setminus A_i$ above $A_i$.
If no, then $P$ is clearly not group-separable. If yes, we recurse
on $A_i$ and $A\setminus A_i$; we declare a success if 
at each level of recursion we manage 
to split each non-singleton set into two clone sets. 

Now, clearly, if at any step our recursive procedure fails to split
a set of alternatives into two clone sets, then the input profile is not 
group-separable. To prove the converse, we need to argue that, if the algorithm
succeeds, then any non-singleton set 
$A'\subseteq A$---and not just the sets explicitly considered
by our algorithm---can be split into two clone sets. Consider some such set $A'$.
Let $C$ be the minimal set among the sets considered by the algorithm such that
$A'\subseteq C$ (we may have $C=A$). Suppose that our algorithm partitions
$C$ as $B$ and $C\setminus B$. Note that by our choice of $C$ we have
$A'\cap B\neq\varnothing$ and $A'\cap(C\setminus B)\neq\varnothing$.
But then each voter who ranks $B$ above $C\setminus B$ places
$A'\cap B$ above $A'\cap (C\setminus B)$, and
each voter who ranks $C\setminus B$ above $B$ places
$A'\cap (C\setminus B)$ above $A'\cap B$, i.e., 
both $A'\cap B$ and $A'\cap (C\setminus B)$ form clone sets
in $P|_{A'}$. This completes the proof.
\end{proof}

Note that the procedure in the proof of \Cref{thm:gs-alg} 
implicitly constructs a binary tree $T$ whose set of leaves 
is $A$. We can interpret this tree as a PQ-tree; however, $T$
need not be a clone decomposition tree of $P$, as there 
may exist sets of alternatives that are clone sets
with respect to $P$, but do not correspond to any of the nodes
of $P$. Moreover, the tree $T$ is not unique.
For instance, suppose that $P$ consists of a single vote, 
namely, $a\succ b \succ c \succ d$. In the first iteration, 
the algorithm may split $A=\{a, b, c, d\}$ as $\{a, b\}$
and $\{c, d\}$, and then split both of these sets into
singletons, resulting in a balanced binary tree. Alternatively, 
in the first iteration it may split $A$ as $\{a\}$, 
and $\{b, c, d\}$, then split $\{b, c, d\}$ as $\{b\}$
and $\{c, d\}$, and finally split $\{c, d\}$, as $\{c\}$
and $\{d\}$; the corresponding tree is a caterpillar.
In either case, there are clones in $P$ that are not represented
by the resulting tree. In fact, the canonical clone
decomposition tree for this profile is not binary:
it is simply a tree that has a Q-node as a root and all four
alternatives as its children.

\subsection{Nearly Structured Preferences}
\label{sec:recog:almost}
Most preference profiles are not structured according to the notions that we have considered so far. This 
is true both in a probabilistic sense (see \Cref{sec:further-topics} on counting and probability) and for 
preferences observed in practice \citep{mattei2012empirical}. For example, PrefLib, the widely used 
reference library of preference data, does not contain any profiles of linear orders that are 
single-peaked~\citep{MW-trends}. Still, we might hope that preference data is \emph{almost} structured, in 
the sense of being very close to a structured profile according to some metric. In this section, we will 
consider several notions of being almost structured, and discuss the computational complexity of finding a 
structured profile that is as close as possible to an input profile.

\subsubsection{Voter and Alternative Deletion}

Suppose we hold an election over some numerical quantity, such as a tax rate. As we discussed in 
\Cref{sec:def:sp}, we may expect everyone's preferences to be single-peaked with 
respect to the natural axis. Yet we find this is not the case; the actual preferences are not 
single-peaked. What happened? A first suspicion could be that one of the voters got confused while voting, 
or maybe that a few of the voters do not understand what a tax is, or that some voters hold non-standard 
theories of taxation. If this is the reason, then we would expect that most of the voters have submitted 
single-peaked preferences, and we only need to delete very few voters from the input profile to obtain a 
single-peaked profile. This leads us to the following computational problem, defined for any preference 
domain $\Gamma$.

\problem{$\Gamma$ Voter Deletion}{A preference profile $P$, and an integer $k\ge 0$.}
{Can we delete at most $k$ voters from $P$ to obtain a profile that satisfies $\Gamma$?}

Another reason why our profile may fail to be single-peaked could be the presence of 
a small number of alternatives that do not quite fit onto our one-dimensional axis, 
or, alternatives for which insufficient information is available for voters to place them on the axis. 
For example, voters may reject the tax rate $20\%$ out of principle as it corresponds 
to the undesirable status quo; deleting this option could reveal a single-peaked profile.
This example suggests the following problem.

\problem{$\Gamma$ Alternative Deletion}{A preference profile $P$ over $A$, and an integer $k\ge 0$.}
{Can we delete at most $k$ alternative from $A$ so that $P|_{A'}$ satisfies $\Gamma$?}

The complexity of these problems has been studied for several domain restrictions. 
In particular, for single-peaked and single-crossing preferences we obtain very intuitive results: 
for the former domain (which is defined in terms of an order of alternatives)
the alternative deletion problem is easy and the voter deletion problem is hard, 
whereas for the latter domain (which is defined in terms of an ordering of the voters), 
the alternative deletion problem is hard and the voter deletion problem is easy. 
In what follows, let $\Gammasp$ denote the domain of single-peaked preference profiles, 
	     and let $\Gammasc$ denote the domain of single-crossing preference profiles. 

\begin{theorem}[\citealp{erd-lac-pfa:j:nearly-sp}, Theorem 6.16, \citealp{przedmojski2016algorithms}]
	$\Gammasp$ alternative deletion can be solved in time $O(m^3n)$.
\end{theorem}
The algorithm of \citet{erd-lac-pfa:j:nearly-sp} proceeds by dynamic programming, 
building up `good' axes; it is based on ideas described in \Cref{sec:recog:sp}. 
The original time bound of $O(m^6n)$ was subsequently improved to $O(m^3n)$
by \citet{przedmojski2016algorithms}. 

\begin{theorem}[\citealp{bredereck2016nicelystructured}, Theorem 6]
	$\Gammasc$ voter deletion can be solved in time $O(m^2n^3)$.
\end{theorem}
\begin{proof}
	The algorithm constructs a single-crossing subprofile $P'$ of $P$ with the maximum number of voters. 
	It starts by guessing the voter $i$ that will be the leftmost voter of $P'$ 
	in its single-crossing order, and then constructs an acyclic digraph $D$ on the voter set of $P$. 
	The arcs of $D$ are defined as follows: there is an arc from $j$ to $k$ 
	if $k$ disagrees on more pairs with $i$ than $j$ does, i.e., if
	$\Delta(v_i, v_j) \subseteq \Delta(v_i,v_k)$.
	The voters on a path in $D$ starting at $i$ then correspond to single-crossing subprofiles 
	of $P$ that start with $v_i$. Since $D$ is acyclic, the algorithm then just needs to output 
	a longest path of $D$ starting at $i$; this can be accomplished by means of dynamic programming.
\end{proof}

\begin{theorem}[\citealp{erd-lac-pfa:j:nearly-sp}, Theorem 6.4; \citealp{bredereck2016nicelystructured}, Corollary 1]
	$\Gammasp$ voter deletion is NP-complete.
\end{theorem}
\begin{proof}%
	We sketch 
	\citeauthor{bredereck2016nicelystructured}'s proof, 
	which is a reduction from \problemname{Vertex Cover}. 
	Let $G = (V,E)$ be a graph, where $V = \{\nu_1,\dots,\nu_n\}$ and $E = \{e_1,\dots,e_m\}$, 
	and let $k\ge 1$ be the target size of the vertex cover.
	
	For each edge $e_j\in E$, we introduce three edge alternatives $a_j, b_j, c_j$. 
	For each vertex $\nu_i\in V$ we construct a voter $i$ and set
	\[ \{a_1,b_1,c_1\} \succ_i \dots \succ_i \{a_m,b_m,c_m\}.  \]
	Further, for each edge $e_j = \{\nu_r, \nu_s\}$ with $r < s$, we set
	\begin{alignat*}{4}
		c_j &\succ_r a_j &&\succ_r b_j, \\
		b_j & \succ_s c_j &&\succ_s a_j, \\
		a_j &\succ_i b_j &&\succ_i c_j \quad\text{ for all $i\not\in\{r,s\}$}.
	\end{alignat*}
	This completes the description of the constructed profile $P$.
	Observe that, for $e_j = \{\nu_r, \nu_s\}\in E$, $P|_{\{a_j,b_j,c_j\}}$ is not single-peaked, 
	but can be made single-peaked by (1) deleting voter $r$, (2) deleting voter $s$, 
	or (3) deleting all voters except $r$ and $s$.
	As a consequence, it can be shown that $G$ has a vertex cover of size at most $k$ 
	if any only if we can delete at most $k$ voters to make $P$ single-peaked.
\end{proof}

The hardness proof for $\Gammasc$ alternative deletion is substantially more
complicated and proceeds by a reduction from \problemname{Max2Sat}.

\begin{theorem}[\citealp{bredereck2016nicelystructured}]
	$\Gammasc$ alternative deletion is NP-complete.
\end{theorem}

\begin{theorem}[\citealp{bredereck2016nicelystructured}]
	$\Gamma$ voter deletion and $\Gamma$ alternative deletion are NP-complete for value-restricted, group-separable, and best-, medium-, and worst-restricted preferences.
\end{theorem}

In \Cref{tbl:deletion} we give an overview of these complexity results as well as
list approximation results obtained by \citet{elkind2014detecting}: 
`$\alpha$-approx.' refers to a constant-factor approximation algorithm for the given problem, i.e., a 
polynomial-time algorithm that constructs a profile in the target domain by performing
at most $\alpha$ times the optimal number of voter/alternative deletions.
These approximation algorithms rely on characterizations of domain restrictions via forbidden 
subprofiles (cf.~\Cref{sec:subprofiles}).

\begin{table}
	\centering
	\begin{tabular}{lll}
		\toprule
		\textbf{Domain} $\Gamma$ & \textbf{Voter Deletion} & \textbf{Alternative Deletion} \\
		\midrule
		Single-peaked & NP-complete (2-approx.) & polynomial time $O(m^3n)$ \\
		Single-crossing & polynomial time $O(n^3m^2)$ &  NP-complete (6-approx.) \\
		\midrule
		Value-restricted & NP-complete (3-approx.) & NP-complete (3-approx.) \\
		Best/medium/worst-restricted & NP-complete (2-approx.) & NP-complete (3-approx.) \\
		Group-separable & NP-complete (2-approx.) & NP-complete (4-approx.) \\
		\midrule
		1-Euclidean & open & open \\
		Single-peaked on a tree & open & open\\ %
		Single-peaked and single-crossing & open %
		 & open \\ %
		$d$-dimensional single-peaked & open & open \\
		\bottomrule
	\end{tabular}
	\caption{Complexity results for various voter and alternative deletion problems.}
	\label{tbl:deletion}
\end{table}

Alternative deletion has the drawback that even a small number of removed alternatives yields 
a significant loss of information about the voters' preferences. 
In the context of single-peaked preferences, 
\citet{erd-lac-pfa:j:nearly-sp} propose \emph{local alternative deletion}, 
where in each vote a small number of alternatives may be removed. 

\problem{$\Gammasp$ Local Alternative Deletion}{A preference profile $P$ over $A$, 
and an integer $k\ge 0$.}
{Does there exist an axis $\lhd$ such that for each vote $v\in P$ there exists
a subset $A'\subseteq A$ of size at least $|A|-k$ with the property that $v|_{A'}$
is single-peaked on $\lhd|_{A'}$?}

Note that this is a more flexible notion than alternative deletion, as voters do not have to `agree' 
which alternatives are the outliers.
\citet{sui2013multi} successfully employ local alternative deletion to argue that two sets of 
Irish voting data are close to being two-dimensional single-peaked.
\citet[Theorem 6.7]{erd-lac-pfa:j:nearly-sp} show that \problemname{$\Gammasp$ Local Alternative Deletion} is NP-complete. This result is somewhat surprising, given the easiness
result for \problemname{$\Gammasp$ Alternative Deletion}.

Generalizing \problemname{$\Gammasp$ Local Alternative Deletion} from $\Gammasp$ to other domains is not 
entirely straightforward. For example, for the domain of preferences single-peaked on a tree, 
we cannot simply replace the axis $\lhd$ in the definition with a tree $T$, as a restriction
of $T$ to a subset of alternatives $A'$ need not be a tree. Indeed, to the best of our knowledge, 
this problem has not been formally defined or studied for domains other than $\Gammasp$.

\subsubsection{Clone Sets and Width Measures}

Suppose a city is planning to open a new library somewhere along Main Street, and is asking its residents 
for their opinions on where and how to build the library. There are several potential sites available for 
the new building, and for each site, several designs have been proposed. Voters are 
asked to rank these plans. Should we expect the voters to provide single-peaked rankings?

The answer is not clear. On the one hand, most people will prefer the library to be built as close to them 
as possible, and thus it seems likely that preferences over the \emph{locations} on Main Street will be 
single-peaked. On the other hand, the choice of design is merely a question of 
taste, and so single-peakedness seems unlikely.

If the voters view the location of the library as the more important feature, then 
we can think of each vote as consisting of `blocks': for each location, the different proposals for this 
location are ranked consecutively (in some order) within the preference ranking. Formally, we will say 
that the proposals form an \emph{interval}.

\begin{definition}
	A set $I\subseteq A$ of alternatives forms an \defemph{interval} (also known as a \defemph{clone set}) 
	of the profile $P$ if for every vote $v_i\in P$, every pair of alternatives $a, b\in I$
	and every alternative $c\in A\setminus I$ it is not the case that $a\succ_i c\succ_i b$.
\end{definition}

Notice that the entire set $A$ is always an interval, and so are all the singletons.
If $I$ is an interval of $P$, then we can \emph{contract} this interval 
by replacing all the alternatives in $I$ by a single super-alternative. 
Using this operation, we can define a structural restriction that captures our library example.

\begin{definition}[\citealp{cor-gal-spa:c:spwidth}]
	\label{def:spwidth}
	A profile $P$ over $A$ has \defemph{single-peaked width} $k$ if the set $A$ can be partitioned 
	into intervals $I_1,\dots, I_q$ so that $|I_j| \le k$ for each $j\in [q]$ and the profile $P'$ 
	obtained by contracting each of the intervals $I_j$, $j\in [q]$, is single-peaked.
\end{definition}

While \Cref{def:spwidth} is formulated for the single-peaked domain, 
it extends naturally to other domain restrictions. 

\addtocounter{definition}{-1}
\begin{definition}[continued]
	Let $\Gamma$ be a domain restriction. A profile $P$ over $A$ has \defemph{$\Gamma$-width} $k$ 
	if the set $A$ can be partitioned into intervals $I_1,\dots, I_q$ so that $|I_j| \le k$ 
	for each $j\in [q]$ and the profile $P'$ obtained by contracting each of the intervals 
	$I_j$, $j\in [q]$, is in $\Gamma$.
\end{definition}

As we will see in \Cref{sec:problems}, some winner determination algorithms that work for 
single-peaked (respectively, single-crossing) input profiles continue to work for profiles of 
bounded single-peaked (respectively, single-crossing) width, 
i.e., $\Gamma$-width is a very useful concept from an algorithmic point of view.
Moreover, in contrast to most problems concerning voter and alternative deletion, 
it is possible to determine the single-peaked width or 
the single-crossing width of a profile in polynomial time.

\begin{examplebox}{Single-peaked width}{spwidth}
	\newcommand{\cloneone}{$a,b,c,d$}
	\newcommand{\clonetwo}{$e,f,g$}
	\newcommand{\clonethree}{$h,i$}
	\begin{center}
	\begin{tikzpicture}
	[clone/.style={rounded corners}]
	\matrix (m) [matrix of math nodes, column sep=1ex] {
		\toprule
		v_1 & v_2 & v_3 \\
		\midrule
		a & a & e \\
		b & b & f \\
		c & c & g \\
		d & d & a \\
		g & i & b \\
		f & h & c \\
		e & e & d \\		
		h & f & h \\
		i & g & i \\
		\bottomrule \\
	};
	\draw[clone] ($(m-2-1.north west)+(0,2pt)$) rectangle ($(m-5-1.south east)+(0,2pt)$);
	\draw[clone] ($(m-6-1.north west)+(0,2pt)$) rectangle ($(m-8-1.south east)+(0,3pt)$);
	\draw[clone] ($(m-9-1.north west)+(0,1.5pt)$) rectangle ($(m-10-1.south east)+(0,2pt)$);
	\draw[clone] ($(m-2-2.north west)+(0,2pt)$) rectangle ($(m-5-2.south east)+(0,3pt)$);
	\draw[clone] ($(m-6-2.north west)+(0,1pt)$) rectangle ($(m-7-2.south east)+(0,2pt)$);
	\draw[clone] ($(m-8-2.north west)+(0,2pt)$) rectangle ($(m-10-2.south east)+(0,2pt)$);
	\draw[clone] ($(m-2-3.north west)+(0,2pt)$) rectangle ($(m-4-3.south east)+(0,2pt)$);
	\draw[clone] ($(m-5-3.north west)+(0,2pt)$) rectangle ($(m-8-3.south east)+(0,2pt)$);
	\draw[clone] ($(m-9-3.north west)+(0,2pt)$) rectangle ($(m-10-3.south east)+(0,2pt)$);
	\matrix (contracted) [matrix of math nodes, column sep=1ex, right=2cm of m] {
		\toprule
		v_1 & v_2 & v_3 \\
		\midrule
		a & a & e \\
		e & h & a \\
		h & e & h \\
		\bottomrule \\
	};
	\draw[-latex, ultra thick] (m.east) -- (contracted.west);
	\begin{scope}[yscale=0.65,xscale=0.95,xshift=7cm,yshift=-1.3cm]
	\def\xmin{1}
	\def\xmax{5}
	\def\ymin{0}
	\def\ymax{3}
	\draw[->] (\xmin -1,\ymin) -- (\xmax+1,\ymin) node[right] {};
	\foreach \x/\xtext in {1/\clonetwo, 3/\cloneone, 5/\clonethree}
	\draw[shift={(\x,\ymin)}] (0pt,2pt) -- (0pt,-2pt) node[below] {\strut$\{$\xtext$\}$};
	\foreach \x/\xtext in {1, 3}
	\node[below] at (\x+1,\ymin) {$\strut\lhd$};  
	\draw[thick,blue!30!white] (1,3)--(3,5)--(5,4);
	\foreach \x/\y/\z in {1/3/\clonetwo,3/5/\cloneone,5/4/\clonethree}
	\node[fill=blue!30!white, rectangle, inner sep=0.6mm] at (\x,\y) {\z};
	\draw[thick,green!50!black!50!white] (1,2)--(3,4)--(5,1);
	\foreach \x/\y/\z in {1/2/\clonetwo,3/4/\cloneone,5/1/\clonethree}
	\node[fill=green!50!black!50!white, rectangle, inner sep=0.6mm] at (\x,\y) {\z};
	\draw[thick,red!50!white] (1,4)--(3,3)--(5,2);
	\foreach \x/\y/\z in {1/4/\clonetwo,3/3/\cloneone,5/2/\clonethree}
	\node[fill=red!50!white, rectangle, inner sep=0.6mm] at (\x,\y) {\z};
	\end{scope}
	\end{tikzpicture}
	\end{center}
\end{examplebox}

\begin{theorem}[\citealp{cornaz2013kemeny}]
	The single-peaked width and the single-crossing width of a profile with $n$ voters 
	and $m$ alternatives can be computed in $O(nm^3)$ time.
\end{theorem}
The algorithms of \citet{cornaz2013kemeny} proceed by
enumerating all intervals of the input profile, and then calculating a PQ-tree that encodes all 
rankings for which every interval of the profile forms an interval 
of the ranking (following \citealp{elkind2012clone}). 
The algorithm then manipulates this PQ-tree, deciding which of the nodes need 
to be collapsed so as to make the profile single-peaked or single-crossing, respectively.%
\footnote{The published paper omits some proof details.}
For domain restrictions $\Gamma$ other than $\Gammasp$ and $\Gammasc$, the concept 
of $\Gamma$-width and the associated computational challenges have not yet been explored.

Another closeness measure that is based on intervals was proposed by \citet{elkind2012clone}. They 
wish to obtain a structured profile by contracting as few alternatives as possible, which they call 
`optimal decloning'.

\begin{theorem}[\citealp{elkind2012clone}, Theorems 5.8 and 6.4]
	Given a profile $P$ and an integer $k\ge 1$, it is NP-complete to decide whether we can obtain 
	a single-crossing profile $P'$ with at least $k$ alternatives from $P$ by contracting intervals. 
	In contrast, it is decidable in polynomial time whether $P$ can be transformed into 
	a single-peaked profile $P'$ with at least $k$ alternatives by contracting intervals.
\end{theorem}

\subsubsection{Swap Distance}
A well-known measure of distance between two linear orders is the Kendall-tau distance $K$ (cf.\ 
\Cref{def:kendall}), which counts how many swaps of adjacent alternatives are necessary to 
transform one linear order into another.
The Kendall-tau distance can also be used to evaluate closeness to preference domains; 
an attractive feature of this measure is that it offers a fine-grained perspective 
on how much a given profile has to be modified---in contrast to `coarser' measures 
such as deleting voters or alternatives.

\citet{fal-hem-hem:j:nearly-sp} and 
\citet{erd-lac-pfa:j:nearly-sp} 
present two natural notions of closeness that are based on the Kendall-tau distance:
we can count the overall number of required swaps (global swap distance) or 
the number of required swaps per vote (local swap distance).
		
\begin{theorem}[\citealp{erd-lac-pfa:j:nearly-sp}, Theorems 6.8 and 6.9]
		Given a profile $P$ and an integer $k\geq 1$, it is NP-complete to decide whether there 
		exists a single-peaked profile $P'$ whose global swap distance from $P$ is at most $k$. 
		The same holds for the local swap distance.
\end{theorem}

These problems are also NP-complete for the single-crossing domain 
\citep[Theorems~1 and 2]{lakhani2019correlating}. \citet{jaeckle2018recognising} show 
that the local swap problem for single-crossing preferences can be solved in XP 
time with respect to the parameter $k$.

\begin{open}
What is the computational complexity of determining the global/local swap distance 
to other preference domains? Are there FPT algorithms for these problems? 
\end{open}		

For single-peakedness specifically, another fine-grained distance measure that was proposed by \citet{escoffier2021measuring} is the \emph{forbidden triples} distance that counts the number of valleys (see \Cref{fig:valley}) induced by a given axis. \citet{escoffier2021measuring} show that this distance is well-behaved in experiments, but like other fine-grained distances, it is NP-hard to compute.

\subsubsection{Voter and Alternative Partition}

Another set of closeness measures is based on partitioning voters or 
alternatives~\citep{erd-lac-pfa:j:nearly-sp}. Consider a situation where 
each alternative can be characterized by a pair of real-valued parameters. 
Some voters consider the first parameter to be more important, the 
others the second parameter. Then it may be possible to split the set of voters into 
two sets so that each set is single-peaked with respect to its own axis.
The underlying computational problem is the following:

\problem{$\Gamma$ Voter Partition}
{A preference profile $P$, and an integer $k\ge 1$.}
{Can we partition the set of voters $N$ into $k$ sets $N_1,\dots,N_k$ 
so that for all $j\in [k]$ the profile $(v_i : i\in N_j)$ belongs to $\Gamma$?}
		
For the single-peaked domain, this problem is NP-hard for $k \ge 3$, but can be solved in polynomial time for $k = 2$ via a reduction to 2SAT, which uses the forbidden subprofile characterization of single-peaked profiles.
\begin{theorem}[\citealp{erd-lac-pfa:j:nearly-sp}, Theorem 6.5]
\problemname{$\Gammasp$ Voter Partition} is NP-complete, for each fixed $k \ge 3$.
\end{theorem}
\begin{theorem}[\citealp{yang2020comp}, Theorem 8]
	\problemname{$\Gammasp$ Voter Partition} can be solved in polynomial time for $k = 2$.
\end{theorem}
The 2SAT technique also works for some other domains.
\begin{theorem}[\citealp{kraiczy2022explaining}]
	For the domains $\Gamma$ of value-restricted or of group-separable preferences, \problemname{$\Gamma$ Voter Partition} can be solved in polynomial time for $k = 2$.
\end{theorem}
					
A similar question can be posed for partitioning alternatives.

\problem{$\Gamma$ Alternative Partition}
{A preference profile $P=(v_1,\dots,v_n)$, and an integer $k\ge 1$.}
{Can we partition the set of alternatives $A$ into $k$ sets $A_1,\dots,A_k$ 
so that for all $j\in [k]$ the profile $P|_{A_j}$ belongs to $\Gamma$?}

The computational complexity of this problem for single-peaked preferences is unknown;
in contrast, for single-crossing preferences this problem is known to be hard 
even for $k=3$.

\begin{theorem}[\citealp{jaeckle2018recognising}]
	\problemname{$\Gammasc$ Alternative Partition} is NP-complete, for each $k\ge 3$.
\end{theorem}
		
\begin{open}
What is the computational complexity of \problemname{$\Gammasp$ Alternative Partition}? 
Furthermore, what is the complexity of \problemname{$\Gamma$ Voter Partition} 
and \problemname{$\Gamma$ Voter Partition} for other preference domains $\Gamma$?
\end{open}		

To conclude, we briefly discuss two notions of closeness that can be viewed
as local versions of alternative/voter partition, but are defined 
for specific restricted domains and do not generalize easily to other domains.

\citet{yang2015multipeak} consider $k$-peaked profiles: 
a profile $P$ 
over a set of alternatives $A$ is said to be \emph{$k$-peaked} if there is an axis 
$\lhd$ such that each voter $i$ can partition $A$ into $k_i\le k$ pairwise
disjoint subsets $A_1, \dots, A_{k_i}$ so that for each $j=1, \dots, k_i$
the set $A_j$ forms a contiguous subset of $\lhd$ and the restriction of $v_i$
to $A_j$ is single peaked on $\lhd|_{A_j}$. This notion can be seen as a local 
analog of alternative partition, in that each voter is allowed to choose
their own partition of $\lhd$. However, \citet{yang2015multipeak} focus on the complexity
of election control in this domain rather than the recognition problem, 
and leave the recognition problem open.

In a similar spirit, \citet{misra2017complexity} discuss $k$-crossing profiles:
a profile $P$ over a set of alternatives $A$ is said to be \emph{$k$-crossing} if 
the voters can be reordered so that in the reordered profile $P'$ each pair 
of alternatives `crosses' at most $k$ times, i.e., for each pair of alternatives $(a, b)$
the voters can be split into $k+1$ groups so that the voters
in each group form a contiguous block in $P'$ and agree on $(a, b)$. 
This notion can be viewed as a local variant of the voter partition 
problem, in the sense that for each pair of alternatives we may choose
a different partition of the reordered profile $P'$ into $k+1$ groups. 
Again, \citet{misra2017complexity} study the complexity of multi-winner voting
in this domain rather than the recognition problem for the domain itself.

\subsubsection{Fixed Order of Voters or Alternatives}
The domain $\Gammasp$ consists of all profiles that are single-peaked on 
\emph{some} axis. Alternatively, we can fix an axis $\lhd$ and 
consider the domain $\Gamma_\lhd$ of all profiles that are single-peaked
on $\lhd$. We can then ask if a given profile $P$ is close 
to being in $\Gamma_\lhd$, i.e., whether we can make a small number of changes
to $P$ to obtain a profile in $\Gamma_\lhd$. 
In a similar fashion, we can ask if we can make a small number of changes to $P$ 
in order to obtain a profile that is single-crossing in the given order.

\citet{erd-lac-pfa:j:nearly-sp} 
study the complexity of making a given profile $P$
single-peaked on a fixed axis $\lhd$, for many of the distance notions 
discussed earlier in this chapter, and obtain polynomial-time algorithms for each
distance measure they consider (Section 6.3 of their paper). 
For voter deletion, the algorithm is trivial: one can simply delete 
the voters whose preferences are not single-peaked on $\lhd$.
For the alternative deletion problem, \citet{erd-lac-pfa:j:nearly-sp} 
argue that this problem can be reduced 
to \problemname{$\Gammasp$ Alternative Deletion} by adding two votes 
that rank the alternatives according to $\lhd$ and its reverse to the input profile:
this ensures that the algorithm for \problemname{$\Gammasp$ Alternative Deletion}
does not benefit from considering axes other than $\lhd$.
However, this problem also admits a direct $O(nm^3)$ algorithm,\footnote{We are grateful
to Andrei Constantinescu for this observation.}
which we describe below.

\begin{theorem}\label{thm:alt-del-fixedorder}
Given an $n$-voter profile $P$ over an alternative set $A$, $|A|=m$,
and an axis $\lhd$, we can compute a maximum-size set $A'$ such that
$P|_{A'}$ is single-peaked on $\lhd|_{A'}$ in time $O(nm^3)$.
\end{theorem}
\begin{proof}
Assume without loss of generality that the axis $\lhd$
is given by $a_1\lhd\dots\lhd a_m$. We proceed by dynamic programming.
For each pair of indices $j, \ell$ with $1\le j < \ell\le m$, 
let $s(j, \ell)$ be the maximum size of a set of alternatives $C$
such that $C\subseteq\{a_1, \dots, a_\ell\}$, $a_j, a_\ell\in C$
and $P|_C$ is single-peaked on $\lhd|_C$. Clearly, we have $s(1, 2)=2$, 
and the size of the target set $A'$ is given by $\max_{1\le j< \ell\le m}s(j, \ell)$;
the set $A'$ itself can be computed by standard dynamic programming techniques.

We will now explain how to compute $s(j, \ell)$ when $\ell>2$.
We say that an index $k$, $1\le k<j$ is \emph{$(j, \ell)$-good}
if for each voter $i\in N$ the restriction of $v_i$ to $\{a_k, a_j, a_\ell\}$
does not form a valley with respect to $\lhd$, i.e., if $a_k\succ_i a_j$
or $a_\ell\succ_i a_i$. We then have
$$
s(i, j) = 1 + \max_{k: k\text{ is $(j, \ell)$-good}}s(k, j).
$$
Indeed, our definition of a good index ensures that no triple
of alternatives that appear consecutively in the set that is implicitly 
constructed by the algorithm forms a valley in any of the voter's preferences, 
i.e., the restriction of $P$ onto this set has the `no local valleys' property (see \Cref{prop:sp-equiv}~(2)).
\end{proof}

For local alternative deletion, the votes are processed one by one;
for each vote, the problem reduces to guessing a new peak,
splitting the axis $\lhd$ at this peak, and finding a maximum-length
increasing subsequence in the left part and a maximum-length decreasing
subsequence in the right part. 

To evaluate the swap-based distance measures, a key step
is to compute the minimum number of swaps required to 
make a given vote single-peaked on $\lhd$. 
\citet{fal-hem-hem:j:nearly-sp} and \citet{erd-lac-pfa:j:nearly-sp} describe
dynamic programming algorithms for this problem,
whose running time is, respectively, $O(m^4)$ and $O(m^3)$;
below, we present a slightly modified version of the procedure
proposed by \citet{erd-lac-pfa:j:nearly-sp}.

\begin{theorem}
Given a vote $v$ over a set of alternatives $A=\{a_1, \dots, a_m\}$
and an axis $\lhd$, we can compute 
$\min_{u\in \Gamma_\lhd}K(u, v)$ in time $O(m^3)$.
\end{theorem}
\begin{proof}
Assume without loss of generality that the axis $\lhd$
is given by $a_1\lhd\dots\lhd a_m$.

For each $i, j$ with $1\le i< j\le m$, let 
$A_{i, j}=\{a_1, \dots, a_i, a_j, \dots, a_m\}$,  
let $S_L(i, j)$ (respectively, $S_R(i, j)$)
be the set of all preference orders over 
$A_{i, j}$ that are single-peaked on the restriction of $\lhd$ to $A_{i, j}$
and rank $a_i$ first (respectively, rank $a_j$ first).
For $X\in\{L, R\}$, let $s_X(i, j)= \min_{u\in S_X(i, j)}K(v|_{A_{i, j}}, u)$;
it will be convenient to set $s_X(i, j)=+\infty$ if $i=0$ or $j=m+1$.

Note that for each $i\in [m-1]$ the set 
$S_L(i, i+1)$ consists of all preference orders
over $A$ that are single-peaked on $\lhd$ and rank $a_i$ first, 
whereas the set $S_R(m-1, m)$ consists of the (unique) preference
order over $A$ that is single-peaked on $\lhd$ and ranks $a_m$ first, 
so
$$
\min_{u\in\Gamma_\lhd}K(u, v) =\min\{s_R(m-1, m), \min_{i\in [m-1]}s_L(i, i+1)\}.
$$

We will now explain how to compute $s_X(i, j)$ for all $X\in\{L, R\}$ 
and $1\le i<j\le m$.
For $i=1, j=m$, we have $s_L(i, j)=0$, $s_R(i, j)=1$ if $v$ ranks $a_1$ above $a_m$
and $s_L(i, j)=1$, $s_R(i, j)=0$ otherwise. Now, fix $i, j$ with $1\le i< j\le m$,
and suppose we have computed
$s_X(k, \ell)$ for all $X\in\{L, R\}$ and all $k, \ell$ with $\ell-k>j-i$.
Let $v'=v|_{A_{i, j}}$, and let $\lhd'$ be the restriction of $\lhd$ to $A_{i, j}$.
To compute $s_L(i, j)$, we first compute the cost of moving $a_i$
to the top position in $v'$; if $a_i$ is ranked in position $r$
in $v'$, this requires $r-1$ swaps. The remaining alternatives in $v'$ 
are now ordered according to $v|_{A(i-1, j)}$; 
we need to reorder them to obtain a vote over $A(i, j)$ that is single-peaked on $\lhd'$.
Note that the length-2 prefix of this vote must form a contiguous segment 
of $\lhd'$, so its second position must be occupied 
by $a_{i-1}$ or $a_j$.
Thus, we obtain
$$
s_L(i, j)=r-1+\min\{s_L(i-1, j), s_R(i-1, j)\}.
$$
By a similar argument, 
$$
s_R(i, j)=r'-1+\min\{s_L(i, j+1), s_R(i, j+1)\},
$$
where $r'$ is the position of $a_j$ in $v'$.
Altogether, we need to compute $O(m^2)$ quantities, 
and each of them can be computed in time $O(m)$,
which implies our bound on the running time.
\end{proof}

\citet{lakhani2019correlating} 
consider the problem of modifying a given profile to make 
it single-crossing in the given order. This problem turns out
to be quite challenging: they obtain NP-hardness results for alternative deletion, 
alternative partition, and both local and global swaps. 

\citeauthor{lakhani2019correlating}'s hardness results for 
alternative deletion and alternative partition are based on the notion of \emph{crossing graph}
\citep{cohen}. The \emph{crossing graph} of a profile $P$ on a set of alternatives 
$A$ is an undirected graph $G(P)$ that has $A$ as its set of vertices; there is an edge
connecting $a$ and $b$ if the pair $(a, b)$ violates the single-crossing condition.
\citet{cohen} describe an efficient algorithm that, 
given an undirected $m$-vertex graph $G$ with no loops and parallel edges,
builds a profile $P$ with at most $2m+1$ voters such that $G(P)=G$.
This construction immediately implies hardness of alternative deletion
and partition: deleting a set of $k$ alternatives to make the profile $P$ 
single-crossing corresponds to finding a vertex cover of size at most $k$ in $G(P)$,
and partitioning $A$ into at most $k$ sets so that the restriction of $A$ to each
set is single-crossing corresponds to coloring $G(P)$ with at most $k$ colors.

\begin{examplebox}
	{The crossing graph.}
	{crossing}
	\begin{minipage}{0.21\linewidth}
	\begin{tabular}{cccc}
	\toprule
	$v_1$ & $v_2$ & $v_3$ & $v_4$ \\
	\midrule
	$a$ & $a$ & $b$ & $a$ \\
	$b$ & $d$ & $a$ & $b$ \\
	$c$ & $c$ & $c$ & $c$ \\
	$d$ & $b$ & $e$ & $d$ \\
	$e$ & $e$ & $d$ & $e$ \\
	\bottomrule
	\end{tabular}
	\end{minipage}
	\hfill
	\begin{minipage}{0.47\linewidth}
	The reader can verify that the graph $G$ on the right is the crossing
	graph of the profile on the left: e.g., $G$ contains the edge $\{a, b\}$
	and we have $a\succ_1 b$, $b\succ_3 a$ and $a\succ_4 b$, i.e., $a$ and $b$
	cross more than once. 
	\end{minipage}
	\qquad
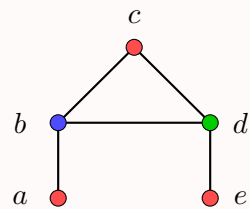
	\begin{minipage}{0.22\linewidth}
	\begin{tikzpicture}
	\draw[thick] (1,0) -- (1,1) -- (2,2) -- (3,1) -- (3,0);
	\draw[thick] (1,1) -- (3,1);
	\draw[fill=red!70!white] (1,0) circle (3pt);
	\draw[fill=blue!70!white] (1,1) circle (3pt);
	\draw[fill=red!70!white] (2,2) circle (3pt);
	\draw[fill=green!80!black] (3,1) circle (3pt);
	\draw[fill=red!70!white] (3,0) circle (3pt);
	\node at (0.5,0) {$a$};
	\node at (0.5,1) {$b$};
	\node at (2,2.4) {$c$};
	\node at (3.4,1) {$d$};
	\node at (3.4,0) {$e$};
	\end{tikzpicture}
	\end{minipage}

	\vspace{6pt}

	The set $\{b, d\}$ forms a vertex cover of $G$;
	thus, if we delete $b$ and $d$ from $G$, the induced graph on $\{a, c, e\}$
	has no edges. Accordingly, the restriction of $P$ to $\{a, c, e\}$ 
	is single-crossing. Further, as $G$ contains a cycle of length $3$, it is
	not 2-colorable, but it admits a 3-coloring, as shown in the figure.
	This corresponds to partitioning $A$ into three
	sets ($\{a, c, e\}$, $\{b\}$ and $\{d\}$) so that the restriction of $P$
	to each of these sets is single-crossing.  
\end{examplebox}

On the positive side, the problem of finding
the minimum number of voters to remove to make the profile 
single-crossing in the given order turns out to be easy: 
one can use (a simplification of) 
\citet{bredereck2016nicelystructured}'s 
algorithm for \problemname{$\Gammasc$ Voter Deletion}, 
which is based on finding the longest
path in a directed acyclic graph. 
\citet{lakhani2019correlating} also propose a polynomial-time algorithm
for the variant of the voter partition problem where each part
has to form a contiguous subprofile of the input profile $P$, as well as for the problem
of finding the minimum number of voter swaps required to make 
$P$ single-crossing in the given order.

\subsubsection{Axiomatic Approach to Nearly Structured Preferences}

As we have seen, there are many possible measures for a profile to be close to being structured. Which of these measures is the most appropriate? \citet{escoffier2021measuring} (extended in the thesis of \citealp{tydrichova2023structural}, Sec. 4.4) consider this question from an axiomatic perspective. They focus on single-peakedness, and conceptualize the problem through \emph{axis selection rules}, which given a profile output an axis $\lhd$ that ``best'' explains the profile. For example, the voter deletion rule outputs an axis on which the maximum number of voters is single-peaked. \citet{tydrichova2023structural} compares different such rules through their axiomatic behavior. For example, she shows that the voter deletion rule satisfies a reinforcement axiom as well as a stability axiom. \citet{delemazure2024comparing} perform a similar analysis for approval preferences based on the CI property (see \Cref{sec:dichotomous}). The same kind of approach could fruitfully be employed for many other notions of structure, such as (multidimensional) Euclidean preferences and preferences single-peaked on trees, or even to single-crossing preferences (for rules that select an ordering of the voters).
 
\newpage
\section{Characterizations by Forbidden Subprofiles}
\label{sec:subprofiles}

A classic result in graph theory is Kuratowski's theorem, which characterizes planar graphs as 
graphs that do not contain subdivisions of $K_5$ and $K_{3, 3}$ as their subgraphs \citep{kuratowski}. 
In a similar fashion, some restricted domains can be defined by a set of forbidden subprofiles: 
a profile belongs to a restricted domain if and only if it does not contain these subprofiles.
In this section, we introduce several characterization results that are based on forbidden subprofiles,
and discuss some applications of this theory.

\subsection{Single-Peaked Preferences}

We start with the following theorem, which characterizes the single-peaked domain in terms of forbidden substructures.

\begin{theorem}[\citealp{ballester2011characterization}]\label{thm:bal-har-sp}
A profile $P$ over $A$ is single-peaked if and only if 
there do not exist alternatives $a,b,c,d\in A$ and voters $i,j\in N$ such that
\begin{align}
\{a,d\} \succ_i b  \succ_i c, \qquad\text{ and }\qquad
\{c,d\} \succ_j b  \succ_j a,
\end{align}
and there do not exist alternatives $a,b,c\in A$ and voters $i,j,k\in N$ such that
\begin{align}
\{b,c\} \succ_i a, \qquad \{a,c\}\succ_j b, \qquad\text{ and }\qquad \{a,b\}\succ_k c.
\end{align}
\end{theorem}
\begin{proof}
We will first argue that if a profile $P$ violates condition~(1) or condition~(2) then it is not single-peaked.
Suppose for the sake of contradiction that $P$ is single-peaked with respect to some axis $\lhd$,
but violates condition (1). 
We can assume without loss of generality that $A=\{a, b, c, d\}$: otherwise we can consider
the restriction of $P$ to $\{a, b, c, d\}$, which is single-peaked as long as $P$ is.
By \Cref{prop:sp-equiv}~(4), we have $a\lhd \{b, d\}\lhd c$ or $c\lhd\{b, d\}\lhd a$;
we can assume without loss of generality that $a\lhd \{b, d\}\lhd c$. But if $a\lhd b \lhd d \lhd c$
then $\succ_i$ is incompatible with $\lhd$ (as $a, d$ form a prefix of $\succ_i$, 
but do not form an interval with respect to $\lhd$), and if 
$a\lhd d \lhd b \lhd c$
then $\succ_j$ is incompatible with $\lhd$ (as $c, d$ form a prefix of $\succ_j$, 
but do not form an interval with respect to $\lhd$).
We obtain a contradiction.
Now, suppose that $P$ is single-peaked with respect to $\lhd$,  
but violates condition (2). Again, we can assume that $A=\{a, b, c\}$.
One of the alternatives $a, b, c$ appears between the other two on $\lhd$. This creates a valley
in the preferences of the voter who ranks this alternative last, a contradiction again.

For the converse direction, we consider the execution of \Cref{alg:sp}, which recognizes 
single-peaked profiles.
We will show that if this algorithm decides that the input profile $P$ is not single-peaked, 
then $P$ fails condition~(1) or condition~(2).

It is straightforward to see that if \Cref{alg:sp} 
terminates in line~\ref{sp-alg:phase1:not-worst} or~\ref{sp-alg:phase2:not-worst} then condition~(2) is not satisfied.

If \Cref{alg:sp} terminates in line~\ref{sp-alg:phase2:contrad}, 
let $x$ be the alternative that is contained in both $L$ and $R$.
Thus, there exists a voter $i$ with $\ell\succ_i x\succ_i r$ and a voter $j$ with $r\succ_j x \succ_j \ell$.
Furthermore, since $|A'|>1$ and $i,j\in N_x(A')$, there exists an alternative $y$ with $y\succ_i x$ and $y\succ_j x$.
Thus, the alternatives $\ell,r,x,y$ and voters $i,j$ witness that condition (1) is not satisfied.

Now suppose that \Cref{alg:sp} terminates in line~\ref{sp-alg:phase2:LRtoolarge}. 
Without loss of generality, we can assume that $|L|=\{x,y\}$.
Then there exist voters $i,j$ with $r\succ_i x \succ_i \ell$ and $r\succ_j y \succ_j \ell$.
Since $i\in N_x(A')$, we have $y\succ_i x$, and since $j\in N_y(A')$, we have $x\succ_j y$.
Since $r$ has already been placed during an earlier step, there has to exist a voter $k$ with $\{x,y\}\succ_k r$.
Thus, candidates $r,x,y$ and voters $i,j,k$ witness that condition~(2)
is violated.
We conclude that if a profile is not single-peaked, then it violates condition~(1) or condition~(2).
\end{proof}

While \Cref{thm:bal-har-sp}
does not explicitly list forbidden subprofiles, it defines them implicitly: these are all profiles with two voters 
and four alternatives witnessing a violation of the first condition and all profiles with three voters and three 
alternatives witnessing a violation of the second condition. This list is essentially finite.
Formally, we say that a profile $P$ over $A$ is \emph{isomorphic} to a profile $P'$ over $A'$ 
if $P'$ can be obtained from $P$ by renaming alternatives and permuting the votes.
In particular, if $P$ is isomorphic to $P'$ then $|A|=|A'|$ and $P$ and $P'$ contain the same number of voters.
Abusing the terminology somewhat, we say that a profile $P$ \emph{contains $P'$ as a subprofile} 
if there exists a profile $P''$ 
that is isomorphic to $P'$ and can be obtained from $P$ by alternative and voter deletion.
With these definitions in hand, we can restate \Cref{thm:bal-har-sp} purely in terms of forbidden subprofiles.

\begin{theorem}\label{thm:bal-har-sp-2}
A profile $P$ over $A$ is single-peaked if and only if $P$ does not contain the following profiles as subprofiles:
\begin{center}
	\begin{tabular}{ccc}
	\toprule
	$v_1$ & $v_2$ & $v_3$ \\
	\midrule
	$a$ & $b$ & $c$ \\
	$b$ & $c$ & $a$ \\
	$c$ & $a$ & $b$ \\
	\bottomrule \\
	\end{tabular}
\quad\qquad
	\begin{tabular}{ccc}
	\toprule
	$v_1$ & $v_2$ & $v_3$ \\
	\midrule
	$a$ & $c$ & $a$ \\
	$b$ & $b$ & $c$ \\
	$c$ & $a$ & $b$ \\
	\bottomrule \\
	\end{tabular}
\quad\qquad
	\begin{tabular}{ccc}
	\toprule
	$v_1$ & $v_2$ \\
	\midrule
	$d$ & $d$  \\
	$a$ & $c$  \\
	$b$ & $b$  \\
	$c$ & $a$  \\
	\bottomrule
	\end{tabular}
\quad\qquad
	\begin{tabular}{ccc}
	\toprule
	$v_1$ & $v_2$ \\
	\midrule
	$d$ & $c$  \\
	$a$ & $d$  \\
	$b$ & $b$  \\
	$c$ & $a$  \\
	\bottomrule
	\end{tabular}	
\quad\qquad
	\begin{tabular}{ccc}
	\toprule
	$v_1$ & $v_2$ \\
	\midrule
	$a$ & $c$  \\
	$d$ & $d$  \\
	$b$ & $b$  \\
	$c$ & $a$  \\
	\bottomrule
	\end{tabular}		
\end{center}
\end{theorem}

\Cref{thm:bal-har-sp,thm:bal-har-sp-2} admit natural analogs for the single-caved domain (just reverse all the forbidden subprofiles). By analyzing their recognition algorithm, \citet{peters2017spoc} give a characterization of preferences single-peaked on a circle by forbidden subprofiles. These are structurally similar to the ones in  \Cref{thm:bal-har-sp,thm:bal-har-sp-2}, but are slightly larger.

\subsection{Single-Crossing Preferences}

The single-crossing domain, too, has been characterized in terms of forbidden subprofiles 
\citep{bredereck2013characterization}. Note that in the theorem statement, the alternatives and voters mentioned are not necessarily distinct.

\begin{theorem}[\citealp{bredereck2013characterization}]
A profile $P$ over $A$ is single-crossing if and only if  
there do not exist alternatives $a,b,c,d,e,f\in A$ and voters $i,j,k\in N$ such that
\begin{align*}
b \succ_i a,\ c \succ_i d,\ e  \succ_i  f,\ \qquad
a \succ_j b,\ d \succ_j c,\ e  \succ_j  f,\ \qquad\text{ and }\qquad
a \succ_k b,\ c \succ_k d,\ f  \succ_k  e,
\end{align*}
and there do not exist alternatives $a,b,c,d\in A$ and voters $i,j,k,\ell\in N$ such that
\begin{align*}
a \succ_i b,\ c\succ_i d,\ \qquad
b \succ_j a,\ c\succ_j d,\ \qquad
a \succ_k b,\ d\succ_k c,\ \qquad\text{ and }\qquad
b \succ_\ell a,\ d\succ_\ell c.
\end{align*}
\end{theorem}

This theorem, too, can be phrased purely in terms of forbidden subprofiles (similarly to \Cref{thm:bal-har-sp-2}).
However, in the case of the single-crossing domain there are 30 non-isomorphic
forbidden subprofiles, and therefore we do not list them all.

\subsection{Value-Restricted and Group-Separable Preferences}
As we have seen in \Cref{prop:vr-cycle}, the value-restricted 
domain also allows a characterization via forbidden subprofiles:
As it is characterized by excluding the Condorcet cycle on three candidates, it follows that 
the value-restricted domain is the largest domain definable by forbidden subprofiles that 
guarantees a Condorcet winner (for an odd number of candidates, cf.\ 
\Cref{prop:vr-transitive}).

The group-separable domain has also been characterized in terms of forbidden subprofiles.
\begin{theorem}[\citealp{ballester2011characterization}]\label{thm:bal-har-gs}
	A profile $P$ over $A$ is group-separable if and only if 
	there do not exist alternatives $a,b,c,d\in A$ and voters $i,j\in N$ such that
	\begin{align*}
		a\succ_i b  \succ_i c \succ_i d \quad\text{ and }\quad 
		b  \succ_j d \succ_j a \succ_j c,
	\end{align*}
	and there do not exist alternatives $a,b,c\in A$ and voters $i,j,k\in N$ such that
	\begin{align*}
		(a \succ_i b \succ_i c, \: a \succ_j c \succ_j b, \: b \succ_k a \succ_k c) 
		\quad
		\text{ or }
		\quad
		(a \succ_i b \succ_i c, \: b \succ_j c \succ_j a, \: c \succ_k a \succ_k b). 
	\end{align*}
\end{theorem}
The second condition ensures that the domain is medium-restricted, so that for each triple of alternatives, at least one of those alternatives is never ranked in between the other two.

\subsection{Euclidean Preferences}

So far we have seen four examples of domains that can be
characterized via a finite set of forbidden subprofiles, namely, 
the single-peaked domain, single-peaked on circles, the single-crossing domain, and the value-restricted domain. 
In contrast, a recent result of \citet{chen15onedimensional} 
shows that the 1-Euclidean domain cannot be characterized in this manner.

\begin{theorem}[\citealp{chen15onedimensional}]
\label{thm:euclid-forbidden}
There is no finite set of forbidden subprofiles that characterizes the 1-Euclidean domain.
\end{theorem}
\begin{proof}[Proof idea.]
	\citet{chen15onedimensional} construct an infinite family of profiles that are minimally non-1-Euclidean. For each 
	$n$, they build an $n$-voter profile $P_n$ over the set of alternatives $A = \{a_1,b_1,\dots, a_{n}, 
	b_{n}\}$. The preference order of voter $i$ ensures that for every potential embedding $x : N\cup A \to \mathbb R$
	witnessing that $P$ is 1-Euclidean 
	we have $\|x(a_i) - x(b_i)\| < \|x(a_{i+1}) - x(b_{i+1})\|$, where subscripts are taken modulo $n$. Taken together,  
	these inequalities imply 
        $\|x(a_1) - x(b_1)\| < \|x(a_2) - x(b_2)\| < \dots < \|x(a_n) - x(b_n)\| < \|x(a_1) - x(b_1)\|$, which 
	is impossible. Hence the constructed profile is not 1-Euclidean. However, whenever we delete a voter, one of the 
	constraints of the cycle vanishes, allowing us to embed the profile into $\mathbb R$. 
	These profiles are thus minimally non-1-Euclidean.

	Now, suppose that the 1-Euclidean domain can be characterized by a finite set of forbidden subprofiles~$\mathcal S$.
	There exist an $n\ge 0$ such that each profile in $\mathcal S$ has fewer than $n$ voters. However, 
	$P_n\not\in\mathcal S$, since it contains $n$ voters, and every subprofile of $P_n$ is not in $\mathcal S$, 
	because it is 1-Euclidean, a contradiction.
\end{proof}

However, the 1-Euclidean domain can be characterized by an \emph{infinite} set of forbidden subprofiles.
In fact, this is the case for every hereditary domain (i.e., a domain closed under deletion of voters and alternatives;
see \Cref{def:hereditary}).

\begin{proposition}[\citealt{lackner2017likelihood}, Proposition 5]\label{prop:inf-set-of-forbidden-subprofiles}
A domain is hereditary if and only if it can be characterized by a (possibly
infinite) set of forbidden subprofiles.
\end{proposition}
\begin{proof}
Suppose that $\mathcal D$ is not a hereditary domain. Then there is a profile $P\in \mathcal D$
such that some subprofile $P'$ of $P$ does not belong to $\mathcal D$.
However, every subprofile of $P'$ is also a subprofile of $P$, which means
that $\mathcal D$ cannot be characterized via forbidden subprofiles.

Conversely, consider a hereditary domain $\mathcal D$, and let $\mathcal S$
be the (infinite) set of all profiles that do not belong to $\mathcal D$.
Note that if a profile belongs to $\mathcal D$, then none of its subprofiles
is in $\mathcal S$, exactly because $\mathcal D$ is a hereditary domain.
Thus, $\mathcal S$ offers the desired characterization.
\end{proof}

\begin{open}
	Give an explicit characterization of the 1-Euclidean domain by infinitely many forbidden subprofiles.
\end{open}

\Cref{prop:inf-set-of-forbidden-subprofiles} 
offers a straightforward way to identify domains that are characterized by a 
(possibly infinite) set of forbidden subprofiles. For instance, it immediately
implies that the domain of preferences single-peaked on a tree does not admit
such a characterization, whereas the domain of $d$-Euclidean preferences
does, for every $d\ge 1$.
In contrast, we are not aware of a general technique
that can distinguish between domains that can be characterized by finitely many
forbidden subprofiles (such as the single-peaked domain) and those that can 
only be characterized by an infinite set (such as the 1-Euclidean domain).
A complexity-theoretic approach offers a partial solution: 
a domain that can be characterized by a finite number of forbidden subprofiles
is polynomial-time recognizable (e.g., by a simple brute-force algorithm),
so if recognizing a restricted domain is known to be computationally hard, 
we can conclude that this domain does not admit a finite characterization (subject to a complexity assumption).

For example, \Cref{thm:recognizing-d-Euclidean} shows that for $d\ge 2$ recognizing whether a profile is $d$-Euclidean 
is NP-hard. Thus, unless P=NP, we know that for each $d\ge 2$ the $d$-Euclidean domain is not characterizable by a finite set 
of forbidden subprofiles. We note that \citet{peters2016recognising} proves the same result without referring 
to any complexity assumption.

Let us end this section by briefly mentioning some applications of 
characterizations via forbidden subprofiles.
Such characterizations have been shown to be useful for detecting profiles that are close to 
being in the respective restricted domain
\citep{elkind2014detecting,bredereck2016nicelystructured}; more details can be found in \Cref{sec:recog:almost}.
Forbidden subprofiles can also be used to prove general results about arbitrary domain restrictions:
\citet{lackner2017likelihood} obtain a combinatorial result counting the number of profiles of a given size 
that belong to a restricted domain characterized by a small forbidden subprofile.

\newpage
\section{Winner Determination}
\label{sec:problems}

Ordinal preferences are often used as inputs to group decision problems. In \emph{voting}, voters submit preference 
rankings, and the aim is to identify an alternative, a set of alternatives, 
or a ranking of the alternatives 
that best represent the voters' joint preferences. Over time, many variations of 
this setting have been studied, and many voting rules have been proposed. 
For each such voting rule, the key computational
problem is \emph{winner determination}, 
i.e., computing the output of the rule given the input preferences.

For many popular voting rules, this problem is straightforward to solve, and indeed many voting rules are \emph{defined} by a 
specification of a winner determination algorithm. For example, the \emph{Plurality} rule 
computes the score of each alternative as the number of voters that rank this alternative first, 
and outputs the alternative(s) with the highest score; clearly, the running time of this procedure
is linear in $n+m$. However, 
as observed by \citet{bartholdi1989voting}, some voting rules only specify their winning 
alternatives implicitly, and na\"ive winner determination algorithms require exponential time. 
Indeed, \citeauthor{bartholdi1989voting} prove that it is NP-hard to decide 
whether a given alternative is winning under the \emph{Dodgson voting rule},
and similar hardness results have subsequently been obtained for many (otherwise) 
very attractive voting rules.

As \citeauthor{bartholdi1989voting}~argue, a good voting rule should admit an efficient winner determination algorithm, for 
otherwise ``a candidate's mandate might have expired before it was ever recognized''! Indeed, in practice, 
voters demand the votes to be counted within days or hours. 
If we nevertheless wish to use voting rules for which computing the output is hard, 
we need to find ways to mitigate this computational complexity. Several 
popular strategies exist for this purpose. \citet{bartholdi1989voting} suggest an integer linear programming formulation 
that captures the winner determination problem for the Dodgson voting rule, 
and we may hope that practical instances will be solved quickly by modern solvers. 
We could employ parameterized analysis to develop algorithms whose running time is exponential, 
say, in the number of alternatives, but not the number of voters; this technique can be useful in elections
with a small number of alternatives. More controversially, we may use approximation algorithms:
if a voting rule is defined in terms of a scoring procedure so that it outputs the alternative(s)
with the maximum/minimum score, we may be able to design an efficient algorithm that finds an alternative
whose score is close to optimal. This approach may be acceptable in low-stake elections;
moreover, we can view the resulting approximation algorithms as new voting rules, and some of these rules
are quite attractive \citep{caragiannis2014dodgson,skowron2015approx}.

In this section, we will explore yet another strategy: 
we demonstrate that for many voting rules the winner determination problem
becomes easy when voters' preferences belong to one of the restricted preference domains discussed in this survey.
This approach fits within the framework of identifying 
``islands of tractability'', i.e., classes of inputs on which a given problem can be solved in polynomial 
time: just like many hard graph-theoretic problems become computationally easy if the input graph is a tree, 
many preference aggregation problems become easy if the input profile is single-peaked or single-crossing.
Also, while real-life preferences are rarely single-peaked or single-crossing, they are often not that far
from belonging to some restricted domain, in terms of distance measures discussed in \Cref{sec:recog:almost}
(see, e.g., \citealp{sui2013multi}),
and some of the positive results for structured profiles extend to almost structured profiles; we will mention
several results of this type later in this section. 

In the remainder of this section, we will focus on single-winner voting rules, which aim
to output a single winning alternative, and multi-winner rules, which elect a fixed-size set of winners. 
We will also briefly mention some results for social welfare functions, 
i.e., mappings that output a ranking of the alternatives.

\subsection{Single-Winner Rules}

This section studies single-winner voting rules, so-called \emph{social choice functions}. 
Formally, a social choice function is a mapping $f$ that for every preference profile $P$ selects a non-empty set 
$f(P) \subseteq A$ of \emph{winning alternatives}. The interpretation is that the alternatives 
in $f(P)$ are tied for winning, and some tie-breaking mechanism will later decide the winner.
For instance, under Plurality each alternative gets 1 point from each voter 
that ranks it first, and under the Borda rule, each alternative gets $m-i$ points from each voter who ranks it 
in position $i$; in both cases, the alternative(s) with the largest number of points 
are considered to be the election winners.

A large variety of social choice functions have been discussed in the literature
(see, e.g., \citealp{brandt2015handbook-chapter2}).
For most of them, an element of the set $f(P)$ can be computed in polynomial time. 
Indeed, for many rules, a stronger statement is true: given an alternative, 
we can decide in polynomial time whether it belongs to the set 
$f(P)$. This holds, for instance, for all scoring rules (a large class of rules that includes 
both the Borda rule and Plurality), 
for {Plurality with Runoff}, for {Copeland's rule}, for {Minimax}  and for {Schulze's method}. 

However, as observed by \citet{bartholdi1989voting}, there are appealing social choice functions 
for which one cannot find an element of $f(P)$ in polynomial time unless P=NP. The key examples are 
the Dodgson rule \citep{hemaspaandra1997dodgson}, the Young rule \citep{rothe2003exact}, and 
the Kemeny winner rule \citep{hemaspaandra2005complexity}. 
For each of these rules, deciding whether a given alternative is a 
winner is $\Theta_2^p$-complete.
Further examples are provided by some \emph{tournament solutions}, i.e., social choice functions
that only depend on the majority relation induced by a profile: specific tournament solutions that have a
hard winner determination problem are the Banks set, the 
minimal extending set, and the tournament equilibrium set
\citep{brandt2015handbook-chapter3}. 
We will now argue that for many of these rules,
winner determination becomes much easier if voters' preferences belong to one of the restricted domains
discussed in this survey.

\subsubsection{Condorcet Extensions}
Recall the definitions of weak and strong Condorcet winners from \Cref{sec:prel}: these are the alternatives preferred
to every other alternative by a weak (respectively, strong) majority of voters. In many settings, it is desirable
to elect Condorcet winners whenever they exist.
Formally, we say that a voting rule $f$ is a \emph{Condorcet extension} if  
given a profile $P$ with a strong Condorcet winner $c$, $f$ outputs $c$ as the unique winner at $P$.
A \emph{weak-Condorcet extension} is a voting rule that outputs the set of 
all weak Condorcet winners whenever this set is non-empty. 
By construction, every weak-Condorcet extension is also a Condorcet extension.

Many of the social choice functions commonly discussed in the literature are Condorcet extensions. 
Below we define three such rules; see also \citet{fishburn1977condorcet} for an overview. 

\begin{description}
\item[The Dodgson rule]
Charles Dodgson, better known by his pen name Lewis Carroll, proposed a Condorcet extension in 1876, though apparently 
unaware of Condorcet's work. According to Dodgson's method, the score of an alternative is
defined as the minimum number of swaps of adjacent alternatives in the preference profile needed 
to make this alternative a strong Condorcet winner, and the winners are the alternatives with the smallest score.
While this rule is not particularly attractive from an axiomatic perspective \citep{brandt2009dodgson}, 
it offers an interesting approach to winner determination, and can be seen as an instantiation of the 
\emph{distance rationalizability} framework \citep{brandt2015handbook-chapter8}. 
Moreover, it is an important example of a voting
rule with a hard winner determination problem: 
\citet{bartholdi1989voting} proved that checking whether a given alternative is a Dodgson winner is NP-hard, 
and subsequently \citet{hemaspaandra1997dodgson} proved that this problem is $\Theta_2^p$-complete, 
which was the first $\Theta_2^p$-completeness result in the computational social choice literature. 
\item[The Young rule]
\citet{young1977extending} proposed a voting method that is somewhat similar in spirit to 
the Dodgson rule. According 
to the Young rule, an alternative's score is the number of voters that need to be removed to make that alternative
a weak Condorcet winner; the alternatives with the smallest score are election winners.
Again, deciding whether a given alternative is a Young winner is 
$\Theta_2^p$-complete \citep{rothe2003exact,brandt2015bypassing}. There is also a variant of this rule,
referred to as strongYoung \citep{brandt2015bypassing},  
where the score of an alternative $a$ is the number of voters that need to be removed to make $a$ a \emph{strong}
Condorcet winner (this number may be $+\infty$ for some $a\in A$, but there is always an alternative 
whose strongYoung score is finite).
\item[The Kemeny rule]
The Kemeny rule \citep{kemeny1959mathematics} is a \emph{social welfare function}, i.e., a mapping that 
returns a set of linear orders over $A$. 
This mapping can then be transformed into a social choice function:
we say that an alternative is a \emph{Kemeny winner} if it is ranked first
in some linear order produced by the Kemeny social welfare function.

Given a preference profile $P$, we define the \emph{Kemeny score} of a linear order $\succ$ as
\[ 
\sum_{i\in N} | {\succ_i} \cap {\succ} | =  
\sum_{i\in N} |\{ (a,b)\in A \times A : a\succ_i b \text{ and } a \succ b \}|, 
\]
and output the ranking(s) with the maximum score. Intuitively, these are the rankings
that maximize  agreement with the input profile.
Equivalently, we can define the rule as minimizing the number of \emph{disagreements}
with the input profile.
Computing a Kemeny ranking is equivalent to solving the feedback arc set problem in the weighted majority graph 
induced by the preference profile \citep{bartholdi1989voting}, and therefore it is NP-hard to find a 
Kemeny ranking. Deciding whether a given alternative is a Kemeny winner is $\Theta_2^p$-complete
\citep{hemaspaandra2005complexity}.
\end{description}

For weak-Condorcet extensions, winner determination is easy
as long as the input profile belongs to a restricted domain that guarantees
the existence of a weak Condorcet winner: one can simply output all weak Condorcet winners. If 
the majority relation is antisymmetric (this happens, for instance, if the number of voters is odd), 
every weak Condorcet winner is also a strong Condorcet winner, so in this case we also get positive
results for Condorcet extensions. These observations are due to \citet[Theorem~3.2]{brandt2015bypassing}.

\begin{proposition}\label{prop:condorcet-compute}
	Consider a social choice function $f$, a profile $P$ over $A$ that is single-peaked
        on a tree, single-crossing on a tree, or group-separable, and an alternative $a\in A$. 
	Then if $f$ is a weak-Condorcet extension, 
	or if $f$ is a Condorcet extension and the majority relation associated with $P$ 
	is antisymmetric, 
	we can decide in polynomial time whether $a\in f(P)$.  
\end{proposition}

The Kemeny rule and the Young rule are weak-Condorcet extensions, 
and therefore by \Cref{prop:condorcet-compute} the winner determination problem is easy for these rules, 
as long as voters' preferences are single-peaked on a tree, single-crossing on a tree, or group-separable.
Also, tournament solutions are defined on tournament graphs, which correspond to profiles with antisymmetric
majority relation, and hence \Cref{prop:condorcet-compute} applies
to many tournament solutions as well. 
However, some well-studied Condorcet extensions fail to be weak-Condorcet extensions; 
for example, this is the case for the Dodgson rule. 
For these rules, \Cref{prop:condorcet-compute} only offers an easiness result
for an odd number of voters, and the case of the even number of voters has to be handled separately.
In the following sections, we discuss the Dodgson rule, the Young rule, 
and the Kemeny rule in more detail. 

\subsubsection{The Dodgson Rule}
While the Dodgson rule is clearly a Condorcet extension, it is not a weak-Condorcet extension 
\citep{brandt2015bypassing}. Thus, \Cref{prop:condorcet-compute}
does not tell us whether we can efficiently compute Dodgson winners for single-peaked 
profiles with an even number of voters. \citet{brandt2015bypassing} answer this question in the positive, 
presenting a greedy algorithm that returns all Dodgson winners when given a single-peaked profile. The correctness 
proof for this algorithm is somewhat involved. The key insights behind the algorithm are that every Dodgson winner 
must also be a weak Condorcet winner, and that there is always an optimal sequence of swaps that only modifies the 
preference relations of at most two voters. 
\citet[Theorem~16]{fitzsimmons2020dodgson} show that we can efficiently compute Dodgson winners
for single-crossing preferences, using a similar algorithm.

Beyond this result, it is natural to ask whether it is possible to compute the Dodgson \emph{score} of 
a given alternative in polynomial time, for profiles that belong to restricted domains (this problem is NP-hard in the general case). 
Being able to calculate scores is useful if one 
would like to use the Dodgson rule in order to \emph{rank} the alternatives.
\citet[Theorem 13]{fitzsimmons2020dodgson} prove that calculating Dodgson scores is easy for single-peaked preferences, 
which provides an alternative algorithm for computing Dodgson winners for single-peaked profiles.

\begin{open}
	Can one efficiently compute Dodgson scores for other restricted domains, such as single-crossing preferences? Can one efficiently find Dodgson winners for profiles that are 
	single-crossing on a tree or single-peaked on a tree, when the number of voters is even?
\end{open}

\subsubsection{The Young Rule}
Since the Young rule is a weak Condorcet 
extension, its winner determination problem becomes easy for all the domain restrictions mentioned in 
\Cref{prop:condorcet-compute}. While strongYoung is not a weak-Condorcet extension, 
the winners under this rule can also be computed in polynomial-time for single-peaked or single-crossing profiles.
This follows from the observation that for these domains the strongYoung score of the winning alternative does
not exceed 1: it suffices to remove at most one voter to engineer a situation where there is a unique median voter.

A related problem of interest is to compute the Young (strongYoung) score of a given alternative. 
For preferences that are single-peaked (or even single-peaked on a circle), 
\citet{peters2017spoc} present a simple counting-based algorithm that solves 
this problem in polynomial time. For single-crossing preferences, 
this problem is also polynomial-time solvable, because the effect of deleting 
voters is simply to shift around the position of the median voter \citep{magiera2017hard}.
In contrast, for group-separable preferences both variants of this problem are NP-hard,
though they admit an FPT algorithm with respect
to the height of the clone decomposition tree \citep{FaliszewskiKO22};
see the discussion in \Cref{sec:cc}.

\subsubsection{The Kemeny Rule}
As \Cref{prop:condorcet-compute} captures the complexity of computing Kemeny \emph{winners} for many restricted domains, 
in what follows we focus on computing Kemeny \emph{rankings}.
 
When the majority relation is transitive, the Kemeny ranking is unique and coincides with the majority relation:
every pair of alternatives $(a, b)$ contributes either $n_{ab} = |\{i\in N: a\succ_i b\}|$ or 
$n_{ba} = |\{i\in N: b\succ_i a\}|$ to the score of a ranking $\succ$, 
and the majority relation is the only ranking for which 
every pair $(a, b)$ contributes $\max\{n_{ab}, n_{ba}\}$. 
Thus, if the number of voters is odd and voters' 
preferences are single-peaked or single-crossing, the Kemeny ranking can be computed easily. 
For an even number of voters and single-peaked or single-crossing preferences, 
the strict part of the majority relation is transitive, and, by the same argument as above, 
the set of Kemeny rankings is exactly the set of all its linearizations.

For preferences single-peaked on trees, computing Kemeny rankings remains NP-hard. 
Indeed, we can reduce the problem of computing a Kemeny ranking for general preferences
to that of computing a Kemeny ranking for preferences single-peaked on a star, 
by adding a new alternative that is ranked first by all voters: this transformation
makes the profile single-peaked on a star, and there is a one-to-one correspondence between Kemeny
rankings for the original profile and Kemeny rankings for the modified profile \citep[p.~245]{peters2020spt-journal}. 
Hardness results also hold for preferences that are single-peaked on a circle \citep{peters2017spoc} because
this preference domain does not impose any restriction on the majority relation (see \Cref{sec:def:spoc}).
\begin{open}
	Can Kemeny rankings be computed efficiently for preferences 
	that are single-peaked on ``nice'' trees, e.g., trees with few leaves?
\end{open}

For multidimensional $d$-Euclidean preferences, $d \ge 2$, \citet{EST21} prove that deciding whether 
there exists a ranking with at most a given Kemeny score is NP-complete,
and hence finding an optimum ranking is NP-hard. 
They assume that a $d$-Euclidean embedding is provided in the input.
Their result also holds for $d$-Euclidean preferences defined with respect to the $\ell_1$ and $\ell_\infty$
metrics, as well as for Slater rankings.
This result is obtained by showing that McGarvey's theorem continues to hold for these preference domains.
\citet{hamm2021computing} consider a variant of this problem, where given as input a preference profile
and a $d$-Euclidean embedding (with respect to the $\ell_2$ metric), we are asked to find the ranking
with the highest Kemeny score among all rankings that are compatible with the given embedding.
They show that this problem can be solved in polynomial time for each fixed $d$.
\begin{open}
	What is the complexity of winner determination for other aggregation rules (such as Dodgson, Young, or multi-winner rules) under the $d$-Euclidean domain?
\end{open}

The Kemeny rule behaves like a median. \citet{kemeny1959mathematics} also defined a second rule that behaves like a \emph{mean}. \citet{lederer2024squared} study this rule under the name \emph{Squared Kemeny}, since it is obtained by squaring the distances appearing in the definition of the normal Kemeny rule. \citet{lederer2024squared} show that Squared Kemeny is also NP-hard. Since Squared Kemeny is not a function of the weighted majority margins, it is not clear that it would become easy to compute on restricted classes of preferences that induce a transitive majority relation, such as single-peakedness. Thus, it is an open problem to identify domains where Squared Kemeny rankings becomes easy to compute. The same problem is also open for the \emph{equalitarian Kemeny rule}, which find the ranking that minimizes the maximum distance to any voter, which is again NP-hard \citep{biedl2009complexity}.

\subsubsection{Sequential Rules}
An important class of single-winner rules that we have not discussed so far is that of sequential, or multi-step, 
rules. These are rules defined by sequential procedures, 
where one may have to break ties at each step: 
examples of such rules include 
Ranked Pairs, Instant Runoff Voting (also known as Single Transferable Vote), the Baldwin rule, and the Nanson rule.

If these ``intermediate'' ties are always broken 
according to a fixed order of alternatives, there is a unique winner for each profile, 
and it can be computed in polynomial time; however, the resulting rule is not neutral. 
On the other hand, if we define $f(P)$ to be the set of all alternatives that win for some way of 
breaking intermediate ties (``parallel-universe tie-breaking''), then it is still easy to compute \emph{some} element of $f(P)$,
but for many of these rules checking whether a given alternative is in $f(P)$ becomes NP-hard 
\citep{conitzer2009preference,brill12price,mattei2014hard}.
Thus, in terms of their computational complexity, 
these rules are positioned between ``easy'' rules, such as Plurality, and ``hard'' rules, 
such as the Dodgson rule. 
It remains open whether computing all winners under Instant Runoff Voting or Ranked Pairs becomes easy
when voters' preferences belong to one of the restricted domains considered in this survey;
no results of this type are known so far.
\begin{open}
	What is the complexity of checking whether an alternative is a winner under Instant Runoff Voting or Ranked Pairs under parallel-universe tie-breaking for profiles that are single-peaked, or that have some other structure?
\end{open}
\citet{jiang2017practical} study Instant Runoff Voting experimentally and find that on 
random single-peaked profiles, it is likely to select the Condorcet winner as the unique winner. \citet{tomlinson2023ballot} show that Instant Runoff Voting is better behaved under ballot truncation on single-peaked, single-crossing, and 1-Euclidean profiles. \citet{tomlinson2024moderating} show that on 1-Euclidean profiles, Instant Runoff Voting never elects candidates in extreme positions, if there exists at least one candidate in a location bounded away from the extremes; \citet{tomlinson2025exclusionzones} show that a similar result need not hold in higher dimensions.

Two sequential elimination methods that are similar to Instant Runoff Voting are the Baldwin rule (which at each step eliminates an alternative with the lowest Borda score) and the Coombs method (which at each step eliminates an alternative that is bottom-ranked by the highest number of voters). Both have NP-hard winner determination problems under parallel-universe tie-breaking and unrestricted preferences \citep{mattei2014hard}. The Baldwin rule is a weak-Condorcet extension (because a weak Condorcet winner is never a Borda loser unless all remaining alternatives are weak Condorcet winners, see \citealp{brandt2015handbook-chapter2}, Footnote 38), and hence it is easy to compute for preference domains such as single-peaked preferences. The Coombs method is not a Condorcet extension, but on single-peaked profiles with an odd number of voters, it always selects the Condorcet winner \citep[Prop.~2]{grofman2004coombs} and thus is easy to compute in this case.

\subsection{Multi-Winner Rules}
In many group decision problems the goal is to select multiple alternatives rather than a single winner. 
Voting rules used for this purpose are called 
\emph{multi-winner voting rules}, or \emph{committee selection rules}. 
The input to such rules is a preference profile $P$ over $A$ and an integer $k$; 
the output is a non-empty set of \emph{committees}, which are size-$k$ subsets of $A$. 
Multi-winner voting rules can be used to select governing bodies, decide how to allocate limited advertising space, 
or to shortlist applicants to be interviewed for a job \citep{elkind2017properties}.
The desiderata for committee selection rules depend on the target application:
choosing movies to be included in the in-flight entertainment system is very different
from deciding which applicants to accept into a selective PhD program. Accordingly, the
literature has proposed a wide variety of committee selection rules. 
For a general overview on this research topic, we refer the reader to the surveys by \citet{FSST-trends} and \citet{lackner2023abc}.

While for some multi-winner voting rules
one can compute a winning committee in polynomial time, for others this problem is computationally hard.
In the rest of this section, we focus on two committee selection rules whose aim is 
to select a \emph{representative committee} (the Chamberlin--Courant rule and the Monroe rule), and discuss the complexity of computing winning committees
when voters' preferences belong to a restricted domain.

\subsubsection{The Chamberlin--Courant Rule}\label{sec:cc}

\citet{cha-cou:j:cc} proposed a rule 
that aims to identify a committee that is maximally representative. 
Under their approach, each size-$k$ committee is assigned a score, and committees
with the highest score are declared to be the election winners. The score of a given
committee $W$ is computed by asking each voter to evaluate $W$, and aggregating the resulting scores;
under the utilitarian variant of this rule (this is the variant that was 
proposed by \citealp{cha-cou:j:cc}) the individual voters' scores are added up, 
and under the egalitarian variant \citep{betzler2013computation}
the score of a committee $W$ is computed
as the minimum of voters' scores. It remains to explain how voters evaluate committees.
In the classic variant of the rule, the evaluations are based on Borda scores:
if a voter ranks her most preferred member of $W$ in position $i$ in her vote, 
she assigns a score of $m-i$ to $W$. More generally, each vector $(w_1, \dots, w_m)$
with $w_1\ge\dots\ge w_m$ can be used to define a committee selection rule, 
which we will call $\mathbf w$-Chamberlin--Courant:
a voter's score for a committee $W$ is $w_i$ if her most preferred member of $W$
appears in the $i$-th position in her preference ranking.
Formally, the utilitarian and the egalitarian score of a committee $W$
under the $\mathbf w$-Chamberlin--Courant rule are defined as, respectively, 
\newcommand{\CCut}{\operatorname{CC}_{\mathrm{util}}}
\newcommand{\CCeg}{\operatorname{CC}_{\mathrm{egal}}}
\[ 
\CCut(W) = \sum_{i\in N} \max \{ w_{\rank_i(a)} : a \in W \}, \quad
\CCeg(W) = \min_{i\in N} \max \{ w_{\rank_i(a)} : a \in W \}.
\]
Intuitively, under both variants of the Chamberlin--Courant rule, each voter is represented by 
a single member of the committee. She can choose her representative
within the committee, and thus will choose the committee member she prefers most. Her satisfaction is measured by how much she likes her representative.

The winner determination problem associated with the Chamberlin--Courant rule is usually translated into a decision 
problem by considering it as an optimization problem, so that we ask: ``given a profile $P$, a target committee size 
$k$, and integer $B$, does there exist a committee $W$ whose Chamberlin--Courant score is at least $B$?''. This problem 
is NP-complete when using Borda scores (\citealp{lu2011budgeted}; see also \citealp{peters2020spt-journal}). Thus, we will analyze restrictions 
of this problem to structured preference profiles.

\paragraph{Single-Peaked Preferences}

\citet{betzler2013computation} show that a winning committee under the 
utilitarian Chamberlin--Courant rule 
can be identified in polynomial time if the input preferences are single-peaked. 
Their algorithm is a dynamic program which builds up an optimal committee 
by moving along the axis from left to right.

\begin{theorem}[\citealp{betzler2013computation}, Theorem~8]
	\label{thm:cc-sp}
	For any non-increasing scoring vector $\mathbf w$, any target committee size $k$, 
	and a single-peaked preference profile, we can find a winning committee according 
	to the utilitarian variant of the Chamberlin--Courant rule in $O(m^2n)$ time.
\end{theorem}

\begin{proof}
	Let $P$ be a profile that is single-peaked with respect to the axis 
        $a_1 \lhd a_2 \lhd \dots \lhd a_m$. We can find this axis in linear time using the 
        algorithm in \Cref{sec:recog:sp}. For each alternative $a_t \in A$, and every 
        integer $j$ with $j \le \min\{k, t\}$, define
	\[ 
        z(a_t, j) := \max \{ \CCut(W) : W \subseteq \{ a_1,\dots,a_t \}, 
                                       |W| = j, \text{ and } a_t \in W  \} 
        \]
	to be the highest Chamberlin--Courant score attainable by a size-$j$ committee that only 
        uses alternatives that appear in the first $t$ positions on the axis and includes $a_t$.
	Then, the Chamberlin--Courant score of an optimum committee is given by 
        $\max_{a_t \in A} z(a_t, k)$. 
	The values $z(a_t, j)$ can be computed using the following recurrence, valid for 
        $t \ge 2$, $2 \le j \le k$, $j \le t$:
	\[ 
        z(a_t, j) = \max_{j-1\le p\le t-1} \left\{ z(a_p, j-1) + 
        \sum_{i\in N} \max\{0, w_{\rank_i(a_t)} - w_{\rank_i(a_p)}\} \right\}. 
        \]
	To see that this recurrence is correct, consider a size-$j$ committee $W$ whose rightmost member 
        is $a_t$. Let $W' = W \setminus \{a_t\}$, and let $a_p$ be the rightmost member of $W'$. 
        Now, each voter $i$ who was represented by a alternative $a_q$ with $q<p$ in $W'$
	is represented by the same alternative in $W$: if she preferred $a_t$ to $a_q$, 
	then $a_q\lhd a_p\lhd a_t$
        would form a valley in her preferences. In particular, we have 
	$w_{\rank_i(a_t)} - w_{\rank_i(a_p)}\le 0$. On the other hand, a voter $i$
        that was represented by $a_p$ in $W'$ may be represented by either $a_p$ or $a_t$ in $W$;
	her gain from switching from $a_p$ to $a_t$ is $w_{\rank_i(a_t)} - w_{\rank_i(a_p)}$
	(if this value is negative, she prefers not to switch). 
	Hence the increase in the Chamberlin--Courant score of $W$ over $W'$ is just 
        $w_{\rank_i(a_t)} - w_{\rank_i(a_p)}$ summed over all voters $i$ who prefer 
        $a_t$ to $a_p$. For implementation details, see \citet[Thm.~8]{betzler2013computation}.
\end{proof}

\citet[Theorem 6]{sornat2022near} provide a faster algorithm for this problem, based on the Minimum Weight $t$-Link Path problem, running in $O(nm\log n + m^{1 + o(1)})$, which is almost linear.
Yet another algorithm (with worse running time) involves formulating the winner determination
problem under this rule as a totally unimodular integer linear program 
\citep{peters2018single,peters2017spoc}.

There also exist tractability results for preferences that are nearly single-peaked.
\citet{cor-gal-spa:c:spwidth} show that the dynamic programming approach of \Cref{thm:cc-sp} 
can be extended to profiles of bounded single-peaked width.
\citet{misra2017complexity} give algorithms for computing a Chamberlin--Courant winning committee
for profiles that can be made single-peaked by deleting a few voters or alternatives: 
their algorithms run in FPT time with respect to the number of alternatives to be deleted, 
and in XP time with respect to the number of voters to be deleted. 
Their algorithm for alternative deletion runs in time $O^*(2^d)$ where $d$ is the number of alternatives to be deleted; \citet{sornat2022near} show that this dependence on $d$ is best-possible under the Strong Exponential Time Hypothesis.
For voter deletion, \citet{chen2023efficient} improve the XP result by giving an FPT algorithm.
While the deletion-based measures for near single-peakedness lead to tractability results, 
\citet{misra2017complexity} show NP-hardness for a generalization of single-peaked preferences 
where up to three peaks are allowed. 

\citet{sonar2020complexity} study the problem of deciding 
whether a given alternative appears in some optimal Chamberlin--Courant committee. 
They show that this problem is $\Theta_2^p$-complete in general, but becomes polynomial-time 
solvable for single-peaked preferences. They leave it as an open question 
to determine whether the same is true for other preference domains.

For the egalitarian version of the Chamberlin--Courant rule, \citet{betzler2013computation} show that it can be evaluated in polynomial time using a simple greedy algorithm. 
To see why this is the case, suppose we wish to decide whether there is a committee whose 
egalitarian Chamberlin--Courant score is at least $B$. This is equivalent to 
looking for a committee where each voter's representative is ranked in 
position $b$ or higher, for an appropriate value of $b$. 
As voters' preferences are single-peaked, for each voter the set 
of alternatives ranked in position $b$ or higher forms an interval of the axis. 
So our problem is equivalent to asking whether 
there are $k$ alternatives that together cover (or ``stab'') all the intervals. 
As \citet{betzler2013computation} note, this interval stabbing problem 
is equivalent to a clique cover problem, which can be solved in linear time 
\citep{golumbic2004algorithmic}. Explicitly, the algorithm proceeds from left to right 
along the axis; if it reaches the rightmost point of the interval of some voter who is
not yet represented by any of the already-selected alternatives, it adds the alternative
associated with this point to the committee. The algorithm reports success if it reaches
the end of the axis without selecting more than $k$ alternatives. 
\citet{elkind2015owa} show tractability for single-peaked preferences
for a family of Chamberlin--Courant variants that interpolate between 
the utilitarian and the egalitarian objective.

\paragraph{Single-Crossing Preferences}

\citet{skowron2013complexity} study the complexity of the Chamberlin--Courant rule for 
single-crossing preferences, and obtain a positive result for this setting: one can find an 
optimal committee in polynomial time if the input profile is single-crossing.
Their approach is 
based on an interesting structural observation concerning the allocation of representatives to 
voters. For each member $a$ of a committee $W$ obtained under the Chamberlin--Courant rule, we 
can think of the voters represented by $a$ as $a$'s ``district''. For single-crossing preferences, 
the voter space has a one-dimensional structure, so a natural shape for each such district 
would be an interval of the voter ordering. 
\citet{skowron2013complexity} find that there always exists an optimal Chamberlin--Courant 
committee whose districts have this shape. All of their results hold 
both for the utilitarian and for the egalitarian version of this rule,
so in what follows we do not distinguish between the two versions.

\begin{proposition}[\citealp{skowron2013complexity}, Lemma 5]
	\label{thm:cc-sc-intervals}
	Let $P = (v_1, \dots, v_n)$ be a profile that is single-crossing with respect 
	to the given order, and fix a non-increasing scoring vector $\mathbf w$ 
	and a target committee size $k$. Then for any committee $W$ that is optimal
	for the Chamberlin--Courant rule 
	and for each $a\in W$, the set $\{ i \in N : a \succ_i b \text{ for all } b \in W \}$ 
	of voters who are represented by $a$ forms an interval of $(1,\dots, n)$.
\end{proposition}
\begin{proof}
	Suppose not, and alternative $a\in W$ represents voters $i$ and $\ell$, but does not represent 
        voter $j$, where $i < j < \ell$. Suppose that $j$ is represented by $b\in W$. 
        Then we must have $a \succ_i b$, $b\succ_j a$, and $a \succ_\ell b$, a contradiction 
        with $P$ being single-crossing with respect to the given order.
\end{proof}

\Cref{thm:cc-sc-intervals} serves as a basis for a dynamic
programming algorithm for the Chamberlin--Courant rule that runs in time $O(n^2mk)$ \citep[Theorem 6]{skowron2013complexity}.
\citet{andrei20} propose a different dynamic programming formulation 
that results in an $O(nmk)$ algorithm;
for the utilitarian version of the Chamberlin--Courant rule, 
they further improve the running time to
$nm2^{O(\sqrt{\log k\log\log n})}$ by interpreting the winner determination
problem as a $k$-link path problem with concave Monge weights.

\citet{skowron2013complexity} also show that their approach extends 
to profiles with bounded single-crossing width \citep{cornaz2013kemeny}, by proving that
the problem of finding an optimal Chamberlin--Courant committee is in FPT with respect 
to the single-crossing width. In particular, for profiles with bounded
single-crossing width they establish an analog of 
\Cref{thm:cc-sc-intervals} saying that the voters represented by the alternatives within 
a given clone set again form an interval; with this result in hand, one can again 
apply dynamic programming to find an optimal committee.

Just as for single-peaked preferences, \citet{misra2017complexity} give algorithms for computing 
a winning committee under the Chamberlin--Courant rule 
for profiles that can be made single-crossing by deleting a few  
voters or alternatives: their algorithms run in FPT time with respect to the number of 
alternatives to be deleted, and in XP time with respect to the number of voters to be deleted. 
\citet{chen2023efficient} give an FPT algorithm for voter deletion.
\citet{misra2017complexity} also show an NP-hardness result for a generalization of the single-crossing domain where 
each pair of alternatives is allowed to cross up to three times.

\paragraph{Preferences Single-Peaked on Trees}
For preferences single-peaked on trees, the egalitarian version 
of the Chamberlin--Courant rule remains polynomial-time computable.
The argument is similar to the interval-stabbing argument presented above, 
except that we now need to stab subtrees rather than subintervals 
\citep{peters2020spt-journal}.

However, in contrast to the situation for 
single-peaked or single-crossing profiles, the utilitarian 
version of the Chamberlin--Courant rule 
remains NP-complete for preferences that are single-peaked on a tree, 
even for Borda scores \citep{peters2020spt-journal}. 
One way around this result is to consider profiles that are single-peaked on trees that have 
additional restrictions imposed on them, such as trees that have few non-leaf vertices.

For example, consider the Chamberlin--Courant rule with Borda scores for a profile $P$ that is 
single-peaked on a star with center vertex $c\in A$. As we saw in \Cref{ex:spt-star}, all 
voters in $P$ rank $c$ in one of the first two positions in their vote. 
With this restriction, it is easy to 
find an optimal Chamberlin--Courant committee: including alternative $c$ in the 
committee guarantees very good representation to every voter, and the remaining spots 
in the committee can be filled greedily, by counting the number of alternatives' appearances
in the top position \citep[for details, see][]{peters2020spt-journal}. Interestingly, there are some 
non-Borda scoring vectors for which the Chamberlin--Courant rule remains hard even for preferences 
single-peaked on a star \citep{peters2020spt-journal}.

Generalizing the above argument, \citet{peters2020spt-journal} show that the Chamberlin--Courant rule 
with Borda scores admits a polynomial-time algorithm for preferences single-peaked on trees 
that have a bounded number of internal (non-leaf) vertices. They also 
give a polynomial-time algorithm for trees that have a bounded number of leaves, using a 
generalization of the dynamic program of \citet{betzler2013computation} for the case of the 
line. It would be interesting to identify further classes of trees that admit efficient 
algorithms for the Chamberlin--Courant rule.

\paragraph{Preferences Single-Crossing on Trees}
In contrast to the case of preferences single-peaked on trees, for preferences
single-crossing on trees the Chamberlin--Courant rule admits a polynomial-time
winner determination algorithm irrespective of the structure of the underlying tree.
This was first claimed by \citet{clearwater2015generalizing}, who put
forward a dynamic programming algorithm for this problem. However, \citet{andrei20}
observed that the dynamic program of \citet{clearwater2015generalizing}
may have exponentially many variables, 
and proposed a different dynamic programming formulation, which results in an 
$O(nmk)$ algorithm, both for the utilitarian and for the egalitarian version of the
Chamberlin--Courant rule.
Their approach is based on an analog of \Cref{thm:cc-sc-intervals}
for trees, showing that the `district' of each representative can be assumed
to form a subtree of the underlying tree.

\paragraph{Group-Separable Preferences}
Recall the  group-separable profile $P_{\textrm{gs-max}}^m$ 
over a set of $m$ alternatives 
constructed in the proof of \Cref{prop:gs-count}, 
and the bijection between the votes in this profile 
and the subsets of $A\setminus\{a_1\}$.
\citet{FaliszewskiKO22} use this bijection to show that the group-separable
domain is sufficiently rich for the Chamberlin--Courant rule to remain
computationally difficult on profiles from this domain. Below, 
we use the same approach to derive a simple hardness proof for 
the Chamberlin--Courant rule with the $2$-approval scoring vector,
i.e.,  $\mathbf{w} = (1, 1, 0, \dots, 0)$; essentially, we
show that (a variant of) the hardness proof of \citet{PRZ08} for this rule goes
through for the group-separable domain.

\begin{theorem}
Given a group-separable profile $P$ with $n$ voters and a parameter $k$, 
it is NP-hard to decide if there exists a committee of size $k$ whose
utilitarian $\mathbf w$-Chamberlin--Courant score for ${\mathbf w}=(1, 1, 0, \dots, 0)$
is at least $n$. 
\end{theorem}
\begin{proof}
Recall that an instance of \textsc{Vertex Cover} is a pair $\langle G, k\rangle$,
where $G$ is an undirected graph and $k$ is a positive integer;
it is a yes-instance if $G$ admits a vertex cover of size $k$.
Given an instance $\langle G, k\rangle$ of \textsc{Vertex Cover} with vertex set $V$,
$|V|=m$, and edge set $E$, $|E|=n$, 
we construct a group-separable profile $P$ over the set of alternatives $V$
that contains one voter for each edge 
$\{v_i, v_j\}$; this voter ranks $v_i$ and $v_i$ in top two positions.
Such a profile can be constructed by picking the required
voters from $P_{\textrm{gs-max}}^m$; as the group-separable
domain is hereditary, we are guaranteed to obtain a group-separable profile.
Now, a voter that corresponds to an edge $e = \{v_i, v_j\}$
assigns a score of $1$ to a committee $W\subseteq V$
if and only if the set of vertices $W$ covers $e$. Thus, the Chamberlin--Courant
score of a committee of size $k$ is $n$ if and only if this committee 
corresponds to a vertex cover of $G$.
\end{proof}
Note that for the weight vector $(1, 1, 0, \dots, 0)$ the utilitarian Chamberlin--Courant
score is equal to $n$ if and only if the egalitarian Chamberlin--Courant score 
is strictly positive. We obtain the following corollary. 
\begin{corollary}
Given a group-separable profile $P$ with $n$ voters and a parameter $k$, 
it is NP-hard to decide if there exists a committee of size $k$ whose
egalitarian $\mathbf w$-Chamberlin--Courant score for ${\mathbf w}=(1, 1, 0, \dots, 0)$
is at least $1$. 
\end{corollary}
The clone decomposition tree of $P_{\textrm{gs-max}}^m$ is a caterpillar, so its height
is $m$. \citet{FaliszewskiKO22} show that the problem of computing an optimal utilitarian 
Chamberlin--Courant committee for group-separable preferences 
is in FPT with respect to the height of the clone decomposition tree
(recall that the clone decomposition tree is 
essentially unique and polynomial-time computable). 
Specifically, they describe a dynamic programming-based algorithm
that runs in time that is linear in $2^h$, where $h$ 
is the height of the clone decomposition tree, and polynomial in 
the number of alternatives $m$, the number of voters $n$, 
and the committee size $k$.

\paragraph{Euclidean preferences}
For 1-Euclidean preferences, w can compute both the utilitarian and the egalitarian Chamberlin--Courant rule in polynomial time using the above results on single-peaked preferences.
For $d$-Euclidean preferences, $d \ge 2$, \citet{sonar2022multiwinner} show that computing the outcome of the egalitarian Chamberlin--Courant rule becomes NP-hard and hard to approximate up to a constant factor. They also give some polynomial-time approximation algorithms. The complexity of the utilitarian variant for 2-Euclidean preferences remains open, though \citet{godziszewski2021analysis} show that an approval-based variant of that rule is NP-hard.

\paragraph{Variants of the Chamberlin--Courant Rule}

Under the Chamberlin--Courant rule, each voter only obtains utility from their 
most-preferred committee member. Alternatively, we could allow voters to also obtain utility 
from their second-most-preferred committee member, etc. This approach leads to the class
of \emph{OWA-based rules} \citep{skowron2016finding}
(where `OWA' stands for `ordered weighted average'), 
which generalize the Chamberlin--Courant rule. 
To calculate a voter $i$'s utility, an OWA rule with a list of weights 
$(\lambda_1, \lambda_2, \dots)$
first orders (sorts) the numbers $(w_{\rank_i(a)})_{a\in W}$ and then 
takes a weighted average with weights $\lambda_1, \lambda_2, \dots$; 
the goal is to find a committee that maximizes the sum of voters' utilities.
\citet{peters2017spoc} show that these rules are polynomial-time computable 
for single-peaked preferences.

\subsubsection{The Monroe Rule}

While the Chamberlin--Courant rule excels at providing good representation to as many voters as 
possible, it may fail to represent the voters \emph{proportionally}. Indeed, in many situations some 
committee members will have to represent a large fraction of the voters, while other committee 
members will be responsible for a few voters only. Consider, for example, a profile in which five
voters rank $a\in A$ first, whereas each of the alternatives $b,c,d,e,f\in A$ is ranked first 
by a single voter, so that there are ten voters in total. If we aim for a committee of size $k=6$, 
then the optimal Chamberlin--Courant committee will be $W = \{a,b,c,d,e,f\}$. Notice that $a$ 
represents five times as many voters as the other alternatives. This may be undesirable: on the 
one hand, alternative $a$ may complain that they have excessive responsibilities, and, on the other 
hand, the five voters who rank $a$ first may justifiably demand that half of the six committee 
seats be filled with alternatives that they rank highly---it is possible, for instance, that 
these five voters place $b,\dots,f$ at the bottom of their rankings.

\citet{cha-cou:j:cc} address the latter problem by suggesting that the committee $W$ use 
\emph{weighted voting} to make internal decisions, and then $a$ would receive as much voting 
weight as the other five committee members together. (They also consider the effect of using power 
indices to determine the voting weights.) A different proposal, which directly ensures proportional 
representation, is due to \citet{monroe1995fully}.

The Monroe rule searches for committees maximizing essentially the same objective function as 
in the Chamberlin--Courant scheme. However, Monroe requires that (up to rounding issues) every 
committee member represents the same number of voters.

To define the Monroe rule formally, we will need to make explicit the assignment of voters to 
representatives. For a committee $W$, such an assignment is just a function $\Phi : N \to W$, 
so that $\Phi(i)$ is the representative of voter $i$; note that $\Phi(i)$ need not be $i$'s
most preferred alternative in $W$. An assignment is \emph{balanced} if 
$\lfloor n/k \rfloor \le |\Phi^{-1}(a)| \le \lceil n/k \rceil$ for each committee member $a\in 
W$, so that every committee member represents essentially the same number of voters in $\Phi$. 
Given a scoring vector $\mathbf w$, the Monroe score of an assignment $\Phi$ is $\sum_{i\in N} 
w_{\rank(\Phi(i))}$. The Monroe score of a committee $W$ is the Monroe score of the best 
balanced assignment of voters to the members of $W$. The Monroe rule then returns all 
committees with the highest score.

As before, we can replace the sum in the definition of Monroe scores by a minimum over all 
voters $i\in N$ to obtain an egalitarian variant of the Monroe rule.

Unsurprisingly, given that evaluating the Chamberlin--Courant rule is hard, it is also hard to 
find an optimal committee under the Monroe rule for general preferences, for both its 
utilitarian and its egalitarian variant \citep{betzler2013computation}. As we will see, the 
addition of the balancedness constraint is an additional challenge from an algorithmic perspective, 
even if the input preferences are structured.

\paragraph{Single-Peaked Preferences}

\citet{betzler2013computation} show that the \emph{egalitarian} version of 
the Monroe rule admits a polynomial-time algorithm for single-peaked preferences 
and any choice of the scoring vector $\mathbf w$. Their algorithm is rather involved 
and exploits a connection to the 1-dimensional rectangle stabbing problem. The problem admits an FPT (resp. XP) algorithm for nearly single-peaked profiles with respect to voter (resp. alternative) deletion \citep{chen2023efficient}.

For the utilitarian version of the Monroe rule, \citet{betzler2013computation} give an example of a scoring vector 
$\mathbf w$ for which evaluating the Monroe rule remains hard even for single-peaked preferences. 
However, it remains a challenging open problem to determine whether 
hardness also holds for Borda scores:

\begin{open}
What is the computational complexity of computing a winning committee under the Monroe rule with 
Borda scores if voters' preferences are single-peaked?
\end{open}

\paragraph{Single-Crossing Preferences}

\citet{skowron2013complexity} show that the Monroe rule with Borda scores remains NP-complete 
to evaluate for single-crossing profiles, by giving a very involved reduction from the 
\emph{unrestricted} version of the winner determination problem under the Monroe rule. For the 
egalitarian version, the complexity is open. However, \citet{skowron2013complexity} show that 
for \emph{narcissistic} single-crossing profiles, an efficient algorithm is available, and 
\citet{elkind2014characterization} extend this result to profiles that are both single-crossing and 
single-peaked. These two algorithms are faster than the algorithm by 
\citet{betzler2013computation} for single-peaked profiles (see the proof of \Cref{thm:cc-sp}).

\paragraph{Other Restricted Domains}

For other domain restrictions, such as preferences single-peaked on trees, no results 
for the Monroe rule are known. 
\begin{open}
	What is the complexity of computing the utilitarian or egalitarian version of the Monroe rule  for preferences single-peaked on trees?
\end{open}
A hardness result for 
the utilitarian version of the Monroe rule over this domain 
should be easier to obtain than for the elusive case of single-peakedness on a line.
 
\newpage
\section{Manipulation and Control}
\label{sec:manip}

Algorithmic aspects of strategic behavior in elections are a major topic within 
computational social choice.
As with many voting problems, the computational complexity of determining how to influence an 
election often decreases if structure in preferences is assumed. However, in contrast to the 
results considered in \Cref{sec:problems}, here it is not necessarily the
case that a decrease in complexity is desirable.
Indeed, NP-hardness results
for manipulation are often interpreted as `barriers' to strategic behavior, 
so polynomial-time algorithms for structured preferences 
are viewed as negative results in this context, 
showing that for such preferences these barriers may disappear.

\citet{walsh2007uncertainty} was the first to consider the complexity
of strategic behavior for restricted preference domains, showing
that a popular voting rule known as Single Transferable Vote
remains NP-hard to manipulate when voters' preferences are single-peaked;
this result holds if there are at least three alternatives and the voters are weighted, 
with weights given in binary. The first systematic study of the complexity
of manipulation and control in single-peaked elections was undertaken by 
\citet{faliszewski2011shield}, and by now there is a large body
of research that considers various forms of strategic behavior, 
focusing primarily (though not exclusively) on single-peaked and nearly
single-peaked preferences.

In the following we will focus on two computational problems that have been 
fundamental in the study of strategic voting:
coalitional manipulation and control.
There is a large research literature on these and related problems 
(such as cloning and bribery); 
we refer the reader to the surveys of \citet{faliszewski2010ai}, 
\citet{faliszewski2010using}, \citet{brandt2015handbook-chapter6}, 
\citet{brandt2015handbook-chapter7}, and \citet{hemaspaandra2016complexity}.
Furthermore, we restrict ourselves to two simple voting rules: Borda and Veto. Both Borda and 
Veto are scoring rules, i.e., alternatives receive points depending on their positions in 
voters' rankings and the alternative(s) with the largest number of points win the election.
For Borda, each alternative gets $m-i$ points from each voter who ranks it in position $i$.
For Veto, alternatives get $1$ point from each voter who ranks them in position $1, \dots, m-1$,
and $0$ points from voters that rank them last. Towards the end of this section, 
we provide a brief summary of results for other voting rules and 
other forms of strategic behavior.

\paragraph{Coalitional Manipulation}
As a consequence of the seminal result of \citet{gibbard1973manipulation} 
and \citet{satterthwaite1975strategy}, any reasonable voting rule is manipulable, 
in particular Borda and Veto.
Hence, there are preference profiles in which some voter can change the outcome of the election in their favor if they misreport their true preferences.
If a group of voters form a coalition with the intention of manipulation, 
they may be able to influence the outcome of the election even if no single voter
is pivotal.
This form of \emph{coalitional manipulation} is captured by the following computational problem.

\problem{Unweighted $\mathcal{R}$ Coalitional Manipulation}{A preference profile $P$, a distinguished alternative $c$, and an integer $k\ge 0$.}{Can we add $k$ new voters (manipulators) to $P$ so as to make $c$ the winner according to voting rule $\mathcal{R}$?}

The \problemname{Weighted Coalitional Manipulation} is defined analogously,  
but with both voters in the original profile $P$ and the manipulators 
having non-negative integer weights (given in binary).

Note that these two problems can also be viewed from the perspective of winner determination given 
incomplete information: given a profile in which $k$ voters have not yet declared their 
preferences, is it still possible that alternative $c$ wins?

Let us first consider the case of unrestricted preferences, i.e., arbitrary preference profiles.
\problemname{Unweighted $\mathcal{R}$ Coalitional Manipulation} is solvable in polynomial time for the Veto rule \citep{zuckerman2009algorithms} and is NP-complete for Borda \citep{betzler2011unweighted,davies2011complexity}.
The \problemname{Weighted Coalitional Manipulation} problem 
is NP-complete for both Veto and Borda; indeed, it is NP-complete for essentially all scoring rules
except for Plurality \citep{conitzer2007elections,hemaspaandra2007dichotomy,procaccia2007junta}.

We have already seen many examples of hard computational social choice problems that 
become polynomial-time solvable for single-peaked preferences.
For coalitional manipulation we have to be a bit careful in defining
what the corresponding computational question actually is. 
The approach that is common in the literature~\citep{walsh2007uncertainty,faliszewski2011shield} 
is to assume that a single-peaked axis $\lhd$ is part of the input, 
so that $P$ is single-peaked with respect to $\lhd$ 
and the manipulators must submit votes that are also single-peaked with respect to $\lhd$.
In the following, if we consider manipulation problems given single-peaked profiles, 
we assume this model and the associated extra input.

It is not hard to see that \problemname{Unweighted Veto Coalitional Manipulation} 
for single-peaked profiles is in P: observe that only the two outermost alternatives 
on $\lhd$ can be vetoed. For Borda elections, \citet{yang2016exact} show that 
\problemname{Unweighted Borda Coalitional Manipulation} for single-peaked profiles 
can be solved in polynomial time for one or two manipulators; the general question 
remains as an interesting open problem.

\begin{open}
What is the computational complexity of \problemname{Unweighted Borda Coalitional Manipulation} 
for single-peaked profiles with three or more manipulators?
\end{open}

For the weighted problem, the computational landscape is far better explored: 
\problemname{Weighted $\mathcal{R}$ Coalitional Manipulation} is in P for Veto and NP-complete for Borda.
This follows from a very general theorem by~\citet{brandt2015bypassing}, which provides a complexity dichotomy for all scoring rules.

One may then wonder
if this easiness result for Veto extends to preferences that are nearly single-peaked. 
The impact of nearly single-peaked preferences on coalitional manipulation and control was first studied by
\citet{fal-hem-hem:j:nearly-sp}.
They obtain a dichotomy results for profiles that can be made single-peaked 
by deleting voters (cf.\ \Cref{sec:recog:almost}).

\begin{theorem}[\citealp{fal-hem-hem:j:nearly-sp}]\label{thm:veto-nearly-manip}
Let $m\geq 3$, and consider $m$-alternative profiles 
that can be made single-peaked by deleting at most $\ell$ voters. 
For such profiles
\problemname{Unweighted Veto Coalitional Manipulation} 
can be solved in polynomial time if $\ell\leq m-3$ and is NP-complete for $\ell>m-3$.
\end{theorem}

We see that manipulating the Veto rule is easy as long as the profile is close to being single-peaked, 
and becomes computationally hard if we need to delete many voters to obtain a single-peaked profile.
Interestingly, the statement of \Cref{thm:veto-nearly-manip} extends to alternative 
deletion, i.e., we have polynomial-time solvability for profiles that 
can be made single-peaked by deleting $\ell$-alternatives
if $\ell\leq m-3$ and NP-completeness for $\ell>m-3$ \citep{erdelyi2015manipulation}. 
\citet{erdelyi2015manipulation} further explore these questions for many of the closeness 
measures introduced in \Cref{sec:recog:almost} and confirm the intuition that 
easiness results only hold for almost single-peaked profiles.
Further work on almost single-peaked 
profiles has been done by \citet{menon2016reinstating}, who study the complexity
of coalitional manipulation when voters are allowed to submit partial orders of a certain form, 
namely \emph{top-truncated ballots}, and \citet{yang2015manipulation}, 
who focuses on single-peaked width (cf. \Cref{sec:recog:almost}).

\paragraph{Control}
Strategic behavior in elections may also originate from an authority with the power 
to directly influence the election, which is referred to as \emph{control}. 
In particular, this authority might be able to exclude some voters from participating 
in the election so as to make a certain alternative win.
This motivates the following computational problem:

\problem{$\mathcal{R}$ Constructive Control by Deleting Voters 
($\mathcal{R}$-CCDV)}
{A preference profile $P$, a distinguished alternative $c$, 
and an integer $k\ge 0$.}
{Can we delete $k$ voters from $P$ so as to make $c$ the winner 
according to voting rule $\mathcal{R}$?}

Constructive control by deleting alternatives is defined in a similar manner.
One can also consider control by adding voters or alternatives.
In this case, an instance of the problem consists 
of a profile $P$ over $A$ and a pool of additional voters/alternatives
that can be added; for adding voters, we assume that
each voter in the pool has a ranking of the alternatives in $A$, 
and for adding alternatives, it is assumed
that each voter in the original profile has a ranking over the set $A\cup A'$, 
where $A'$ is the set of potential new alternatives, 
and if we choose to add a subset $B\subseteq A'$ of the alternatives, 
each voter forms their vote by restricting of their ranking of $A\cup A'$
to $A\cup B$. 

CCDV is NP-complete both for Borda~\citep{russell2007complexity} and 
for Veto elections~\citep{lin2012solving}.
If we now turn to single-peaked profiles, we encounter an interesting situation: 
\citet{yang2017complexity} shows that \problemname{Borda-CCDV} remains hard even for 
single-peaked preferences, but can be solved in polynomial time 
for single-caved preferences (cf.\ \Cref{def:single-caved}).
In contrast, \problemname{Veto-CCDV} is solvable in polynomial time for single-peaked preferences; 
single-caved preferences have not been studied in this context so far.
For Copeland$^{\alpha}$-rules ($\alpha \in [0,1)$) the complexity of the CCDV (and the CCAV problem about adding voters) is open.

Further work has established a detailed complexity map of control problems 
for almost single-peaked 
preferences~\citep{fal-hem-hem:j:nearly-sp,yang2014approval,yang2014controlling,yang2015multipeak,yang2020comp}.

\citet{magiera2017hard} study the complexity of control 
for single-crossing preferences. They consider 
adding/deleting voters/alternatives, and investigate 
both constructive and destructive control, 
for the Plurality rule
and the Condorcet rule (i.e., the rule that outputs the Condorcet winner
if it exists and no winners otherwise). 
They identify a number of settings where the respective control
problem is NP-hard for general preferences, but becomes
polynomial-time solvable for single-crossing preferences.
\citet{bulteau2015} extend some of these results to the more general problem 
of combinatorial voter control, where the attacker's cost
for adding/deleting a group of voters is not necessarily equal
to the sum of her costs for adding/deleting individual voters in that group.

For group-separable preferences, \citet{FaliszewskiKO22} investigate the complexity of constructive
control by adding/deleting voters/alternatives; 
again, they focus on the Plurality rule and the Condorcet rule. 
Just as for the winner determination
problems discussed in \Cref{sec:problems}, 
for each of the settings they consider, they obtain an NP-hardness result
for general group-separable preferences and an FPT result with respect
to the height of the clone decomposition tree.

\paragraph{Bribery and Other Forms of Strategic Behavior}
The \emph{bribery} problem was introduced by \citet{faliszewski2009hard};
in this problem, an external party (briber) can pay each voter
to change their vote according to the briber's instructions.
\citet{brandt2015bypassing} find that this problem, while NP-hard for several voting rules, becomes polynomial-time solvable for single-peaked elections.
Subsequently, \citet{elkind2009swap} proposed two more fine-grained models, where the price of changing 
each vote depends on the nature of the changes requested; these are called swap bribery and shift bribery.
The computational complexity of these two models have been analysed by \citet{EFGR20}
for single-peaked and single-crossing preferences,
and they obtained polynomial-time algorithms for a number of rules, 
including Plurality, Borda and Condorcet-consistent rules.
 
\newpage
\section{Further Topics}
\label{sec:further-topics}

In the following we list topics that concern domain restrictions and their computational 
aspects, but have not been discussed elsewhere in this survey.

\subsection{Weak Orders}

\paragraph{Single-Peaked Preferences}
Extending the definition of single-peaked preferences to weak orders is not as straightforward 
as one might believe. Indeed, there are several natural ways of defining single-peaked weak 
orders.
All these definitions collapse to the usual definition of single-peaked preferences 
for linear orders, but differ in their treatment of ties.
\begin{itemize}
	\item[(a)] In his original definition, \citet{black1948rationale,black1958theory} allows ties `across' the peak, 
	but no ties at the peak\footnote{\citet{arr:b:sc} states Black's definition slightly differently. In Arrow's definition, we are allowed to tie two top alternatives;
		cf.\ the discussion by \citet{dummett1961stability}.}
	 and no ties between two alternatives on the same side 
	of the peak. 
	For example, a voter may be indifferent between the alternatives immediately to the left 
	and immediately to the right of the voter's peak. 
	Written down formally, this is exactly the same as \Cref{def:sp}.
	Note that a voter can never report a tie
	between more than two alternatives.
	\item[(b)] In \emph{single-plateaued 
		preferences} \citep{black1958theory,moulin1984generalized}, the vote may have ties among several most-preferred alternatives, but no other ties are allowed. The set of most-preferred alternatives must form an interval of the axis $\lhd$. Some other authors allow additional ties between two alternatives on opposite sides of the top plateau.
	\item[(c)] Finally, \emph{possibly single-peaked preferences} \citep{lackner2014incomplete} are profiles of weak orders
	where it is possible to break ties so as to obtain a profile of linear orders that is single-peaked.
	Equivalently, a profile of weak orders is possibly single-peaked if there is an axis $\lhd$ such that for each vote $\pref$ and each $c \in A$, the set $\{a \in A : a \pref c\}$ is an interval of $\lhd$; this is the same condition that appears in \Cref{prop:sp-equiv}. 
\end{itemize}

The last notion is the most general among these three definitions: 
every profile that is single-peaked according to (a) or (b) is also single-peaked according to (c). 
In addition, definition (c) naturally generalizes to a notion of single-peakedness for
profiles of partial orders, which are single-peaked if they can be extended to a single-peaked
profile of linear orders \citep{lackner2014incomplete}.

See \Cref{fig:sp-weak} for examples.

\begin{figure}
	\centering
	\begin{subfigure}{0.3\textwidth}
		\centering
		\begin{tikzpicture}[yscale=0.5,xscale=0.65]
		
		\def\xmin{1}
		\def\xmax{7}
		\def\ymin{0}
		\def\ymax{7}
		
		\draw[step=1cm,black!20,very thin] (\xmin,\ymin) grid (\xmax,\ymax);

		\draw[->] (\xmin -0.5,\ymin) -- (\xmax+0.5,\ymin) node[right] {};
		\foreach \x/\xtext in {1/a, 2/b, 3/c, 4/d, 5/e, 6/f, 7/g}
		\draw[shift={(\x,\ymin)}] (0pt,2pt) -- (0pt,-2pt) node[below] {$\strut\xtext$};
		\foreach \x/\xtext in {1, 2,3,4,5,6}
		\node[below] at (\x+0.5,\ymin) {$\strut\lhd$};  
		
		\foreach \x/\y in {4/6,5/7,3/4, 6/6, 2/3, 7/3, 1/1}
		\node[fill=blue, circle, inner sep=0.6mm] at (\x,\y) {};
		\draw[thick,blue] (1,1)--(2,3)--(3,4)--(4,6) -- (5,7)--(6,6)--(7,3);
		
		\foreach \x/\y in {4/4,5/3, 6/2, 3/5, 2/6, 7/1, 1/2}
		\node[fill=green!50!black, circle, inner sep=0.6mm] at (\x,\y) {};
		\draw[thick,green!50!black] (1,2)--(2,6)--(3,5)--(4,4) -- (5,3)--(6,2)--(7,1);  
		
		\foreach \x/\y in {1/3, 2/4, 3/6, 4/7, 5/5, 6/3, 7/2}
		\node[fill=red!50!black, circle, inner sep=0.6mm] at (\x,\y) {};
		\draw[thick,red!50!black]  (1,3)--(2,4)--(3,6)--(4,7) -- (5,5)--(6,3)--(7,2);
		
		\end{tikzpicture}
		\caption{Single-peaked with Ties}
	\end{subfigure}
	\quad 
	\begin{subfigure}{0.3\textwidth}
		\centering
		\begin{tikzpicture}[yscale=0.5,xscale=0.65]
			
			\def\xmin{1}
			\def\xmax{7}
			\def\ymin{0}
			\def\ymax{7}
			
			\draw[step=1cm,black!20,very thin] (\xmin,\ymin) grid (\xmax,\ymax);

			\draw[->] (\xmin -0.5,\ymin) -- (\xmax+0.5,\ymin) node[right] {};
			\foreach \x/\xtext in {1/a, 2/b, 3/c, 4/d, 5/e, 6/f, 7/g}
			\draw[shift={(\x,\ymin)}] (0pt,2pt) -- (0pt,-2pt) node[below] {$\strut\xtext$};
			\foreach \x/\xtext in {1, 2,3,4,5,6}
			\node[below] at (\x+0.5,\ymin) {$\strut\lhd$};  
			
			\foreach \x/\y in {4/6,5/6,3/4, 6/6, 2/3, 7/2, 1/1}
			\node[fill=blue, circle, inner sep=0.6mm] at (\x,\y) {};
			
			\draw[thick,blue] (1,1)--(2,3)--(3,4)--(4,6) -- (5,6)--(6,6)--(7,2);
			
			\foreach \x/\y in {4/4,5/3, 6/2, 3/5, 2/5, 7/1, 1/5}
			\node[fill=green!50!black, circle, inner sep=0.6mm] at (\x,\y) {};
			
			\draw[thick,green!50!black] (1,5)--(2,5)--(3,5)--(4,4) -- (5,3)--(6,2)--(7,1);  
			
			\foreach \x/\y in {1/2,2/4,3/7, 4/7, 5/5, 6/3, 7/1}
			\node[fill=red!50!black, circle, inner sep=0.6mm] at (\x,\y) {};
			
			\draw[thick,red!50!black]  (1,2)--(2,4)--(3,7)--(4,7) -- (5,5)--(6,3)--(7,1);
			
		\end{tikzpicture}
		\caption{Single-Plateaued}
	\end{subfigure}
	\quad
	\begin{subfigure}{0.3\textwidth}
		\centering
		\begin{tikzpicture}[yscale=0.5,xscale=0.65]
		
		\def\xmin{1}
		\def\xmax{7}
		\def\ymin{0}
		\def\ymax{7}
		
		\draw[step=1cm,black!20,very thin] (\xmin,\ymin) grid (\xmax,\ymax);

		\draw[->] (\xmin -0.5,\ymin) -- (\xmax+0.5,\ymin) node[right] {};
		\foreach \x/\xtext in {1/a, 2/b, 3/c, 4/d, 5/e, 6/f, 7/g}
		\draw[shift={(\x,\ymin)}] (0pt,2pt) -- (0pt,-2pt) node[below] {$\strut\xtext$};
		\foreach \x/\xtext in {1, 2,3,4,5,6}
		\node[below] at (\x+0.5,\ymin) {$\strut\lhd$};  
		
		\foreach \x/\y in {4/6,5/7,3/2, 6/6, 2/2, 7/2, 1/1}
		\node[fill=blue, circle, inner sep=0.6mm] at (\x,\y) {};
		\draw[thick,blue] (1,1)--(2,2)--(3,2)--(4,6) -- (5,7)--(6,6)--(7,2);
		
		\foreach \x/\y in {4/2,5/2, 6/2, 3/5, 2/5, 7/1, 1/2}
		\node[fill=green!50!black, circle, inner sep=0.6mm] at (\x,\y) {};
		\draw[thick,green!50!black] (1,2)--(2,5)--(3,5)--(4,2) -- (5,2)--(6,2)--(7,1);  
		
		\foreach \x/\y in {1/4,2/4,3/7, 4/7, 5/5, 6/3, 7/3}
		\node[fill=red!50!black, circle, inner sep=0.6mm] at (\x,\y) {};
		\draw[thick,red!50!black]  (1,4)--(2,4)--(3,7)--(4,7) -- (5,5)--(6,3)--(7,3);
		
		\end{tikzpicture}
		\caption{Possibly Single-Peaked}
	\end{subfigure}
	\caption{Different definitions for single-peakedness on weak orders}\label{fig:sp-weak}
\end{figure}

To recognize profiles that are single-peaked in one of these senses, \citet{fitzsimmons2020incomplete} present reductions to the consecutive ones problem, meaning that the recognition problems can be solved in polynomial time. For profiles of partial orders, \citet{fitzsimmons2020incomplete} prove that the recognition problem of possible single-peakedness is NP-complete, though they show that for a \emph{given} axis $\lhd$, one can check in polynomial time whether a profile of partial orders can be extended to linear orders that are single-peaked on $\lhd$.

Some of the desirable properties of single-peaked preferences extend to these generalizations. For example, Condorcet winners exist for profiles satisfying Black's condition (a), by the same proof as for linear orders (see \Cref{prop:median-voter-theorem}). The case for single-plateaued preferences is more complicated; \citet{barbera2007indifferences} provides a detailed discussion. For possibly single-peaked preferences, Condorcet winners definitely do not exist. \citet[Table~9.1]{Fish73a} provides an example with 5 voters, where $b \succ_1 a \succ_1 c$, $c \succ_2 b \succ_2 a$, $c \succ_3 b \succ_3 a$, $a \succ_4 b \sim_4 c$, and $a \succ_5 b \sim_5 c$. This profile is possibly single-peaked on the axis $a \lhd b \lhd c$. Its majority relation satisfies $a \succ_\text{maj} c \succ_\text{maj} b \succ_\text{maj} a$, which is not transitive.

Regarding winner determination, the algorithms based on total unimodularity for multi-winner voting rules discussed by \citet{peters2017spoc} continue to work for possibly single-peaked profiles of weak orders.
\citet{fitzsimmons2015complexity,FitzsimmonsH16} and \citet{menon2016reinstating} investigated the complexity 
of manipulation and bribery when the input preference profile 
consists of weak orders and is single-peaked.

\paragraph{Single-Crossing Preferences}
Like for single-peaked preferences, there are several ways to generalize single-crossing preferences to weak orders. \citet{elkind2015incompletecrossing} discuss three natural definitions. They say a profile $P = (v_1, \dots, v_n)$ of weak orders is
\begin{itemize}
	\item[(a)] \emph{weakly single-crossing} in the given order if for every pair $a,b\in A$, the sets $I_1 = \{ i \in [n] : a \succ_i b\}$, $I_2 = \{ i \in [n] : a \sim_i b\}$, and $I_3 = \{ i \in [n] : b \succ_i a\}$ form intervals of $[n]$, and $I_2$ is between $I_1$ and $I_3$.
	\item[(b)] \emph{seemingly single-crossing} in the given order if for every pair $a,b\in A$, if $i < k$ are such that $a \succ_i b$ and $a \succ_k b$, then for every $i \le j \le k$ we have either $a \succ_j b$ or $a \sim_j b$. In other words, as we scan from left to right, ignoring voters indifferent between $a$ and $b$, we see that the voters with $a \succ b$ are on one end, and the voters with $b \succ a$ are on the other.
	\item[(c)] \emph{possibly single-crossing} if there is some way to break ties in $P$ so as to obtain a profile of linear orders that is single-crossing.
\end{itemize}
Then $P$ is weakly (resp. seemingly or possibly) single-crossing if there exists a permutation of the votes in $P$ such that the resulting profile is weakly (resp. seemingly or possibly) single-crossing in the given order.

\citet{elkind2015incompletecrossing} consider the computational complexity of recognizing profiles satisfying these conditions. They find that weakly single-crossing profiles can be recognized in polynomial time using similar techniques as in the case for linear orders. On the other hand, by reduction from the betweenness problem, they show that recognizing seemingly and possibly single-crossing profiles is NP-complete. Intriguingly, they find that these problems are not even necessarily easy if we fix the ordering of the voters.

\paragraph{Single-Peaked on a Tree or Circle}

It is possible to define single-peakedness for trees and circles (and other graphs) using the same recipe that we saw for possibly single-peaked preferences: require that any upper contour set $\{a \in A : a \pref c\}$ is connected in the underlying graph. To recognize profiles that are single-peaked in such a sense, for circles we can use algorithms for the circular ones property \citep{booth1976testing,peters2017spoc}, and for trees we can use one of a number of different algorithms for ``tree convexity'' or ``hypergraph acyclicity'' \citep{trick1988induced,conitzer2004combinatorial,tarjan1984simple,sheng2012review}. For the case of trees, an interesting open problem is whether we can decide whether a profile of weak orders is single-peaked on a tree with additional structural properties, in the spirit of \citet{peters2016nicetrees}. 
\begin{open}
	For a profile of weak orders that is single-peaked on a tree, can we efficiently find such a tree with the minimum or maximum number of leaves? With minimum max-degree?
\end{open}
For circles, the algorithms based on total unimodularity for multi-winner voting rules discussed by \citet{peters2017spoc} continue to work for weak orders.

\subsection{Dichotomous Preferences}
\label{sec:dichotomous}
An important special case of weak preferences are
dichotomous preferences, which distinguish between approved and disapproved alternatives.
Consequently, a vote based on dichotomous preferences corresponds to a subset of (approved) alternatives.
The concepts of single-peaked and single-crossing preferences have been adapted  
to dichotomous preferences 
\citep{dietrich2010majority,list2003possibility,faliszewski2011shield,elkind2015structure},
and \citet{yang2019tree} has considered dichotomous preferences
that are single-peaked or single-crossing on trees.
The most prominent definitions are the candidate interval and voter interval domains:

\begin{definition}
Consider a profile $P=(v_1,\dots,v_n)$, where $v_i\subseteq A$ for $i\in[n]$.
We say that $P$ satisfies \defemph{Candidate Interval (CI)} if alternatives
can be ordered so that each of the sets $v_i$
forms an interval of that ordering.
We say that $P$ satisfies \defemph{Voter Interval (VI)} if the voters in $P$
can be reordered so that for every alternative $c$
the voters that approve $c$ form an interval of that ordering.
\end{definition}

Both CI and VI preferences can be recognized in polynomial time via a reduction to the 
\textsc{Consecutive Ones} problem \citep{faliszewski2011shield,elkind2015structure}, 
see also \Cref{sec:recog:sp}.
\citet{terzopoulou2020restricted} provide forbidden subprofile characterizations for these 
and other domains, and consider the computational complexity of deciding whether a partially 
specified profile can be completed so as to fall into one of these domains.
\citet{constantinescu2023recovering} give an algorithm for recognizing profiles that are possibly single-crossing.

These domains also prove to be useful for algorithmic purposes.
Proportional Approval Voting (PAV), proposed by \citet{Thie95a}, is a multi-winner voting rule 
defined as follows: 
Given a desired committee size $k$, PAV outputs all committees $W$ of size $k$ that maximize 
$\sum_{i\in N} h(|W\cap v_i|)$, where $h$ is the harmonic series 
$h(j)=1 + \frac 1 2 + \dots + \frac 1 j$.
PAV is known to satisfy desirable proportionality axioms~\citep{AzizBCEFW17,LacknerS17}, 
but is NP-hard to compute \citep{aziz2015computational,skowron2016finding}.
However, for CI preferences a winning committee under PAV can be computed in polynomial time. 
This was shown by \citet{peters2018single} using an approach based on totally unimodular matrices 
(see also \citealt{peters2017spoc} for a generalization to circular preferences, and \citealt{sornat2022near} for a generalization to nearly CI preferences).
Interestingly, this approach does not extend to VI preferences.

\begin{open}
What is the computational complexity of computing a PAV winner given VI preferences?
\end{open}

\citet{elkind2015structure} do, however, show that VI preferences are useful from a parameterized
complexity perspective: they identify two natural parameters for which computing PAV winners is 
para-NP-hard when not restricting the preference domain, but obtain XP and FPT algorithms with respect to these parameters for VI preferences.
\citet{liu2016param} obtain easiness results for another prominent approval-based 
committee selection rule both for CI and for VI preferences.
For a discussion of other structured domains of dichotomous preferences, we refer the reader 
to the overview by \citet{elkind2015structure} and, for analogs to multidimensional 
single-peakedness, to the work of \citet{peters2016recognising}.
\citet{lackner2023abc} review further algorithmic results for these domains.

\citet{pierczynski2022core} show that for CI and for VI preferences, there always exists a committee that is in the \emph{core}, which is a stability notion guaranteeing proportional representation. Whether such committees exist without any domain restrictions is a famous open problem \citep{lackner2023abc}.

\subsection{Elicitation}
Preference elicitation is concerned with the acquisition of preference data from (typically 
human) agents, a topic clearly relevant to computational social choice.
An overview of this research agenda can be found in a handbook chapter by \citet{brandt2015handbook-chapter10}.
The general goal of preference elicitation is to design efficient elicitation protocols and an often-used measure for the efficiency of a protocol is its communication complexity.

It is natural to expect that structured preferences are easier to elicit.
Indeed, it takes 
$\Theta(\log m!) =\Theta(m\log m)$ comparison queries to elicit an arbitrary ranking of $m$ alternatives.
In contrast, if we know that the ranking being elicited is single-peaked with respect to
a given axis $\lhd$, we can elicit it using $m-1$ queries in a bottom-up fashion:
we ask the voter to compare the two endpoints of the axis, place the less preferred of the two
alternatives in the last position in the voter's ranking, and continue recursively.

Positive communication complexity results are not limited to single-peaked preferences: 
if preferences are single-peaked (on trees), single-crossing (on trees), 
or $d$-dimensional Euclidean, communication complexity of preference elicitation is considerably lower 
than in the general case~\citep{conitzer2009eliciting,jamieson2011active,dey-misra-sc,dey-misra-spt,dey2016recognizing}; this line of research has been surveyed by \citet{ELP-trends}.

\subsection{Counting and Probability}
Domain restrictions impose constraints on the mathematical structure of preference profiles. 
From a combinatorial viewpoint, it is thus natural to ask how many preference profiles (of a 
given size) belong to a given restricted domain.
Assuming that all preference profiles are equally likely, such a result yields the probability 
that a profile belongs to this restricted domain.
\citet{lackner2017likelihood} study this question for the single-peaked domain and 
obtain an asymptotically tight result counting the number of single-peaked profiles
\citep[see also][pp. 581--585]{durand2003finding}.
\citet{chen2017number} obtain an exact enumeration result for single-peaked profiles where 
the sets of voters and candidates coincide and voters rank themselves first, i.e., 
voters are narcissistic (cf. \Cref{sec:def:spsc}).
Related results are obtained by \citet{brown2014single} for single-peaked preferences over 
a multi-attribute domain.
\citet{lackner2017likelihood} further study the likelihood that a profile is single-peaked 
if drawn according to the P\'{o}lya urn or the Mallows model, and \citet{karpov2020likelihood} 
considers the likelihood for several natural distributions related to the impartial culture.
\citet{karpov2019number} analyses the number of group-separable profiles.

\subsection{Matching and Assignment}

\paragraph{Stable Roommates}
In the \emph{stable roommates} problem, a set of $2n$ agents need to be matched into pairs. Each agent has a preference over potential matching partners. A matching is \emph{stable} if there is no pair of agents who both strictly prefer this pair over their pair in the matching. In contrast to the stable marriage setting (where agents come in two types), stable matchings need not exist for the stable roommates problem \citep{GaSh62a}. If there are no ties in the preferences, one can decide in $O(n^2)$ time whether one exists and if so find one \citep{GuIr89a}. In the presence of ties, the problem becomes NP-complete \citep{Ronn90a}.

\citet{bar-tri:j:sp} studied the stable roommates problem with single-peaked preferences. They assume that there is a common ordering $\lhd$ of the set of agents, that each agent's preferences over matching partners is single-peaked with respect to $\lhd$, and that each agent $i$'s peak is located at $i$'s position in the axis (``narcissistic'' single-peaked preferences).
They show that in this case, there always exists a stable matching and that it is unique. They also give a simple algorithm for finding this matching. \citet{bredereck2020stable} show that this algorithm runs in time $O(n^2)$. (\citet{bar-tri:j:sp} had claimed an $O(n)$ runtime.)

For narcissistic single-crossing preferences, \citet{bredereck2020stable} show that the same result holds: there exists a unique stable matching, and it can be found in $O(n^2)$ time. They also extend the results for single-peaked and single-crossing preferences to the case with ties allowed, for which the same results hold except that uniqueness is not guaranteed. For incomplete preferences with ties, they show NP-completeness results.

\paragraph{Assignment}
In the assignment problem, each of $n$ agents needs to receive exactly one of $n$ objects. Each agent has preferences over these objects. In \emph{random assignment}, we look for methods that, given the preferences, select a probability distribution of assignments. A famous impossibility theorem of \citet{BoMo01a} states that there is no randomized method that satisfies efficiency, strategyproofness, and equal treatment of equals. The former two axioms are defined via stochastic dominance, and the third is a weakening of anonymity. \citet{Kasa12a} proved that this impossibility holds even for single-peaked preferences over objects. \citet{ChCh16a} show that it holds even for single-peaked preference profiles in which all agents have the same peak.

A variant of the (deterministic) assignment problem, called the \emph{house assignment problem} \citep{ShSc74a}, starts with an initial assignment of houses to agents (endowments). These can then be traded among the agents. For unrestricted preferences, the Top Trading Cycle algorithm \citep{ShSc74a} is the only one that satisfies Pareto optimality, individual rationality, and strategyproofness. For single-peaked preferences, however, other rules satisfying these properties are known \citep{bade2019matching,beynier2020optimal}, and swap dynamics have been studied for this restricted domain \citep{beynier2021swap,BrWi19a}.

\subsection{Other Domain Restrictions} 
There are a number of domain restrictions that have 
received little or no attention in the computational social choice literature. Among those are Level-1 Consensus preferences~\citep{nitzan2017level} and intermediate preferences on a median 
graph~\citep{grandmont1978intermediate,demange2012majority}; see also the book by 
\citet{gaertner2001domain} and the surveys by \citet{puppe2024maximal} and \citet{karpov2022structured} for additional domains. A particularly interesting domain restriction is top 
monotonicity~\citep{barbera2011top}, which generalizes both single-peaked and single-crossing 
preferences (see also \citealp{pierczynski2022core}). \citet{magiera2019recognizing} show, using a reduction to 2-SAT, that the 
recognition problem for top monotonicity is solvable in polynomial time; previously \citet{aziz2014testing} had shown 
that the problem is NP-hard for partial orders. It remains an open question whether these 
domains are algorithmically useful, for example for winner determination problems.

\newpage
\section{Research Directions}
\label{sec:conclusion}

Throughout the survey, we have mentioned open problems, which pose gaps in our understanding of preference restrictions and their computational benefits.
Now, we would like to highlight a more general research directions that we consider promising.

\paragraph{Multidimensional Restrictions}
Most of this survey has focused on domain restrictions that are in some sense one-dimensional:
single-peaked, single-crossing and 1-Euclidean preferences are all defined by a linear order or 
an embedding into the real line.
Multidimensional analogs of these notions have received much less attention in the computational social choice literature.
In particular, little is known about computational benefits of such higher-dimensional restrictions.
For example, it is not known whether the Kemeny rule is computable in polynomial time 
on two-dimensional single-peaked profiles (cf. \Cref{sec:def:multidim-sp}) or on 2-Euclidean profiles (cf. \Cref{sec:def:euclid}).
Even if NP-hard voting problems remain hard for these domains, it might be that better approximation algorithms can be found than for general preferences.
Multidimensional domain restrictions offer many challenging research questions, but faster algorithms for these classes are very desirable: these algorithms would be applicable to a much larger class of preferences than algorithms for one-dimensional restrictions.

\paragraph{New Preference Restrictions}
Another direction is to consider completely new restrictions.
Domains suggested in the social choice literature usually guarantee the existence of a Condorcet winner, 
but this is not a necessarily relevant property for algorithmic purposes.
Inspiration could be found by adapting structural concepts from graph theory, 
such as restrictions resembling treewidth.
For a systematic study of domain restrictions, the framework of forbidden subprofiles could prove to be valuable, as well as related mathematical theories such as 0-1-matrix containment \citep{furedi1992davenport,klinz1995permuting} or permutation patterns \citep{kitaev2011patterns,vatter:permutation-cla:} (the connection between forbidden subprofiles and permutation patterns has been established in \citet{lackner2017likelihood}).
Preference profiles, seen as tuples of linear orders, are mathematically rich structures and there is hope for a similarly 
diverse and powerful classification of structure as exists for graph classes---along with algorithmic applications of 
these structural restrictions.

\paragraph{Experiments and Real-World Data}
A pressing issue is the compatibility of theoretical structural restrictions (such as those discussed in this survey) and structural properties of real-world data sets. Little work has been done on tackling this practical challenge. It can be expected that real-world data such as the data sets found at \texttt{preflib.org} \citep{mattei2013preflib} rarely belong to a strict mathematical restriction such single-peaked or single-crossing. However, the introduction of distances to these restrictions (as discussed in \Cref{sec:recog:almost}) may yield practically useful notions of structure. 
First analyses in this respect have been performed by \citet{sui2013multi} and \citet{przedmojski2016algorithms}. 
As can be expected, the conclusions drawn from experimental work depend on the chosen data set:
\citet{sui2013multi} study Irish election data and conclude that their data set are almost two-dimensional single-peaked preferences.
\citet{przedmojski2016algorithms} finds some data sets that are close to single-peaked preferences, although these typically have few alternatives.

\citet{boehmer2022collecting} report results from a larger-scale analysis of real-world preference data. In their collection of 7582 profiles, only very few belong to a restricted domain (1.3\% are single-peaked, 2.3\% are single-crossing, and 1.6\% are group-separable). They report that only some profiles are close to these domains under the voter or alternative deletion nearness measures.
In addition, profiles that are close to one domain are typically also close to another. They report that these profiles are typically quite degenerate (votes are very similar to each other; single-peaked
profiles seems to be single-peaked with respect to many different axis
and voters have only few different top-choices; in single-crossing profiles only few candidates ever change their pairwise ordering). 6.3\% of their profiles are value-restricted.

\citet{SzufaFSST20} and \citet{BoehmerBFNS21} propose a new approach to compare \emph{statistical cultures} (probability distributions over profiles), including statistical cultures that generate preferences within restricted domains (e.g., single-peaked profiles).
They generate ``maps'' that visualize the similarity of statistical cultures and domain restrictions, and thereby offer new insights into the structure of preferences.
In particular, this approach offers another way to evaluate the similarity between real-world data sets and restricted preference domains.
In conclusion, it is of high importance to continue to analyze structural properties of available data sets and statistical cultures, and thus enrich the theoretical research done on this topic.

\paragraph{Beyond Voting}
Finally, the work on structured preferences has mostly focused on voting-related topics: winner determination, 
manipulation, control, etc.
Given the advances that have been made in these fields, it could prove to be worthwhile to investigate the impact of 
structured preferences in other fields of social choice; matching, fair division and judgment aggregation are natural 
candidates.

\section*{Acknowledgments}
This work was supported by the European Research Council (ERC) under grant number 639945
(ACCORD). 
Martin Lackner was additionally supported by the Austrian Science Fund FWF, grant P31890.
We thank Niclas Boehmer, Jiehua Chen, Zack Fitzsimmons, Clemens Puppe, Olivier Spanjaard, and Yongjie Yang for feedback.

\newpage
\bibliographystyle{ACM-Reference-Format}

\end{document}